\documentclass[10pt,reqno]{amsart}

\usepackage[margin=1in]{geometry}
\usepackage{amsmath}
\usepackage{amsthm, amsfonts,amssymb,euscript}
\usepackage{mathtools}
\usepackage{enumitem}
\usepackage{bm}
\usepackage{amssymb}
\usepackage{graphicx}
\usepackage{caption}
\usepackage{subcaption}
\usepackage{amsfonts}
\usepackage{float}
\usepackage{color}
\usepackage{slashed}
\usepackage{mathrsfs}
\usepackage{upgreek}
\usepackage{pifont}
\usepackage{bbm}
\usepackage[colorlinks=true]{hyperref}

\usepackage[usenames,dvipsnames,svgnames,table]{xcolor}
\hypersetup{linkcolor=BrickRed, urlcolor=green, citecolor=blue, linktoc=page}

\numberwithin{equation}{section}

\newtheorem*{proposition*}{Proposition}
\newtheorem*{theorem*}{Theorem}
\newtheorem*{conjecture*}{Conjecture}
\newtheorem*{claim*}{Claim}
\newtheorem*{lemma*}{Lemma}
\newtheorem*{corollary*}{Corollary}
\newtheorem{theorem}{Theorem}[section]
\newtheorem{proposition}[theorem]{Proposition}

\newtheorem{lemma}[theorem]{Lemma}
\newtheorem{corollary}[theorem]{Corollary}

\newtheorem*{definition*}{Definition}

\newtheorem*{assumption*}{Assumption}

\newtheorem*{remark*}{Remark}
\newtheorem{remark}{Remark}[section]
\newtheorem*{condition*}{Condition}
\newtheorem{condition}{Condition}

\usepackage{xargs}   
\usepackage[colorinlistoftodos,prependcaption,textsize=tiny]{todonotes}
\newcommandx{\unsure}[2][1=]{\todo[linecolor=red,backgroundcolor=red!25,bordercolor=red,#1]{#2}}
\newcommandx{\change}[2][1=]{\todo[linecolor=blue,backgroundcolor=blue!25,bordercolor=blue,#1]{#2}}
\newcommandx{\info}[2][1=]{\todo[linecolor=OliveGreen,backgroundcolor=OliveGreen!25,bordercolor=OliveGreen,#1]{#2}}
\newcommandx{\improvement}[2][1=]{\todo[linecolor=Plum,backgroundcolor=Plum!25,bordercolor=Plum,#1]{#2}}
\newcommandx{\thiswillnotshow}[2][1=]{\todo[disable,#1]{#2}}

\newcommand{\la}{\langle}
\newcommand{\ra}{\rangle}
\newcommand{\R}{\mathbb{R}}
\newcommand{\s}{\mathbb{S}}
\newcommand{\C}{\mathbb{C}}

\newcommand{\Z}{\mathbb{Z}}
\newcommand{\N}{\mathbb{N}}

\newcommand{\snabla}{\slashed{\nabla}}

\newcommand{\Sc}{\textrm{\ding{34}}}

\newcommand{\Lbar}{\underline{L}}

\newcommand{\hpsi}{\widehat{\psi}}

\newcommand{\h}{\mathbbm{h}}
\newcommand{\re}{\textnormal{Re}\,}
\newcommand{\im}{\textnormal{Im}\,}

\newcommand{\sign}{\,\textnormal{sign}\,}
\newcommand{\q}{q}

\DeclareFontFamily{U}{mathx}{}
\DeclareFontShape{U}{mathx}{m}{n}{<-> mathx10}{}
\DeclareSymbolFont{mathx}{U}{mathx}{m}{n}
\DeclareMathAccent{\widecheck}{0}{mathx}{"71}

\allowdisplaybreaks

\usepackage[foot]{amsaddr}

\begin{document}

\author{Dejan Gajic$^*$$^1$}
\address{$^{1}$\small Institut f\"ur Theoretische Physik, Universit\"at Leipzig, Br\"uderstrasse 16, 04103 Leipzig, Germany }
\email{$^*$dejan.gajic@uni-leipzig.de}
\title[Charged scalar fields on Reissner--Nordstr\"om II: late-time tails and instabilities]{Charged scalar fields on Reissner--Nordstr\"om spacetimes II: late-time tails and instabilities}
\maketitle

\begin{abstract}This is the second part of a series of papers deriving the precise, late-time behaviour and (in)stability properties of charged scalar fields on near-extremal Reissner--Nordstr\"om spacetimes via energy estimates.  In this paper, we use purely physical-space based methods to establish the precise late-time behaviour of solutions to the charged scalar field equation in the form of oscillating and decaying late-time tails that satisfy inverse-power laws, assuming global integrated energy decay estimates, which are proved in the companion paper \cite{gaj26a}.

This paper provides the first pointwise decay estimates for charged scalar fields on black hole backgrounds without an assumption of smallness of the scalar field charge.

We also prove the existence of asymptotic instabilities for the radiation field along future null infinity and, in the extremal case, also along the future event horizon. Both the energy methods and the precise late-time asymptotics derived in this paper are expected to play an important role in future nonlinear studies of black hole dynamics in the context of the spherically symmetric (Einstein--)Maxwell--charged scalar field equations, as well as in the context of extremal Kerr spacetimes.
\end{abstract}

\tableofcontents
\section{Introduction}
Late-time dynamics around sub-extremal Kerr black holes solutions \cite{kerr63} to the Einstein vacuum equations
\begin{equation}
\label{eq:ee}
{\rm Ric}[g]=0	
\end{equation}
are intimately related to the following three geometric properties of Kerr spacetimes $g=g_{M,a}$: the \emph{red-shift} along null generators of the event horizon, the existence of \emph{trapped null geodesics} and the presence of an \emph{ergoregion}. 

The latter two manifest themselves in the study of wave equations on Kerr as energy concentration around trapped null geodesics for finite times \cite{janpaper} and energy amplification in certain time-frequency regimes  known as \emph{rotational superradiance} \cite{zel71,zel72}. In contrast, red-shift emerges as a favourable energy decay mechanism near the event horizon \cite{redshift}.

The study of late-time dynamics of waves on extremal Kerr black holes, members of the Kerr family that attain a maximal angular momentum $(|a|=M)$, becomes considerably more difficult, due to a \emph{strong coupling} of superradiance with the \emph{degeneration} of red-shift and the existence of trapped null geodesics arbitrarily close to the event horizon. The absence of red-shift and the coupling of geometric obstructions to decay may be thought of as the source of pointwise growth or asymptotic instabilities along the event horizon in the context of linear wave equations \cite{aretakis4,zimmerman1,zimmerman4,gaj23}.

A robust control of the strong coupling between the lack of red-shift, the trapping of null geodesics and superradiance for \emph{linear} wave equations on fixed (near-)extremal Kerr backgrounds is fundamental for an understanding of the ultimate fates of dynamical spacetime solutions to \eqref{eq:ee} arising from small initial data perturbations of extremal Kerr data.

A shortcut towards shedding light on the nonlinear effects of the above strong coupling of difficulties in black hole dynamics, bypassing a full resolution in the extremal Kerr setting, is the consideration of small perturbations of (near-)extremal Reissner--Nordstr\"om initial data in the context of the spherically symmetric Einstein--Maxwell-(massless) charged scalar field (EMCSF) system of equations:
\begin{align}
	\label{eq:eemax}
{\rm Ric}[g]-\frac{1}{2}R[g]g=&\: 8\pi \left(\mathbb{T}^{EM}[g,F]+\mathbb{T}^{SF}[g,\phi]\right),\\
	\label{eq:max1}
dF=&\:0,\\
	\label{eq:max2}
\nabla_{\nu}F^{\nu \mu}=&\:4\pi\mathfrak{q}\im(\phi \overline{\:(^AD)^{\mu}\phi}) ,\\
	\label{eq:CSFnonlin}
(g^{-1})^{\mu \nu}\:(^AD)_{\mu}\:(^AD)_{\nu}(\phi)=&\:0,
\end{align}
where $F$ is a Faraday tensor, describing electromagnetic fields, $\phi$ is a scalar field with associated charge parameter (or coupling constant) $\mathfrak{q}$, $\mathbb{T}^{EM}[g,F]$ and $\mathbb{T}^{SF}[g,\phi]$ are the electromagnetic and scalar field stress-energy tensors, respectively, and $\:^AD=\nabla-i\mathfrak{q}A\otimes(\cdot)$ is the gauge derivative, with the 1-form $A$ denoting an electromagnetic potential satisfying $dA=F$.

In the present paper, we perform a precise, quantitative analysis of the \emph{linearization} of \eqref{eq:eemax}--\eqref{eq:CSFnonlin} around electrically charged (sub-)extremal Reissner--Nordstr\"om solutions:
\begin{equation*}
 \left(g_{M,Q}=-\Omega^2dt^2+\Omega^{-2}dr^2+r^2\slashed{g}_{\s^2},\quad F_{Q}=-Qr^{-2}dt\wedge dr,\quad \phi\equiv 0\right),
 \end{equation*}
 with $\Omega^2=1-2Mr^{-1}+Q^2r^{-2}$, which reduces to a single equation in spherical symmetry:
\begin{align}
\label{eq:CSFintro}
	(g^{-1}_{M,Q})^{\mu \nu}\:(^AD)_{\mu}\:(^AD)_{\nu}\phi=0,
\end{align}
where $dA=F_Q$. Extremality corresponds to $|Q|=M$ in this setting. The equation \eqref{eq:CSFintro} is called the \emph{charged scalar field equation} on a Reissner--Nordstr\"om background. Together with the companion paper \cite{gaj26a}, the present paper forms the first mathematically rigorous, global analysis of \eqref{eq:CSFintro} in both the sub-extremal and extremal cases without any restrictions on the size of the charge parameter $\mathfrak{q}$. This may be thought of as, in particular, a first step towards a global analysis of \eqref{eq:eemax}--\eqref{eq:CSFnonlin} for small perturbations of (near-)extremal Reissner--Nordstr\"om data. However, in contrast with the motivating problem \eqref{eq:eemax}--\eqref{eq:CSFnonlin}, we will \underline{not} assume spherical symmetry of $\phi$ in our analysis of \eqref{eq:CSFintro}.

A global understanding of \eqref{eq:CSFintro} relies on methods that can treat the degeneration of red-shift, the presence of trapped null geodesics as well as \emph{charged superradiance} \cite{bek73, gib75}, which may be thought of as a close analogue of rotational superradiance.

In the first part of our series \cite{gaj26a}, we establish global, \emph{integrated energy estimates} via a combination of physical-space and Fourier-space based analysis in the setting of (near-)extremal Reissner--Nordstr\"om, which may be thought of as weak decay statements that serve as the foundation for the stronger, more precise decay statements derived in the present paper. The main results in \cite{gaj26a} immediately extend to the full sub-extremal range of parameters under the additional assumption of mode stability on the real axis.

In the present paper, we obtain a complete understanding of precise, leading-order late-time behaviour and instability phenomena, starting from integrated energy estimates and using purely physical-space-based methods.

The main results of the present paper can be split into the following two theorems:

\begin{theorem}[Late-time asymptotics; informal version of Theorem \ref{thm:mainthmpoint}]
\label{thm:introtails}
	Let $r^{-1}\psi$ be a solution to \eqref{eq:CSFintro} with $\mathfrak{q}\neq 0$ with respect to the conformally smooth electromagnetic gauge $A=-\frac{Q}{r}d\tau$, with $\tau$ a time function whose level sets are horizon-intersecting and asymptotically hyperboloidal spacelike hypersurfaces and with $\partial_{\tau}$ a Killing vector field normalized at infinity that is timelike away from the event horizon.  Assume that $\psi$ arises from the evolution of smooth, compactly supported initial data. Then, for $1-\frac{|Q|}{M}\ll 1$ or $|\mathfrak{q}Q|\ll 1$:
	\begin{align*}
		\psi(\tau,r,\theta,\varphi)\sim \sum_{\substack{\ell\in \N_0\\\ell(\ell+1)<\mathfrak{q}^2Q^2}}\sum_{\substack{m\in \Z\\ |m|\leq \ell}}\left[\Psi_{\ell m}^{\infty}(\tau,r;\mathfrak{q}, \mathfrak{I}_{\ell m}[\psi])+\Psi^+_{\ell m}(\tau,r;\mathfrak{q},\mathfrak{H}_{\ell m}[\psi])\right]Y_{\ell m}(\theta,\varphi)\quad (\tau\to \infty),
	\end{align*}
	where $\Psi^{+}_{\ell m}(\tau,r;\mathfrak{q},\mathfrak{H}_{\ell m}[\psi])=e^{-i\mathfrak{q}Qr_+^{-1}\tau}\Psi_{\ell m}^{\infty}\left(\tau,r_++\frac{r_+^2}{r-r_+};-\mathfrak{q},\mathfrak{H}_{\ell m}[\psi]\right)$, with $r_+$ the event horizon area radius, and where $\Psi_{\ell m}^{\infty}$ are globally-defined functions satisfying the following late-time asymptotics with $r_0\geq r_+$:
	\begin{align}
	\label{eq:introasympPsiinfty1}
	\Psi^{\infty}_{\ell m}(\tau,\infty;\mathfrak{q},\mathfrak{I}_{\ell m}[\psi])\sim &\: \mathfrak{I}_{\ell m}[\psi]\tau^{-\frac{1}{2}-\frac{1}{2}\beta_{\ell}+i\mathfrak{q}Q},\quad & (\beta_{\ell}\in (0,1)),\\
		\label{eq:introasympPsiinfty2}
		\Psi^{\infty}_{\ell m}(\tau,\infty;\mathfrak{q},\mathfrak{I}_{\ell m}[\psi])\sim &\: \mathfrak{I}_{\ell m}[\psi]\frac{\tau^{-\frac{1}{2}+i\mathfrak{q}Q}}{\log \tau},\quad & (\beta_{\ell}=0),\\
			\label{eq:introasympPsiinfty3}
		\Psi^{\infty}_{\ell m}(\tau,\infty;\mathfrak{q},\mathfrak{I}_{\ell m}[\psi])\sim &\: \mathfrak{I}_{\ell m}[\psi]\tau^{-\frac{1}{2}-\sigma\frac{\beta_{\ell}}{2}+i\mathfrak{q}Q}\sum_{n=0}^{\infty}\frac{\zeta^{\sigma n}\tau^{-\sigma n\beta_{\ell}}}{\Gamma\left(\frac{1}{2}+i\mathfrak{q}Q-(n+\frac{1}{2})\sigma\beta_{\ell}\right)},\quad & (\beta_{\ell}\in i(0,\infty)),\\
			\label{eq:introasympPsiinfty4}
		\Psi^{\infty}_{\ell m}(\tau,r_0;\mathfrak{q},\mathfrak{I}_{\ell m}[\psi])\sim &\:\mathfrak{I}_{\ell m}[\psi]\mathfrak{w}_{\ell}(r_0;\mathfrak{q}Q)\tau^{-1-\beta_{\ell}}, \quad & (\beta_{\ell}\in (0,1)),\\
			\label{eq:introasympPsiinfty5}
		\Psi^{\infty}_{\ell m}(\tau,r_0;\mathfrak{q},\mathfrak{I}_{\ell m}[\psi])\sim &\:\mathfrak{I}_{\ell m}[\psi]\mathfrak{w}_{\ell}(r_0;\mathfrak{q}Q)\frac{\tau^{-1}}{\log^2\tau},\quad & (\beta_{\ell}=0),\\
		\label{eq:introasympPsiinfty6}
			\Psi^{\infty}_{\ell m}(\tau,r_0;\mathfrak{q},\mathfrak{I}_{\ell m}[\psi])\sim &\:\mathfrak{I}_{\ell m}[\psi]\mathfrak{w}_{\ell}(r_0;\mathfrak{q}Q)\tau^{-1-\sigma\beta_{\ell}}\\ \nonumber
			\times &\sum_{n=0}^{\infty}\zeta^{\sigma n}\frac{\tau^{-\sigma n\beta_{\ell}}}{\Gamma(\sigma \beta_{\ell}+1) \Gamma(\frac{1}{2}-\frac{1}{2}\sigma \beta_{\ell}+i\mathfrak{q}Q)\Gamma(-(n+1)\sigma\beta_{\ell})}, & (\beta_{\ell}\in i(0,\infty)).
	\end{align}
	Here,
	\begin{itemize}
	\item $\beta_{\ell}:=\sqrt{(2\ell+1)^2-4\mathfrak{q}^2Q^2}\in [0,1)\cup i(0,\infty)$,
	\item $\sigma=\sign(\mathfrak{q}Q)$, $\zeta=\zeta(\mathfrak{q}Q,\beta_{\ell})$ is a complex constant (defined in Theorem \ref{thm:mainthmpoint}) and $\Gamma(\cdot)$ is the gamma function,
	\item $\mathfrak{I}_{\ell m}[\psi]$ and $\mathfrak{h}_{\ell m}[\psi]$ are linear functionals of (derivatives of) initial data at $\tau=0$, so that each $\mathfrak{I}_{\ell m}[\psi]$ vanishes only for a codimension-$1$ set of initial data, $\mathfrak{H}_{\ell m}[\psi]\equiv 0$ when $|Q|<M$ and in the case $|Q|=M$, each $\mathfrak{H}_{\ell m}[\psi]$ also vanishes only for a codimension-$1$ set of initial data,
	\item $\mathfrak{w}_{\ell}Y_{\ell m}$ are time-independent solutions to \eqref{eq:CSFintro} that are smooth at $r=r_+$ and satisfy $\mathfrak{w}_{\ell}(r_+;\mathfrak{q}Q)=1$. 
		\end{itemize}
		In the case that $1-\frac{|Q|}{M}\gtrsim 1$ and $|\mathfrak{q}Q|\gtrsim 1$, the above results remain valid under the condition of the validity of mode stability on the real axis.
\end{theorem}

\begin{theorem}[Instabilities and energy concentration; informal version of Theorem \ref{thm:mainpointwinst}]
\label{thm:introinst}
	Under the assumptions of Theorem \ref{thm:introtails}, the following asymptotic instability and non-decay results hold: 
	\begin{enumerate}[label=\emph{(\roman*)}]
	\item For all $k\in \N_0$, there exist constants $\mathfrak{c}_k>0$, such that for all $\ell\in \N_0$ with $\ell(\ell+1)<\mathfrak{q}^2Q^2$ and $m\in \Z$ with $|m|\leq \ell$:
\begin{align}
\label{eq:introlowboundbetanonzeroRNhor}	
		\limsup_{\tau\to \infty}(1+\tau)^{-k-\frac{1}{2}+\frac{\re \beta_{\ell}}{2}}|D_X^{k+1}\psi_{\ell m}|(\tau,r_+)\geq &\:|\mathfrak{c}_k| |\mathfrak{H}_{\ell m}[\psi]| \quad &(|Q|=M,\:\beta_{\ell}\neq 0),\\
		\label{eq:introlowboundzerobetaRNhor}	
		\limsup_{\tau\to \infty}\log(1+\tau)(1+\tau)^{-k-\frac{1}{2}}|D_X^{k+1}\psi_{\ell m}|(\tau,r_+)\geq &\:|\mathfrak{c}_k| |\mathfrak{H}_{\ell m}[\psi]|\quad &(|Q|=M,\:\beta_{\ell}=0),\\
		\label{eq:introlowboundbetanonzeroRNinf}	
		\limsup_{\tau\to \infty}(1+\tau)^{-k-\frac{1}{2}+\frac{\re \beta_{\ell}}{2}}|(r^2D_X)^{k+1}\psi_{\ell m}|(\tau,\infty)\geq &\:|\mathfrak{c}_k| |\mathfrak{I}_{\ell m}[\psi]| \quad &(\beta_{\ell}\neq 0),\\
		\label{eq:introlowboundzerobetaRNinf}	
		\limsup_{\tau\to \infty}\log(1+\tau)(1+\tau)^{-k-\frac{1}{2}}|(r^2D_X)^{k+1}\psi_{\ell m}|(\tau,\infty)\geq &\:|\mathfrak{c}_k| |\mathfrak{I}_{\ell m}[\psi]|\quad &(\beta_{\ell}=0),
		\end{align}
with $X$ a radial vector field tangential to $\Sigma_{\tau}$, the level sets of $\tau$. 
\item When $|\mathfrak{q}Q|>\frac{1}{2}$ ($\beta_0\in i(0,\infty)$), \emph{energy concentration} occurs near $\mathcal{H}^+$ in the case $|Q|=M$ and, with additional $r$-weights, also near $\mathcal{I}^+$ when $|Q|\leq M$, in the following sense: there exists a constant $\mathfrak{c}_{\rm int}>0$, such that for all $r_+\leq r_H<r_I<\infty$:
\begin{align}
\label{eq:introenblowup1}
	\limsup_{\tau\to \infty}\int_{\Sigma_{\tau}\cap\{r\leq r_H\}} \mathbb{T}^{SF}[g_{M,Q},r^{-1}\psi](\mathbf{n}_{\tau},\mathbf{n}_{\tau})\,d\mu_{\Sigma_{\tau}}\geq &\: \mathfrak{c}_{\rm int}\sum_{\substack{\ell\in \N_0\\ \ell+\frac{1}{2}<|\mathfrak{q}Q|}}\sum_{\substack{m\in \Z\\ |m|\leq \ell}}|\mathfrak{H}_{\ell m}[\psi]|^2,\\
	\label{eq:introendecy1}
	\lim_{\tau\to \infty}\int_{\Sigma_{\tau}\cap\{r> r_H\}} \mathbb{T}^{SF}[g_{M,Q},r^{-1}\psi](\mathbf{n}_{\tau},\mathbf{n}_{\tau})\,d\mu_{\Sigma_{\tau}}= &\: 0,\\
	\label{eq:introenblowup2}
\limsup_{\tau\to \infty}\int_{\Sigma_{\tau}\cap \{r\geq r_I\}} r^2|D_{r}\psi|^2\,\sin\theta d\theta d\varphi dr\geq &\: \mathfrak{c}_{\rm int}\sum_{\substack{\ell\in \N_0\\ \ell+\frac{1}{2}<|\mathfrak{q}Q|}}\sum_{\substack{m\in \Z\\ |m|\leq \ell}}|\mathfrak{I}_{\ell m}[\psi]|^2,
\end{align}
with induced volume form $d\mu_{\tau}$ on $\Sigma_{\tau}$ and $\mathbf{n}_{\tau}$ the corresponding future-directed normal. 
\end{enumerate}
In the case that $1-\frac{|Q|}{M}\gtrsim 1$ and $|\mathfrak{q}Q|\gtrsim 1$, the above results remain valid under the condition of the validity of mode stability on the real axis.
\end{theorem}
Precise and more general versions of the above theorems are stated later in the paper in the form of Theorems \ref{thm:mainthmpoint} and \ref{thm:mainpointwinst}. In the next section we provide some initial remarks.

\subsection{Initial remarks}
In this section, we place Theorems \ref{thm:introtails} and \ref{thm:introinst} in a broader context, discuss some new dynamical aspects of charged scalar fields that follow from the theorems, and we mention some remaining open problems. We will also discuss some related previous work. An overview of further, related previous work can be found in \S \ref{intro:prevresults}.

\subsubsection{Neutral scalar fields}
\label{eq:intronetural}
	In the case $\mathfrak{q}=0$, \eqref{eq:CSFintro} reduces to the wave equation on Reissner--Nordstr\"om: 
	\begin{equation}
	\label{eq:introwaveeq}
		\square_{g_{M,Q}}\phi=0.
	\end{equation}
While there is no superradiance present in this setting, the absence of red-shift alone already leads to the existence of an asymptotic instability in the extremal case, known as the \emph{Aretakis instability} \cite{aretakis2}, which can be stated as follows: for $k\geq 1$, there exists a constant $c_k>0$ such that
	\begin{equation}
	\label{eq:aretakisinst}
		||\partial_r^{k+1}\psi(\tau,r_+,\cdot)||_{L^2(\s^2)}\geq c_k H_0[\psi]\tau^{k},
	\end{equation}
	with $H_0[\psi]=\int_{\s^2}\partial_r(r\psi)|_{\Sigma_0\cap \mathcal{H}^+}\,\sin \theta d\theta d\varphi$. Observe that the $\tau$-growth in \eqref{eq:aretakisinst} is \emph{slower} than the $\tau$-growth in \eqref{eq:introlowboundbetanonzeroRNhor}	 and \eqref{eq:introlowboundzerobetaRNhor}. Furthermore, whereas $H_0[\psi]$ vanishes for initial data supported away from the event horizon and the instabilities become weaker (by losing a power of $\tau$ growth \cite{aretakis2012}), we will show that the constants $\mathfrak{h}_{\ell m}[\psi]$ are generically \underline{non-vanishing}, even for data supported away from the horizon.
	
	A precise quantitative analysis of \eqref{eq:introwaveeq} was performed in \cite{paper4}, where an analogue of Theorem \ref{thm:introtails} was obtained. There, it was also shown that the non-degenerate energy decays in $\tau$, in contrast with \eqref{eq:introenblowup1}.
	
	Another important difference with the present paper is that the instability \eqref{thm:mainpointwinst} follows from the existence of conservation laws along the event horizon and does not rely on a quantitative understanding of the asymptotics derived in \cite{paper4}. On the other hand, Theorem \ref{thm:introtails} is crucial for proving Theorem \ref{thm:introinst}. The inability to rely on conservation laws in a derivation of instabilities is a key motivation for deriving precise late-time asymptotics in the present paper.

	In view of the above comparisons, the instabilities in Theorem \ref{thm:introinst} may be thought of as being considerably \emph{stronger} than the Aretakis instability.

	In the neutral scalar field setting, there are related results in the context of the full nonlinear EMCSF system \eqref{eq:eemax}--\eqref{eq:CSFnonlin} with $\mathfrak{q}=0$, in which $Q$ is not a dynamical quantity. See the recent \cite{aku26} for a complete stability statement in the form of an analysis of the moduli space of initial data leading to Reissner--Nordstr\"om spacetimes with different masses, providing, in particular, a mathematical confirmation of earlier numerical work in \cite{harvey2013}. We also refer to the earlier codimension-1 stability result of extremal Reissner--Nordstr\"om in \cite{aku24} and the related nonlinear analysis in the neutral scalar field setting (outside of spherical symmetry) in \cite{yannis1,aagnonlin,au25}. 

\subsubsection{Comparison with the extremal Kerr setting}
	The instabilities along the event horizon $\mathcal{H}^+$ ($r=r_+$) in the $|Q|=M$ case in Theorem \ref{thm:introinst} bear a close resemblance to the \emph{azimuthal instabilities} along the extremal Kerr ($|a|=M$) event horizon in \cite{gaj23}[Theorem 4.2], if we restrict to $\beta_{\ell}\in i(0,\infty)$ and replace $\beta_{\ell}$ with:
	\begin{equation*}
		\sqrt{1+4\Lambda_{m\ell}(a\upomega_+)-8m^2},
	\end{equation*}
	where $\Lambda_{m\ell}(a\upomega_+)$ are Boyer--Lindquist oblate spheroidal eigenvalues at the superradiant treshold frequency $m\upomega_+$, labelled by the azimuthal number $m\in \Z\setminus\{0\}$ and $\ell\in \N_0$, with $\ell\geq |m|$. These instabilities were first discovered via a heuristic Fourier-based analysis around $\omega=m\upomega_+$ in \cite{zimmerman1} and then \cite{zimmerman4}.
	
	One may think of the aziumuthal number $m$ as playing the role of the dimensionless charge parameter $\mathfrak{q}Q$, see also \cite{gaj26a}[\S 1.3] for a closer comparison in Fourier space.
		
	In contrast with Theorem \ref{thm:introinst}, \cite{gaj23}[Theorem 4.2] relies on a global integrability assumption. Failure of this integrability assumption would correspond to the (stronger) instability scenarios (b) and (c) in \cite{gaj23}[Corollary 4.3]. Furthermore, the present paper also treats the cases $\beta_{\ell}\in [0,1)$. 
	
	The improvements in the present paper when compared to \cite{gaj23} are closely related to the application of uniform estimates that apply in both the near-extremal and extremal settings and are derived in \cite{gaj26a}. The new methods introduced in the present paper remain applicable in the setting in \cite{gaj23} and will be applied to improve the results in \cite{gaj23} in future work.
	
		Similarly, the late-time tails in Theorem \ref{thm:introtails} in the $|Q|=M$ case, away from future null infinity $\mathcal{I}^+$ are analogues of the late-time tails derived in \cite{gaj23}[Theorem 4.1]. However, while \cite{gaj23}[Theorem 4.1] is conditional on a priori qualitative global integrability assumptions, Theorem \ref{thm:introtails} is entirely unconditional when $0\leq 1-\frac{|Q|}{M}\ll 1$.
		
		The late-time tails along $\mathcal{I}^+$, on the other hand, are particular to the setting of CSF and differ from those derived in \cite{gaj23}[Theorem 4.1].

\subsubsection{A conformal isometry}
\label{sec:confisom}
	The symmetry between the behaviour along $\mathcal{H}^+$ and $\mathcal{I}^+$ in the $|Q|=M$ case is related to the existence of the Couch--Torrence conformal isometry \cite{couch} $\Phi_{\rm C-T}$. Initial data for $\psi$ imposed on a hypersurface that is invariant under $\Phi_{\rm C-T}$, such that $\psi\circ \Phi_{\rm C-T}=\psi$ will lead to identical behaviour along $\mathcal{I}^+$ and $\mathcal{H}^+$ up to additional change in sign of $\mathfrak{q}$; see Proposition \ref{prop:mainequations}. 
	
	One may interpret the asymptotic instabilities of $r^2$-weighted derivative quantities in \eqref{eq:introlowboundbetanonzeroRNinf}	 and \eqref{eq:introlowboundzerobetaRNinf} as close analogues of horizon instabilities. In fact, the asymptotic instabilities along $\mathcal{I}^+$ are also present in the sub-extremal setting. This is not surprising, since $\mathcal{I}^+$ can be viewed as a Killing horizon with zero surface gravity after performing an appropriate conformal rescaling of the metric.

\subsubsection{Wagging tails and null infinity signatures of extremality}
\label{sec:introwagging}
	In contrast with the uncharged scalar field setting, the late-time asymptotics in Theorem \ref{thm:introtails} feature oscillations in $\tau$ at different levels. 
	
	First of all, note that $e^{-i\mathfrak{q} f}\psi$ is a solution to \eqref{eq:CSFintro} with respect to the electromagnetic potential $A+df$ for any smooth function $f$, so an overall phase factor, by itself, has no gauge-invariant meaning and can be absorbed via an appropriate (time-dependent) redefinition of $A$. The presence or absence of the phase factor $e^{-i\mathfrak{q}Q\tau}$ does play an important, gauge-invariant role when we consider the gauge-invariant time derivatives $|D_{\tau}\psi|$, see \S \ref{sec:introminkrole} for a further discussion.
	
	In the $\beta_{\ell}\in i(0,\infty)$ case however, the gauge-invariant quantity $|\psi|$ also features oscillations in its late-time asymptotics, since the asymptotics are governed by an infinite sum of oscillating terms, see \eqref{eq:introasympPsiinfty3} and \eqref{eq:introasympPsiinfty6}.
	
	\subsubsection{The role of Minkowski and Bertotti--Robinson}
\label{sec:introminkrole}
The projection of $\Psi^{\infty}$ to spherical harmonics with angular frequency $\ell$ has the following structure (see Theorem \ref{thm:mainthmpoint}):
\begin{equation*}
	\Psi^{\infty}_{\ell}=\frac{\mathfrak{w}_{\ell}}{\mathfrak{w}^0_{\ell}}(\Psi_0^{\infty})_{\ell},
\end{equation*}
where
\begin{equation}
\label{eq:CSFminkintro}
	(m^{-1})^{\mu \nu}\:(^{\widehat{A}}D)_{\mu}\:(^{\widehat{A}}D)_{\nu}(r^{-1}\Psi^{\infty}_0)=0,
\end{equation}
with $m$ the Minkowski metric and $\widehat{A}=-Qr^{-1}du$ (with respect to Bondi coordinates $(u,r)$) and where:
\begin{equation*}
	\mathfrak{w}_{\ell}^0=\begin{cases}\alpha_+ r^{\frac{1}{2}+i\q-\frac{1}{2}\beta_{\ell}}+\alpha_- r^{\frac{1}{2}+i\q-\frac{1}{2}\beta_{\ell}}\quad &(\beta_{\ell}\neq 0),\\
		\alpha_+ r^{\frac{1}{2}+i\q}\log r+\alpha_- r^{\frac{1}{2}+i\q}\quad &(\beta_{\ell}= 0),
\end{cases}
\end{equation*}
with $\mathfrak{w}_{\ell}^0Y_{\ell m}$ time-independent solutions to \eqref{eq:CSFminkintro} and $\alpha_{\pm}\in \C$ the constants appearing in the large-$r$ asymptotics of $\mathfrak{w}_{\ell}$:
\begin{equation*}
	\mathfrak{w}_{\ell}=\begin{cases}\alpha_+ r^{\frac{1}{2}+i\q-\frac{1}{2}\beta_{\ell}}(1+O(r^{-1})+\alpha_- r^{\frac{1}{2}+i\q-\frac{1}{2}\beta_{\ell}}(1+O(r^{-1})\quad &(\beta_{\ell}\neq 0),\\
		\alpha_+ r^{\frac{1}{2}+i\q}\log r(1+O(r^{-1})+\alpha_- r^{\frac{1}{2}+i\q}(1+O(r^{-1})\quad &(\beta_{\ell}= 0).
\end{cases}
\end{equation*}

Similarly, the function $\Psi^{+}_{0}(\tau,r,\theta,\varphi):=\Psi^{\infty}_{0}(\tau, r_++\frac{r_+^2}{r-r_+},\theta,\varphi)$ is a solution to:
\begin{equation}
\label{eq:CSFBRintro}
(g_{\rm B-R}^{-1})^{\mu \nu}\:(^{A_0}D)_{\mu}\:(^{A_0}D)_{\nu}(r_+^{-1}\Psi^{+}_0)=0,
\end{equation}
with 
\begin{equation*}
	\left(g_{\rm B-R}=r_+^2(-x^2d\tilde{v}^2-2d\tilde{v}dx+\slashed{g}_{\s^2}),\quad  F=-r_+ d\tilde{v}\wedge dx\right)
\end{equation*}
the Bertotti--Robinson solution (or near-horizon extremal Reissner--Nordstr\"om solution) \cite{ber59,rob59} to the Einstein--Maxwell equations and $A_0=-r_+ \tilde{x}dv$.\footnote{Note that the conformally rescaled Minkowski spacetime $(\R^{1+3},r^{-2}m)$ can be embedded in a Bertotti--Robinson spacetime; see \cite{gvdm24}[\S 1.5.1]. }

\textbf{The leading-order late-time asymptotics of solutions to \eqref{eq:CSFintro} are therefore determined by appropriate solutions to the charged scalar field equation on Minkowski and Bertotti--Robinson (up to the factors $\frac{\mathfrak{w}_{\ell}}{\mathfrak{w}^0_{\ell}}$)!}

Note, however, that while the Bertotti--Robinson solutions determine the asymptotic behaviour of $|\psi|$ along $\mathcal{H}^+$, they do \underline{not} determine the asymptotic behaviour of the tangential gauge derivative $|D_{\tau}\psi|$ along $\mathcal{H}^+$:\footnote{When $\beta_{\ell}\in i(0,\infty)$ the asymptotics include an infinite sum of oscillating factors, which we omit here for ease of notation.}
	\begin{equation*}
		|D_{\tau}\psi_{\ell m}|(\tau,r_+)\sim |\mathfrak{q}Q||\mathfrak{I}_{\ell m}[\psi]|\begin{cases}\tau^{-1-\re \beta_{\ell}}\quad & (\beta_{\ell}\in (0,1)\cup i(0,\infty)),\\
		\frac{\tau^{-1}}{\log^2\tau}\quad & (\beta_{\ell}=0).
\end{cases}
	\end{equation*}
	This illustrates how one must proceed with caution when translating late-time properties of waves on near-horizon geometries to the behaviour along the event horizon of genuine extremal black hole spacetimes.
	
	In particular, the leading-order late-time behaviour of $D_{\tau}\psi_{\ell m}$ at $r=r_+$ therefore agrees in the extremal and sub-extremal cases (in stark contrast with the $\mathfrak{q}=0$ setting!). This implies in particular that in \underline{both} the extremal and sub-extremal cases the stress-energy tensor satisfies:
	\begin{equation*}
		\mathbb{T}^{SF}[g_{M,Q},r^{-1}\psi_{\ell m}]|_{\mathcal{H}^+}(\partial_{\tau},\partial_{\tau})\sim \mathfrak{q}^2Q^2|\mathfrak{I}_{\ell m}[\psi]|^2\begin{cases}\tau^{-2-2\re \beta_{\ell}}\quad & (\beta_{\ell}\in (0,1)\cup i(0,\infty)),\\
\frac{\tau^{-2}}{\log^4 \tau}\quad & (\beta_{\ell}=0).
\end{cases}	
	\end{equation*}
	In other words, the energy density of $\psi$ measured with respect to the Killing vector field $\partial_{\tau}$, \underline{tangential} to $\mathcal{H}^+$ is dictated by late-time tails arising from far-away behaviour of the wave operator, at \emph{null infinity}, and it decays in time.

	On the other hand, we have by \eqref{eq:introlowboundbetanonzeroRNhor} and \eqref{eq:introlowboundzerobetaRNhor} that the energy density $\mathbb{T}^{SF}[g_{M,Q},\psi_{\ell m}]|_{\mathcal{H}^+}(\partial_{r},\partial_{r})$ \underline{transversal} to $\mathcal{H}^+$ is dictated by tails arising from the \emph{event horizon}, described by $\Psi_0^+$, and features growth  in time.
	
	By analogy (recall \S \ref{sec:confisom}), we also obtain along $\mathcal{I}^+$ in the extremal case:
	\begin{equation*}
		|D_T\psi_{\ell m}|(\tau,\infty)\sim |\mathfrak{q}Q||\mathfrak{H}_{\ell m}[\psi]|\tau^{-1-\re \beta_{\ell}}.
	\end{equation*}
Gauge-invariant time derivatives along $\mathcal{I}^+$ therefore decay faster in the sub-extremal case than in the extremal case and in the latter, we can read off the value of the constants $|\mathfrak{H}_{\ell m}[\psi]|$ from the limiting behaviour of $|D_{\tau}\psi_{\ell m}|$ at $\mathcal{I}^+$, thereby providing an ``null infinity signature'' of extremality and the constants $\mathfrak{H}_{\ell m}[\psi]$ responsible for the extremal instability phenomena in Theorem \ref{thm:introinst}; see \cite{aag18, aags23} for results on null infinity signatures in the neutral scalar field setting. 
	
	\subsubsection{Transient instabilities along near-extremal event horizons}
	In the estimates that are implicit in the tilde-notation in Theorem \ref{thm:introtails}, we keep precise track of the dependence on the surface gravity $\kappa_+$ of the background Reissner--Nordstr\"om spacetime; see also the precise statements in Theorem \ref{thm:mainthmpoint}. This allows us to estimate the time interval in the near-extremal setting for which the instabilities of Theorem \ref{thm:introinst} are relevant.

	We show that for $\delta_1\ll \delta_2\ll 1$ and $\kappa_+\leq \delta_1 r_+^{-1}$:
		\begin{align*}
		\sup_{\delta_1 \kappa_+^{-1}\leq \tau\leq \delta_2 \kappa_+^{-1}}(\tau+1)^{-k-\frac{1}{2}+\frac{1}{2}\re\beta_0}|D_X^{k+1}\psi_{\ell m}|(\tau,r_+)\geq&\:  \frac{|\mathfrak{c}_k|}{2}\left|\mathfrak{H}_{\ell m}[\psi_M]\right|\quad &(\beta_{0}\neq 0),\\
		\sup_{\delta_1 \kappa_+^{-1}\leq \tau\leq \delta_2 \kappa_+^{-1}}\log(1+\tau)\tau^{-k-\frac{1}{2}}|D_X^{k+1}\psi_{\ell m}|(\tau,r_+)\geq&\:  \frac{|\mathfrak{c}_k|}{2}\left|\mathfrak{H}_{\ell m}[\psi_M]\right|\quad &(\beta_{0}= 0).
	\end{align*}
	where $\psi_{M}$ is the solution to \eqref{eq:CSFintro} arising from the same initial data as $\psi$, but corresponding to the metric $g_{M,Q}$ with $|Q|=M$ instead of $|Q|<M$; see Proposition \ref{prop:transinstb}.
	
Therefore, $|D_X^{k+1}\psi_{\ell m}|$ must attain a maximum along $\mathcal{H}^+$ greater or equal to a constant that scales like $\kappa_+^{-\frac{1}{2}-k+\frac{1}{2}\re \beta_0}$ in the case $|Q|<M$ when $\beta_0\neq 0$ (and similarly with an additional decaying factor $(\log \kappa_+^{-1})^{-1}$ when $\beta_0=0$). This may be thought of as a \emph{transient instability}. See also the results in \cite{zimmerman2} that are consistent with the above observations.

	\subsubsection{Sub-extremality and mode (in)stability}
	\label{sec:intromodeinstb}
		The results in Theorems \ref{thm:introtails} and \ref{thm:introinst} apply also away from extremality under the additional assumption of mode stability on the real axis; see Condition \ref{cond:quantmodestabp2} for the precise assumption. Roughly, this corresponds to the absence of oscillating, finite-energy solutions with an exponential time dependence.
				
		Mode stability in the charged scalar field context remains an open problem away from extremality. See also \cite{gaj26a}[Appendix B] for a partial mode stability result and  \cite{gaj26a}[Remark 1.4] for a further discussion.
				
		In the extremal case, and by continuity, in the near-extremal case, mode stability is shown to hold in \cite{gaj26a}, which allows for a proof of the desired integrated energy estimates in these settings in \cite{gaj26a}. This stands in sharp contrast with the presence of growth and instabilities along $\mathcal{H}^+$ in Theorem \ref{thm:introinst}, which can be associated with the behaviour in Fourier space at frequencies $\omega=\mathfrak{q}Qr_+^{-1}$, which are actually \emph{excluded} in the mode stability statements as modes are ill-defined at those frequencies in the extremal case.
		
		Note that for sufficiently small $|\mathfrak{q}Q|$, mode stability follows straightforwardly from mode stability of the neutral scalar field; see also \cite{gaj26a}[\S 4].
		
		We also note that in the coupled Maxwell--charged scalar field system of equations \eqref{eq:max1}--\eqref{eq:CSFnonlin}, where\\ $g=g_{M,Q}$, the sum of stress energy tensors $\mathbb{T}^{EM}[g,F]+\mathbb{T}^{SF}[g,\phi]$ satisfies a conservation law and corresponding integrals along appropriate horizon-intersecting asymptotically hyperboloidal or null  hypersurfaces are non-negative definite and non-increasing, implying in particular that $|\phi|$ must be bounded in time and small for initially small perturbations of the Reissner--Nordstr\"om Faraday tensor.

		Boundedness of $|\phi|$ in the nonlinear setting alone, however, does not imply the absence of exponentially growing modes in the linearized setting, as it is possible that $|\phi|$ does not decay (sufficiently fast) and hence, that the linearized equations do not approximate the nonlinear equations at late-times. We note that for small $|\mathfrak{q} Q|$, decay of $|\phi|$ has been proved in \cite{vdm22} and the numerical results in \cite{hp98c,br16,av25} are consistent with decay of $|\phi|$ for large $|\mathfrak{q}Q|$.
				
	In view of the above, a proof of mode (in)stability away from extremality in the context of \eqref{eq:CSFintro} remains an interesting open problem for understanding \eqref{eq:eemax}--\eqref{eq:CSFnonlin} within the context of small perturbations of sub-extremal Reissner--Nordstr\"om initial data with large $|\mathfrak{q}Q|$, \underline{away from extremality}.

	\subsubsection{Late-time tails for $\psi_{\geq \ell}$}
	Theorem \ref{thm:introtails} determines the leading-order late-time asymptotics for $\psi$ supported on spherical harmonics with angular parameter greater or equal to $\ell$, $\psi_{\geq \ell}$, provided $\ell(\ell+1)<\mathfrak{q}^2Q^2$. However, the methods in the present paper also extend straightforwardly to $\ell(\ell+1)\geq \mathfrak{q}^2Q^2$. The key difference is the necessity to consider higher-order time integrals $(TK)^{-k}$, $k\geq 2$; see the discussion in \S \ref{sec:sketchpf}. We omit this analysis in the present paper.

\subsection{Further previous literature}
\label{intro:prevresults}

In this section, we will provide references to literature pertaining to late-time tails and precise late-time asymptotics in the Reissner--Nordstr\"om setting. We refer to \cite{gaj26a}[\S 1.1] for further references pertaining to integrated energy estimates.

\subsubsection{Numerics and heuristics}
A study of late-time tails for \eqref{eq:CSFintro} on sub-extremal Reissner--Nordstr\"om black holes was instigated by \cite{hp98b}, see also the earlier work restricted to small $|\mathfrak{q}Q|$ in \cite{hp98a}. In these works a heuristic analysis was performed in Fourier space, restricted to a neighbourhood of $\omega=0$. When $\beta_{\ell}\in (0,1)$, the tails in \cite{hp98b} are consistent with Theorem \ref{thm:introtails}. When $\beta_{\ell}\in i(0,\infty)$, the present paper elucidates in addition the precise oscillatory structure in the tails suggested in \cite{hp98b}, which are precisely captured by the tail function $\Psi^{\infty}$.

Late-time tails in the extremal setting were first investigated along the event horizon (corresponding to initial data supported away from the horizon) in the setting of \eqref{eq:CSFintro} in \cite{zimmerman3}, via a mixed heuristic-numerical analysis in performed in Fourier space around the frequency $\omega=\mathfrak{q}Q$. The instabilities suggested by \cite{zimmerman3} are consistent with Theorems \ref{thm:introtails} and \ref{thm:introinst}. The present work also covers the case $\beta_{\ell}=0$ and provides also, in particular, a precise characterization of the oscillatory, late-time asymptotics along $\mathcal{H}^+$ in the form of the tail function $\Psi^+$. Furthermore, as the Fourier analysis \cite{zimmerman3} does not consider frequencies near $\omega=0$, their dominant role in the late-time tails of the derivatives $|D_{\tau}\psi|$ (captured by the tail function $\Psi^{\infty})$ was not observed; see \S \ref{sec:introwagging} for a further discussion.

Late-time tails in the setting of the nonlinear EMCSF equations in spherical symmetry \eqref{eq:eemax}--\eqref{eq:CSFnonlin}, have also been studied numerically in the sub-extremal setting, starting from \cite{hp98c} and more recently in \cite{br16,av25} reporting results consistent with the analysis of the linearized system in the present paper when $\beta_{\ell}\neq 0$. 

In the extremal setting, nonlinear EMCSF equations were first studied numerically in \cite{gp26}. See also \cite{gp25} for earlier numerical work on the Maxwell--charged scalar field equations (EMCSF with $F=F_Q$). The late-time behaviour in the context of spacetimes settling down to extremal Reissner--Nordstr\"om that is observed numerically is consistent with the linear asymptotics derived in the present paper.

\subsubsection{Late-time tails for neutral scalar fields}
Late-time tails in the context of the wave equation on black hole spacetimes (neutral scalar fields) have a rich history, starting from heuristic work of Price \cite{Price1972}. The decay rates corresponding to sub-extremal black hole backgrounds are referred to as ``Price's law''. On sub-extremal Reissner--Nordstr\"om spacetimes, the existence and nature of late-time tails was first proved in \cite{paper1,paper2}. In \cite{hintzprice,aagkerr}, late-time tails were proved to exist also on sub-extremal Kerr spacetimes. Late-time tails for scalar fields supported on higher spherical harmonics consistent with Price's law were established on Reissner--Nordstr\"om in \cite{hintzprice,aagprice}, see also \cite{aagkerr} for a discussion on higher spherical harmonics in Kerr. We refer to \cite{aagkerr} for further discussion on the late-time tails literature in the black hole setting. Note that the rates in late-time tails change if one considers \emph{dynamical} instead of stationary spacetimes or adds nonlinearities to the wave equation. Furthermore, the rates change when considering initial data that are motivated by gravitational $N$-body problems with no incoming radiation. We point the reader to \cite{gk22, lukoh24,gk25} and references therein.

Let us note here that the late-time tails in Theorem \ref{thm:introtails} are less sensitive to the fall-off in $r$ of the initial data when compared to the neutral setting.

Late-time tails on extremal black holes were first established in \cite{paper4} on extremal Reissner--Nordstr\"om and then in \cite{gaj23} on extremal Kerr. See \S \ref{eq:intronetural} for a further discussion on this setting.

\subsubsection{Charged scalar fields}
The mathematical literature on late-time time tails and precise late-time asymptotics in the context of charged scalar fields is more sparse. In contrast with neutral scalar fields, late-time tails are already present for charged scalar fields on Minkowski\footnote{Note that the linearization of the Maxwell--Klein Gordon equations around the Minkowski solution results in the standard (neutral) wave equation. One can nevertheless consider the gauge derivative $D=d-i\mathfrak{q}QA\otimes(\cdot)$ corresponding to the Reissner--Nordstr\"om electric charge $Q$ and consider the associated charged scalar field equation, which we refer to as the charged scalar field equation on Minkowski.} by \cite{gvdm24}, see also the related sharp decay results in the massless Dirac--Coulomb system established in \cite{bgrm25}. 

In the black hole setting, decay estimates for solutions to the spherically symmetric Maxwell--Klein Gordon equations on sub-extremal Reissner--Nordstr\"om were established in the small-$|\mathfrak{q}Q|$ setting in \cite{vdm22}; see also \S \ref{sec:intromodeinstb}.

Let us also note that the behaviour of solutions to the spherically symmetric EMCSF system in dynamically sub-extremal black hole interiors has been studied, \emph{assuming} late-time decay along the event horizon that is consistent with Theorem \ref{thm:introtails}. In \cite{vdm18}, $C^0$-stability and $C^2$-instability of the Cauchy horizon was established for $|\mathfrak{q}Q_f|<\frac{1}{2}$, with $Q_f$ the charge of the Reissner--Nordstr\"om black hole that the corresponding dynamical black holes solution settles down to in the black hole exterior. The $|\mathfrak{q}Q_f|<\frac{1}{2}$ condition is motivated by expectation that in that case $|\phi|_{\mathcal{H}^+}$ decays integrably along $\mathcal{H}^+$ in that case, which is consistent with Theorem \ref{thm:introtails}.

 In \cite{vdMK21}, the $|\mathfrak{q}Q_f|\geq\frac{1}{2}$ case was considered and $C^0$-stability of the Cauchy horizon was established, relying on the oscillation present in $\psi|_{\mathcal{H}^+}$ with respect to the gauge $A=-Q(r^{-1}-r_+^{-1})d\tau$, which follows directly from Theorem \ref{thm:introtails} (after changing the gauge $A=-Qr^{-1}d\tau$ to $A=-Q(r^{-1}-r_+^{-1})d\tau$).

In \cite{dejanjon1}, the extremal setting, the behaviour for solutions to the spherically symmetric EMCSF system were studied in the interiors of dynamical black holes settling down to extremal Reissner--Nordstr\"om, with a smallness condition on $|\mathfrak{q}Q|$. It was shown that Cauchy horizons are stable in $C^0\cap H^1$ and are sufficiently regular so that the spacetime can be extended as a solution to the EMCSF system across the horizons. The initial data assumptions for the energy decay along $\mathcal{H}^+$ in \cite{dejanjon1} are consistent with the results in the present paper; see also the discussion in \S \ref{sec:introminkrole}.

 Finally, we note that some of the difficulties in establishing late-time tails in the charged scalar field setting, assuming suitable integrated energy estimates, are already present in the setting of neutral scalar fields with additional asymptotically inverse-square potentials. See \cite{gaj22a,hin23} for proofs on the existence of late-time tails and references therein for previous work on decay estimates in this setting. In particular, the philosophy of proving late-time tails by subtracting appropriate global tail functions and integrating in time, introduced in \cite{gaj22a}, remains applicable to the setting of the present paper.

\subsection{Sketch of the proofs of Theorems \ref{thm:introtails} and \ref{thm:introinst}}
\label{sec:sketchpf}
We will consider energy densities of functions $\hpsi$ that are defined as follows:
\begin{equation*}
	\hpsi=\psi-\Psi
\end{equation*}
where $r^{-1}\psi$ is a solution to \eqref{eq:CSFintro} with respect to the conformally smooth electromagnetic gauge $A=-\frac{Q}{r}d\tau$ and $\Psi$ is the \emph{tail function} from Theorem \ref{thm:introtails}, which, as we will sketch here, approximates $\psi$ as $\tau\to \infty$.

The coefficients $\mathfrak{I}_{\ell m}[\psi]$ and $\mathfrak{h}_{\ell m}[\psi]$ are complex numbers that can be explicitly determined from initial data for $\psi$ at $\tau=0$ and we will illustrate how they are determined in Step 2 below. 

For convenience, we introduce the dimensionless charge parameter:
\begin{equation*}
	\q=\mathfrak{q}Q.
\end{equation*}

The relevant energy densities take the following form:
\begin{equation*}
	\mathcal{E}_p[\hpsi]\sim (\Omega^{-1}r)^{p}\Omega^{2}|\partial_r\hpsi|^2+(\Omega^{-1} r)^{\min\{p,0\}}r^{-2}\left[|\partial_{\tau}\hpsi|^2+|\snabla_{\s^2}\hpsi|^2+|\hpsi|^2\right],
\end{equation*}
where $p\in \R$, $\Omega^2=1-\frac{2M}{r}+\frac{Q^2}{r^2}$ and $\snabla_{\s^2}$ is the covariant derivative with respect to the unit round sphere.

Note that in the case $p=2$ and for bounded $r$, $\mathcal{E}_p[\hpsi]\sim \mathbb{T}(\mathbf{n}_{\tau},\mathbf{n}_{\tau})$, so $\mathcal{E}_2[\hpsi]$ controls a non-degenerate energy density. Furthermore, away from $\mathcal{H}^+$, $\mathcal{E}_p[\hpsi]$ may be viewed as an $r^p$-weighted energy density, as first introduced in \cite{newmethod}.

We write:
\begin{equation*}
	(g^{-1}_{M,Q})^{\mu \nu}\:(^AD)_{\mu}\:(^AD)_{\nu}(r^{-1}\hpsi)=-(g^{-1}_{M,Q})^{\mu \nu}\:(^AD)_{\mu}\:(^AD)_{\nu}(r^{-1}\Psi)=:G_A.
\end{equation*}

\subsubsection{Step 0: integrated energy estimates}
The starting point of the present paper is the following uniform integrated energy estimate, established in \cite{gaj26a}. Let $1<  p<\max\{1+\re \beta_{0},2\}+\epsilon$, with $\epsilon>0$ arbitrarily small, then for all $\tau_1<\tau_2$:
\begin{multline}
\label{eq:introied}
 \int_{\Sigma_{\tau_2}}\mathcal{E}_{p-2\epsilon}[\hpsi]\,d\sigma dr+	\int_{\tau_1}^{\tau_2} \int_{\Sigma_{\tau}}\upzeta(r)\cdot \mathcal{E}_{p-1-2\epsilon}[\hpsi]+r^{-2}|\hpsi|^2\,d\sigma dr d\tau\leq C \int_{\Sigma_{\tau_1}}\mathcal{E}_{p}[\hpsi]\,d\sigma dr\\
 +C\int_{\tau_1}^{\tau_2} \int_{\Sigma_{\tau}}(\ldots)|r^3G_A|^2+(1-\upzeta)|\partial_{\tau}G_A|^2\,d\sigma dr,
\end{multline}
where $\upzeta$ is a function supported away from a small neighbourhood of the photon sphere at $r=r_{\sharp}$ and $(\ldots)$ denotes appropriate weights that we will not state in the present sketch; see however Theorem \ref{thm:iedpaper1} for a detailed version.

When applied to $\psi$, with $r^{-1}\psi$ a solution to \eqref{eq:CSFintro}, \eqref{eq:introied} implies the following energy boundedness and weak energy decay estimates:
\begin{align}
\label{eq:introenbound}
	 \int_{\Sigma_{\tau_2}}\mathcal{E}_{p-2\epsilon}[\hpsi]\,d\sigma dr\leq &\: C \int_{\Sigma_{\tau_1}}\mathcal{E}_{p}[\hpsi]\,d\sigma dr,\\
	 \label{eq:introendec}
	  \lim_{k\to \infty} \int_{\Sigma_{\tau_k}}\upzeta\cdot \mathcal{E}_{p-1-2\epsilon}[\hpsi]\,d\sigma dr\to &\:0\quad \textnormal{for some monotonically increasing sequence $\{\tau_k\}$},
\end{align}
where the latter follows from the mean-value theorem and the factor $\upzeta$ can be removed after additional commutation with $\partial_{\tau}$ or angular derivatives. Note in particular the loss in $\epsilon$ in the uniform energy estimate! This loss can be removed with an additional smallness assumption on the charge parameter $\q$: $|\q|<\frac{1}{4}$.

In the present paper, we will improve the decay estimate \eqref{eq:introendec}, but we will not do so directly at the level of $\widehat{\psi}$. Instead, we first consider $(TK)^k(\hpsi)$, where $T=\partial_{\tau}$ and
\begin{equation*}
	K=T+iq r_+^{-1}\mathbf{1}.
\end{equation*}
The underlying idea is that $(TK)^k\hpsi$ should decay faster for each application of $TK$. We illustrate this as follows: along $r=r_0$, we expect (ignoring coefficients and oscillations and possible $\log \tau$ factors) by Theorem \ref{thm:introtails}:
\begin{align*}
	\psi_{\ell=0}(\tau,r_0,\cdot)\sim c_{r_0}\tau^{-s}\quad (|Q|>M),\\
	\psi_{\ell=0}(\tau,r_0,\cdot)\sim \tau^{-s}+e^{-i\q r_+^{-1}\tau}\tau^{-s}\quad (|Q|=M),
\end{align*}
for appropriate values of $s$. In the sub-extremal case, $T^k\psi$ should therefore decay with a rate $\tau^{-s-k}$, whereas in the extremal case we need to apply $(TK)^{k}$ to obtain $\tau^{-s-k}$, due to the presence of the oscillating factor $e^{-i\mathfrak{q}Qr_+^{-1}\tau}$. Since we are considering the extremal and sub-extremal cases simultaneously, we therefore simply consider $(TK)^{k}$ in both cases.
\subsubsection{Step 1: energy decay in time for $TK$-derivatives}
We will discuss here the case $\re \beta_0=0$, which is the most difficult case, as then $p$ in \eqref{eq:introied} satisfies $1<p<1+\epsilon$, so that \eqref{eq:introied} provides effectively a single weighted integrated estimate by taking, for example, $p=1+\frac{\epsilon}{2}$.

It is straightforward to show that \eqref{eq:introied} also holds when $\hpsi$ is replaced with $\mathbf{Z}^{\gamma}\hpsi$, which schematically denotes:
\begin{equation*}
	\mathbf{Z}^{\gamma}\hpsi=\snabla_{\s^2}^{\gamma_1}((r-r_+)\partial_r)^{\gamma_2}\partial_{\tau}^{\gamma_2}\hpsi,\quad |\gamma|=k.
\end{equation*}
Then we observe that
\begin{equation}
\label{eq:fromTKtoZ}
	\mathcal{E}_p[(TK)\hpsi]\lesssim \sum_{|\gamma|\leq 1}\mathcal{E}_{p-2}[\mathbf{Z}^k\hpsi]+\mathcal{E}_{p-2}[\mathbf{Z}^kT\hpsi]+(\Omega^{-1}r)^{2-p}(r^{-2}|r^3G_A|^2+r^{-2}|r^3TG_A|^2).
\end{equation}
Assuming sufficient decay on the inhomogeneity $G_A$ (which is derived in \S \ref{sec:boxPsismall}), we can therefore extend the range of $p$ from $p\in \{1+\frac{\epsilon}{2}\}$ to $p\in \{1+\frac{\epsilon}{2},3+\frac{\epsilon}{2}\}$ by considering $(TK)\hpsi$ instead of $\hpsi$.

We will establish (assuming again sufficient decay of $G_A$): for $1<p<2$,
\begin{equation}
\label{eq:edecaytimeinvintro}
	\int_{\Sigma_{\tau}}\mathcal{E}_p[(TK)\hpsi]\,d\sigma dr d\tau\lesssim \tau^{-3+p+\nu},
\end{equation}
with $\nu>0$ arbitrarily small (for $\epsilon$ sufficiently small). This step involves an interpolation argument, where we partition the $r$-interval as follows:
\begin{equation*}
	\{(r^{-2}(r-r_+)\leq (1+\tau)^{-\sigma}\}\cup \{(r^{-2}(r-r_+)>(1+\tau)^{-\sigma}\}.
\end{equation*}
We start with $\sigma=0$ and apply \eqref{eq:fromTKtoZ} to obtain:
\begin{equation*}
	\int_{\Sigma_{\tau}}\mathcal{E}_{1+\epsilon}[(TK)\hpsi]\,d\sigma dr d\tau\lesssim \tau^{-1+\nu}.
\end{equation*}

Then we take $\sigma=\frac{1}{2}$ and use the already established energy decay estimate to obtain an improvement. We consider iteratively $\sigma=1-2^{-k}$ and take $k$ sufficiently large to obtain the desired energy decay estimate. One may view these estimates as a replacement of the energy decay estimates introduced in \cite{newmethod}, applicable with only limited available $r^p$-weighted energy estimates. See the proof of Proposition \ref{prop:edecay} for details.

In this step it is important to note that even though the argument involves a large number of iterations, the total number of derivatives on the right-hand side of the estimate is bounded and \emph{independent} of the number of iterations.

The improved decay for $(TK)^k\hpsi$ then follows analogously and requires the consideration of $\mathbf{Z}^k\hpsi$ and $\mathbf{Z}^kT\hpsi$.

Observe that the above argument resembles the method for deriving energy decay estimates in \cite{gaj22a} and \cite{gaj23},  with the key differences being:
\begin{itemize}
	\item The consideration of a non-trivial inhomogeneity, in contrast with \cite{gaj22a},\\
	\item Commutation with $TK$ instead of $T$, in contrast with \cite{gaj23}.
\end{itemize}

\subsubsection{Step 2: constructing time integral initial data}
Note that Step 1 only produces energy decay estimates for $(TK)^k\hpsi$ with $k\geq 1$, whereas we need suitably strong energy decay estimate for $\hpsi$ itself.

The main technique here is the operation of the ``time inversion'' operator $(TK)^{-1}$, which we can interpret as \emph{time integration} in the following way:
\begin{equation}
\label{eq:timeintintro}
	(TK)^{-1}(f)=-\frac{r_+}{i \q}\left( K^{-1}f-T^{-1}f\right),
\end{equation}
with
\begin{align*}
(T^{-1}f)(\tau,r,\theta,\varphi)=&-\int_{\tau}^{\infty}f(\tau',r,\theta,\varphi)\,d\tau',\\
(K^{-1}f)(\tau,r,\theta,\varphi)=&-e^{-i \q r_+^{-1} \tau}\int_{\tau}^{\infty}e^{i q  r_+^{-1} \tau'}f(\tau',r,\theta,\varphi)\,d\tau'.
\end{align*}
In view of the fact that $\psi$ does not decay integrably in time when $\re\beta_0=0$ (see Theorem \ref{thm:introtails}) we cannot hope to define $(TK)^{-1}(\psi)$. For this reason, we consider $(TK)^{-1}(\widehat{\psi})$, which, is well-defined since, as we will show, $\hpsi$ does decay integrably in $\tau$ along constant $r$ hypersurfaces (away from $r=r_+$).

However, \emph{a priori}, we do not have any pointwise decay for $\hpsi$, so we can only interpret the above time integrals \emph{formally}. That is to say, we will construct functions  $(TK)^{-1}(\hpsi)$ that only \emph{a foriori} can be shown to satisfy the above integral identities.

Instead, we define $T^{-1}\hpsi$ and $K^{-1}\hpsi$ as solutions to the equations:
\begin{align*}
	(g^{-1}_{M,Q})^{\mu \nu}\:(^AD)_{\mu}\:(^AD)_{\nu}(r^{-1}T^{-1}\hpsi)=&\:T^{-1}G_A,\\
	(g^{-1}_{M,Q})^{\mu \nu}\:(^AD)_{\mu}\:(^AD)_{\nu}(r^{-1}K^{-1}\hpsi)=&\:K^{-1}G_A,
\end{align*}
arising from initial data $(T^{-1}\hpsi, T(T^{-1}\hpsi))|_{\Sigma_0}$ and $(K^{-1}\hpsi, T(K^{-1}\hpsi))|_{\Sigma_0}$, respectively, such that
\begin{align*}
	T(T^{-1}\hpsi)|_{\Sigma_0}=&\:\hpsi|_{\Sigma_0},\\
	K(K^{-1}\hpsi)|_{\Sigma_0}=&\:\hpsi|_{\Sigma_0}.
\end{align*}
The above conditions determine equations for $T^{-1}\hpsi|_{\Sigma_0}$ and $K^{-1}\hpsi|_{\Sigma_0}$ of the form
\begin{align*}
\mathcal{L}(T^{-1}\hpsi|_{\Sigma_0})=r\widetilde{G}_A,	\\
\underline{\mathcal{L}}(K^{-1}\hpsi|_{\Sigma_0})=r\underline{\widetilde{G}}_A,
\end{align*}
where $\mathcal{L}$ is a second-order differential operator that contains only $\partial_r,\partial_{\theta},\partial_{\varphi}$ derivatives and the $\partial_{\tau}$-derivatives, which can be expressed in terms of $\hpsi|_{\Sigma_0}$ and $T\hpsi|_{\Sigma_0}$, form part of $r\widetilde{G}_A$.

For this procedure to make sense, we need $G_A$ to decay suitably fast in $\tau$, so that \eqref{eq:timeintintro} with $f=G_A$ is well-defined.

In particular, the operator $\mathcal{L}$ takes the following form:
\begin{align*}
	\mathcal{L}f=\partial_r(\Omega^2\partial_r f)+r^{-2}\slashed{\Delta}_{\s^2} f-\frac{d\Omega^2}{dr} r^{-1}f+2i \q(1+O(r^{-2}))r^{-1}\partial_r f+\left[i \q r^{-2}+O(r^{-2})\right]f.
\end{align*}

In the case $\mathfrak{q}Q=0$, the differential operators $\mathcal{L}$ and $\underline{\mathcal{L}}$ may be thought of a degenerate elliptic operators and $r$-weighted, degenerate elliptic estimates apply. When $\mathfrak{q}Q\neq 0$, the principal symbol of $\mathcal{L}$ remains unchanged, but the coercivity of the corresponding (degenerate) energies is destroyed.

Restricting to high angular frequencies $\psi|_{\ell\geq L}$ with $L\in \N_0$ suitably large and integrating by parts $r^{-s}(\Omega^{-1}r)^{-p'}\partial_r(r^{-1}T^{-1}\psi) \mathcal{L}(T^{-1}\psi)$ we maintain coercivity and obtain $r$-weighted elliptic-type estimates, see \S \ref{sec:ellipticesthighl} for the details. This does not, however, work for bounded angular frequencies.

We therefore do not construct $T^{-1}\psi$ (and similarly, $K^{-1}\psi$) by inverting $\mathcal{L}$. Instead, we first construct the spherical harmonic modes $(T^{-1}\psi)_{\ell}$ by considering an appropriate \emph{twisted operator} $\widecheck{\mathcal{L}}_{\ell}$ and use the elliptic-type estimates valid at high angular frequencies to be able to sum the spherical harmonic modes in $\ell$.

The twisted operator $\widecheck{\mathcal{L}}_{\ell}$ is defined as follows:
\begin{equation*}
	\widecheck{\mathcal{L}}_{\ell}(f):=\partial_r(\Omega^2\mathfrak{w}_{\ell}^2 e^{-2i\q r_+^{-1}\int_{r_0}^rr'^{-1}\Omega^{-2}(r')(1-\h(r'))\,dr'}\partial_r(\mathfrak{w}_{\ell}^{-1}f)),
\end{equation*} 
so that
\begin{equation*}
	\widecheck{\mathcal{L}}_{\ell}((T^{-1}\psi)_{\ell}|_{\Sigma_0})=\mathfrak{w}_{\ell} e^{-2i\q r_+^{-1}\int_{r_0}^rr'^{-1}\Omega^{-2}(r')(1-\h(r'))\,dr'}r(\widetilde{G}_A)_{\ell}.
\end{equation*}
Here, $\mathfrak{w}_{\ell}$ are the functions appearing in Theorem \ref{thm:introtails}. By considering $\mathfrak{w}_{\ell}^{-1}T^{-1}\hpsi_{\ell}$, we therefore arrive at an ODE which can simply be solved by integrating in the $r$-direction.

We start integrating at $r=r_+$ and assume the corresponding boundary term at $r=r_+$ vanishes, which is our first boundary condition. The second boundary condition is that $\mathfrak{w}_{\ell}^{-1}T^{-1}\psi_{\ell}|_{\Sigma_0}$ vanishes as $r\to \infty$.

When $\re \beta_{\ell}<1$, we need the following condition for the above boundary conditions to be compatible:
\begin{equation}
\label{eq:fixconstI}
	\int_{r_+}^{\infty}\mathfrak{w}_{\ell} e^{-2i\mathfrak{q}Qr_+^{-1}\int_{r_0}^rr'^{-1}\Omega^{-2}(r')(1-\h(r'))\,dr'}r(\widetilde{G}_A)_{\ell}\,dr=0.
\end{equation}
\textbf{The requirement \eqref{eq:fixconstI} fixes the constants $\mathfrak{I}_{m\ell}[\psi]$ appearing in the definition of $\Psi^{\infty}$!}

 An analogous procedure for $K^{-1}\psi$, fixes the constants $\mathfrak{h}_{m\ell}[\psi]$ and hence the entire function $\Psi$, which determines the late time asymptotics. We refer to \S \ref{eq:timeintellipticboundl} for the details of the above argument.
 
 Finally, to obtain $T^{-1}\hpsi$ and $K^{-1}\hpsi$, we need to sum $(T^{-1}\hpsi)_{\ell}$ and $(K^{-1}\hpsi)_{\ell}$ in $\ell$. The convergence of this sum follows from the above mentioned elliptic estimates that are valid for large $\ell$; see \S \ref{sec:ellipticesthighl}.

We conclude in particular finiteness of the relevant energy norms of $(TK)^{-1}\hpsi$ appearing on the right-hand side of the energy decay estimates in Step 1.

\subsubsection{Step 3: energy decay and late-time asymptotics}
In Step 2, we constructed $(TK)^{-1}\psi|_{\Sigma_0}$ and obtained finiteness of weighted energies along $\Sigma_0$. Applying \eqref{eq:edecaytimeinvintro} to $(TK)^{-1}\hpsi$, we then obtain for $1<p<2$:
\begin{equation}
\label{eq:edecayintro}
	\int_{\Sigma_{\tau}}\mathcal{E}_{p}[\hpsi]\,d\sigma dr d\tau\lesssim \tau^{p-3+\nu}.
\end{equation}
See \S \ref{sec:edectimint} for the details.

By a standard argument, \eqref{eq:edecayintro} results in the pointwise decay estimate:
\begin{equation*}
		|\psi-\Psi|=|\hpsi|\lesssim \tau^{-1+\nu},
\end{equation*}
which already implies the late-time asymptotics along $\mathcal{I}^+$ stated in Theorem \ref{thm:introtails}, and, in the case $|Q|=M$, also along $\mathcal{H}^+$, because $\Psi$ decays slower than $\tau^{-1+\nu}$. See \S \ref{sec:tails} for the details.

We can increase the range of $p$ in \eqref{eq:edecayintro} to $-\re \beta_0<p<3$, by applying once more the elliptic estimates from Step 2. This then results in an additional pointwise estimate which degenerates at $r=r_+$ and $r=\infty$:
\begin{equation*}
	|\psi-\Psi|=|\hpsi|\lesssim (r\Omega)^{-1+\re\beta_0}\tau^{-\frac{3}{2}-\re\beta_0+\nu},
\end{equation*}
from which we can conclude late-time asymptotics away from $\mathcal{H}^+$ and $\mathcal{I}^+$.

To conclude precise late-time asymptotics in the sub-extremal case $|Q|<M$ along $\mathcal{H}^+$, we apply in addition standard red-shift energy estimates that remove the degeneracy at $r=r_+$ in the above estimates, at the expense of introducing $\kappa_+$-dependence in the constants.

\subsubsection{Step 4: instabilities}
The existence of pointwise instabilities along $\mathcal{H}^+$ and $\mathcal{I}^+$ may be thought of as a consequence of the precise late-time asymptotics for $\psi$ that follow from Step 3. In the sketch below, we will therefore write:
\begin{equation*}
	\psi=\Psi+\ldots,
\end{equation*}
and we will ignore the terms in $ +\ldots$, as they will not contribute to leading-order in the growth estimates below. It therefore remains to investigate the behaviour of $r^2X\Psi$ and of $\mathcal{E}_2[\Psi]$ to conclude the instabilities in Theorem \ref{thm:introinst}. In view of the symmetry of $\Psi$ under the conformal Couch--Torrence isometry, it suffices to consider $r^2\Psi|_{\mathcal{I}^+}$.

When $\mathfrak{I}_{\ell m}[\psi]\neq 0$ or $\mathfrak{H}_{\ell m}[\psi]\neq 0$, it easily follows from the expressions of $\Psi$ in terms of $\mathfrak{w}_{\ell}$ and $\Psi_0$ and the asymptotics of $\Psi_0$ that in the case $ \beta_{\ell}\in [0,1)$:
\begin{align*}
	\int_{\Sigma_\tau\cap \{r\geq R\}}\mathcal{E}_2[\Psi]\,d\sigma dr\sim &\: (1+\tau)^{-\re \beta_{\ell}}\quad (\beta_{\ell}\neq 0),\\
		\int_{\Sigma_\tau\cap \{r\geq R\}}\mathcal{E}_2[\Psi]\,d\sigma dr\sim &\: \frac{1}{\log^2(2+\tau)}\quad (\beta_{\ell}=0).
\end{align*}
The leading-order behaviour for $r^2X\Psi|_{\mathcal{I}^+}$ similarly follows in a straightforward manner when $\beta_{\ell}\in \R$.

The energy and pointwise growth are more subtle when $\beta_{\ell}\in i(0,\infty)$, since $r^2X\Psi$ then features oscillating terms that can in principle cancel at certain values of $\tau$.

The details of the derivation of instabilities can be found in \S \ref{sec:instab}.

\subsection{Overview of the remainder of the paper}
We outline the structure of the remaining sections of the paper:
\begin{itemize}
	\item In \S\ref{sec:prelim}, we provide the geometric preliminaries and notation that is relevant for the subsequent analysis. We introduce the relevant coordinate chart and the main equations expressed in this coordinate chart. This section moreover contains the global integrated estimates that we will assume in the present paper and which are derived in the companion paper \cite{gaj26a}; see \S \ref{sec:iedassm}.
	\item We introduce the notion of \emph{stationary solutions} in \S \ref{sec:statsol}, as well as a detailed derivation of the relevant properties of the functions $\mathfrak{w}_{\ell}$ that appeared in the preceding sections.
	\item In \S \ref{sec:precstatthm}, we state precise versions of Theorems \ref{thm:introtails} and \ref{thm:introinst}, making use of the notation introduced in \S\S \ref{sec:prelim}--\ref{sec:statsol}.
	\item In \S \ref{sec:horp}, we derive additional, higher-order commuted integrated energy estimates by applying $(\Omega^{-1}r)^p$-weighted energy estimates near $\mathcal{H}^+$ and $\mathcal{I}^+$.
	\item We derive time-decay for weighted energies for $TK$-derivatives  of solutions to \eqref{eq:CSFintro} in \S \ref{sec:edecaytimeder}. \textbf{These are the main decay estimates of the paper!}
	\item In \S \ref{sec:approxsol}, we construct the functions $\Psi$, which determine the leading late-time behaviour of $\psi$. We show that $(g^{-1}_{M,Q})^{\mu \nu}({}^AD)_{\mu}({}^AD)_{\nu}(r^{-1}\Psi)$ decays suitably fast in $r$ and $\tau$.
	\item We then subtract $\psi-\Psi$ and construct time integrals $(TK)^{-1}(\psi-\Psi)$ in \S \ref{sec:idatatimeint}. We additionally derive elliptic-type estimates in this section.
	\item In \S \ref{sec:edectimint}, we apply the energy decay estimates of \S \ref{sec:edecaytimeder} to the time integrals constructed in \S \ref{sec:idatatimeint} to obtain energy decay for the difference $\psi-\Psi$.
	\item In \S\ref{sec:tails}, we conclude the leading-order pointwise behaviour of $\psi$ (late-time tails) by converting the energy decay estimates of \S \ref{sec:edectimint} into pointwise decay estimates for the difference $\psi-\Psi$. This concludes the proof of Theorem \ref{thm:introtails}.
	\item We derive growth and instability in \S \ref{sec:instab} and prove Theorem \ref{thm:introinst}.
	\item In Appendix \ref{sec:tailfunctmink}, we construct solutions to \eqref{eq:CSFminkintro} in Minkowski, $\Psi_0^{\infty}$ that determine the late-time asymptotics in Reissner--Nordstr\"om. This section is crucial for establishing the decay properties of $(g^{-1}_{M,Q})^{\mu \nu}({}^AD)_{\mu}({}^AD)_{\nu}(r^{-1}\Psi)$ and determining the late-time behaviour of $\Psi$ in different regions in spacetime.
\end{itemize}

\subsection{Acknowledgments}
 We thank Mihalis Dafermos, Christoph Kehle, Ryan Unger and Peter Zimmerman for interesting and helpful discussions relating to the setting of the present paper. The author acknowledges funding through the ERC Starting Grant 101115568.
 
\section{Preliminaries}
\label{sec:prelim}
In this section, we will define the relevant geometric concepts and notation, and we introduce the charged scalar field equation. We will also state  the global integrated energy decay and energy boundedness estimates that follow from \cite{gaj26a}.
\subsection{Reissner--Nordstr\"om}
The \emph{Reissner--Nordstr\"om black hole exteriors} are spacetimes $(\mathcal{M}_{M,Q},g_{M,Q})$, with $M>0$, $Q\in \R$, $|Q|\leq M$,
\begin{align*}
\mathcal{M}_{M,Q}=&\:\R_v\times [r_+,\infty)_r\times \s^2_{(\theta,\varphi)},\quad \textnormal{with $\s^2$ the round unit sphere},\\
g_{M,Q}=&-\Omega^2 dv^2+2dvdr+r^2\slashed{g}_{\s^2},\quad \textnormal{where}\\
\Omega^2(r)=&\: 1-\frac{2M}{r}+\frac{Q^2}{r^2}=r^{-2}(r-r_+)(r-r_-),\quad r_-\leq r_+,\\
\slashed{g}_{\s^2}:=&\:d\theta^2+\sin^2\theta d\varphi^2,
\end{align*}
and the global causal vector field $T=\partial_v$ is future-directed.\footnote{The coordinates $(\theta,\varphi)$ cover $\S^2$ away from a half great circle connecting the north and south poles.}

The \emph{future event horizon} $\mathcal{H}^+$ is defined as the boundary:
\begin{equation*}
\mathcal{H}^+:=\partial \mathcal{M}_{M,Q}=\{r=r_+\}.
\end{equation*}

The corresponding inverse or dual metric is:
\begin{equation*}
g_{M,Q}^{-1}=\Omega^2\partial_r\otimes \partial_r+\partial_v\otimes \partial_r+\partial_r\otimes \partial_v+ r^{-2}\slashed{g}_{\s^2}^{-1},
\end{equation*}
with $\slashed{g}_{\s^2}^{-1}$ the inverse metric associated to the unit round sphere metric $\slashed{g}_{\s^2}$.

The \emph{surface gravity} $\kappa_+$ is defined as follows:
\begin{align*}
\kappa_+:=\frac{1}{2}\frac{d\Omega^2}{dr}(r_+)=\frac{Mr_+-Q^2}{r_+^3}.
\end{align*}
When $|Q|=M$, we have that $r_+=M$ and $\kappa_+=0$. In that case, we say the Reissner--Nordstr\"om black hole exterior is \emph{extremal}. If $|Q|>M$, we have that $r_+>M$ and $\kappa_+>0$, and we say the Reissner--Nordstr\"om black hole exterior is \emph{sub-extremal}. We will refer to sub-extremal spacetimes as \emph{near-extremal}, when the dimensionless quantity $r_+\kappa_+$ is assumed to be small but non-zero.

In order to compare more easily the analysis near $r=r_+$ and for $r\geq R\gg r_+$, it will be convenient to introduce the following alternative radial coordinates:
\begin{align*}
	\rho_+:=&\: r_+^{-1}-r^{-1},\\
	\rho_{\infty}:=&\: r^{-1}.
\end{align*}
Note that $\rho_+,\rho_{\infty}\in [0,r_+^{-1})$. We state the following lemma from \cite{gaj26a}:
\begin{lemma}[Lemma 2.2 from \cite{gaj26a}]
\label{lm:metricest}
Applying standard ``Big-O notation'' (see for example \cite{gaj26a}[\S 2.2]), we can estimate: 
	\begin{align*}
	r^{-2}\Omega^2(r(\rho_+))=&\:2\kappa_+ \rho_++(1-6\kappa_+ r_+)\rho_+^2+(-2r_++6\kappa_+ r_+^2)\rho_+^3+r_+^2\left(1-2\kappa_+ r_+\right)\rho_+^4,\\
	r^{-2}\Omega^2(r(\rho_{\infty}))=&\:\rho_{\infty}^2-2(r_+-\kappa_+ r_+^2)\rho_{\infty}^3+r_+^2\left(1-2\kappa_+ r_+\right)\rho_{\infty}^4.
		\end{align*}
\end{lemma}

The \emph{photon sphere radius} is defined as follows:
\begin{equation*}
r_{\sharp}:=\frac{3M}{2}+\frac{1}{2}\sqrt{9M^2-8Q^2},
\end{equation*}
and it is the solution to the equaton $\frac{d}{dr}(r^{-2}\Omega^2)(r)=0$ in $(r_+,\infty)$.

We define the \emph{tortoise coordinate} $r_*: (r_+,\infty)\to \R$ as the solution to the ODE:
\begin{align*}
	\frac{dr_*}{dr}=&\:\frac{1}{\Omega^2},\\
	r_*(r_{\sharp})=&\:0.
\end{align*}
\subsection{Foliations}
We introduce the following smooth \emph{foliation-defining functions}:
\begin{align*}
\widetilde{\h}: [r_+,\infty)\to  [0,\infty),\quad& 2-\widetilde{\h}\Omega^2\geq 0,\\
\h: [r_+,\infty)\to  [0,\infty),\quad &\h(r):=2-\widetilde{\h}\Omega^2(r).
\end{align*}
Then we define the time function $\tau: \mathcal{M}_{M,Q}\to \R$ as follows:
\begin{equation*}
\tau(v,r,\theta,\varphi):=v-\int_{r_+}^r \widetilde{\h}(r')\,dr'.
\end{equation*}
Denote $\Sigma_{\tau'}:=\{\tau=\tau'\}\subset \mathcal{M}_{M,Q}$. Then $\mathcal{M}_{Q,M}=\bigcup_{\tau'\in \R}\Sigma_{\tau'}$. Note that the level sets $\Sigma_{\tau}$ intersect $\mathcal{H}^+$ at $v=\tau$.

The tuple $(\tau,r,\theta,\varphi)$ is a well-defined coordinate chart on $\mathcal{M}_{Q,M}$. Now define $u: \mathring{\mathcal{M}}_{M,Q}\to \R$, with $\mathring{\mathcal{M}}_{M,Q}$ the interior of the manifold with boundary $\mathcal{M}_{M,Q}$, as follows:
\begin{equation*}
u=v-2r_*.
\end{equation*}
Since $r_*(r_{\sharp})=0$, we have that $u(\tau,r_{\sharp},\theta,\varphi)=v(\tau,r_{\sharp},\theta,\varphi)=\tau+\int_{r_+}^{r_{\sharp}} \widetilde{\h}(r')\,dr'$ and
\begin{equation*}
u(\tau,r,\theta,\varphi)=\tau+\int_{r_+}^{r_{\sharp}}  \widetilde{\h}(r')\,dr'+\int_{r_{\sharp}}^r\left[\widetilde{\h}(r')-2\Omega^{-2}\right](r')\,dr'=\tau+\int_{r_+}^{r_{\sharp}}  \widetilde{\h}(r')\,dr'-\int_{r_{\sharp}}^r\Omega^{-2}\h(r')\,dr'.
\end{equation*}
We conclude that $\lim_{r\to \infty}u(\tau,r,\theta,\varphi)$ is well-defined if $\Omega^{-2}\h$ is integrable in $[r_0,\infty)$, with $r_0>r_+$ arbitrary. We will therefore assume that $h=O(r^{-1-\delta})$ for some $\delta>0$. For the sake of convenience, we consider the also the following further restrictions: let $r_I>r_{\sharp}$, then
\begin{align}
\label{eq:foliassm1}
\h(r)=&\:0,\quad &r\geq r_I>r_{\sharp},\\
\label{eq:foliassm2}
\Omega^{-2}\h(r)=&\:h_0 r^{-2}+O_{\infty}(r^{-3}),\quad &h_0>0.
\end{align}
It will also be convenient to consider: let $r_+<r_H<r_{\sharp}$, then
\begin{align}
\label{eq:foliassm3}
\widetilde{\h}(r)=&\:0,\quad &r\leq r_H<r_{\sharp},\\
\label{eq:foliassm4}
\widetilde{\h}(r)>&\:0,\quad &r_+\leq r\leq r_{\sharp}
\end{align}
If \eqref{eq:foliassm1} holds, then $\Sigma_{\tau}\cap\{r>r_{I}\}$ are null hypersurfaces. Similarly, if \eqref{eq:foliassm3} holds, then $\Sigma_{\tau}\cap\{r<r_{H}\}$ are null hypersurfaces. If \eqref{eq:foliassm2} and \eqref{eq:foliassm4} hold, then $\Sigma_{\tau}$ are globally spacelike and \emph{asymptotically hyperboloidal}.

Then $g_{M,Q}$ attains the following form with respect to $(u,\rho_{\infty},\theta,\varphi)$ coordinates:
\begin{equation*}
g_{M,Q}=r^2(\Omega^2\rho_{\infty}^2du^2+2dud\rho_{\infty}+\slashed{g}_{\s^2}).
\end{equation*}
We moreover have that: $\mathcal{M}_{M,Q}\cong \mathcal{H}^+\cup \left(\R_u\times (0,r_+^{-1})_{\rho_{\infty}}\times \s^2\right)$. Then we extend $\mathcal{M}_{M,Q}$ by defining:
\begin{equation*}
\widehat{\mathcal{M}_{M,Q}}:=\mathcal{H}^+\cup (\R_u\times [0,r_+^{-1})_{\rho_{\infty}}\times \s^2).
\end{equation*}
We denote the level set $\{\rho_{\infty}=0\}$ of $\widehat{\mathcal{M}_{M,Q}}$ as follows:
\begin{equation*}
\mathcal{I}^+:=\{\rho_{\infty}=0\}\subset \widehat{\mathcal{M}_{M,Q}}.
\end{equation*}
We refer to $\mathcal{I}^+$ as \emph{future null infinity}.

Note that with the choices \eqref{eq:foliassm1} or \eqref{eq:foliassm2}, $(\tau,\rho_{\infty},\theta,\varphi)$ cover the extended manifold $\widehat{\mathcal{M}_{M,Q}}$ (away from half great circles on the spheres of constant $\rho_{\infty},\tau$, which are not covered by the chart $(\theta,\varphi)$).

We moreover denote with $S^2_{\tau',r'}$, $S^2_{v',r'}$ and $S^2_{v',r'}$ the 2-spheres that form the intersection of the level set $\{r=r'\}$ with the level sets $\{\tau=\tau'\}$, $\{v=v'\}$, or $\{u=u'\}$.

When $\kappa_+=0$, the map $(v, \rho_+,\theta,\varphi)\mapsto (u,\rho_{\infty},\theta,\varphi)$ is a \emph{conformal transformation} of $(\mathring{\mathcal{M}}_{M,Q},g_{M,Q})$, which is called a \emph{Couch--Torrence transformation}.

Define:
\begin{equation*}
s_{\infty}=r-r_+\quad \textnormal{and}\quad s_{+}=r_+^{2}(r-r_+)^{-1}=r_+^{2}s_{\infty}^{-1}.
\end{equation*}

 Note that $s_+,s_{\infty}\in (0,\infty)$. Then the Couch--Torrence transformation can then alternatively be identified with the map: $(v,s_+,\theta,\varphi)\mapsto (u, r_+^{2}s_+^{-1}, \theta,\varphi)$.

In slight abuse of notation, we will denote in the paragraph below by $\h,\widetilde{\h},r,\Omega^2:(0,\infty)_{s_+}\to [0,\infty)$ the composition functions defined as follows: $\h(s_+)=\h(r(s_+))$, $\widetilde{\h}(s_+)=\widetilde{\h}(r(s_+))$, $r(s_+)=r_+^2s_{+}^{-1}+r_+$ and $\Omega^2(s_+)=\Omega^2(r(s_+))=(r_+^{-1}s_++1)^{-2}$. In the case $|Q|=M$, it will be convenient to introduce the functions $\h,\widetilde{\h}^+: (0,\infty)\to [0,1]$, defined as follows:
\begin{align}
\label{eq:relhhplus1}
\widetilde{\h}^+(s_{\infty}):=&\:(r^2\h)(r_+^2s_{\infty}^{-1}),\\
\label{eq:relhhplus2}
\h^+(s_{\infty}):=&\: (\Omega^2\widetilde{\h})(r_+^2s_{\infty}^{-1})=2-\Omega^2(r_+^2s_{\infty}^{-1})\widetilde{\h}^+(s_{\infty}).
\end{align}
Note that 

When integrating over the unit  round sphere $\s^2$, we will make use of the notation:
\begin{equation*}
d\sigma:=\sin\theta d\theta d\varphi.
\end{equation*}

\subsection{Key vector fields}
The following vector fields will play key roles in the analysis in the remainder of the article. With respect to $(v,r,\theta,\varphi)$ coordinates:
\begin{align*}
T:=&\:\partial_v,\\
X:=&\: \partial_r+\widetilde{\h}\partial_v,\\
\Lbar:=&\: -\frac{\Omega^2}{2}\partial_r,\\
L:=&\: T-\Lbar,\\
Y_*:=&L-\underline{L}=\Omega^2X+(1-\widetilde{\h}\Omega^2)T.
\end{align*}
With respect to $(\tau,r,\theta,\varphi)$ coordinates, we have that $T=\partial_{\tau}$ and $X=\partial_r$.

We will also denote for $\q\in \R$:
\begin{align*}
K:=&\: T+i\q r_+^{-1}\mathbf{1},
\end{align*}
where we will define $\q$ according to \eqref{eq:deflittleq} below.
\subsection{Charged scalar field equation}
\label{sec:CSF}
The inhomogeneous charged scalar field equation for $\phi: \mathcal{M}_{M,Q}\to \C$ on $(\mathcal{M}_{M,Q},g_{M,Q})$ is:
\begin{equation}
\label{eq:CSF}
(g^{-1}_{M,Q})^{\mu \nu}(^AD_{\mu}) (^AD_{\nu})\phi=G_A,
\end{equation}
where $A\in \Omega^1(\mathcal{M}_{M,Q})$ satisfies $dA=F_Q$ and is called an \emph{electromagnetic gauge}, with $F_Q\in \Omega^2(\widehat{\mathcal{M}}_{M,Q})$ the Reissner--Nordstr\"om \emph{Faraday tensor}, which is defined as follows: $F_Q:=-\frac{Q}{r^2}dv\wedge dr$, with $Q\in \R$ the total \emph{electric charge} in the Reissner--Nordstr\"om black hole exterior. Furthermore,
\begin{equation*}
^AD:=\nabla-i\mathfrak{q} A\otimes (\cdot) ,
\end{equation*}
 is the \emph{electromagnetic gauge derivative}, with $\nabla$ the Levi-Civita covariant derivative corresponding to $g_{M,Q}$ with $\mathfrak{q}\in \R$ the \emph{scalar field charge parameter} or \emph{charge coupling constant}. We introduce the following (dimensionless) charge parameter:
 \begin{equation}
 \label{eq:deflittleq}
 \q:=\mathfrak{q}Q.
 \end{equation}
 
In \eqref{eq:CSF}, we applied the following notational convention:
 \begin{align*}
(^AD_{\mu}) (^AD_{\nu})T:=&\:((^AD)^2T)_{\mu \nu}\quad \textnormal{if $T$ is a tensor field}.
 \end{align*}
 It will moreover be convenient to denote:
 \begin{equation*}
  ^AD_{Y}:=\nabla_Y-i\mathfrak{q} A(Y)\mathbf{1}\quad \textnormal{for any vector field $Y$}.
 \end{equation*}
 
 While we are mainly interested in \eqref{eq:CSF} with $G_A\equiv 0$, we will consider also non-vanishing $G_A$ in part of the analysis.
 
In the remainder of the article, it will be more convenient to work with the quantity
\begin{equation*}
\psi:= r\cdot \phi.
\end{equation*}
Then \eqref{eq:CSF} is equivalent to:
 \begin{equation}
\label{eq:rCSF}
r(g^{-1}_{M,Q})^{\mu \nu}(^AD_{\mu}) (^AD_{\nu})(r^{-1}\psi)=rG_A.
\end{equation}
 
 Note that we can alternatively express $F_Q:=-\frac{Q}{r^2}du \wedge dr=Q du\wedge dx$ in $\widehat{\mathcal{M}}_{M,Q}\setminus \mathcal{H}^+$, so $F_Q\in \Omega^2(\widehat{\mathcal{M}}_{M,Q})$. We can choose $A\in  \Omega^1(\widehat{\mathcal{M}}_{M,Q})$, for example, we can consider \eqref{eq:foliassm1} or \eqref{eq:foliassm2} and take $A=\widehat{A}$ with:
 \begin{equation*}
\widehat{A}:=-\frac{Q}{r}d\tau.
\end{equation*}
\textbf{In most of the analysis in the present paper, we will restrict to $A=\widehat{A}$.}

In order to state the integrated energy estimates below, it will also be convenient to consider $A\notin  \Omega^1(\mathcal{M}_{M,Q})$ but $A\in  \Omega^1(\mathring{\mathcal{M}}_{M,Q})$. For example, we can consider $A=\widetilde{A}$ with
\begin{equation*}
\widetilde{A}:=-\frac{Q}{r}dt.
\end{equation*}
In that case, \eqref{eq:CSF} is only valid in $\mathring{\mathcal{M}}_{M,Q}$.

Suppose that $A'=A+df$ for some $f\in C^{\infty}(\mathring{\mathcal{M}}_{M,Q})$ and $\phi'=e^{i\mathfrak{q}f}\phi$. Then for any vector field $Y$, we have that:
\begin{equation*}
^{A'}D\phi'=e^{i \mathfrak{q}f}(^AD\phi).
\end{equation*}
Hence $\phi$ is a solution to \eqref{eq:CSF} if and only if $\phi'$ satisfies:
\begin{equation*}
(g^{-1}_{M,Q})^{\mu \nu}(^{A'}D)_{\mu}(^{A'}D)_{\nu}\phi'=e^{i \mathfrak{q}f}G_A=:G_{A'}.
\end{equation*}

We refer to $\phi'$ and $G_{A'}$ as \emph{gauge transformations} of $\phi$ and $G_A$ respectively. In particular, the norms $|\phi|$ and $|G_A|$ are gauge invariant.

We can relate the electromagnetic gauge choices $\widehat{A}$ and $\widetilde{A}$ as follows:
\begin{equation*}
\widetilde{A}=-\frac{Q}{r}d(v-r_*)=-\frac{Q}{r}(d\tau-\Omega^{-2}dr+\widetilde{\h}dr)=\widehat{A}+d\left(\int_{r_{\sharp}}^r\frac{Q}{r'}(\Omega^{-2}-\widetilde{\mathbbm{h}})(r')\,dr'\right).
\end{equation*}

An important role will be played by \eqref{eq:CSF} with $g_{M,Q}$ replaced by the Minkowski metric $m=g_{0,0}$ and $G_A\equiv 0$, but $A$ as in \eqref{eq:CSF}:
\begin{equation}
	\label{eq:CSFmink}
	(m^{-1})^{\mu \nu}(^AD_{\mu}) (^AD_{\nu})(r^{-1}\psi)=0.
\end{equation}
In other words, we consider the limits $M\downarrow 0$ and $Q\to 0$ in \eqref{eq:CSF}, but keep $\mathfrak{q}Q$ constant.

The following proposition provides more explicit expressions for \eqref{eq:CSF} with respect to $\widehat{A}$ in terms $\psi$:
\begin{proposition}
\label{prop:mainequations}
Let $\psi$ be a solution to \eqref{eq:rCSF} with $A=\widehat{A}$. Then:
\begin{multline}
\label{eq:maineqradfield}
rG_{\widehat{A}}=X(\Omega^2X\psi)+r^{-2}\slashed{\Delta}_{\s^2}\psi-\h \widetilde{\h} T^2\psi-2(1-\h)(T+i\q r^{-1})X\psi\\
+\left(\frac{d\h}{dr}-2i \q r^{-1} \h\widetilde{\h}\right)T\psi-\left[\frac{d\Omega^2}{dr} r^{-1}-i \q r^{-1}\frac{d\h}{dr}-\q^2  \mathbbm{h}\widetilde{\mathbbm{h}}r^{-2}-i \q r^{-2}(1-\h)\right]\psi.
\end{multline}
With respect to the coordinate chart $(\tau,\rho_{\infty},\theta,\varphi)$ we moreover obtain:
\begin{multline}
\label{eq:maineqradfieldconf}
r^2(rG_{\widehat{A}})=\partial_{\rho_{\infty}}(\Omega^2r^{-2}\partial_{\rho_{\infty}}\psi)+\slashed{\Delta}_{\s^2}\psi-\rho_{\infty}^{-2}\h \widetilde{\h} T^2\psi+2(1-\h)\partial_{\rho_{\infty}}T\psi+2i\q(1-\h)\rho_{\infty} \partial_{\rho_{\infty}}\psi\\
+\left(-\frac{d\h}{d\rho_{\infty}}-2i \q \rho_{\infty}^{-1}\h\widetilde{\h}\right)T\psi-\left[r\frac{d\Omega^2}{dr} +i\q \rho_{\infty}\frac{d\h}{d\rho_{\infty}}-\q^2 \mathbbm{h}\widetilde{\mathbbm{h}}-i \q (1-\h)\right]\psi
\end{multline}
and with respect to the coordinate chart $(\tau,\rho_+,\theta,\varphi)$, we obtain:
\begin{multline}
\label{eq:maineqradfieldconfhor}
r^2(rG_{\widehat{A}})=\partial_{\rho_+}(\Omega^2r^{-2}\partial_{\rho_+}\psi)+\slashed{\Delta}_{\s^2}\psi-r^2\h \widetilde{\h} K^2\psi-2(1-\h)\partial_{\rho_+}K\psi+ 2i\q (1-\h)\rho_+ \partial_{\rho_+}\psi\\
+\left(r^2\frac{d\h}{dr}+2i \q \rho_+ r^2\h\widetilde{\h}\right)K\psi-\left[r\frac{d\Omega^2}{dr} +i\q  \rho_+r^2 \frac{d\h}{dr}-\q^2 \rho_+^2r^2\mathbbm{h}\widetilde{\mathbbm{h}}-i \q (1-\h)\right]\psi.
\end{multline}
In particular, when $\kappa_+=0$, then
\begin{multline}
\label{eq:maineqradfieldhor}
r^2(rG_{\widehat{A}})=\partial_{\rho_{\infty}}(\Omega^2r^{-2}\partial_{\rho_{\infty}}\psi)+\slashed{\Delta}_{\s^2}\psi-\rho_{\infty}^{-2}\h \widetilde{\h} T^2\psi+2(1-\h)\partial_{\rho_{\infty}}T\psi+2i\q(1-\h)\rho_{\infty} \partial_{\rho_{\infty}}\psi\\
+\left(-\frac{d\h}{d\rho_{\infty}}-2i \q \rho_{\infty}^{-1} \h\widetilde{\h}\right)T\psi-\left[2\rho_+\rho_{\infty} +i\q \rho_{\infty}\frac{d\h}{d\rho_{\infty}}-\q^2  \mathbbm{h}\widetilde{\mathbbm{h}}-i \q (1-\h)\right]\psi\\
=\partial_{\rho_+}(\Omega^2r^{-2}\partial_{\rho_+}\psi)+\slashed{\Delta}_{\s^2}\psi-\rho_+^{-2}\h^+ \widetilde{\h}^+ K^2\psi+2(1-\h^+)\partial_{\rho_+}K\psi+2i(-\q)(1-\h_+)\rho_+ \partial_{\rho_+}\psi\\
+\left(-\frac{d\h^+}{d\rho_+}-2i (-\q) \rho_+^{-1}\h^+\widetilde{\h}^+\right)K\psi-\left[2\rho_{\infty}\rho_+  +i(-\q) \rho_+ \frac{d\h^+}{d\rho_+}-(-\q)^2  \mathbbm{h}^+\widetilde{\mathbbm{h}}^+-i (-\q) (1-\h^+)\right]\psi.
\end{multline}
Note the symmetry between \eqref{eq:maineqradfieldhor} and \eqref{eq:maineqradfieldconf}: we can obtain the expressions on the very RHS of \eqref{eq:maineqradfieldhor} from the expressions to the right of the first equality in \eqref{eq:maineqradfieldconf} by interchanging $\rho_+$ and $\rho_{\infty}$, replacing $\q$ with $-\q$, $T$ with $K$, $\h$ with $\h_+$ and $\widetilde{\h}$ with $\widetilde{\h}_+$.
\end{proposition}

\subsection{Spherical harmonics}
Let $f\in L^2(\s^2)$. Then we can decompose:
\begin{align*}
f=\sum_{\ell \in \N_0}\sum_{m\in \Z, |m|\leq \ell}f_{\ell m}Y_{\ell m},
\end{align*}
with $f_{\ell m}\in \C$ and $Y_{\ell m}\in L^2(\s^2)$ spherical harmonics, which form an orthonormal basis of  $L^2(\s^2)$, satisfying $\partial_{\varphi}Y_{\ell m}=imY_{\ell m}$ and $\slashed{\Delta}_{\s^2}Y_{\ell m}=-\ell(\ell+1)Y_{\ell m}$. We also denote:
\begin{align*}
f_{\ell}=&\:\sum_{m\in \Z, |m|\leq \ell}f_{ \ell m}Y_{\ell m},\\
f_{\geq \ell}=&\:\sum_{\ell'=\ell}^{\infty}\sum_{m\in \Z, |m|\leq \ell'}f_{ \ell' m}Y_{\ell' m}.
\end{align*}
We also apply the above notation to functions on $\mathcal{M}_{M,Q}$ or $\Sigma_{\tau}$.

The following parameter plays fundamental role in the present paper when considering $\psi_{\geq \ell}$, with $\psi$ a solution to \eqref{eq:CSF}:
\begin{equation*}
\beta_{\ell}:=\sqrt{4\ell(\ell+1)+1-4q^2}=\sqrt{(2\ell+1)^2-4q^2}.
\end{equation*}

\begin{lemma}
\mbox{}
\begin{enumerate}[label=\emph{(\roman*)}]
\item\emph{(Poincar\'e inequality on $\s^2$)} Let $f\in H^1(\s^2)$, then:
\begin{equation}
\label{eq:poincares2}
||f_{\geq \ell}||_{L^2(\s^2)}\leq \frac{1}{\sqrt{\ell(\ell+1)}}||f||_{H^1(\s^2)}.
\end{equation}
\item\emph{(Sobolev inequality on $\s^2$)} Let $f\in H^2(\s^2)$. Then there exists a numerical constant $C>0$, such that:
\begin{equation}
\label{eq:poincares2}
||f||_{L^{\infty}(\s^2)}\leq C ||f||_{H^2(\s^2)}\leq C\sum_{\ell\in \N_0}(\ell+1)^2||f_{\ell}||_{L^2(\s^2)}.
\end{equation}
\end{enumerate}
\begin{proof}
A straightforward consequence of the expansion $f=\sum_{\ell\in \N_0}f_{\ell}$.
\end{proof}
\end{lemma}

\subsection{Commutator vector fields}
Let $\gamma\in \N_0^3$. We introduce the following shorthand notations for differential operators acting on functions:
\begin{align}
\label{eq:commvf1}
D_{\mathbf{Z}}^{\gamma}:=&(^AD_{\s^2})_{\s^2}^{\gamma_1} ((r-M)(^AD)_{X})^{\gamma_2}(^A D_T)^{\gamma_3},\\
\label{eq:commvf2}
\mathbf{Z}^{\gamma}:=&\:\snabla_{\s^2}^{\gamma_1}((r-M)X)^{\gamma_2}T^{\gamma_3}.
\end{align}
with $\snabla_{\s^2}$ the Levi-Civita covariant derivative with respect to $\slashed{g}_{\s^2}$ and with $D_{\s^2}=\snabla_{\s^2}-i\mathfrak{q}\slashed{A}$, where $\slashed{A}\in \Omega^1(S^2_{\tau,r})$ satisfies $\slashed{A}=A_{\theta}d\theta+A_{\varphi}d\varphi$.

In particular, when $A=\widehat{A}$, we have that:
\begin{equation*}
D_{\mathbf{Z}}^{\gamma}=\snabla_{\s^2}^{\gamma_1} ((r-M)X)^{\gamma_2}D_T^{\gamma_3}.
\end{equation*}

Given a vector field $Y$ and a function $f$, we moreover apply the following notational conventions: for any $k\in \N$,
\begin{align*}
 D_Y D_{\mathbf{Z}}^{\gamma}f:=&\:(^AD_{\s^2})_{\s^2}^{\gamma_1} (^AD)_Y((r-M)(^AD)_{X})^{\gamma_2}(^A D_T)^{\gamma_3}f,\\
  D_{\s^2}^kD_{\mathbf{Z}}^{\gamma}f:=&\:(^AD_{\s^2})_{\s^2}^{\gamma_1+k} ((r-M)(^AD)_{X})^{\gamma_2}(^A D_T)^{\gamma_3}f,\\
  D_Y \mathbf{Z}^{\gamma}f:=&\:\snabla_{\s^2}^{\gamma_1} D_Y((r-M)X)^{\gamma_2}T^{\gamma_3}f,\\
    Y \mathbf{Z}^{\gamma}f:=&\:\snabla_{\s^2}^{\gamma_1} Y((r-M)X)^{\gamma_2}T^{\gamma_3}f,\\
   \snabla_{\s^2}^k \mathbf{Z}^{\gamma}f:=&\:\snabla_{\s^2}^{\gamma_1+k} Y((r-M)X)^{\gamma_2}T^{\gamma_3}f.
 \end{align*}

\subsection{Energy boundedness and integrated energy decay}
\label{sec:iedassm}
We first state a local existence and uniqueness result and then a global integrated energy estimate, coming from \cite{gaj26a}.
\begin{theorem}(Existence and uniqueness for the charged scalar field equation; \cite{gaj26a}[Theorem 2.5])
\label{thm:gwp}

Let $A\in \Omega^1(\mathcal{M}_{M,Q})$ and $G_A \in C^{\infty}(\mathcal{M}_{M,Q})$. Consider the initial data pair:
\begin{equation*}
(\upphi,T\upphi)\in C_c^{\infty}(\Sigma_0)\times C_c^{\infty}(\Sigma_0).
\end{equation*}
Then there exists a unique solution $\phi\in C^{\infty}(\mathcal{M}_{M,Q})$ to \eqref{eq:CSF} with respect to $A$, such that $(\phi|_{\Sigma_0}, T\phi|_{\Sigma_0})=(\upphi,T\upphi)$.

If we take $A\in \Omega^1(\widehat{\mathcal{M}}_{M,Q})$ and $G_A\in C^{\infty}(\widehat{\mathcal{M}}_{M,Q})$, then $\psi=r\phi\in C^{\infty}(\widehat{\mathcal{M}_{M,Q}})$.
\end{theorem}

We define the energy densities $\mathcal{E}_p[\psi]$ as follows:
\begin{equation*}
\mathcal{E}_p[\psi]:= (\Omega^{-1}r)^{p}\Omega^{2}|D_X\psi|^2+(\Omega^{-1}r)^{\min\{p,0\}}(\h\widetilde{\h}|D_T\psi|^2+r^{-2}|D_{\s^2}\psi|^2+r^{-2}|\psi|^2).
\end{equation*}
\textbf{Note that $\mathcal{E}_p[\psi]$ are gauge invariant!}

When $A=\widehat{A}$, we can express:
\begin{equation*}
\mathcal{E}_p[\psi]= (\Omega^{-1}r)^{p}\Omega^{2}|X\psi|^2+(\Omega^{-1}r)^{\min\{p,0\}}(\h\widetilde{\h}|D_T\psi|^2+r^{-2}|\snabla_{\s^2}\psi|^2+r^{-2}|\psi|^2).
\end{equation*}

Given $\gamma=(\gamma_1,\gamma_2,\gamma_3)\in \N_0^3$ consider the operators $\mathbf{Z}^{\gamma}$ and $D_{\mathbf{Z}}^{\gamma}$ as defined in \eqref{eq:commvf1} and \eqref{eq:commvf2}. Then we define the following higher-order energy densities:
\begin{align*}
\mathcal{E}_p[\mathbf{Z}^{\gamma}\psi]:=&\:(\Omega^{-1}r)^{p}\Omega^{2}|D_X \mathbf{Z}^{\gamma}\psi|^2+(\Omega^{-1}r)^{\min\{p,0\}}(\h\widetilde{\h}|D_T\mathbf{Z}^{\gamma}\psi|^2+|\slashed{D}_{\s^2}\mathbf{Z}^{\gamma}\psi|^2+r^{-2}|\mathbf{Z}^{\gamma}\psi|^2),\\
\mathcal{E}_p[D_{\mathbf{Z}}^{\gamma}\psi]:=&\:(\Omega^{-1}r)^{p}\Omega^{2}|D_X D_\mathbf{Z}^{\gamma}\psi|^2+(\Omega^{-1}r)^{\min\{p,0\}}(\h\widetilde{\h}|D_TD_\mathbf{Z}^{\gamma}\psi|^2+|\slashed{D}_{\s^2}D_\mathbf{Z}^{\gamma}\psi|^2+r^{-2}|D_\mathbf{Z}^{\gamma}\psi|^2).
\end{align*}

Note that $\mathcal{E}_p[D_{\mathbf{Z}}^{\gamma}\psi]$ are gauge invariant, but $\mathcal{E}_p[\mathbf{Z}^{\gamma}\psi]$ are not. When $A=\widehat{A}$, we can estimate for any $N\in \N_0$: 
\begin{equation*}
\sum_{|\gamma|\leq N}\mathcal{E}_p[D_{\mathbf{Z}}^{\gamma}\psi]\sim \sum_{|\gamma|\leq N}\mathcal{E}_p[\mathbf{Z}^{\gamma}\psi],
\end{equation*}
using that $D_X=X$, $D_{\s^2}=\snabla_{\s^2}$ and $D_T=T-iqr^{-1}\mathbf{1}$.

\begin{theorem}[Weighted energy boundedness and integrated energy decay; Theorem 3.1 in \cite{gaj26a}]
\label{thm:iedpaper1}
Let $\psi$ be a solution to \eqref{eq:rCSF}.  Let $\epsilon,\delta,\eta ,\kappa_0>0$, $N,\ell\in \N_0$ and $r_+<r_H<r_I<\infty$. Consider $1< p<\min\{1+\re \beta_{\ell},2\}+\epsilon$ and let $\upzeta$ be a smooth function satisfying $\upzeta(r)=1$ for $r\notin (r_{\sharp}-2\eta,r_{\sharp}+2\eta)$ and $\upzeta(r)=0$ for $r\in (r_{\sharp}-\eta,r_{\sharp}+\eta)$. 

Assume that $0\leq \kappa_+\leq \kappa_1$ or $|\q|\leq q_1$. Then, for suitably small $\kappa_1$ or $q_1$, there exist a constant\\ $C=C(q_1,\epsilon,\delta,\eta, \kappa_1, p, r_H,r_I,\h,N)>0$, such that for any $\tau_1\leq \tau_2$:
\begin{multline}
\label{eq:iedpaper1}
\sup_{\tau\in [\tau_1,\tau_2]}\sum_{k\leq N}\int_{\Sigma_{\tau}} \mathcal{E}_{p-2\epsilon}[D_T^k\psi_{\geq \ell}]\,d\sigma dr+\sum_{k_1+k_2+k_3=k}\int_{\tau_1}^{\tau_2}\int_{\Sigma_{\tau}\cap\{r_H\leq r\leq r_I\}} |D_{\s^2}^{k_1}D_{Y_*}^{k_2+1}D_T^{k_3}\psi_{\geq \ell}|^2+|D_{\s^2}^{k_1}D_{Y_*}^{k_2}D_T^{k_3}\psi_{\geq \ell}|^2\\
+\upzeta(r)(|D_{\s^2}^{k_1}D_{Y_*}^{k_2}D_T^{k_3+1}\psi_{\geq \ell}|^2+|D_{\s^2}^{k_1+1}D_{Y_*}^{k_2}D_T^{k_3}\psi_{\geq \ell}|^2)\,d\sigma dr d\tau\\
+\int_{\tau_1}^{\tau_2}\int_{\Sigma_{\tau}\setminus\{r_H\leq r\leq r_I\}} (\rho_++\kappa_+)r^{-1}(r^{-1}\rho_+)^{2\epsilon}\mathcal{E}_{p}[D_T^k\psi_{\geq \ell}]+(\rho_++\kappa_+)r^{-3}(r^{-1}\Omega)^{\delta-2}|D_T^{k+1}\psi_{\geq \ell}|^2\,d\sigma drd\tau\\
\leq C\sum_{k\leq N}\int_{\Sigma_{\tau_1}} \mathcal{E}_{p}[D_T^k\psi_{\geq \ell}]\,d\sigma dr\\
+\int_{\tau_1}^{\tau_2+r_+}\int_{\Sigma_{\tau}} \max\{(r^{-1}\Omega)^{-p}\rho_+^{1-2\epsilon}r^{-1+2\epsilon},1\}r^{-2}|r^3D_T^k(G_{A})_{\geq \ell}|^2+(1-\upzeta)|D_T^{k+1}(G_A)_{\geq \ell}|^2\,d\sigma dr \,d\tau.
\end{multline}
The estimate \eqref{eq:iedpaper1} remains valid when $\kappa_+>\kappa_1$ and $|\q|\geq \q_1$ under the assumption of Condition \ref{cond:quantmodestabp2} below.
\end{theorem}

We state the following condition, motivated in \cite{gaj26a}]\S 7.4]:

\begin{condition}[Quantitative mode stability on the real axis away from extremality]
\label{cond:quantmodestabp2}
Let $\omega_1$ and $\ell_1\in \N_0$ be arbitrary. Let $(\omega,\ell)\in \R \times \N_0$, with $|\omega|\leq \omega_1$ and $\ell\leq \ell_1$. Let $\kappa_+> \kappa_1$, with $\kappa_1$ the constant from Theorem \ref{thm:iedpaper1}.

Consider the ODE:
\begin{equation}
\label{eq:radialODE}
u''(r_*)+(\omega^2-V_{\omega  \ell}(r_*))u(r_*)=0,
\end{equation}
with
\begin{equation*}
	V_{\omega  \ell}(r_*(r))=\ell(\ell+1)\Omega^2(r)r^{-2}-{\q}^2\rho_{\infty}^2(r)+2{\q} \omega \rho_{\infty}(r)+r^{-1}\Omega^2(r)\frac{d\Omega^2 }{dr}
\end{equation*}
Let $u_+,u_{\infty}$ be solutions to \eqref{eq:radialODE} that are normalized as follows:
\begin{equation*}
\lim_{r_*\to -\infty}e^{i(\omega-\q r_+^{-1}) r_*+i \q \int_{0}^{r_*}\rho_+(r_*')\,dr_*'}u_+(r_*)=\lim_{r_*\to \infty}e^{-i\omega r_*+i \q \int_{0}^{r_*}\rho_{\infty}(r_*')\,dr_*'}u_{\infty}(r_*)=1.
\end{equation*}
and consider their corresponding Wronskian:
\begin{equation*}
\mathfrak{W}(\omega,\ell):=u_{+}(r_*)\frac{du_{\infty}}{dr_*}(r_*)-u_{\infty}(r_*)\frac{du_{+}}{dr_*}(r_*).
\end{equation*}

 Then there exists a constant $K=K(\omega_1,\ell_1,\kappa_1)>0$ such that:
\begin{equation*}
|\mathfrak{W}|^{-1}\leq  K.
\end{equation*}
\end{condition}

\begin{remark}
Note that Condition \ref{cond:quantmodestabp2} can be interpreted as a quantitative version of a \emph{mode stability} statement on the real axis, i.e.\ the absence of purely oscillating mode solutions to \eqref{eq:CSF} away from extremality; see also the discussion in \S \ref{sec:intromodeinstb}.
\end{remark}

\subsection{Further notation}
We appeal to standard notational conventions regarding multiplications of uniform constants $C$ and inequalities $\lesssim$, $\gtrsim$ and $\sim$; we refer the reader \cite{gaj26a}[\S 2.2] for a detailed overview.

\section{Stationary solutions}
\label{sec:statsol}
Before we give a precise formulation of the main theorem, we first introduce the notion of stationary (time-independent) solutions to \eqref{eq:rCSF} with $\q\neq 0$, $A=\widehat{A}$ and $G_{A}=0$, which play a key role in the main theorems.
\begin{proposition}
\label{prop:statsol}
	Let $\ell\in \N_0$, $m\in \Z$ with $|m|\leq \ell$. Fix $r_+=1$. Let $\mathfrak{w}_{\ell }, \mathfrak{w}^{\pm}_{\ell }: [r_+,\infty)\to \C$, such that $\mathfrak{w}_{\ell }Y_{\ell m}$ and $\mathfrak{w}^{\pm}_{\ell }Y_{\ell m}$ are solutions to \eqref{eq:rCSF} with $A=\widehat{A}$ and $G_A=0$. We write $\mathfrak{w}^{\pm}(r;\q)$ and  $\mathfrak{w}(r;\q)$ when we need to emphasize the value of the corresponding charge parameter $\q$.
	
	Then:
	\begin{enumerate}[label=\emph{(\roman*)}]
	\item We can fix $\mathfrak{w}_{\ell }, \mathfrak{w}^{\pm}_{\ell }$ uniquely by setting:
	\begin{align*}
		\mathfrak{w}_{\ell }(r_+)=&\:1,\\
		\mathfrak{w}^{-}_{\ell }(r)=&\:r^{\frac{1}{2}+i\q- \frac{1}{2}\beta_{\ell}}(1+O(r^{-1}))\quad (r>2),\\
		\mathfrak{w}^{+}_{\ell }(r)=&\:r^{\frac{1}{2}+i\q+ \frac{1}{2}\beta_{\ell}}(1+O(r^{-1}))\quad (\beta_{\ell}\neq 0)\quad (r>2),\\
		\mathfrak{w}^{+}_{\ell }(r)=&\:r^{\frac{1}{2}+i\q}\log r(1+O(r^{-1}))\quad (\beta_{\ell}= 0)\quad (r>2).
	\end{align*}
	\item We can express:
	\begin{equation*}
		\mathfrak{w}_{\ell }={\alpha}_+\mathfrak{w}_{\ell }^++{\alpha}_-\mathfrak{w}_{\ell }^-,
	\end{equation*}
	with ${\alpha}_{\pm}\in \C$, such that in the case $\beta_{\ell}\in i(0,\infty)$:
		\begin{multline}
		\label{eq:alphaminplusratiobetaim}
		\sign(q)\log\left|\frac{ \alpha_-}{\alpha_+}\right|=-\frac{\pi}{2}\left(\im \beta_{\ell}-\frac{1}{\pi}\log \left(\frac{\cosh (\pi(|q|-\frac{1}{2}\im \beta_{\ell}))}{\cosh (\pi(|q|+\frac{1}{2}\im \beta_{\ell}))}\right)\right)\\
		+\sign(q)\frac{\pi}{2}\left( \im \beta_{\ell}-\frac{1}{\pi}\log \left(\frac{\cosh(\pi(\frac{|\q|}{2}\frac{r_-+r_+}{r_+^2\kappa_+}+\frac{1}{2}\im \beta_{\ell}))\cosh^2(\pi(|\q|-\frac{1}{2}\im \beta_{\ell}))}{\cosh(\pi(\frac{|\q|}{2}\frac{r_-+r_+}{r_+^2\kappa_+}-\frac{1}{2}\im \beta_{\ell}))\cosh^2(\pi(|\q|+\frac{1}{2}\im \beta_{\ell}))}\right)\right)\\
		\substack{\kappa_+\downarrow 0\\ \to}-\frac{\pi}{2}\left(\im \beta_{\ell}+\frac{1}{\pi}\log \left(\frac{\cosh (\pi(|q|-\frac{1}{2}\im \beta_{\ell}))}{\cosh (\pi(|q|+\frac{1}{2}\im \beta_{\ell}))}\right)\right).
	\end{multline}
	\item Denote with $\alpha_{\pm}(\beta_{\ell})$ the values of $\alpha_{\pm}$ when $\beta_{\ell}\neq 0$ and $q$ is fixed. Then $\alpha_-(\beta_{\ell})=\alpha_+(-\beta_{\ell})$ and we have that:
\begin{equation*}
	\lim_{\beta_{\ell}\to 0}\mathfrak{w}_{\ell}(r)=\lim_{\beta_{\ell}\to 0}(\alpha_+\beta_{\ell})r^{\frac{1}{2}+iq}\left[\log r+\frac{d}{d\beta_{\ell}}\bigg|_{\beta_{\ell}=0}\frac{\alpha_{-}(\beta_{\ell})}{\alpha_{+}(\beta_{\ell})}\right].
\end{equation*}
Therefore $\alpha_{\pm}$ take the following values when $\beta_{\ell}=0$:
\begin{align*}
	\alpha_{+}=&\:-\lim_{\beta_{\ell}\to 0}(\alpha_+(\beta_{\ell})\beta_{\ell}),\\
	\alpha_{-}=&\:\alpha_+ \frac{1}{i}\frac{d}{ds}\bigg|_{s=0}\frac{\alpha_{-}(is)}{\alpha_{+}(is)}.
\end{align*}
In particular,
\begin{multline}
\label{eq:alphaminplusratiobetazero}
	\sign(\q)\im \left(\frac{\alpha_-}{\alpha_+}\right)=-\sign(\q)\re\left(\frac{d}{ds}\Big|_{s=0} \left(\frac{\alpha_-(is)}{\alpha_+(is)}\right)\right)=-\sign(q)\frac{ \pi}{2}(1-\tanh(\pi|q|))\\
	+\frac{\pi}{2}\sign(\q)\left(1 -\tanh\left(\pi \frac{|\q|}{2}\frac{r_-+r_+}{r_+^2 \kappa_+}\right)-2\tanh(\pi |q|)\right)\\
	\substack{\kappa_+\downarrow 0\\ \to} -\sign(q)\frac{ \pi}{2}(1+\tanh(\pi|q|).
\end{multline}
	\item The function $\mathfrak{w}_{\ell }$ is non-vanishing.
	\end{enumerate}
\end{proposition}
\begin{proof}
Let
\begin{align*}
	w_{\ell }:=&\:e^{-i\q\int_{r^{\sharp}}^{r}r'^{-1}\,dr'+i\q\int_{r^{\sharp}}^{r}r'^{-1}\Omega^{-2}(r')(\Omega^2-1+\h)(r')\,dr'}\mathfrak{w}_{\ell },\\
	w_{\ell }^{\pm}:=&\:e^{-iq\int_{r^{\sharp}}^{r}r'^{-1}\,dr'+i\q\int_{r^{\sharp}}^{r}r'^{-1}\Omega^{-2}(r')(\Omega^2-1+\h)(r')\,dr'}\mathfrak{w}_{\ell }^{\pm}.
	\end{align*}
	Then $w_{\ell }Y_{\ell m}$ and $w_{\ell m}^{\pm}Y_{\ell m}$ are solutions to \eqref{eq:CSF} with $A=\widetilde{A}$ and $G_A=0$. In other words, they are solutions to \eqref{eq:maineqradfield} with $\widetilde{\h}=\Omega^{-2}$.

Note that we can equivalently express:
\begin{align*}
	w_{\ell }:=&\:e^{i\q r_+^{-1}r_*-i\q\int_{r^{\sharp}}^{r}(\Omega^{-2} \rho_+)(r')\,dr'-i\q\int_{r_{\sharp}}^{r}r'^{-1}\widetilde{\h}(r')\,dr'}\mathfrak{w}_{\ell },\\
	w_{\ell }^{\pm}:=&\:e^{i\q r_+^{-1}r_*-i\q\int_{r^{\sharp}}^{r}(\Omega^{-2} \rho_+)(r')\,dr'-i\q\int_{r_{\sharp}}^{r}r'^{-1}\widetilde{\h}(r')\,dr'}\mathfrak{w}_{\ell }^{\pm}.
	\end{align*}
	
	 Then \eqref{eq:maineqradfield} with $\widetilde{\h}=\Omega^{-2}$ is equivalent to:
	\begin{equation}
	\label{eq:relevanteqforstatsoln}
		\frac{d}{dr}\left(\Omega^2 r^2 \frac{d}{dr}(r^{-1}w_{\ell})\right)+\left[\q^2 r^2(r^{2}\Omega)^{-2}-\ell(\ell+1)\right]r^{-1}w_{\ell}=0.
	\end{equation}
	Assume for convenience that $r_+=1$. 
	
	We first consider the case $\kappa_+=0$. By changing variables to $z=\frac{2i\q}{r-1}$, we obtain the equation:
	\begin{equation*}
		\frac{d^2(r^{-1}w_{\ell})}{dz^2}+\left[-\frac{1}{4}-\frac{i\q}{z}+\frac{\q^2-\ell(\ell+1)}{z^2}\right]r^{-1}w_{\ell}=0.
	\end{equation*}
	This is the \emph{Whittaker equation}; see for example \cite{NIST:DLMF}[\S 13.4]. Using the asymptotic properties of Whittaker functions as $|z|\to \infty$, it follows that: we can express:
	\begin{equation*}
		r^{-1}w_{\ell}(r)=c_{\q}\cdot W_{-i\q,\frac{\beta_{\ell}}{2}}(2i\q(r-1)^{-1}),
	\end{equation*}	
	for some constant $c_{\q}\in \C\setminus\{0\}$. The Whittaker function $W_{-i\q,\frac{\beta_{\ell}}{2}}(2i\q(r-1)^{-1})$ has the following asymptotics as $z\to 0$ when $\beta_{\ell}\in i(0,\infty)$ (see \cite{NIST:DLMF}[\S 13.14]):
	\begin{multline*}
		c_{\q}^{-1}w_{\ell}(r)=r W_{-i\q,\frac{\beta_{\ell}}{2}}(2i\q(r-1)^{-1})=\frac{\Gamma(\beta_{\ell})}{\Gamma(\frac{1}{2}+\beta_{\ell}+i\q)}(2i\q)^{\frac{1}{2}-\frac{\beta_{\ell}}{2}}r^{\frac{1}{2}+\frac{\beta_{\ell}}{2}}(1+O(r^{-1}))\\
		+(2i\q)^{\frac{1}{2}+\frac{\beta_{\ell}}{2}}\frac{\Gamma(-\beta_{\ell})}{\Gamma(\frac{1}{2}-\beta_{\ell}+i\q)}r^{\frac{1}{2}-\frac{\beta_{\ell}}{2}}(1+O(r^{-1}))\\
		=\frac{\Gamma(\beta_{\ell})}{\Gamma(\frac{1}{2}+\beta_{\ell}+i\q)}(2i\q)^{\frac{1}{2}}|2\q|^{\frac{i}{2}\im \beta_{\ell}}e^{-\frac{\pi}{4}\sign(\q)\im \beta_{\ell}}r^{\frac{1}{2}+\frac{\beta_{\ell}}{2}}(1+O(r^{-1}))\\
		+(2i\q)^{\frac{1}{2}}|2\q|^{-\frac{i}{2}\im \beta_{\ell}}e^{\frac{\pi}{4}\sign(\q)\im \beta_{\ell}}\frac{\Gamma(-\beta_{\ell})}{\Gamma(\frac{1}{2}-\beta_{\ell}+i\q)}r^{\frac{1}{2}-\frac{\beta_{\ell}}{2}}(1+O(r^{-1})).
	\end{multline*}
	We therefore conclude that
	\begin{equation*}
		\frac{\alpha_-}{\alpha_+}=|2\q|^{-i\im \beta_{\ell}}\frac{\Gamma(-\beta_{\ell})}{\Gamma(\beta_{\ell})}\frac{\Gamma(\frac{1}{2}+\beta_{\ell}+i\q)}{\Gamma(\frac{1}{2}-\beta_{\ell}+i\q)}e^{-\frac{\pi}{2}\sign(\q).\im \beta_{\ell}}
	\end{equation*}
	
	From this, it follows that:
		\begin{multline*}
		\log\left|\frac{ \alpha_-}{\alpha_+}\right|=-\frac{\pi}{2}\sign(q)\im \beta_{\ell}+\frac{1}{2}\log \left|\frac{\Gamma(-\beta_{\ell})}{\Gamma(+\beta_{\ell})}\right|^2+\frac{1}{2}\log \left|\frac{\Gamma(\frac{1}{2}+\frac{1}{2}\beta_{\ell}+iq)}{\Gamma(\frac{1}{2}-\frac{1}{2}\beta_{\ell}+iq)}\right|^2\\
		=-\frac{\pi}{2}\sign(q)\im \beta_{\ell}+\frac{1}{2}\sign(q)\log \left(\frac{\cosh(\pi(|q|-\frac{1}{2}\im \beta_{\ell}))}{\cosh(\pi(|q|+\frac{1}{2}\im \beta_{\ell}))}\right).
	\end{multline*}

	Now consider the case $\kappa_+>0$. Fix $r_+=1$. Write $2\kappa_+:=r_+-r_-$. Then we introduce the following change of variable and rescalings:
	\begin{align*}
		z=&\:\frac{r-r_+}{r-r_-},\\
		r^{-1}w_{\ell}(r(z))=&\:z^{i\frac{\q}{2} \kappa_+^{-1}}(1-z)^{\frac{1}{2}-\beta_{\ell}}F(z).
	\end{align*}
to obtain from \eqref{eq:relevanteqforstatsoln} the following equation for $F$:
\begin{equation*}
	z(1-z)\frac{d^2F}{dz^2}+(c-(a+b+1)z)\frac{dF}{dz}-abF=0,
\end{equation*}
with $a=i\frac{\q}{2}\kappa_+^{-1}(1-r_-)+\frac{1-\beta_{\ell}}{2}$, $b=i\frac{\q}{2}\kappa_+^{-1}(1+r_-)+\frac{1-\beta_{\ell}}{2}$, $c=1+i\q\kappa_+^{-1}$. The above ODE is a hypergeometric equation, see \cite{NIST:DLMF}[\S 15.10], and we can therefore write:
\begin{equation*}
	r^{-1}w_{\ell}(r)=c_{\q,\kappa_+}\left(\frac{r-1}{r-r_-}\right)^{i\frac{\q}{2} \kappa_+^{-1}}\left(\frac{1-r_-}{r-r_-}\right)^{\frac{1}{2}-\beta_{\ell}} {}_2\mathbf{F}_1\left(a,b,c\; \frac{r-1}{r-r_-}\right),
\end{equation*}
with ${}_2\mathbf{F}_1\left(a,b,c\; \frac{r-r_+}{r-r_-}\right)$ the regularized Gauss hypergeometric function. It is straightforward to verify that $\mathfrak{w}_{\ell}$ is smooth at $r=1$.

Furthermore, by applying \cite{NIST:DLMF}[\S 15.8.4] we obtain the following large-$r$ asymptotics ($z\uparrow 1$) :
\begin{multline*}
	\pi(\sin(\pi \beta_{\ell}))^{-1}c_{\q,\kappa_+}^{-1}w_{\ell}(r)=\frac{r\left(\frac{r-r_-}{1-r_-}\right)^{-\frac{1}{2}+\frac{\beta_{\ell}}{2}}}{\Gamma(\frac{1}{2}+\frac{\beta_{\ell}}{2}+i\frac{\q}{2}(1+r_-)\kappa_+^{-1})\Gamma(\frac{1}{2}+\frac{\beta_{\ell}}{2}+i\q)\Gamma(1-\beta_{\ell})}(1+O(r^{-1}))\\
	-\frac{r\left(\frac{r-r_-}{1-r_-}\right)^{-\frac{1}{2}-\frac{\beta_{\ell}}{2}}}{\Gamma(\frac{1}{2}-\frac{\beta_{\ell}}{2}+i\q)\Gamma(\frac{1}{2}-\frac{\beta_{\ell}}{2}+i\frac{\q}{2} (1+r_-)\kappa_+^{-1})\Gamma(1+\beta_{\ell})}(1+O(r^{-1})).
\end{multline*}
We therefore conclude that:
\begin{equation*}
	\frac{\alpha_-}{\alpha_+}=(1-r_-)^{\beta_{\ell}}\frac{\Gamma(1-\beta_{\ell})}{\Gamma(1+\beta_{\ell})}\frac{\Gamma(\frac{1}{2}+\frac{\beta_{\ell}}{2}+i\q)}{\Gamma(\frac{1}{2}-\frac{\beta_{\ell}}{2}+i\q)}\frac{\Gamma(\frac{1}{2}+\frac{\beta_{\ell}}{2}+i\frac{\q}{2}(1+r_-)\kappa_+^{-1})}{\Gamma(\frac{1}{2}-\frac{\beta_{\ell}}{2}+i\frac{\q}{2}(1+r_-)\kappa_+^{-1})}
\end{equation*}
	From this, it follows that for $\beta_{\ell}\in i(0,\infty)$:
		\begin{multline*}
		\log\left|\frac{ \alpha_-}{\alpha_+}\right|=\frac{1}{2}\log \left|\frac{\Gamma(1-\beta_{\ell})}{\Gamma(1+\beta_{\ell})}\right|^2+\frac{1}{2}\log \left|\frac{\Gamma(\frac{1}{2}+\frac{1}{2}\beta_{\ell}+i\q)}{\Gamma(\frac{1}{2}-\frac{1}{2}\beta_{\ell}+i\q)}\right|^2+\frac{1}{2}\log \left|\frac{\Gamma(\frac{1}{2}+\frac{1}{2}\beta_{\ell}+i\frac{\q}{2}(1+r_-)\kappa_+^{-1})}{\Gamma(\frac{1}{2}-\frac{1}{2}\beta_{\ell}+i\frac{\q}{2}(1+r_-)\kappa_+^{-1})}\right|^2\\
		=\frac{1}{2}\sign(q)\log \left(\frac{\cosh(\pi(|\q|-\frac{1}{2}\im \beta_{\ell}))}{\cosh(\pi(|\q|+\frac{1}{2}\im \beta_{\ell}))}\right)+\sign(q)\frac{1}{2}\log \left(\frac{\cosh(\pi(\frac{|\q|}{2}\frac{r_-+1}{\kappa_+}-\frac{1}{2}\im \beta_{\ell}))}{\cosh(\pi(\frac{|\q|}{2}\frac{r_-+1}{\kappa_+}+\frac{1}{2}\im \beta_{\ell}))}\right)\\
		=-\frac{\pi}{2}\sign(q)\im \beta_{\ell}-\frac{1}{2}\sign(q)\log \left(\frac{\cosh(\pi(|\q|-\frac{1}{2}\im \beta_{\ell}))}{\cosh(\pi(|\q|+\frac{1}{2}\im \beta_{\ell}))}\right)\\
		+\frac{1}{2}\sign(q)\left(\pi \im \beta_{\ell}-\log \left(\frac{\cosh(\pi(\frac{|\q|}{2}\frac{r_-+1}{\kappa_+}+\frac{1}{2}\im \beta_{\ell}))\cosh^2(\pi(|\q|-\frac{1}{2}\im \beta_{\ell}))}{\cosh(\pi(\frac{|\q|}{2}\frac{r_-+1}{\kappa_+}-\frac{1}{2}\im \beta_{\ell}))\cosh^2(\pi(|\q|+\frac{1}{2}\im \beta_{\ell}))}\right)\right).
	\end{multline*}
	Note that the last term vanishes in the extremal limit $\kappa_+\to 1$.

Now consider (iii). We apply Taylor's theorem to obtain:
\begin{equation*}
	\frac{1}{2}\log \left(\frac{\cosh(\pi(|q|-\frac{1}{2}\im \beta_{\ell}))}{\cosh(\pi(|q|+\frac{1}{2}\im \beta_{\ell}))}\right)=-\frac{\pi}{2} \tanh(\pi |q|)\im \beta_{\ell}+O(|\beta|^2).
\end{equation*}
	Therefore,
	\begin{multline*}
		\re\left(\frac{d}{ds}\Big|_{s=0} \left(\frac{\alpha_-(is)}{\alpha_+(is)}\right)\right)=-\re\left(\frac{d}{ds}\Big|_{s=0}\log \left(\frac{\alpha_-(is)}{\alpha_+(is)}\right)\right)=-	\frac{d}{ds}\Big|_{s=0}\log \left|\frac{\alpha_-(is)}{\alpha_+(is)}\right|\\
		=+\frac{\pi}{2}\sign(q)-\pi \sign(q)\tanh(\pi |q|)+\frac{1}{2}\pi\sign(\q)\left(1 -2\tanh\left(\pi \frac{|\q|}{2}\frac{r_-+1}{\kappa_+}\right)-4\tanh(\pi |q|)\right).
	\end{multline*}
	
	We now rewrite for $s=\im \beta_{\ell}$:
\begin{multline*}
	\mathfrak{w}_{\ell}(r)=\frac{1}{i}\alpha_+(is) s r^{\frac{1}{2}+iq}\left[\frac{1}{s}e^{\frac{is}{2}\log r}+\frac{1}{s}\frac{\alpha_-(is)}{\alpha_+(is)}e^{-i\frac{s}{2}\log r }\right]\\
	=\frac{1}{i}\alpha_+(is)s r^{\frac{1}{2}+iq}e^{i\frac{s}{2}\log r}\left[\frac{1}{s}(e^{+i\frac{s}{2}\log r }-1)-\frac{1}{s}(e^{-i\frac{s}{2}\log r }-1)+\frac{1}{s}\left(\frac{\alpha_-(is)}{\alpha_+(is)}-(-1)\right)e^{-i\frac{s}{2}\log r }\right]
\end{multline*}
	By (ii), we have that $\lim_{\beta_{\ell}\to 0}\beta_{\ell}\alpha_+(\beta_{\ell})=-\lim_{\beta_{\ell}\to 0}\beta_{\ell}\alpha_-(\beta_{\ell})$ is well-defined, so after taking the limit $s \downarrow 0$, we obtain:
	\begin{equation*}
	\lim_{s\to 0}\mathfrak{w}_{\ell}(r)=-\lim_{s\to 0}(\alpha_+(is)is)r^{\frac{1}{2}+iq}\left[\log r+\frac{1}{i}\frac{d}{ds}\bigg|_{s=0}\frac{\alpha_{-}(is)}{\alpha_{+}(is)}\right].
\end{equation*}
Since this limit must correspond to a stationary solution in the $\beta_{\ell}=0$ case, we conclude the identities for $\alpha_{\pm}$ when $\beta_{\ell}=0$.
	
	Finally, we turn to (iv). Write $w=w_{\ell }$. Then for $(\cdot)':=\frac{d}{dr_*}$, we have that $w''(r_*)=V_{\ell }(r(r_*))w(r_*)$, with
	\begin{equation*}
	V_{\ell}(r)=\ell(\ell+1)\Omega^2r^{-2}-{\q}^2r^{-2}+\Omega^2r^{-1}\frac{d\Omega^2 }{dr}.
	\end{equation*}
	This follows by \cite{gaj26a}[Proposition 5.1]. Since $V_{\ell}$ is real-valued, we obtain $\re(i\overline{w} w'')=0$, so:
	\begin{equation*}
	\frac{d}{dr_*}\re\left(i\overline{w}w'\right)=0.
	\end{equation*}
	Then 
	\begin{equation*}
	\lim_{r_*\to- \infty}\re\left(i\overline{w}w'\right)(r_*)=-\q r_+^{-1}|w|(-\infty)=-\q r_+^{-1}.
	\end{equation*} 
	Hence, $\re\left(i\overline{w}w'\right)(r_*)=qr_+^{-1}\neq 0$ for all $r_*\in \R$, which implies that $\mathfrak{w}_{\ell}$ is non-vanishing.
\end{proof}

\begin{remark}
For $\q\neq 0$, the stationary solution $\mathfrak{w}_{\ell}^{-}Y_{\ell m}$ is not smooth in $\widehat{\mathcal{M}}_{M,Q}$ at $\mathcal{I}^+=\{x=0\}$. This stands in contrast with the $q=0$ case, where $r^{-1}\mathfrak{w}_{\ell}^{-}(r)=Q_{\ell}(r_+^{-2}\kappa_+^{-1}(r-r_+)+1)$ is a Legendre polynomial of the second kind when $\kappa_+>0$ and $r^{-1}\mathfrak{w}_{\ell}^{-}(r)=(r-r_+)^{-1-\ell}$ when $\kappa_+=0$, so $\mathfrak{w}_{\ell}^{\pm }Y_{\ell m}$ are both smooth at $\mathcal{I}^+$.
\end{remark}

\begin{corollary}
\label{cor:kappanotsmallstatsoln}
When $\beta_{\ell}\in i[0,\infty)$, the following estimates hold for all $\kappa_+\geq 0$: there exists an $\eta_0>0$, such that:
\begin{align}
\label{eq:condimagbeta}
\sign(q) \log \left|\frac{\alpha_-}{\alpha_+}\right|\leq&-\frac{\pi}{2}\im \beta_{\ell}-\eta_0 \im \beta_{\ell}-\frac{1}{2}\log \left(\left|\frac{\cosh(\pi(|q|-\frac{1}{2}\im \beta_{\ell})}{\cosh(\pi(|q|+\frac{1}{2}\im \beta_{\ell})}\right|\right)\quad &(\beta_{\ell}\in i(0,\infty)),\\
\label{eq:condzerobeta}
	\sign(q) \im \left(\frac{\alpha_-}{\alpha_+}\right)\leq &-\frac{\pi}{2}\left(1-\tanh(\pi |q|)\right)+\eta_0 \quad &(\beta_{\ell}=0).
\end{align}
\end{corollary}
\begin{proof}
	Consider \eqref{eq:alphaminplusratiobetazero}. Using the monotonicity of $\tanh$ and the fact that $|q|\geq \frac{1}{2}$, we obtain the desired inequality by using that:
	\begin{equation*}
		1-3\tanh\left(\frac{\pi}{2}\right)=\frac{2(2-e^{\pi})}{e^{\pi}+1}<0,
	\end{equation*}
	using that $e^{\pi}>e^3>e^{\log 2}$.
	
	Now consider \eqref{eq:alphaminplusratiobetaim}. Write $\beta_{\ell}=is$. Note that:
	\begin{multline*}
	\im \beta_{\ell}-\frac{1}{\pi}\log \left(\frac{\cosh(\pi(|\q|\frac{r_-+1}{1-r_-}+\frac{1}{2}\im \beta_{\ell}))\cosh^2(\pi(|\q|-\frac{1}{2}\im \beta_{\ell}))}{\cosh(\pi(|\q|\frac{r_-+1}{1-r_-}-\frac{1}{2}\im \beta_{\ell}))\cosh^2(\pi(|\q|+\frac{1}{2}\im \beta_{\ell}))}\right) \leq  \im \beta_{\ell}-\frac{3}{\pi}\log\left(\frac{\cosh(\pi(|\q|+\frac{1}{2}\im \beta_{\ell}))}{\cosh(\pi(|\q|-\frac{1}{2}\im \beta_{\ell}))}\right)\\
	\to 0 \quad \im \beta_{\ell}\to 0.
		\end{multline*}
		Furthermore,
		\begin{multline*}
			\frac{d}{ds}\left(s-\frac{3}{\pi}\log\left(\frac{\cosh(\pi(|\q|+\frac{1}{2}s))}{\cosh(\pi(|\q|-\frac{1}{2}s))}\right)\right)=1-\frac{3}{2}\tanh(\pi(|\q|+s))-\frac{3}{2}\tanh(\pi(|\q|-s))\\
			\leq 1-3 \tanh(\pi |\q|)\leq 1-3 \tanh(\frac{\pi}{2})<0,
		\end{multline*}
		where we used that $2|q|-|\im \beta_{\ell}|>0$ when $\im \beta_{\ell}\neq 0$ to arrive at the first inequality above.
		Hence,
		\begin{equation*}
			\left(\im \beta_{\ell}-\frac{3}{\pi}\log\left(\frac{\cosh(\pi(|\q|+\frac{1}{2} \beta_{\ell}))}{\cosh(\pi(|\q|-\frac{1}{2} \beta_{\ell}))}\right)\right)<0. \qedhere
		\end{equation*}
\end{proof}

\section{Precise statements of the main theorems}
\label{sec:precstatthm}
In this section, we give precise statements of the main results proved in the paper. Theorem \ref{thm:mainthmpoint} is the main result of the present paper and it establishes the precise, global late-time asymptotics of solutions $\psi$ to \eqref{eq:CSF}. The theorem illustrates how the late-time asymptotics are dictated by a global function $\Psi$ that have a null infinity contribution $\Psi^{\infty}$ and, in the $|Q|=M$ case, also an event horizon contribution $\Psi^{+}$ and are closely related to solutions to \eqref{eq:CSFmink}.

In the $|Q|=M$ case, when considering $\psi|_{\mathcal{I}^+}$, $\Psi^{\infty}_{\ell m}$ dominate the late-time asymptotics, whereas for $\psi|_{\mathcal{H}^+}$, $\Psi^{+}_{\ell m}$ dominate the late-time asymptotics. We also show that when considering instead $D_T\psi$, the reverse is true: $\Psi^{+}_{\ell m}$ now dominate the late-time asymptotics of $D_T\psi|_{\mathcal{I}^+}$, whereas $\Psi^{+}_{\ell m}$ dominate the late-time asymptotics of $D_T\psi|_{\mathcal{H}^+}$.
\begin{theorem}[Late-time pointwise asymptotics]
\label{thm:mainthmpoint}
 Let $\psi$ be a solution to \eqref{eq:rCSF} with $\mathfrak{q}\neq 0$ and\\ $A=\widehat{A}=-Qr^{-1}d\tau$, and write $\q=\mathfrak{q}Q$. Let $r_H>r_+$. Let $\beta_{\ell}=\sqrt{(2\ell+1)^2-4q^2}$. If $\kappa_+>\kappa_1$ and $|\q|>q_1$, assume moreover that Condition \ref{cond:quantmodestabp2} holds.
\begin{enumerate}[label=\emph{(\roman*)}]
\item Let $N_1,N_2\in \N_0$, $l\in \{0,1\}$ and $\nu,\epsilon>0$. For arbitrary small $\nu>0$, there exists a constant $C=C(\nu,\epsilon, N_1,N_2)>0$, such that:
\begin{align}
\label{eq:mainthmpoint1}
	\sum_{|\gamma|\leq  N_2}||D_T^l\mathbf{Z}^{\gamma}(TK)^{N_1}(\psi-\Psi)||_{L^{\infty}(S^2_{\tau,r})}\\ \nonumber
\leq C  \sum_{k\leq 2}\sqrt{E_{N_1+1,N_2+1+l,\epsilon}[\snabla_{\s^2}^k\psi]}(\tau+1)^{-1+\frac{1}{2}\nu-N_1-\frac{l}{2}}(r^{-1}\Omega\tau+1)^{-\frac{1}{2}-\frac{1}{2}\re \beta_{0}},\\
\label{eq:mainthmpoint2}
	\sum_{|\gamma|\leq  N_2}||X^l\mathbf{Z}^{\gamma}(TK)^{N_1}(\psi-\Psi)||_{L^{\infty}(S^2_{\tau,r})}\\ \nonumber
\leq C\kappa_+^{-\frac{1}{2}}  \sum_{k\leq 2}\sqrt{E_{N_1+1,N_2+1+l,\epsilon}[\snabla_{\s^2}^k\psi]}(\tau+1)^{-1+\frac{1}{2}\nu-N_1-\frac{l}{2}}(r^{-1}\tau+1)^{-\frac{1}{2}-\frac{1}{2}\re \beta_{0}}\\
+Ce^{-c\kappa_+\tau}\sum_{|\gamma|\leq N_2}\sqrt{\int_{\Sigma_{0}\cap \{r\leq r_H\}}\mathcal{E}_2[\snabla_{\s^2}^kX^l\mathbf{Z}^{\gamma}(TK)^{N_1}\psi]\,d\sigma dr}\quad (|Q|<M),
\end{align}
with $\Psi=\Psi^{\infty}+\Psi^{+}$, where
\begin{align*}
\Psi^{\infty}(\tau,r,\theta,\varphi):=&\:\sum_{\substack{\ell\in\N_0\\ \ell(\ell+1)<\q^2}}\frac{\mathfrak{w}_{\ell}(r;\q)}{\mathfrak{w}_{\ell}^0(r;\q)}\sum_{\substack{m\in \Z\\ |m|\leq \ell}}\mathfrak{I}_{\ell m}[\psi](\Psi^{\infty}_0)_{\ell m}(\tau,r;\q)Y_{\ell m}(\theta,\varphi),\\
\Psi^{+}(\tau,r,\theta,\varphi):=&\:e^{-i\q r_+^{-1}\tau}\sum_{\substack{\ell\in\N_0\\ \ell(\ell+1)<\q^2}}\frac{\mathfrak{w}_{\ell}\left(r_++\frac{r_+^2}{r-r_+};-\q\right)}{\mathfrak{w}_{\ell}^0\left(r_++\frac{r_+^2}{r-r_+};-\q\right)}\\
\times &\: \sum_{\substack{m\in \Z\\ |m|\leq \ell}}\mathfrak{H}_{\ell m}[\psi](\Psi^{\infty}_0)_{\ell m}\left(\tau,r_++\frac{r_+^2}{r-r_+};-\q\right)Y_{\ell m}(\theta,\varphi),
\end{align*}
where:
\begin{itemize}
\item 	$\mathfrak{I}_{\ell m}[\psi]$ and $\mathfrak{h}_{\ell m}[\psi]$ are integrals that are linear functionals of derivatives of initial data $(\psi,\partial_{\tau}\psi)|_{\tau=0}$ and $\mathfrak{h}_{\ell m}[\psi]=0$ for all data if $|Q|<M$; see the precise definitions in Proposition \ref{prop:timeinvffixedmodeconstr} and Corollary \ref{cor:timeinvffixedmodeconstrhor},
\item The constants $\mathfrak{I}_{\ell m}[\psi]$ and, in the case $|Q|=M$, also $\mathfrak{H}_{\ell m}[\psi]$, each vanish only for a codimension-1 subset of initial data,
\item the constants $E_{N_1+1,N_2+1+l,\epsilon}[\snabla^k_{\s^2}\psi]$ are weighted, initial data energy norms that are defined in Proposition \ref{prop:energydecayptieminv},
\item the functions $\mathfrak{w}_{\ell,\pm }Y_{\ell m}$ are $\tau$-independent solutions to \eqref{eq:CSF} with $A=\widehat{A}$, satisfying $\mathfrak{w}_{\ell}=\alpha_+  \mathfrak{w}_{\ell}^{+}+\alpha_- \mathfrak{w}_{\ell}^{-}$, with 
\begin{align*}
		\mathfrak{w}_{\ell}^{-}(r;\q)=&\:r^{\frac{1}{2}+i\q- \frac{1}{2}\beta_{\ell}}(1+O(r^{-1}))\quad (r>2r_+),\\
		\mathfrak{w}^{+}_{\ell }(r;\q)=&\:r^{\frac{1}{2}+i\q+ \frac{1}{2}\beta_{\ell}}(1+O(r^{-1}))\quad (\beta_{\ell}\neq 0)\quad (r>2r_+),\\
		\mathfrak{w}^{+}_{\ell }(r;\q)=&\:r^{\frac{1}{2}+i\q}\log r(1+O(r^{-1}))\quad (\beta_{\ell}= 0)\quad (r>2r_+),
	\end{align*}
	and $\alpha_{\pm}(\q)\in \C$ are constants such that $\mathfrak{w}_{\ell}(r;\q)={\alpha}_+(\q) \mathfrak{w}_{\ell}^{+}(r;\q)+{\alpha}_-(\q) \mathfrak{w}_{\ell}^{-}(r;\q)$ satisfies $\mathfrak{w}_{\ell}(r_+;\q)=1$ (see also Proposition \ref{prop:statsol}), 
	\item the functions $\mathfrak{w}_{\ell,\pm }^0Y_{\ell m}$ are $\tau$-independent solutions to \eqref{eq:CSFmink} with $A=\widehat{A}$ and $\h=0$, satisfying:
\begin{equation*}
		\mathfrak{w}_{\ell}=\begin{cases}
	\alpha_+r^{\frac{1}{2}+i\q+ \frac{1}{2}\beta_{\ell}}+\alpha_-r^{\frac{1}{2}+i\q- \frac{1}{2}\beta_{\ell}}\quad &(\beta_{\ell}\neq 0),\\
	\alpha_+r^{\frac{1}{2}+i\q}\log r+\alpha_-r^{\frac{1}{2}+i\q} &(\beta_{\ell}= 0).
\end{cases}
	\end{equation*} 
\item The functions $(\Psi_0)^{\infty}_{\ell m}Y_{\ell m}$, with $\ell(\ell+1)<\q^2$, are solutions to \eqref{eq:CSFmink} with $A=\widehat{A}$, with:
\begin{equation}
\label{eq:defPisinfty0full}
(\Psi^{\infty}_0)_{\ell m}(\tau,r;\q):= \begin{cases}
\beta_{\ell}\sum_{n=0}^{N_{\beta_{\ell}}}\zeta^n(1+\tau)^{-\frac{1}{2}-(n+\frac{1}{2})\beta_{\ell}+i\q}\\
		\times {}_2\mathbf{F}_1\left(\frac{1}{2}+\frac{\beta_{\ell}}{2}+i\q,\frac{1}{2}-\frac{\beta_{\ell}}{2}+i\q,\frac{1}{2}+i\q-\left(n+\frac{1}{2}\right)\beta_{\ell};-\frac{\tau+1}{2r}\right),\quad \textnormal{when}\: \beta_{\ell}\in (0,1),\\
		\mbox{}\\
		\sign(\q)\beta_{\ell}\sum_{n=0}^{\infty}\zeta^{\sign(\q)n}(1+\tau)^{-\frac{1}{2}-(n+\frac{1}{2})\sign(\q)\beta_{\ell}+i\q}\\
		\times {}_2\mathbf{F}_1\left(\frac{1}{2}+\sign(\q)\frac{\beta_{\ell}}{2}+iq,\frac{1}{2}-\sign(\q)\frac{\beta_{\ell}}{2}+i\q,\frac{1}{2}+i\q-\left(n+\frac{1}{2}\right)\sign(\q)\beta_{\ell};-\frac{\tau+1}{2r}\right),\\ \textnormal{when}\:\beta_{\ell}\in i(0,\infty),\\
		\mbox{}\\
		(\tau+1)^{-\frac{1}{2}+i\q}\int_0^{\infty}(\tau+1)^{-ix}e^{-\eta x}{}_2\mathbf{F}_1\left(\frac{1}{2}+iq,\frac{1}{2}+iq,\frac{1}{2}+iq-ix;-\frac{\tau+1}{2r}\right)\,dx,\\ \textnormal{when}\:\beta_{\ell}=0,
\end{cases}
\end{equation}
with $N_{\beta_{\ell}}=\lceil (\re \beta_{\ell})^{-1}\rceil\in \N_0$, ${}_2\mathbf{F}_1(a,b,c;\cdot )$ denoting regularized Gauss hypergeometric functions, see for example \cite{NIST:DLMF}[\S 15.2] for a definition, and where $\zeta,\eta\in \C$ are defined as follows:
\begin{align*}
	\zeta(\q):=& -\frac{2^{\beta_{\ell}}\Gamma(-\beta_{\ell}+1)\Gamma(\frac{1}{2}+ \frac{\beta_{\ell}}{2}+i\q)}{\Gamma(\beta_{\ell}+1) \Gamma(\frac{1}{2}- \frac{\beta_{\ell}}{2}+i\q)}\frac{\alpha_-(\q)}{\alpha_+(\q)},\\
	\eta(\q):=&\: i \sign(q)\left[\frac{\alpha_-(\q)}{\alpha_+(\q)}-\log 2+2\gamma_{\rm Euler}-\frac{\Gamma'(\frac{1}{2}+i\q)}{\Gamma(\frac{1}{2}+i\q)}\right],
	\end{align*}
where $\gamma_{\rm Euler}$ is the Euler--Mascheroni constant.
\end{itemize}
\item In particular, we obtain along $\mathcal{I}^+$:
\begin{equation*}
(\Psi|_{\mathcal{I}^+})_{\ell m}(\tau)=\mathfrak{I}_{\ell m}[\psi]\begin{cases}\beta_{\ell}(1+\tau)^{-\frac{1}{2}-\frac{\beta_{\ell}}{2}+i\q}\sum_{n=0}^{N_{\beta_{\ell}}}\frac{\zeta^n(\q)(1+\tau)^{-n\beta_{\ell}}}{\Gamma\left(\frac{1}{2}+i\q-(n+\frac{1}{2})\beta_{\ell}\right)}\\
+e^{-i\q r_+^{-1}\tau}O_{\infty}((\tau+1)^{-1-\beta_{\ell}}),\quad &(\beta_{\ell}\in (0,1)),\\
\sign(\q)\beta_{\ell}(1+\tau)^{-\frac{1}{2}-\sign(\q)\frac{\beta_{\ell}}{2}+i\q}\sum_{n=0}^{\infty}\frac{\zeta^{\sign(q)n}(\q)(1+\tau)^{-\sign(\q)n\beta_{\ell}}}{\Gamma\left(\frac{1}{2}+i\q-(n+\frac{1}{2})\sign(\q)\beta_{\ell}\right)}\\
+e^{-i\q r_+^{-1}\tau}O_{\infty}((\tau+1)^{-1}),\quad &(\beta_{\ell}\in i(0,\infty)),\\
\frac{-i}{\Gamma\left(\frac{1}{2}+i\q\right)}(\tau+1)^{-\frac{1}{2}+i\q}\left[(\log(1+\tau))^{-1}+O((\log(1+\tau))^{-2})\right]\\
+e^{-i\q r_+^{-1}\tau}\frac{1}{\log(1+\tau)}O_{\infty}((\tau+1)^{-1}),\quad &(\beta_{\ell}=0).
\end{cases}
\end{equation*}
Let $r_0>r_+$. Then we obtain along $\{r=r_0\}$:
\begin{align*}
(\Psi)_{\ell m}(\tau,r_0)=&-2^{\frac{1}{2}+\frac{1}{2}\beta_{\ell}}\mathfrak{I}_{\ell m}[\psi]\frac{\pi \beta_{\ell}}{\alpha_+(\q)\sin (\pi \beta_{\ell})}(1+\tau)^{-1-\beta_{\ell}}\mathfrak{w}_{\ell}(r_0;\q)\\
\times & \left[\sum_{n=0}^{N_{\beta_{\ell}}}\zeta^n(\q)\frac{(1+\tau)^{-n\beta_{\ell}}}{\Gamma(\beta_{\ell}+1) \Gamma(\frac{1}{2}-\frac{1}{2}\beta_{\ell}+i\q)\Gamma(-(n+1)\beta_{\ell})}\right] \\
&-2^{\frac{1}{2}+\frac{1}{2}\beta_{\ell}}\mathfrak{H}_{\ell m}[\psi]\frac{\pi \beta_{\ell}}{\alpha_+(-\q)\sin (\pi \beta_{\ell})}(1+\tau)^{-1-\beta_{\ell}}e^{-i\q r_+^{-1}\tau}\mathfrak{w}_{\ell}(r_++r_+^{2}(r_0-r_+)^{-1};-\q)\\
\times  &\left[\sum_{n=0}^{N_{\beta_{\ell}}}\zeta^n(-\q)(\beta_{\ell}) \frac{(1+\tau)^{-n\beta_{\ell}}}{\Gamma(\beta_{\ell}+1) \Gamma(\frac{1}{2}-\frac{1}{2}\beta_{\ell}-i\q)\Gamma(-(n+1)\beta_{\ell})}\right]\\
&+O_{\infty}\left((1+\tau)^{-2}\right)+e^{-i\q r_+^{-1}\tau}O_{\infty}\left((1+\tau)^{-2}\right),  \quad (\beta_{\ell}\in (0,1)),\\
(\Psi)_{\ell m}(\tau,r_0)=&-2^{\frac{1}{2}+\frac{1}{2}\sign(\q)\beta_{\ell}}\mathfrak{I}_{\ell m}[\psi]\frac{\pi \beta_{\ell}}{\alpha_+(\q)\sin (\pi \beta_{\ell})}(1+\tau)^{-1-\sign(\q)\beta_{\ell}}\mathfrak{w}_{\ell}(r_0;\q)\\
\times & \left[\sum_{n=0}^{\infty}\zeta^{\sign(\q)n}(\q) \frac{(1+\tau)^{-\sign(\q)n\beta_{\ell}}}{\Gamma(\sign(\q)\beta_{\ell}+1) \Gamma(\frac{1}{2}-\frac{1}{2}\sign(\q)\beta_{\ell}+i\q)\Gamma(-\sign(\q)(n+1)\beta_{\ell})}\right] \\
&-2^{\frac{1}{2}-\frac{1}{2}\sign(\q)\beta_{\ell}}\mathfrak{H}_{\ell m}[\psi]\frac{\pi \beta_{\ell}}{\alpha_+(-\q)\sin (\pi \beta_{\ell})}(1+\tau)^{-1+\sign(\q)\beta_{\ell}}e^{-i\q r_+^{-1}\tau}\mathfrak{w}_{\ell}(r_++r_+^{2}(r_0-r_+)^{-1};–\q)\\
\times&  \left[\sum_{n=0}^{\infty}\zeta^{-\sign(\q)n}(-\q) \frac{(1+\tau)^{\sign(\q)n\beta_{\ell}}}{\Gamma(-\sign(\q)\beta_{\ell}+1) \Gamma(\frac{1}{2}+\frac{1}{2}\sign(\q)\beta_{\ell}-i\q)\Gamma((n+1)\sign(\q)\beta_{\ell})}\right]\\
&+O_{\infty}\left((1+\tau)^{-2}\right)+e^{-iq r_+^{-1}\tau}O_{\infty}\left((1+\tau)^{-2}\right), \quad (\beta_{\ell}\in i(0,\infty)),\\
(\Psi)_{\ell m}(\tau,r_0)=&-\frac{2^{\frac{1}{2}+i\q}}{\Gamma(\frac{1}{2}+i\q)}\mathfrak{I}_{\ell m}[\psi]\mathfrak{w}_{\ell}(r_0;\q)\left[(\log(1+\tau))^{-2}(\tau+1)^{-1}+(\log(1+\tau))^{-3}O_{\infty}((\tau+1)^{-1}\right]\\
&-\frac{2^{\frac{1}{2}-i\q}}{\Gamma(\frac{1}{2}-i\q)}e^{-i\q r_+^{-1}\tau}\mathfrak{H}_{\ell m}[\psi]\mathfrak{w}_{\ell}(r_++r_+^{2}(r_0-r_+)^{-1};-\q)\\
&\times \left[(\log(1+\tau))^{-2}(\tau+1)^{-1}+(\log(1+\tau))^{-3}O_{\infty}((\tau+1)^{-1}\right]\\
	&+O_{\infty}((1+\tau)^{-2})+e^{-i\q r_+^{-1}\tau}O_{\infty}((1+\tau)^{-2}), \quad (\beta_{\ell}=0).
\end{align*}
We obtain along $\mathcal{H}^+$:
\begin{align*}
(\Psi)_{\ell m}(\tau,r_+)=&\:\beta_{\ell}\mathfrak{H}_{\ell m}[\psi]e^{-iq r_+^{-1}\tau}(1+\tau)^{-\frac{1}{2}-\frac{\beta_{\ell}}{2}-i\q} \sum_{n=0}^{N_{\beta_{\ell}}}\frac{\zeta^n(-\q)(1+\tau)^{-n\beta_{\ell}}}{\Gamma\left(\frac{1}{2}-i\q-(n+\frac{1}{2})\beta_{\ell}\right)}\\
&-2^{\frac{1}{2}+\frac{1}{2}\beta_{\ell}}\mathfrak{I}_{\ell m}[\psi]\frac{\pi \beta_{\ell}}{\alpha_+(\q)\sin (\pi \beta_{\ell})}(1+\tau)^{-1-\beta_{\ell}}\\
\times & \left[\sum_{n=0}^{N_{\beta_{\ell}}}\zeta^n(\q)\frac{(1+\tau)^{-n\beta_{\ell}}}{\Gamma(\beta_{\ell}+1) \Gamma(\frac{1}{2}-\frac{1}{2}\beta_{\ell}+i\q)\Gamma(-(n+1)\beta_{\ell})}\right] \\
&+O_{\infty}\left((1+\tau)^{-2}\right), \quad (\beta_{\ell}\in (0,1)),\\
(\Psi)_{\ell m}(\tau,r_+)=&\:-\sign(\q)\beta_{\ell}\mathfrak{H}_{\ell m}[\psi]e^{-iq r_+^{-1}\tau}(1+\tau)^{-\frac{1}{2}+\sign(\q)\frac{\beta_{\ell}}{2}-i\q}\sum_{n=0}^{\infty}\frac{\zeta^{-\sign(q)n}(-\q)(1+\tau)^{\sign(\q)n\beta_{\ell}}}{\Gamma\left(\frac{1}{2}-i\q+(n+\frac{1}{2})\sign(\q)\beta_{\ell}\right)}\\
&-2^{\frac{1}{2}+\frac{1}{2}\sign(\q)\beta_{\ell}}\mathfrak{I}_{\ell m}[\psi]\frac{\pi \beta_{\ell}}{\alpha_+(\q)\sin (\pi \beta_{\ell})}(1+\tau)^{-1-\sign(\q)\beta_{\ell}}\\
\times & \left[\sum_{n=0}^{\infty}\zeta^{\sign(\q)n}(\q) \frac{(1+\tau)^{-\sign(\q)n\beta_{\ell}}}{\Gamma(\sign(\q)\beta_{\ell}+1) \Gamma(\frac{1}{2}-\frac{1}{2}\sign(\q)\beta_{\ell}+i\q)\Gamma(-\sign(\q)(n+1)\beta_{\ell})}\right] \\
&+O_{\infty}\left((1+\tau)^{-2}\right), \quad (\beta_{\ell}\in i(0,\infty)),\\
(\Psi)_{\ell m}(\tau,r_+)=&\: \frac{-i}{\Gamma\left(\frac{1}{2}-i\q\right)}\mathfrak{H}_{\ell m}[\psi]e^{-iq r_+^{-1}\tau}(\tau+1)^{-\frac{1}{2}-i\q}\left[(\log(1+\tau))^{-1}+O((\log(1+\tau))^{-2})\right]\\
&-\frac{2^{\frac{1}{2}+i\q}}{\Gamma(\frac{1}{2}+i\q)}\mathfrak{I}_{\ell m}[\psi]\left[(\log(1+\tau))^{-2}(\tau+1)^{-1}+(\log(1+\tau))^{-3}O_{\infty}((\tau+1)^{-1}\right]\\
	&+O_{\infty}((1+\tau)^{-2}),\quad (\beta_{\ell}=0).
\end{align*}
\item We obtain moreover the following late-time asymptotics for $D_T\Psi$ along $\mathcal{H}^+$ and $\mathcal{I}^+$:
\begin{align*}
(D_T\Psi|_{\mathcal{I}^+})_{\ell m}(\tau)=&\:2^{\frac{1}{2}+\frac{1}{2}\beta_{\ell}}\q r_+^{-1}\mathfrak{H}_{\ell m}[\psi]\frac{\pi \beta_{\ell}}{\alpha_+(-\q)\sin (\pi \beta_{\ell})}(1+\tau)^{-1-\beta_{\ell}}e^{-i\q r_+^{-1}\tau}\mathfrak{w}_{\ell}(r_+;-\q)\\
\times  &\left[\sum_{n=0}^{N_{\beta_{\ell}}}\zeta^n(-\q)(\beta_{\ell}) \frac{(1+\tau)^{-n\beta_{\ell}}}{\Gamma(\beta_{\ell}+1) \Gamma(\frac{1}{2}-\frac{1}{2}\beta_{\ell}-i\q)\Gamma(-(n+1)\beta_{\ell})}\right]\\
+&\:O_{\infty}((\tau+r_+)^{-\frac{3+\beta_{\ell}}{2} }), \quad (\beta_{\ell}\in (0,1),\\
(D_T\Psi|_{\mathcal{I}^+})_{\ell m}(\tau)=&\:2^{\frac{1}{2}-\frac{1}{2}\sign(\q)\beta_{\ell}}\mathfrak{H}_{\ell m}[\psi]\q r_+^{-1}\frac{\pi \beta_{\ell}}{\alpha_+(-\q)\sin (\pi \beta_{\ell})}(1+\tau)^{-1+\sign(\q)\beta_{\ell}}e^{-i\q r_+^{-1}\tau}\mathfrak{w}_{\ell}(r_+;–\q)\\
\times &  \left[\sum_{n=0}^{\infty}\zeta^{-\sign(\q)n}(-\q) \frac{(1+\tau)^{-\sign(\q)n\beta_{\ell}}}{\Gamma(-\sign(\q)\beta_{\ell}+1) \Gamma(\frac{1}{2}+\frac{1}{2}\sign(\q)\beta_{\ell}-i\q)\Gamma((n+1)\sign(\q)\beta_{\ell})}\right]\\
+&\:O_{\infty}((\tau+r_+)^{-\frac{3}{2} }), \quad (\beta_{\ell}\in i(0,\infty)),\\
(D_T\Psi|_{\mathcal{I}^+})_{\ell m}(\tau)=&\:\frac{2^{\frac{1}{2}-i\q}}{\Gamma(\frac{1}{2}-i\q)}\q r_+^{-1}e^{-i\q r_+^{-1}\tau}\mathfrak{H}_{\ell m}[\psi]\mathfrak{w}_{\ell}(r_++r_+^{2}(r_0-r_+)^{-1};-\q)\\
&\times \left[(\log(1+\tau))^{-2}(\tau+1)^{-1}+(\log(1+\tau))^{-3}O_{\infty}((\tau+1)^{-1}\right]\\
+&\:O_{\infty}((\tau+r_+)^{-\frac{3}{2} }), \quad (\beta_{\ell}=0),\\
(D_T\Psi|_{\mathcal{H}^+})_{\ell m}(\tau)=&\:-2^{\frac{1}{2}-\frac{1}{2}\beta_{\ell}}\q r_+^{-1}\mathfrak{I}_{\ell m}[\psi]\frac{\pi \beta_{\ell}}{\alpha_+(\q)\sin (\pi \beta_{\ell})}(1+\tau)^{-1-\beta_{\ell}}\mathfrak{w}_{\ell}(r_+;\q)\\
\times  &\left[\sum_{n=0}^{N_{\beta_{\ell}}}\zeta^n(\q)(\beta_{\ell}) \frac{(1+\tau)^{-n\beta_{\ell}}}{\Gamma(\beta_{\ell}+1) \Gamma(\frac{1}{2}-\frac{1}{2}\beta_{\ell}+i\q)\Gamma(-(n+1)\beta_{\ell})}\right]\\
+&\:O_{\infty}((\tau+r_+)^{-\frac{3+\beta_{\ell}}{2} }), \quad (\beta_{\ell}\in (0,1),\\
(D_T\Psi|_{\mathcal{H}^+})_{\ell m}(\tau)=&\:-2^{\frac{1}{2}+\frac{1}{2}\sign(\q)\beta_{\ell}}\q r_+^{-1}\mathfrak{I}_{\ell m}[\psi]\frac{\pi \beta_{\ell}}{\alpha_+(\q)\sin (\pi \beta_{\ell})}(1+\tau)^{-1-\sign(\q)\beta_{\ell}}\mathfrak{w}_{\ell}(r_+;\q)\\
\times&  \left[\sum_{n=0}^{\infty}\zeta^{\sign(\q)n}(\q) \frac{(1+\tau)^{-\sign(\q)n\beta_{\ell}}}{\Gamma(\sign(\q)\beta_{\ell}+1) \Gamma(\frac{1}{2}-\frac{1}{2}\sign(\q)\beta_{\ell}+i\q)\Gamma(-(n+1)\sign(\q)\beta_{\ell})}\right]\\
+&\:O_{\infty}((\tau+r_+)^{-\frac{3}{2} }), \quad (\beta_{\ell}\in i(0,\infty)),\\
(D_T\Psi|_{\mathcal{H}^+})_{\ell m}(\tau)=&\:-\frac{2^{\frac{1}{2}+i\q}}{\Gamma(\frac{1}{2}+i\q)}\q r_+^{-1}\mathfrak{I}_{\ell m}[\psi]\mathfrak{w}_{\ell}(r_+;\q)\\
&\times \left[(\log(1+\tau))^{-2}(\tau+1)^{-1}+(\log(1+\tau))^{-3}O_{\infty}((\tau+1)^{-1}\right]\\
+&\:O_{\infty}((\tau+r_+)^{-\frac{3}{2} }), \quad (\beta_{\ell}=0).
\end{align*}
\end{enumerate}
\end{theorem}

Theorem \ref{thm:mainthmpoint} is proved in \S \ref{sec:proofthm1} and follows from energy decay estimates, which encompass the following theorem:
\begin{theorem}[Late time energy decay asymptotics]
\label{thm:mainthmenergy}
Let $\psi$ be as in Theorem \ref{thm:mainthmpoint}. If $\kappa_+>\kappa_1$ and $|\q|>q_1$, assume moreover that Condition \ref{cond:quantmodestabp2} holds. Let $N_1\in \N_0$, $N_2\in \N_0$. Let $1< p<2$. For all $\nu>0$ and $\epsilon>0$ suitably small, there exists a constant $C=C(\nu,\epsilon, N_1,N_2,p)>0$ such that for $|\gamma|\leq N_2$:
\begin{equation*}
	\int_{\Sigma_{\tau}}\mathcal{E}_{p-3\epsilon}[\mathbf{Z}^{\gamma}((1
+\tau)TK)^{N_1}\psi]\,d\sigma dr\leq  \begin{cases}CE_{N_1,N_2,\epsilon}[\psi](1+\tau)^{p-2-2\re \beta_{\ell}}\quad (\beta_{\ell}\neq 0),\\
	CE_{N_1,N_2,\epsilon}[\psi]\frac{(1+\tau)^{p-2}}{\log^2(1+\tau)}\quad (\beta_{\ell}\neq 0)
  \end{cases}	
\end{equation*}
Furthermore, for $p\in [0,1-\epsilon]$
\begin{equation*}
	\int_{\Sigma_{\tau}}\mathcal{E}_{p}[\mathbf{Z}^{\gamma}((1
+\tau)TK)^{N_1}\psi]\,d\sigma dr\leq  \begin{cases}CE_{N_1+1,N_2+1,\epsilon}[\psi](1+\tau)^{p-2-2\re \beta_{\ell}}\quad (\beta_{\ell}\neq 0),\\
	CE_{N_1+1,N_2+1,\epsilon}[\psi]\frac{(1+\tau)^{p-2}}{\log^2(1+\tau)}\quad (\beta_{\ell}\neq 0)
  \end{cases}	
\end{equation*}
and
\begin{equation*}
	\int_{\Sigma_{\tau}}\mathcal{E}_{2}[\mathbf{Z}^{\gamma}((1
+\tau)TK)^{N_1}\psi]\,d\sigma dr\leq  \begin{cases}CE_{N_1,N_2+2,\epsilon}[\psi](1+\tau)^{-2\re \beta_{\ell}}\quad (\beta_{\ell}\neq 0),\\
	CE_{N_1,N_2+2,\epsilon}[\psi]\frac{1}{\log^2(1+\tau)}\quad (\beta_{\ell}\neq 0).
  \end{cases}	
\end{equation*}

\end{theorem}
Theorem \ref{thm:mainthmenergy} is proved in \S \ref{sec:pfthm2}. 

Using the precise pointwise and energy decay estimates above, we also derive asymptotic instabilities in the form of pointwise blow-up of ($r$-weighted) transversal derivatives along $\mathcal{H}^+$ and $\mathcal{H}^+$, and the concentration of non-degenerate energies at $\mathcal{H}^+$ in the $|Q|=M$ case (and, in the case of $r$-weighted energies, also at  $\mathcal{I}^+$ in the $|Q|\leq M$ case).

\begin{theorem}[Instabilities and energy concentration]
\label{thm:mainpointwinst}
Let $\psi$ be as in Theorem \ref{thm:mainthmpoint}. 	If $\kappa_+>\kappa_1$ and $|\q|>q_1$, assume moreover that Condition \ref{cond:quantmodestabp2} holds. Let $\nu,\epsilon,\delta>0$ be arbitrarily small. Let $\mathfrak{c}_k\in  \C$ and $\mathfrak{c}_{\rm int}\in (0,\infty)$ be the constants defined in Proposition \ref{prop:Psi0inftygrowthdecay}.	\begin{enumerate}[label=\emph{(\roman*)}]
	\item 
 Let $r_H>r_+$ and $r_I<\infty$ be arbitrary. Then for $\kappa_+\leq \kappa_0$:
	\begin{align}
	\label{eq:estenergyPsi1}
\limsup_{\tau\to \infty}(1+\tau)^{\re \beta_0}\int_{\Sigma_{\tau}\cap\{r\leq r_H\}}\mathcal{E}_2[\psi]\,d\sigma dr\geq &\: \mathfrak{c}_{\rm int}\sum_{\substack{\ell\in \N_0\\ \re \beta_{\ell}=\re \beta_{0}}}\sum_{\substack{m\in \Z\\ |m|\leq \ell}}|\mathfrak{H}_{\ell m}[\psi]|^2\quad &( |Q|=M,\: \beta_{0}\neq 0),\\
\label{eq:estenergyPsi2}
\limsup_{\tau\to \infty}\log^2(1+\tau)\int_{\Sigma_{\tau}\cap\{r\leq r_H\}}\mathcal{E}_2[\psi]\,d\sigma dr\geq &\: \mathfrak{c}_{\rm int}|\mathfrak{H}_{00}[\psi]|^2\quad &( |Q|=M,\: \beta_{0}= 0),\\
	\label{eq:estenergyPsi1inf}
\limsup_{\tau\to \infty}(1+\tau)^{\re \beta_{0}}\int_{\Sigma_{\tau}\cap\{r\geq r_I\}}\mathcal{E}_2[\psi]\,d\sigma dr\geq &\:   \mathfrak{c}_{\rm int}\sum_{\substack{\ell\in \N_0\\ \re \beta_{\ell}=\re \beta_{0}}}\sum_{\substack{m\in \Z\\ |m|\leq \ell}}|\mathfrak{I}_{\ell m}[\psi]|^2\quad &( \beta_{0}\neq 0),\\
\label{eq:estenergyPsi2inf}
\limsup_{\tau\to \infty}\log^2(1+\tau)\int_{\Sigma_{\tau}\cap\{r\geq r_I\}}\mathcal{E}_2[\psi]\,d\sigma dr\geq &\: \mathfrak{c}_{\rm int}|\mathfrak{I}_{00}[\psi]|^2\quad &( \beta_{0}= 0).
	\end{align}
	Furthermore, $\limsup_{\tau\to \infty}\log^2(1+\tau)(1+\tau)^{\re \beta_0}\int_{\Sigma_{\tau}\cap\{r_H\leq r\leq r_I\}}\mathcal{E}_2[\psi]\,d\sigma dr=0$ for all $\beta_0\in [0,1]\cup i(0,\infty)$.
	
	\item Let $k\in \N_0$ and $\ell(\ell+1)<q^2$. Then there exists a constant $C_k>0$, such that for all 
	\begin{align}\label{eq:lowboundbetanonzeroRNhor}	
		\limsup_{\tau\to \infty}(1+\tau)^{-k-\frac{1}{2}+\frac{\re \beta_{\ell}}{2}}|(r^2X)^{k+1}\psi_{\ell m}|(\tau,r_+)\geq &\:|\mathfrak{c}_k| |\mathfrak{H}_{\ell m}[\psi]| \quad &(|Q|=M,\:\beta_{\ell}\neq 0),\\
		\label{eq:lowboundzerobetaRNhor}	
		\limsup_{\tau\to \infty}\log(1+\tau)(1+\tau)^{-k-\frac{1}{2}}|(r^2X)^{k+1}\psi_{\ell m}|(\tau,r_+)\geq &\:|\mathfrak{c}_k| |\mathfrak{H}_{\ell m}[\psi]|\quad &(|Q|=M,\:\beta_{\ell}=0),\\
		\label{eq:lowboundbetanonzeroRNinf}	
		\limsup_{\tau\to \infty}(1+\tau)^{-k-\frac{1}{2}+\frac{\re \beta_{\ell}}{2}}|(r^2X)^{k+1}\psi_{\ell m}|(\tau,\infty)\geq &\:|\mathfrak{c}_k| |\mathfrak{I}_{\ell m}[\psi]| \quad &(\beta_{\ell}\neq 0),\\
		\label{eq:lowboundzerobetaRNinf}	
		\limsup_{\tau\to \infty}\log(1+\tau)(1+\tau)^{-k-\frac{1}{2}}|(r^2X)^{k+1}\psi_{\ell m}|(\tau,\infty)\geq &\:|\mathfrak{c}_k| |\mathfrak{I}_{\ell m}[\psi]|\quad &(\beta_{\ell}=0).
		\end{align}
	\end{enumerate}
\end{theorem}
Theorem \ref{thm:mainpointwinst} is proved in \S \ref{sec:pfthm3}.

\section{Higher-order $(\Omega^{-1} r)^{p}$-weighted energy estimates}
\label{sec:horp}
In this section, we show that we can commute \eqref{eq:iedpaper1} with the differential operators $\{(r-r_+)X,T,\snabla_{\s^2}\}$. Since commutation with $T^k$ follows immediately from the Killing property of $T$ and commutation with $\snabla_{\s^2}$ is an immediate consequence of the fact that we can commute with angular momentum Killing vector fields, since the background spacetime is spherically symmetric, it remains to show that we can commute with $((r-r_+)X)^k$.

By \eqref{eq:maineqradfieldconf}, with $\h\equiv 0$, we obtain:
\begin{equation}
\label{eq:maineqradfieldhzero}
r^3G_{\widehat{A}}=\partial_{\rho_{\infty}}(\Omega^2r^{-2}\partial_{\rho_{\infty}}\psi)+\slashed{\Delta}_{\s^2}\psi+2\partial_{\rho_{\infty}}T\psi+2i\q\rho_{\infty} \partial_{\rho_{\infty}}\psi-\left[r\frac{d\Omega^2}{dr} -i \q \right]\psi.
\end{equation}
By \eqref{eq:maineqradfieldconfhor}, with $\widetilde{\h}\equiv 0$, we obtain:
\begin{equation}
	\label{eq:maineqradfieldhtildezero}
r^3G_{\widehat{A}}=\partial_{\rho_+}(\Omega^2r^{-2}\partial_{\rho_+}\psi)+\slashed{\Delta}_{\s^2}\psi+2\partial_{\rho_+}K\psi- 2i\q\rho_+ \partial_{\rho_+}\psi\\
-\left[r\frac{d\Omega^2}{dr}+i \q \right]\psi.
\end{equation}
\begin{lemma}
	Let $N\in \N_0$.
		\begin{enumerate}[label=\emph{(\roman*)}]
		\item Let $\psi$ be a solution to \eqref{eq:maineqradfieldhzero}. Denote $\psi^{(N)}=(\rho_{\infty}\partial_{\rho_{\infty}})^N\psi$. Then:
	\begin{multline}
	\label{eq:maineqradfieldhoinf}
	\rho_{\infty}^{-1}(\rho_{\infty}\partial_{\rho_{\infty}})^N(\rho_{\infty} r^3G_{\widehat{A}})=\partial_{\rho_{\infty}}(\Omega^2r^{-2}\partial_{\rho_{\infty}}\psi^{(N)})+\slashed{\Delta}_{\s^2}\psi^{(N)}+2\partial_{\rho_{\infty}}T\psi^{(N)}\\
	+2i\q\rho_{\infty} \partial_{\rho_{\infty}}\psi^{(N)}+N(1+O_{\infty}(\rho_{\infty}))\rho_{\infty}\partial_{\rho_{\infty}}\psi^{(N)}+N \slashed{\Delta}_{\s^2}\psi^{(N-1)}\\
-\left[r\frac{d\Omega^2}{dr}-i \q(1+2N) -\frac{1}{2}N(N+1)+O_{\infty}(\rho_{\infty})\right]\psi^{(N)}+N\sum_{n=0}^{N-1} O_{\infty}(\rho_{\infty}^0)\psi^{(n)}\\
+N(N-1)\sum_{n=0}^{N-2} O_{\infty}(\rho_{\infty}^0)\slashed{\Delta}_{\s^2}\psi^{(n)}
	\end{multline}

	\item Let $\psi$ be a solution to \eqref{eq:maineqradfieldhtildezero}. Denote $\psi^{(N)}=(\rho_+\partial_{\rho_+})^N\psi$. Then:
	\begin{multline}
	\label{eq:maineqradfieldhohor}
	\rho_+^{-1}(\rho_+\partial_{\rho_+})^N(\rho_+ r^3G_{\widehat{A}})=\partial_{\rho_+}(\Omega^2r^{-2}\partial_{\rho_+}\psi^{(N)})+\slashed{\Delta}_{\s^2}\psi^{(N)}+2\partial_{\rho_+}K\psi^{(N)}\\
	- 2i\q\rho_+ \partial_{\rho_+}\psi^{(N)}+N(1-6\kappa_+ r_+)(1+O_{\infty}(\rho_+))\rho_+\partial_{\rho_+}\psi^{(N)}+N \slashed{\Delta}_{\s^2}\psi^{(N-1)}\\
-\left[r\frac{d\Omega^2}{dr}+i \q(1+2N) -\frac{1}{2}N(N+1)(1-6\kappa_+ r_+)+O_{\infty}(\rho_+)\right]\psi^{(N)}+N\sum_{n=0}^{N-1} O_{\infty}(\rho_+^0)\psi^{(n)}\\
+N(N-1)\sum_{n=0}^{N-2} O_{\infty}(\rho_+^0)\slashed{\Delta}_{\s^2}\psi^{(n)}.
	\end{multline}

	\end{enumerate}
\end{lemma}
\begin{proof}
We will prove (ii). The proof of (i) proceeds entirely analogously to (ii), with $\rho_+$ replaced by $\rho_{\infty}$ and $\q$ replaced by $-\q$. 

The proof of (ii) proceeds by induction. The $N=0$ case follows directly from \eqref{eq:maineqradfieldhtildezero}. Suppose that \eqref{eq:maineqradfieldhohor} holds for $N\in \N_0$. Then:
\begin{multline*}
	(\rho_+\partial_{\rho_+})^{N+1}(\rho_+ r^3G_{\widehat{A}})=\rho_+\partial_{\rho_+}(\rho_+\partial_{\rho_+}(\Omega^2r^{-2}\rho_+^{-1}\psi^{(N+1)}))+\rho_+\partial_{\rho_+}(\rho_+\slashed{\Delta}_{\s^2}\psi^{(N)})+2\rho_+\partial_{\rho_+}K\psi^{(N+1)}\\
	- 2i\q\partial_{\rho_+}(\rho_+ \psi^{(N+1)})+N(1-6\kappa_+ r_+)(1+O_{\infty}(\rho_+))\rho_+\partial_{\rho_+}(\rho_+ \psi^{(N+1)})+N\rho_+\slashed{\Delta}_{\s^2}\psi^{(N)}\\
-\left[r\frac{d\Omega^2}{dr}+i \q(1+2N) -\frac{1}{2}N(N+1)(1-6\kappa_+ r_+)+O_{\infty}(\rho_+)\right]\rho_+\psi^{(N+1)}+(N+1)\rho_+\sum_{n=0}^{N} O_{\infty}(\rho_+^0)\psi^{(n)}\\
+N(N+1)\rho_+\sum_{n=0}^{N-1} O_{\infty}(\rho_+^0)\slashed{\Delta}_{\s^2}\psi^{(n)}.
	\end{multline*}
	
We can expand further by applying Lemma \ref{lm:metricest}:
\begin{multline*}
	\rho_+\partial_{\rho_+}(\rho_+\partial_{\rho_+}(\Omega^2r^{-2}\rho_+^{-1}\psi^{(N+1)}))=\rho_+\partial_{\rho_+}(\Omega^2r^{-2} \partial_{\rho_+}\psi^{(N+1)})+\rho_+\partial_{\rho_+}(\rho_+\partial_{\rho_+}(\Omega^2r^{-2}\rho_+^{-1})\psi^{(N+1)})\\
	=\rho_+\partial_{\rho_+}(\Omega^2r^{-2} \partial_{\rho_+}\psi^{(N+1)})+\rho_+\partial_{\rho_+}(\rho_+(1-6\kappa_++O_{\infty}(\rho_+))\psi^{(N+1)})\\
	=\rho_+\partial_{\rho_+}(\Omega^2r^{-2} \partial_{\rho_+}\psi^{(N+1)})+\rho_+^2(1-6\kappa_+ r_+)(1+O_{\infty}(\rho_+))\partial_{\rho_+}\psi^{(N+1)}+\rho_+(1-6\kappa_+ r_+)(1+O_{\infty}(\rho_+))\psi^{(N+1)}.
\end{multline*}
Combining the above results in \eqref{eq:maineqradfieldhohor} with $N$ replaced by $N+1$.
\end{proof}

\begin{proposition}
\label{prop:horpest}
Let $N\in \N_0$ and let $\psi$ be a solution to \eqref{eq:rCSF}. Let $1<r_H<r_I<\infty$, $q\in \R$, $\epsilon>0$, $K\in \N_0$ and $1< p<\min\{1+\re \beta_{\ell},2\}+\epsilon$. Then there exist a constant $C=C(\h,\epsilon,p, r_H,r_I,N)>0$, such that for suitably small $\epsilon>0$:
\begin{multline}
\label{eq:horpest}
\sum_{|\gamma|\leq N}\int_{\Sigma_{\tau_2}} \mathcal{E}_{p-2\epsilon}[\mathbf{Z}^{\gamma}\psi_{\geq \ell}]\,d\sigma dr+\int_{\tau_1}^{\tau_2}\int_{\Sigma_{\tau}\cap \{r_H\leq r\leq r_I\}}  \upzeta \mathcal{E}_p[\mathbf{Z}^{\gamma}\psi_{\geq \ell}]+|\mathbf{Z}^{\gamma}\psi_{\geq \ell}|^2\,d\sigma dr d\tau\\
+\int_{\tau_1}^{\tau_2}\int_{\Sigma_{\tau}\setminus \{r_H\leq r\leq r_I\}} (\rho_++\kappa_+)r^{-1}(r^{-1}\rho_+)^{2\epsilon}\mathcal{E}_{p}[\mathbf{Z}^{\gamma}\psi_{\geq \ell}]+(\rho_++\kappa_+)r^{-3}(r^{-1}\Omega)^{\epsilon-2}|D_T\mathbf{Z}^{\gamma}\psi_{\geq \ell}|^2\,d\sigma dr d\tau\\
\leq C\sum_{|\gamma|\leq N}\sum_{n\leq N}\int_{\Sigma_{\tau_1}} \mathcal{E}_{p-2\epsilon}[\mathbf{Z}^{\gamma}\psi_{\geq \ell}]\,d\sigma dr+C\sum_{n_1+n_2=n}\int_{\Sigma_{\tau_1}} \mathcal{E}_{p}[\snabla_{\s^2}^{n_1}T^{n_2}\psi_{\geq \ell}]\,d\sigma dr\\
+C\int_{\tau_1}^{\tau_2}\int_{\Sigma_{\tau}}(r^{-1}\Omega)^{-p+2\epsilon}r^{-1}\rho_+r^{-2}|\mathbf{Z}^{\gamma}( r^3G_{\widehat{A}})_{\geq \ell} |^2\,d\sigma dr d\tau\\
+C\sum_{n_1+n_2=n}\int_{\tau_1}^{\tau_2}\int_{\Sigma_{\tau}} \max\{(r^{-1}\Omega)^{-p}\rho_+^{1-2\epsilon}r^{-1+2\epsilon},1\}r^{-2}|\snabla_{\s^2}^{n_1}T^{n_2}(r^3G_{\widehat{A}})_{\geq \ell}|^2+\Omega^2r^{-2}(1-\upzeta) |\snabla_{\s^2}^{n_1}T^{n_2+1}( r^3G_{\widehat{A}})_{\geq \ell} |^2\,d\sigma dr d\tau.
\end{multline}
\end{proposition}
\begin{proof}
Without loss of generality, we may assume that $\h\equiv 0$ in $\{r\geq r_I\}$ and $\widetilde{\h}\equiv 0$ in $\{r\leq r_H\}$. Indeed, we can simply apply the local energy estimates in \cite{gaj26a}[Theorem 2.6] afterwards to obtain \eqref{eq:horpest} in the more general setting.

In view of the similarity between \eqref{eq:maineqradfieldhohor} and \eqref{eq:maineqradfieldhoinf}, we can also restrict without loss of generality to  $\{r\leq r_H\}$. The estimates in $\{r\geq r_I\}$ proceed analogously to the estimates in $\{r\leq r_H\}$ with $\kappa_+=0$.

We multiply both sides of \eqref{eq:maineqradfieldhohor} with $\chi_{r_H}(r^{-1}\Omega)^{2-p}\partial_{\rho_+}\overline{\psi}^{(N)}$ and take the real part to obtain:
\begin{multline*}
	\chi_{r_H}(r^{-1}\Omega)^{2-p}\re\left[\rho_+^{-1}(\rho_+\partial_{\rho_+})^N(\rho_+ r^3G_{\widehat{A}}) \partial_{\rho_+}\overline{\psi}^{(N)}\right]=\frac{1}{2}\chi_{r_H}(r^{-1}\Omega)^{-p}\partial_{\rho_+}((\Omega^2r^{-2})^2|\partial_{\rho_+}\psi|^2)\\
	-\frac{1}{2}\chi_{R_H}(r^{-1}\Omega)^{2-p}\partial_{\rho_+}(|\snabla_{\s^2}\psi|^2)+{\rm div}_{\s^2}(\ldots)+T(\chi_{r_H}(r^{-1}\Omega)^{2-p}|\partial_{\rho_+}\overline{\psi}^{(N)}|^2)\\
	+N(1-6\kappa_+ r_+)(1+O_{\infty}(\rho_+))\rho_+(r^{-1}\Omega)^{2-p}|\partial_{\rho_+}\psi^{(N)}|^2-N\chi_{r_H}(r^{-1}\Omega)^{2-p}\partial_{\rho_+}\left(\re\left[\snabla_{\s^2}\psi^{(N-1)}\cdot \snabla_{\s^2}\overline{\psi}^{(N)}\right]\right)\\
	+N\chi_{r_H}(r^{-1}\Omega)^{2-p}\re\left[\snabla_{\s^2}\partial_{\rho_+}\psi^{(N-1)}\cdot \snabla_{\s^2} \overline{\psi}^{(N)}\right]\\
	-\chi_{r_H}\left[r\frac{d\Omega^2}{dr}+i \q(1+2N) -\frac{1}{2}N(N+1)(1-6\kappa_+ r_+)+O_{\infty}(\rho_+)\right](r^{-1}\Omega)^{2-p}\re\left[\psi^{(N)}\partial_{\rho_+}\overline{\psi}^{(N)}\right]\\
	+N\chi_{r_H}(r^{-1}\Omega)^{2-p}\sum_{n=0}^{N-1} O_{\infty}(\rho_+^0)\re\left[\psi^{(n)}\partial_{\rho_+}\overline{\psi}^{(N)}\right]-N(N-1)(r^{-1}\Omega)^{2-p}\sum_{n=0}^{N-2} O_{\infty}(\rho_+^0)\re\left[\slashed{\Delta}_{\s^2}\psi^{(n)}\partial_{\rho_+}\overline{\psi}^{(N)}\right]\\
	=\frac{1}{2}\partial_{\rho_+}\Bigg(\chi_{r_H}(r^{-1}\Omega)^{4-p}|\partial_{\rho_+}\psi|^2-\frac{1}{2}\chi_{R_H}(r^{-1}\Omega)^{2-p}|\snabla_{\s^2}\psi|^2+\frac{1}{2}\chi_{R_H}(r^{-1}\Omega)^{2-p}N(N+1)(1-6\kappa_+ r_+)|\psi^{(N)}|^2\\
	-N\chi_{r_H}(r^{-1}\Omega)^{2-p}\re\left[\snabla_{\s^2}\psi^{(N-1)}\cdot \snabla_{\s^2}\overline{\psi}^{(N)}\right]\Bigg)+T(\chi_{r_H}(r^{-1}\Omega)^{2-p}|\partial_{\rho_+}\overline{\psi}^{(N)}|^2)+{\rm div}_{\s^2}(\ldots)\\
	+\left[\frac{p}{4}r^2\frac{d}{dr}(r^{-2}\Omega^2)+N(1-6\kappa_+ r_+)(1+O_{\infty}(\rho_+))\rho_+\right](r^{-2}\Omega^2)^{2-p}|\partial_{\rho_+}\psi|^2\\
	+\left[\frac{2-p}{4}r^2\frac{d}{dr}(r^{-2}\Omega^2)+N(r^{-1}\Omega)^{2}\rho_+^{-1}\right](r^{-1}\Omega)^p|\snabla_{\s^2}\psi^{(N)}|^2-\frac{2-p}{4}N(N-1)(1-6\kappa_+ r_+)r^2\frac{d}{dr}(r^{-2}\Omega^2)(r^{-1}\Omega)^p|\psi^{(N)}|^2\\
+Nr^2\frac{2-p}{2}\chi_{r_H}(r^{-1}\Omega)^{-p}\frac{d}{dr}(r^{-2}\Omega^2)\re\left[\snabla_{\s^2}\psi^{(N-1)}\cdot \snabla_{\s^2}\overline{\psi}^{(N)}\right]\\
+N\chi_{r_H}(r^{-1}\Omega)^{2-p}\sum_{n=0}^{N-1} O_{\infty}(\rho_+^0)\re\left[\psi^{(n)}\partial_{\rho_+}\overline{\psi}^{(N)}\right]-N(N-1)\chi_{r_H}(r^{-1}\Omega)^{2-p}\sum_{n=0}^{N-2} O_{\infty}(\rho_+^0)\re\left[\slashed{\Delta}_{\s^2}\psi^{(n)}\partial_{\rho_+}\overline{\psi}^{(N)}\right]\\
-\chi_{r_H}\left[r\frac{d\Omega^2}{dr}+i \q(1+2N)+O_{\infty}(\rho_+)\right](r^{-1}\Omega)^{2-p}\re\left[\psi^{(N)}\partial_{\rho_+}\overline{\psi}^{(N)}\right]\\
-\frac{r^2}{2}\frac{d\chi_{r_H}}{dr}\left[(r^{-1}\Omega)^{4-p}|\partial_{\rho_+}\psi^{(N)}|^2-|\snabla_{\s^2}\psi^{(N)}|^2+N\frac{2-p}{2}(r^{-1}\Omega)^{-p}\frac{d}{dr}(r^{-2}\Omega^2)\re\left[\snabla_{\s^2}\psi^{(N-1)}\cdot \snabla_{\s^2}\overline{\psi}^{(N)}\right]\right]
\end{multline*}
Using that $\frac{d}{dr}(r^{-2}\Omega^2)=2\kappa_++2(1-6\kappa_+r_+)\rho_++O_{\infty}(\rho_+^2)$ by Lemma \ref{lm:metricest} and integrating the above identity, we obtain the estimate:

\begin{multline}
\label{eq:keyhorpest}
\int_{\Sigma_{\tau_2}\cap\{r\leq r_H\}}(r^{-1}\Omega)^{2-p}	|\partial_{\rho_+}\psi^{(N)}|^2\,d\sigma dr\\
+\int_{\tau_1}^{\tau_2}\int_{\Sigma_{\tau}\cap\{r\leq r_H\}}(r^{-1}\Omega)^{2-p}(\kappa_++\rho_+)\left[	|\partial_{\rho_+}\psi^{(N)}|^2+(2-p)(r^{-1}\Omega)^{-2}|\snabla_{\s^2}\psi^{(N)}|^2\right]\,d\sigma dr\\
\leq \int_{\Sigma_{\tau_1}}\chi_{r_H}(r^{-1}\Omega)^{2-p}	|\partial_{\rho_+}\psi^{(N)}|^2\,d\sigma dr+C\sum_{k+n\leq N}\int_{\tau_1}^{\tau_2}\int_{\Sigma_{\tau_2}\cap\{r\leq r_H\}}\chi_{r_H}(r^{-1}\Omega)^{-p}(\rho_++\kappa_+)|\snabla_{\s^2}^{k}\psi^{(n)}|^2\,d\sigma dr d\tau\\
+C \sum_{k+n\leq N}\int_{\tau_1}^{\tau_2}\int_{\Sigma_{\tau_2}\cap\{r_H\leq r\leq r_H+\frac{1}{2}(r_H-r_+)\}}\left[|\snabla_{\s^2}^{k+1}\psi^{(n)}|^2+|\snabla_{\s^2}^{k}\psi^{(n)}|^2\right]\,d\sigma dr d\tau\\
+C\int_{\tau_1}^{\tau_2}\int_{\Sigma_{\tau}}\chi_{r_H}(r^{-1}\Omega)^{2-p}(\rho_++\kappa_+)^{-1}\rho_+^{-2}|(\rho_+\partial_{\rho_+})^N(\rho_+ r^3G_{\widehat{A}}) |^2\,d\sigma dr d\tau.
\end{multline}
Note that $|\psi^{(N)}|^2=\rho_+^2|\partial_{\rho_+}\psi^{(N-1)}|^2$. Hence, we can estimate via a proof by induction. In the $N=1$ case, the spacetime integrals on the RHS of \eqref{eq:keyhorpest} involving $\psi$ and its derivatives can be estimated by \eqref{eq:iedpaper1}.  Suppose that \eqref{eq:keyhorpest} without the spacetime integrals involving  $\psi$ on the RHS holds for $N$ replaced with $n$, where $n\leq N-1$. Then we can apply \eqref{eq:keyhorpest} to conclude the induction step. The spacetime integrals supported in $\{r_H\leq r\leq r_H+\frac{1}{2}(r_H-r_+)\}$ can finally be absorbed using \eqref{eq:iedpaper1}.
\end{proof}

\begin{corollary}
Let $\gamma\in \N_0$ and $\epsilon>0$. Let $1< p<\min\{1+\re \beta_{\ell},2\}+\epsilon$. Then
\begin{multline}
\label{eq:horpestrmin1pest}
\sum_{|\gamma|\leq N}\int_{\Sigma_{\tau_2}} \mathcal{E}_{p-2\epsilon}[\mathbf{Z}^{\gamma}\psi_{\geq \ell}]\,d\sigma dr+\int_{\tau_1}^{\tau_2}\int_{\Sigma_{\tau}\cap \{r_H\leq r\leq r_I\}}  \upzeta \mathcal{E}_p[\mathbf{Z}^{\gamma}\psi_{\geq \ell}]+|\mathbf{Z}^{\gamma}\psi_{\geq \ell}|^2\,d\sigma dr d\tau\\
+\int_{\tau_1}^{\tau_2}\int_{\Sigma_{\tau}\setminus \{r_H\leq r\leq r_I\}} (\rho_++\kappa_+)r^{-1}(r^{-1}\rho_+)^{2\epsilon}\mathcal{E}_{p}[\mathbf{Z}^{\gamma}\psi_{\geq \ell}]+(\rho_++\kappa_+)r^{-3}(r^{-1}\Omega)^{\epsilon-2}|D_T\mathbf{Z}^{\gamma}\psi_{\geq \ell}|^2\,d\sigma dr d\tau\\
\leq C\sum_{|\gamma|\leq N}\sum_{n\leq N}\int_{\Sigma_{\tau_1}} \mathcal{E}_{p}[\mathbf{Z}^{\gamma}\psi_{\geq \ell}]\,d\sigma dr\\
+C\int_{\tau_1}^{\tau_2}\int_{\Sigma_{\tau}} \max\{(r^{-1}\Omega)^{-p}\rho_+^{1-2\epsilon}r^{-1+2\epsilon},1\}r^{-2}|\mathbf{Z}^{\gamma}(r^3G_{\widehat{A}})_{\geq \ell}|^2+\Omega^2r^{-2}(1-\upzeta) |\mathbf{Z}^{\gamma}T( r^3G_{\widehat{A}})_{\geq \ell} |^2\,d\sigma dr d\tau.
\end{multline}

\end{corollary}

\section{Energy time-decay of time derivatives}
\label{sec:edecaytimeder}
In this section, we establish decay for the energies corresponding to the energy densities $\mathcal{E}_p[D_{\mathbf{Z}}^{\gamma}(TK)^{N_1}\psi]$ (assuming appropriate decay for $G_{\widehat{A}}$), with $|\gamma|\leq N_2$ and $N_1\in \N_1$, $N_2\in \N_0$ arbitrary. While we do not establish energy time-decay for $\mathcal{E}_p[D_{\mathbf{Z}}^{\gamma}\psi]$, we establish time-decay after commuting with the operator $T\circ K$ and show that the decay rate can be improved as we apply $T\circ K$ additional times. 
\begin{proposition}
\label{prop:edecay}
Let $N_1\in \N_1$, $N_2\in \N_0$ and $\q_1\in [0,\infty)$. Let $\psi$ be a solution to \eqref{eq:rCSF} with $A=\widehat{A}$ and $|\q|\leq \q_1$. For all $\nu>0$, there exists an $\epsilon_0>0$, such that for all $0\leq \epsilon\leq \epsilon_0$ and $1+\epsilon\leq p\leq 1+\re\sqrt{1-4\q_1^2}+\epsilon$, there exists a constant $C=C(\h,\nu,\epsilon_0,p, N_1,N_2,\q_1)>0$ such that:
\begin{multline}
\label{eq:energydecay}
(1+\tau)^{2N_1-\nu}\int_{\Sigma_{\tau}}\sum_{|\gamma|\leq N_2}\mathcal{E}_{p}[\mathbf{Z}^{\gamma}(TK)^{N_1}\psi]\,d\sigma dr \\
\leq C\sum_{m\leq N_1}\sum_{|\gamma|+m\leq N_1+N_2}\sum_{j\leq 2N_1}\int_{\Sigma_{0}} \mathcal{E}_{p}[\mathbf{Z}^{\gamma}(TK)^mT^j\psi]\,d\sigma dr\\
+\sup_{\tau\in [0,\infty)}\int_{\Sigma_{\tau}} (1+\tau)^{2-\epsilon+2\max\{m-1,0\}}(\Omega r^{-1})^{2-p}r^{-2}|r^3\mathbf{Z}^{\gamma}T^j(TK)^{\max\{m-1,0\}}G_{\widehat{A}}|^2\,d\sigma dr d\tau\\
+\int_{0}^{\infty} \int_{\Sigma_{\tau}} (1+\tau)^{2m}\Big[\max\{(r^{-1}\Omega)^{-p}\rho_+^{1-2\epsilon}r^{-1+2\epsilon},1\}r^{-2}|r^3\mathbf{Z}^{\gamma}T^j(TK)^mG_{\widehat{A}}|^2\\
+\Omega^2(1-\upzeta)|r\mathbf{Z}^{\gamma}T^{j+1}(TK)^mG_{\widehat{A}}|^2\Bigg]\,d\sigma dr d\tau.
\end{multline}
\end{proposition}
\begin{proof}
We set $A=\widehat{A}$. For the sake of notational convenience, we write $G=G_{\widehat{A}}$. We first consider the case $N_2=0$ and $N_1=1$. Let $0\leq \sigma<1$ and define $\mathcal{D}_{\sigma,\tau_A}:=\{\rho_+\leq \frac{1}{2}(1+\tau_A)^{- \sigma}\}\cap \{r\geq 4(1+\tau_A)^{\sigma}\}$. Let $\epsilon'>0$. Then we can split the following integral:
\begin{multline*}
\int_{\tau_A}^{\tau_B}\int_{\Sigma_{\tau}}\mathcal{E}_{p'}[TK\psi]\,d\sigma dr d\tau=\int_{\tau_A}^{\tau_B}\int_{\Sigma_{\tau} \cap \mathcal{D}_{\sigma,\tau_A} }\mathcal{E}_{p'}[TK\psi]\,d\sigma dr d\tau+\int_{\tau_A}^{\tau_B}\int_{\Sigma_{\tau}\setminus \mathcal{D}_{\sigma,\tau_A}}\mathcal{E}_{p'}[TK\psi]\,d\sigma dr d\tau\\
\lesssim \overbrace{\int_{\tau_A}^{\tau_B}\int_{\Sigma_{\tau}\cap \mathcal{D}_{\sigma,\tau_A}}(\Omega r^{-1})^{2-p'}r^2|XTK\psi|^2\,d\sigma dr d\tau}^{=:J_1}+\overbrace{\int_{\tau_A}^{\tau_B}\int_{\Sigma_{\tau}\cap \mathcal{D}_{\sigma,\tau_A} }|TTK\psi|^2+|\snabla_{\s^2}TK\psi|^2+|TK\psi|^2\,d\sigma dr d\tau}^{=:J_2}\\
+\overbrace{\int_{\tau_A}^{\tau_B}\int_{\Sigma_{\tau} \setminus  (\mathcal{D}_{\sigma,\tau_A}\cup \{r^{\sharp}-\eta\leq r\leq r^{\sharp}+\eta\} ) }\mathcal{E}_{p'}[TK\psi]\,d\sigma dr d\tau}^{=:J_3}+\overbrace{\int_{\tau_A}^{\tau_B}\int_{\Sigma_{\tau} \cap \{r_{\sharp}-\eta\leq r\leq r_{\sharp}+\eta\}  }\mathcal{E}_{0}[TK\psi]\,d\sigma dr d\tau}^{=:J_4}.
\end{multline*}
We will estimate each of the terms $J_i$ separately.

Take $1+\epsilon'\leq p'\leq 1+\re \sqrt{1-4\q^2}+\epsilon'$, where $\epsilon'>0$ will be chosen suitably small.

To estimate $J_1$, we use that by \eqref{eq:maineqradfield} and \eqref{eq:maineqradfieldconfhor}
\begin{equation}
\label{eq:TKvsZ}
r^2|XTK\psi|^2\lesssim \sum_{|\gamma|\leq 1} \mathcal{E}_{0}[\mathbf{Z}^{\gamma}T\psi]+\mathcal{E}_{0}[\mathbf{Z}^{\gamma}\psi]+ r^{-2}|r^3TG|^2+r^{-2}|r^3G|^2.
\end{equation}
Hence,
\begin{multline*}
J_1\lesssim \sum_{|\gamma|\leq 1}\int_{\tau_A}^{\tau_B}\int_{\Sigma_{\tau}\cap \mathcal{D}_{\sigma,\tau_A}} (\Omega r^{-1})^{2-p'}\left[\mathcal{E}_{0}[\mathbf{Z}^{\gamma}T\psi]+\mathcal{E}_{0}[\mathbf{Z}^{\gamma}\psi]\right]\,d\sigma dr d\tau\\
+\int_{\tau_A}^{\tau_B}\int_{\Sigma_{\tau}} (\Omega r^{-1})^{2-p'}r^{-2}(|r^3TG|^2+|r^3G|^2)\,d\sigma dr d\tau\\
\lesssim (1+\tau_A)^{-\sigma(1+p-2\epsilon-p')} \sum_{|\gamma|\leq 1}\int_{\tau_A}^{\tau_B}\int_{\Sigma_{\tau}\cap \mathcal{D}_{\sigma,\tau_A}} (\Omega r^{-1})^{1-p+2\epsilon}\left[\mathcal{E}_{0}[\mathbf{Z}^{\gamma}T\psi]+\mathcal{E}_{0}[\mathbf{Z}^{\gamma}\psi]\right]\,d\sigma dr d\tau\\
+ (1+\tau_A)^{-\sigma(1+p-p')} \int_{\tau_A}^{\tau_B}\int_{\Sigma_{\tau}} (1+\tau)^{\sigma(1+p-p')} (\Omega r^{-1})^{2-p'}(|rTG|^2+|rG|^2)\,d\sigma dr d\tau\\
\lesssim (1+\tau_A)^{-\sigma(1+p-2\epsilon-p')} \sum_{|\gamma|\leq 1}\int_{\tau_A}^{\tau_B}\int_{\Sigma_{\tau}\cap \mathcal{D}_{\sigma,\tau_A}} \mathcal{E}_{p-1-2\epsilon}[\mathbf{Z}^{\gamma}T\psi]+\mathcal{E}_{p-1+2\epsilon}[\mathbf{Z}^{\gamma}\psi]\,d\sigma dr d\tau\\
+ (1+\tau_A)^{-\sigma(1+p-2\epsilon-p')} \int_{\tau_A}^{\tau_B}\int_{\Sigma_{\tau}} (1+\tau)^{\sigma(1+p-p'-2\epsilon)} (\Omega r^{-1})^{2-p'}r^{-2}(|r^3TG|^2+|r^3G|^2)\,d\sigma dr
\end{multline*}
We now apply \eqref{eq:horpestrmin1pest} with $p=p'$, $\epsilon=2\epsilon'$ and $|\gamma|=1$ to obtain:
\begin{multline*}
J_1\lesssim  (1+\tau_A)^{-\sigma(1-4\epsilon')}\sum_{|\gamma|\leq 1}\int_{\Sigma_{\tau_A}} \mathcal{E}_{p' }[\mathbf{Z}^{\gamma}\psi]+\mathcal{E}_{p'}[\mathbf{Z}^{\gamma}T\psi]\,d\sigma dr\\
+ (1+\tau_A)^{-\sigma(1-4\epsilon')}\sum_{|\gamma|\leq 1}\int_{\tau_A}^{\tau_B} \int_{\Sigma_{\tau}}\max\{(r^{-1}\Omega)^{-p'}\rho_+^{1-4\epsilon'}r^{-1+4\epsilon'},1\}r^{-2}(|\mathbf{Z}^{\gamma}( r^3G) |^2+|\mathbf{Z}^{\gamma}( r^3TG) |^2)\\
+(1-\upzeta)\Omega^2r^{-2}|\mathbf{Z}^{\gamma}T^{2}( r^3G) |^2\,d\sigma dr d\tau\\
+ (1+\tau_A)^{-\sigma(1-4\epsilon') } \int_{\tau_A}^{\tau_B}\int_{\Sigma_{\tau}} (1+\tau)^{\sigma(1-4\epsilon') }  (\Omega r^{-1})^{2-p'}r^{-2}(|r^3TG|^2+|r^3G|^2)\,d\sigma dr.
\end{multline*}

We can estimate $J_2$ as follows:
\begin{multline*}
J_2\lesssim (1+\tau_A)^{-\sigma(1-2\epsilon)}\int_{\tau_A}^{\tau_B}\int_{\Sigma_{\tau}\cap \mathcal{D}_{\sigma,\tau_A} } (\Omega r^{-1})^{-1+2\epsilon}\left[|TTK\psi|^2+|\snabla_{\s^2}TK\psi|^2+|TK\psi|^2\right]\,d\sigma dr d\tau\\
\lesssim (1+\tau_A)^{-\sigma(1-2\epsilon)}\int_{\tau_A}^{\tau_B}\int_{\Sigma_{\tau}\cap \mathcal{D}_{\sigma,\tau_A} }(r^{-1}\Omega)^{-1+2\epsilon}(|D_TT^2\psi|^2+|\snabla_{\s^2}D_TT\psi|^2+|D_TT\psi|^2)\\
+(r^{-1}\Omega)^{-1+2\epsilon}\rho_+^{2}(|T\psi|^2+|\snabla_{\s^2}T\psi|^2+|T^2\psi|^2)\,d\sigma dr d\tau.
\end{multline*}
We apply \eqref{eq:horpest} with $|\gamma|=0$, $p=p'$, $\epsilon=2\epsilon'$, to $T\psi$ and $T^2\psi$ to obtain:
\begin{multline*}
J_2\lesssim C (1+\tau_A)^{-\sigma(1-4\epsilon')}\int_{\Sigma_{\tau_A}} \mathcal{E}_{p' }[T\psi]+\mathcal{E}_{p'}[T^2\psi]\,d\sigma dr\\
+C(1+\tau_A)^{-\sigma(1-4\epsilon')}\int_{\tau_A}^{\tau_B} \int_{\Sigma_{\tau}}\max\{(r^{-1}\Omega)^{-p'}\rho_+^{1-4\epsilon'}r^{-1+4\epsilon'},1\}r^{-2}(| r^3TG |^2+| r^3T^2G |^2)\\
+(1-\upzeta)\Omega^2r^{-2}|T^{3}( r^3G) |^2\,d\sigma dr d\tau.
\end{multline*}

We estimate $J_3$:
\begin{equation*}
J_3\lesssim (1+\tau_A)^{\sigma(1+4\epsilon')}\int_{\tau_A}^{\tau_B}\int_{\Sigma_{\tau} \setminus  \{r_{\sharp}-\eta\leq r\leq r_{\sharp}+\eta\}  }\mathcal{E}_{p'-1-4\epsilon'}[TK\psi]\,d\sigma dr d\tau.
\end{equation*}
We apply \eqref{eq:horpest} with $|\gamma|=0$, $p=p'$ and $\epsilon=2\epsilon'$ to $TK\psi$ to obtain:
\begin{multline*}
J_3\lesssim  C (1+\tau_A)^{\sigma(1+4\epsilon')}\int_{\Sigma_{\tau_A}} \mathcal{E}_{p'}[TK\psi]\,d\sigma dr\\
+C(1+\tau_A)^{\sigma(1+4\epsilon')}\sum_{j=0}^1\int_{\tau_A}^{\tau_B} \int_{\Sigma_{\tau}}\max\{(r^{-1}\Omega)^{-p'}\rho_+^{1-4\epsilon'}r^{-1+4\epsilon'},1\}r^{-2}| r^3T^jTKG |^2+(1-\upzeta)\Omega^2r^{-2}| r^3T^{j+1}TKG |^2\,d\sigma dr d\tau.
\end{multline*}
Finally, we can estimate $J_4$ by directly applying \eqref{eq:horpestrmin1pest} with $|\gamma|=0$, $p=1+\epsilon'$ and $\epsilon=2\epsilon'$ to obtain:
\begin{multline*}
J_4\lesssim C \int_{\Sigma_{\tau_A}} \mathcal{E}_{1+\epsilon'}[TK\psi]+ \mathcal{E}_{1+\epsilon'}[T^2K\psi]\,d\sigma dr\\
+C\int_{\tau_A}^{\tau_B} \int_{\Sigma_{\tau}}\max\{(r^{-1}\Omega)^{-1-\epsilon'}\rho_+^{1-4\epsilon'}r^{-1+4\epsilon'},1\}r^{-2}| r^3TKG |^2+(1-\upzeta)\Omega^2r^{-2}|T( r^3TKG) |^2\,d\sigma dr d\tau.
\end{multline*}

We conclude that:
\begin{multline}
\label{eq:mainestedecay}
\int_{\tau_A}^{\tau_B}\int_{\Sigma_{\tau}}\mathcal{E}_{p'}[TK\psi]\,d\sigma dr d\tau\lesssim  (1+\tau_A)^{\sigma(1+4\epsilon')}\int_{\Sigma_{\tau_A}} \mathcal{E}_{p'}[TK\psi]\,d\sigma dr+ \int_{\Sigma_{\tau_A}} \mathcal{E}_{p'}[TK\psi]+ \mathcal{E}_{p'}[T^2K\psi]\,d\sigma dr\\
+(1+\tau_A)^{-\sigma(1-4\epsilon')}\sum_{|\gamma|\leq 1}\int_{\Sigma_{\tau_A}} \mathcal{E}_{p'}[\mathbf{Z}^{\gamma}\psi]+\mathcal{E}_{p'}[\mathbf{Z}^{\gamma}T\psi]\,d\sigma dr\\
+(1+\tau_A)^{-\sigma(1-4\epsilon')}\int_{\tau_A}^{\tau_B} \sum_{j=0}^1\int_{\Sigma_{\tau}}\max\{(r^{-1}\Omega)^{-p'}\rho_+^{1-4\epsilon'}r^{-1+4\epsilon'},1\}(1+\tau)^{2\sigma(1-4\epsilon')}r^{-2}| r^3T^jTKG |^2\\
+(1-\upzeta)\Omega^2r^{-2}(1+\tau)^{2\sigma(1-4\epsilon')}| r^3T^{j+1}TKG |^2\,d\sigma dr d\tau\\
+ (1+\tau_A)^{-\sigma(1-4\epsilon')}\sum_{|\gamma|\leq 1}\int_{\tau_A}^{\tau_B} \int_{\Sigma_{\tau}}\max\{(r^{-1}\Omega)^{-p'}\rho_+^{1-4\epsilon'}r^{-1+4\epsilon'},1\}r^{-2}(|\mathbf{Z}^{\gamma}( r^3G) |^2+|\mathbf{Z}^{\gamma}( r^3TG) |^2)\\
+(1-\upzeta)\Omega^2r^{-2}|\mathbf{Z}T^{2}( r^3G) |^2\,d\sigma dr d\tau\\
+ (1+\tau_A)^{-\sigma (1-4\epsilon')} \int_{\tau_A}^{\tau_B}\int_{\Sigma_{\tau}} (1+\tau)^{\sigma(1-4\epsilon') }  (\Omega r^{-1})^{2-p'}r^{-2}(|r^3TG|^2+|r^3G|^2)\,d\sigma dr\\
\leq (1+\tau_A)^{\sigma(1+4\epsilon')}\int_{\Sigma_{\tau_A}} \mathcal{E}_{p'}[TK\psi]\,d\sigma dr+ \int_{\Sigma_{\tau_A}} \mathcal{E}_{p'}[TK\psi]+ \mathcal{E}_{p'}[T^2K\psi]\,d\sigma dr\\
+(1+\tau_A)^{-\sigma(1-4\epsilon')}\sum_{|\gamma|\leq 1}\int_{\Sigma_{\tau_A}} \mathcal{E}_{p'}[\mathbf{Z}^{\gamma}\psi]+\mathcal{E}_{p'}[\mathbf{Z}^{\gamma}T\psi]\,d\sigma dr+(1+\tau_A)^{-\sigma(1-4\epsilon')}\mathcal{G}[G],
\end{multline}
with
\begin{multline*}
	\mathcal{G}[G]:=\sum_{j\leq 1}\int_{0}^{\infty} \int_{\Sigma_{\tau}}\max\{(r^{-1}\Omega)^{-p'}\rho_+^{1-4\epsilon'}r^{-1+4\epsilon'},1\}(1+\tau)^{2\sigma(1-4\epsilon')}r^{-2}| r^3T^jTKG |^2\\
	+(1-\upzeta)\Omega^2r^{-2}(1+\tau)^{2\sigma(1-4\epsilon')}| r^3T^{j+1}TKG |^2\,d\sigma dr d\tau\\
+ \sum_{|\gamma|\leq 1}\int_{0}^{\infty} \int_{\Sigma_{\tau}}\max\{(r^{-1}\Omega)^{-p'}\rho_+^{1-4\epsilon'}r^{-1+4\epsilon'},1\}r^{-2}|(\mathbf{Z}^{\gamma}( r^3G) |^2+|\mathbf{Z}^{\gamma}( r^3TG) |^2)+(1-\upzeta)\Omega^2r^{-2}|\mathbf{Z}^{\gamma}T^{2}( r^3G) |^2\,d\sigma dr d\tau\\
+  \sup_{\tau\in[0,\infty)}\int_{\Sigma_{\tau}} (1+\tau)^{\sigma(1-4\epsilon')+1+\epsilon'}  (\Omega r^{-1})^{2-p'}r^{-2}(|r^3TG|^2+|r^3G|^2)\,d\sigma dr.
\end{multline*}

\paragraph{\underline{Step 1: $\tau^{-1+2\epsilon'}$ }}
Let $\{\tau_j^{(1)}\}$ be a dyadic sequence. Taking $\sigma=0$, $\tau_A=\tau_j^{(1)}$ and $\tau_{B}=\tau_{j+1}^{(1)}$ in \eqref{eq:mainestedecay} and applying the mean-value theorem (``pigeonhole principle''), we obtain along a dyadic subsequence  $\{\tau_j^{(2)}\}$ of $\{\tau_j^{(1)}\}$:
\begin{multline}
\label{eq:firstedecayest}
(1+\tau_j^{(2)})\int_{\Sigma_{\tau_j^{(2)}}}\mathcal{E}_{p'}[TK\psi]\,d\sigma dr d\tau\lesssim  \int_{\Sigma_{0}} \mathcal{E}_{p'}[TK\psi]\,d\sigma dr+ \int_{\Sigma_{0}}   \mathcal{E}_{p'}[TK\psi]+ \mathcal{E}_{p'}[T^2K\psi]\,d\sigma dr\\
+\sum_{|\gamma|\leq 1}\int_{\Sigma_{0}}  \mathcal{E}_{p'}[\mathbf{Z}^{\gamma}\psi]+\mathcal{E}_{p'}[\mathbf{Z}^{\gamma}T\psi]\,d\sigma dr+\mathcal{G}[G]=:D_{0,p',\epsilon'}[\psi].
\end{multline}
Applying \eqref{eq:horpestrmin1pest} with $|\gamma|=0$, $p=p'$ and $\epsilon=2\epsilon'$, we in fact obtain for all $\tau\geq 0$:
\begin{equation*}
\int_{\Sigma_{\tau}}\mathcal{E}_{p'-4\epsilon'}[TK\psi]\,d\sigma dr d\tau\lesssim (1+\tau)^{-1}D_{0,p',\epsilon'}[\psi].
\end{equation*}
Now we apply an interpolation argument. We split
\begin{multline*}
\int_{\Sigma_{\tau}}\mathcal{E}_{p'}[TK\psi]\,d\sigma dr d\tau\lesssim \int_{\Sigma_{\tau}}(\Omega r^{-1})^{2-p'}r^2|XTK\psi|^2\,d\sigma dr d\tau+ \int_{\Sigma_{\tau}}\mathcal{E}_{p'-4\epsilon'}[TK\psi]\,d\sigma dr d\tau\\
=\int_{\Sigma_{\tau}\cap\{r^{-1}\rho_+\leq (1+\tau)^{-1}\}}(\Omega r^{-1})^{2-p'}r^2|XTK\psi|^2\,d\sigma dr d\tau+\int_{\Sigma_{\tau}\cap\{r^{-1}\rho_+\geq (1+\tau)^{-1}\}}(\Omega r^{-1})^{2-p'}r^2|XTK\psi|^2\,d\sigma dr d\tau\\
+ \int_{\Sigma_{\tau}}\mathcal{E}_{p'-4\epsilon'}[TK\psi]\,d\sigma dr d\tau
\end{multline*}
and then apply \eqref{eq:TKvsZ} to further estimate the second term on the RHS, so that we obtain:
\begin{multline*}
\int_{\Sigma_{\tau}}\mathcal{E}_{p'}[TK\psi]\,d\sigma dr d\tau\lesssim (1+\tau)^{4\epsilon'} \int_{\Sigma_{\tau}}\mathcal{E}_{p'-4\epsilon'}[TK\psi]\,d\sigma dr d\tau\\
+(1+\tau)^{-1+4\epsilon'}\sum_{|\gamma|\leq 1}\int_{\Sigma_{\tau}} \mathcal{E}_{p'-4\epsilon'}[\mathbf{Z}^{\gamma}T\psi]+\mathcal{E}_{p'-4\epsilon'}[\mathbf{Z}^{\gamma}\psi]\,d\sigma dr d\tau\\
+ (1+\tau)^{-1+4\epsilon'}\int_{\Sigma_{\tau}}(1+\tau)^{1-4\epsilon'}(\Omega r^{-1})^{2-p'}r^{-2}(|r^3TG|^2+|r^3G|^2)\,d\sigma dr.
\end{multline*}
By applying \eqref{eq:horpest} with $|\gamma|=1$ to $\psi$ and $T\psi$, $p=p'$ and $\epsilon=2\epsilon'$ once more, we conclude the first global-in-$\tau$ decay result:
\begin{equation}
\label{eq:firstedecayestpart2}
\int_{\Sigma_{\tau}}\mathcal{E}_{p'}[TK\psi]\,d\sigma dr d\tau\lesssim (1+\tau)^{-1+4\epsilon'}D_{0,p',\epsilon'}[\psi].
\end{equation}
Note that the above interpolation argument can also be used to conclude the $N_1=0$ case of the theorem.\\
\\

\paragraph{\underline{Step 2: $\tau^{-\frac{3}{2}}$ decay}}
We will now improve the decay rate. We apply \eqref{eq:mainestedecay} again and we moreover apply \eqref{eq:firstedecayestpart2} to estimate the first two integrals on the RHS. We obtain:
\begin{multline}
\label{eq:edecaystep2}
\int_{\tau_A}^{\tau_B}\int_{\Sigma_{\tau}}\mathcal{E}_{p'}[TK\psi]\,d\sigma dr d\tau\lesssim  (1+\tau_A)^{\sigma-1+4\epsilon'}D_{0,p',\epsilon'}[\psi]+(1+\tau_A)^{-1+4\epsilon'}D_{0,p',\epsilon'}[T\psi]\\
+(1+\tau_A)^{-\sigma(1-4\epsilon')}\sum_{|\gamma|\leq 1}\int_{\Sigma_{\tau_A}} \mathcal{E}_{p'}[\mathbf{Z}^{\gamma}\psi]+\mathcal{E}_{p'}[\mathbf{Z}^{\gamma}T\psi]\,d\sigma dr+(1+\tau_A)^{-\sigma(1-4\epsilon')}\mathcal{G}[G].
\end{multline}
Let $\{\tau_j^{(1)}\}$ be a dyadic sequence. Taking $\sigma=\frac{1}{2}$, $\tau_A=\tau_j^{(1)}$ and $\tau_{B}=\tau_{j+1}^{(1)}$ in \eqref{eq:edecaystep2} and applying the mean-value theorem, we obtain along a dyadic subsequence  $\{\tau_j^{(2)}\}$ of $\{\tau_j^{(1)}\}$:
\begin{multline}
\label{eq:edecaystep2b}
(1+\tau_j^{(2)})^{\frac{3}{2}-4\epsilon'}\int_{\Sigma_{\tau_j^{(2)}}}\mathcal{E}_{p'}[TK\psi]\,d\sigma dr d\tau\lesssim  D_{0,p',\epsilon'}[\psi]+(1+\tau_j^{(2)})^{-\frac{1}{2}}D_{0,p',\epsilon'}[T\psi]\\
+\sum_{|\gamma|\leq 1}\int_{\Sigma_{0}} \mathcal{E}_{p'}[\mathbf{Z}^{\gamma}\psi]+\mathcal{E}_{p'}[\mathbf{Z}^{\gamma}T\psi]\,d\sigma dr+\mathcal{G}[G]\\
=:D_{1,p',\epsilon'}[\psi]+(1+\tau_j^{(2)})^{-\frac{1}{2}}D_{0,p',\epsilon'}[T\psi].
\end{multline}
We repeat the interpolation argument at the end of Step 1 by considering $\{r^{-1}\rho_+\leq (1+\tau)^{-\frac{3}{2}}\}$ and $\{r^{-1}\rho_+\geq (1+\tau)^{-\frac{3}{2}}\}$ to conclude that
\begin{equation*}
\int_{\Sigma_{\tau}}\mathcal{E}_{p'}[TK\psi]\,d\sigma dr \lesssim \tilde{D}_{1,p',\epsilon'}[\psi](1+\tau)^{-\frac{3}{2}+8\epsilon'},
\end{equation*}
with
\begin{equation*}
\tilde{D}_{1,p',\epsilon'}[\psi]:=D_{1,p,\epsilon'}[\psi]+D_{0,p',\epsilon'}[T\psi].
\end{equation*}

\paragraph{\underline{Step 3: $\tau^{-2+\nu}$ decay}}
Note that the only term that does not decay at least like $\tau^{-1+4\epsilon'}$ on the RHS of \eqref{eq:edecaystep2} when $\sigma=1$ is $(1+\tau_A)^{\sigma-1+4\epsilon'}D_{0,p',\epsilon'}[\psi]$.

By Step 2, we have that \eqref{eq:edecaystep2} still holds with $(1+\tau_A)^{\sigma-1+4\epsilon'}D_{0,p',\epsilon'}[\psi]$ replaced by\\ $(1+\tau_A)^{\sigma-\frac{3}{2}+4\epsilon'}\tilde{D}_{1, p',\epsilon'}[\psi]$.

We can therefore take $\sigma=\frac{3}{4}$ and conclude that:
\begin{equation*}
\int_{\Sigma_{\tau}}\mathcal{E}_{p'}[TK\psi]\,d\sigma dr d\tau\lesssim \tilde{D}_{1,p',\epsilon'}[\psi](1+\tau)^{-\frac{7}{4}+12\epsilon'}.
\end{equation*}
Continuing this procedure iteratively, taking $\sigma =1-\frac{1}{2^k}$, we can show:
\begin{equation*}
\int_{\Sigma_{\tau}}\mathcal{E}_{p'}[TK\psi]\,d\sigma dr d\tau\lesssim \tilde{D}_{1,p',\epsilon'}[\psi](1+\tau)^{2-\frac{1}{2k}+4k\epsilon'}.
\end{equation*}

We conclude that for all $\nu>0$, there exists a suitably small $\epsilon'>0$, such that for $1+\epsilon'\leq p'\leq 1+\re\sqrt{1-4q^2}+\epsilon'$:
\begin{multline*}
(1+\tau)^{2-\nu}\int_{\Sigma_{\tau}}\mathcal{E}_{p'}[TK\psi]\,d\sigma dr d\tau\lesssim \sum_{|\gamma|+m\leq 1}\sum_{j\leq 2}\int_{\Sigma_{0}} \mathcal{E}_{p'}[\mathbf{Z}^{\gamma}(TK)^mT^j\psi]\,d\sigma dr\\
+\sum_{j=0}^2\int_{0}^{\infty} \int_{\Sigma_{\tau}}(1+\tau)^{2-4\epsilon'}\left[\max\{(r^{-1}\Omega)^{-p'}\rho_+^{1-4\epsilon'}r^{-1+4\epsilon'},1\}r^{-2}| r^3T^j(TK)G |^2+(1-\upzeta)\Omega^2r^{-2}| r^3T^{1+j}
TKG |^2\right]\,d\sigma dr d\tau\\
+ \sum_{|\gamma|\leq 1}\sum_{j=0}^1\int_{0}^{\infty} \int_{\Sigma_{\tau}}\max\{(r^{-1}\Omega)^{-p'}\rho_+^{1-4\epsilon'}r^{-1+4\epsilon'},1\}r^{-2}(|\mathbf{Z}^{\gamma}( r^3T^jG) |^2+|\mathbf{Z}^{\gamma}( r^3T^{j+1}G) |^2)\\
+(1-\upzeta)\Omega^2r^{-2}|\mathbf{Z}^{\gamma}T^{2+j}( r^3G) |^2\,d\sigma dr d\tau\\
+  \sum_{j=0}^1\sup_{\tau\in[0,\infty)}\int_{\Sigma_{\tau}} (1+\tau)^{2-2\epsilon'}  (\Omega r^{-1})^{2-p'}r^{-2}(|r^3T^{j+1}G|^2+|r^3T^jG|^2)\,d\sigma dr.
\end{multline*}

\paragraph{\underline{Step 4: general $N_1,N_2$}}
Now we consider the general $N_1\in \N_1$, $N_2\in \N_0$ case. We repeat Steps 1 and 2 above but we generalize \eqref{eq:TKvsZ} as follows:
\begin{equation*}
\sum_{|\gamma|\leq N-1}r^2|XTK\mathbf{Z}^{\gamma}\psi|^2\lesssim \sum_{|\gamma|\leq N} \mathcal{E}_{0}[\mathbf{Z}^{\gamma}T\psi]+\mathcal{E}_{0}[\mathbf{Z}^{\gamma}\psi]+r^{-2}\sum_{|\gamma|\leq N-1} |r^3\mathbf{Z}^{\gamma}TG|^2+|r^3\mathbf{Z}^{\gamma}G|^2.
\end{equation*}
We apply the above inequality $N_1$ times.
\end{proof}

\begin{corollary}
\label{prop:edecaygeneralp}           
Let $N_1\in \N_1$, $N_2\in \N_0$ and $\q_1\in [0,\infty)$. Let $|\q|\leq \q_1$. Let $1+\re \sqrt{1-4\q_1^2}< p<2+\re \sqrt{1-4\q_1^2}$. For all $\nu>0$ and $0<\epsilon<p-1-\re \sqrt{1-4\q_1^2}<1$, there exists a constant\\ $C=C(\h,\nu,\epsilon, N_1,N_2,p,\q_1)>0$ such that:
\begin{multline}
\label{eq:energydecayp}
(1+\tau)^{1+\re\sqrt{1-4\tilde{q}^2}-p+\epsilon-\nu+2N_1}\int_{\Sigma_{\tau}}\sum_{|\gamma|\leq N_2}\mathcal{E}_{p}[\mathbf{Z}^{\gamma}(TK)^{N_1}\psi]\,d\sigma dr \\
\leq C\sum_{m\leq N_1}\sum_{|\gamma|+m\leq N_1+N_2}\sum_{j\leq 2N_1}\int_{\Sigma_{0}} \mathcal{E}_{1+\re \sqrt{1-4\tilde{q}^2}+\epsilon}[\mathbf{Z}^{\gamma}(TK)^mT^j\psi]\,d\sigma dr\\
+\sup_{\tau\in [0,\infty)}\int_{\Sigma_{\tau}} (1+\tau)^{2-\epsilon+2\max\{m-1,0\}}(\Omega r^{-1})^{2-p}r^{-2}|r^3\mathbf{Z}^{\gamma}T^j(TK)^{\max\{m-1,0\}}G_{\widehat{A}}|^2\,d\sigma dr d\tau\\
+\int_{0}^{\infty} \int_{\Sigma_{\tau}} (1+\tau)^{2m}\Big[\max\{(r^{-1}\Omega)^{-1-\re \sqrt{1-4\tilde{q}^2}-\epsilon}\rho_+^{1-2\epsilon}r^{-1+2\epsilon},1\}r^{-2}|r^3\mathbf{Z}^{\gamma}T^j(TK)^mG_{\widehat{A}}|^2\\
+\Omega^2(1-\upzeta)|r\mathbf{Z}^{\gamma}T^{j+1}(TK)^mG_{\widehat{A}}|^2\Big]\,d\sigma dr d\tau.
\end{multline}
\end{corollary}
\begin{proof}
We will prove the case $N_1=1$, $N_2=0$. The general case proceeds entirely analogously. We apply an interpolation argument. Split:
\begin{multline*}
\int_{\Sigma_{\tau}}\mathcal{E}_{p}[TK\psi]\,d\sigma dr=\int_{ \Sigma_{\tau}\cap\{r^{-1}\rho_+\leq (1+\tau)^{-1}\}}\mathcal{E}_{p}[TK\psi]\,d\sigma dr\\
+\int_{ \Sigma_{\tau}\cap\{r^{-1}\rho_+> (1+\tau)^{-1}\}}\mathcal{E}_{p}TK\psi]\,d\sigma dr\\
\lesssim  \int_{ \Sigma_{\tau}\cap\{r^{-1}\rho_+\leq (1+\tau)^{-1}\}}(r^{-1}\Omega)^{2-p}r^2|XTK\psi|^2\,d\sigma dr\\
+(1+\tau)^{p-(1+\re \sqrt{1-4\tilde{q}^2} +\epsilon)}\int_{\Sigma_{\tau}}\mathcal{E}_{1+\re \sqrt{1-4\tilde{q}^2}+\epsilon}[\mathbf{Z}^{\gamma}(TK)^{N_1}\psi]\,d\sigma dr.
\end{multline*}
We use that by \eqref{eq:maineqradfield} and \eqref{eq:maineqradfieldconfhor}:
\begin{equation}
\label{eq:TKvsZv2}
r^2|XTK\psi|^2\lesssim \sum_{|\gamma|\leq 1} \mathcal{E}_{0}[\mathbf{Z}^{\gamma}T\psi]+\mathcal{E}_{0}[\mathbf{Z}^{\gamma}\psi]+ |rTG|^2+|rG|^2
\end{equation}
to estimate:
\begin{multline*}
 \int_{ \Sigma_{\tau}\cap\{r^{-1}\rho_+\leq (1+\tau)^{-1}\}}(r^{-1}\Omega)^{2-p}r^2|XTK\psi|^2\,d\sigma dr\\
 \lesssim (1+\tau)^{-3-\sqrt{1-4\tilde{q}^2}+\epsilon+\nu+p}\sum_{|\gamma|\leq 1}\sum_{j\leq 1}\int_{\Sigma_{\tau}}\mathcal{E}_{1+\sqrt{1-4\tilde{q}^2}-\epsilon-\nu}[\mathbf{Z}^{\gamma}T^j\psi]\,d\sigma dr\\
 +(1+\tau)^{-2+\epsilon}\sum_{j\leq 1}\int_{\Sigma_{\tau}} (1+\tau)^{2-\epsilon}(r^{-1}\Omega)^{2-p}[|rTG|^2+|rG|^2]\,d\sigma dr.
\end{multline*}
The first integral on the RHS above can then be estimated by \eqref{eq:horpestrmin1pest} applied to $\psi$ and $T\psi$.
\end{proof}

\section{Approximate solutions}
\label{sec:approxsol}
In this section, we will set $r_+=1$ for convenience. We define the following integral operators:
\begin{align*}
(T^{-1}f)(\tau,r,\theta,\varphi)=&-\int_{\tau}^{\infty}f(\tau',r,\theta,\varphi)\,d\tau',\\
(K^{-1}f)(\tau,r,\theta,\varphi)=&-e^{-i \q \tau}\int_{\tau}^{\infty}e^{i \q \tau'}f(\tau',r,\theta,\varphi)\,d\tau'.
\end{align*}
It is straightforward to check that for $q\neq 0$:
\begin{equation*}
TK\left(\frac{1}{-i\q}\left( K^{-1}f-T^{-1}f\right)\right)=f,
\end{equation*}
so we can express:
\begin{equation*}
	(TK)^{-1}(f)=\frac{1}{-i\q}\left( K^{-1}f-T^{-1}f\right).
\end{equation*}

We will apply the energy decay estimates of Proposition \ref{prop:edecay} and Corollary \ref{prop:edecaygeneralp} to $(TK)^{-1}(\widehat{\psi})=(TK)^{-1}(\psi-\Psi)$, where $\psi$ is a solution to \eqref{eq:CSF} with $G_A$ and $\Psi$ is a function for which $(g^{-1}_{M,e})^{\mu\nu}D^{\widehat{A}}_{\mu}D^{\widehat{A}}_{\nu}\Psi$ decays suitably in $\tau$ and $r$.

The structure of this section is as follows:
\begin{itemize}
	\item We define $\Psi$ here and establish in \S \ref{sec:boxPsismall} decay estimates for $(g^{-1}_{M,e})^{\mu\nu}\:^{\widehat{A}}D_{\mu}^{\widehat{A}}D_{\nu}(r^{-1}\Psi)$. The function $\Psi$ will depend on constants $\mathfrak{I}_{\ell m}[\psi]$ and $\mathfrak{H}_{\ell m}[\psi]$, which will be determined by initial data for $\psi$, with $\mathfrak{H}_{\ell m}[\psi]=0$ when $|Q|<M$.
	\item In \S \ref{sec:idatatimeint}, we interpret $(TK)^{-1}(\widehat{\psi})$ as solutions to the inhomogeneous equation:
	\begin{equation*}
		(g^{-1}_{M,e})^{\mu\nu}\:^{\widehat{A}}D_{\mu}\:^{\widehat{A}}D_{\nu}(r^{-1}(TK)^{-1}\widehat{\psi})=-(TK)^{-1}((g^{-1}_{M,e})^{\mu\nu}\:^{\widehat{A}}D_{\mu}^{\widehat{A}}D_{\nu}(r^{-1}\Psi))
	\end{equation*}
	arising from initial data $((TK)^{-1}(\widehat{\psi}), T(TK)^{-1}(\widehat{\psi}))|_{\Sigma_0}$ that we construct explicitly. In this step, we fix $\mathfrak{I}_{\ell m}[\psi]$ and $\mathfrak{H}_{\ell m} [\psi]$.
\end{itemize}
\subsection{Definition of tail functions}
We define the \emph{null infinity tail function} $\Psi^{\infty}: [0,\infty)_{\tau}\times [r_+,\infty)_r\times \s^2\to \C$ as follows: fix $r_+=1$, then
\begin{align*}
\Psi^{\infty}(\tau,r,\theta,\varphi):=&\:\sum_{\substack{\ell\in \N_0\\\ell(\ell+1)<q^2}}\sum_{|m|\leq \ell} \Psi^{\infty}_{\ell m}(\tau,r)Y_{\ell m}(\theta,\varphi),\\
\Psi^{\infty}_{\ell m}(\tau,r):=&\: \mathfrak{I}_{\ell m}[\psi]\frac{\mathfrak{w}_{\ell}(r)}{\mathfrak{w}^0_{\ell}(r)}(\Psi^{\infty}_0)_{\ell m}(\tau,r;\q),
\end{align*}
where
\begin{equation*}
\mathfrak{w}^0_{\ell}(r;\q):=\begin{cases}
	{\alpha}_+ r^{\frac{1}{2}+i\q+\frac{1}{2}\beta_{\ell}}+{\alpha}_- r^{\frac{1}{2}+i\q-\frac{1}{2}\beta_{\ell}}\quad (\beta_{\ell}\neq 0),\\
	 {\alpha}_+ r^{\frac{1}{2}+i\q}\log r+{\alpha}_- r^{\frac{1}{2}+i\q}\quad (\beta_{\ell}= 0),
\end{cases}	
\end{equation*}
 with ${\alpha}_{\pm}\in \C$ the constants appearing in the following asymptotics as $r\to \infty$:
\begin{equation*}
	\mathfrak{w}_{\ell}(r;\q)=\begin{cases} {\alpha}_+ r^{\frac{1}{2}+i\q+\frac{1}{2}\beta_{\ell}}(1+O(r^{-1}))+{\alpha}_- r^{\frac{1}{2}+i\q-\frac{1}{2}\beta_{\ell}}(1+O(r^{-1}))\quad (\beta_{\ell}\neq 0),\\
	 {\alpha}_+ r^{\frac{1}{2}+i\q}\log r(1+O(r^{-1}))+{\alpha}_- r^{\frac{1}{2}+i\q}(1+O(r^{-1}))\quad (\beta_{\ell}= 0);
	 \end{cases}
\end{equation*}
	 see Proposition \ref{prop:statsol}, and where
\begin{equation*}
(\Psi^{\infty}_0)_{\ell m}(\tau,r;\q):= \begin{cases}
\beta_{\ell}\sum_{n=0}^{N_{\beta_{\ell}}}\zeta^n(1+\tau)^{-\frac{1}{2}-(n+\frac{1}{2})\beta_{\ell}+i\q}\\
		\times {}_2\mathbf{F}_1\left(\frac{1}{2}+\frac{\beta_{\ell}}{2}+i\q,\frac{1}{2}-\frac{\beta_{\ell}}{2}+i\q,\frac{1}{2}+i\q-\left(n+\frac{1}{2}\right)\beta_{\ell};-\frac{\tau+1}{2r}\right)\quad \textnormal{when}\: \beta_{\ell}\in (0,1)),\\
		\mbox{}\\
		\sign(\q)\beta_{\ell}\sum_{n=0}^{\infty}\zeta^{\sign(\q)n}(1+\tau)^{-\frac{1}{2}-(n+\frac{1}{2})\sign(\q)\beta_{\ell}+i\q}\\
		\times {}_2\mathbf{F}_1\left(\frac{1}{2}+\sign(\q)\frac{\beta_{\ell}}{2}+iq,\frac{1}{2}-\sign(\q)\frac{\beta_{\ell}}{2}+i\q,\frac{1}{2}+i\q-\left(n+\frac{1}{2}\right)\sign(\q)\beta_{\ell};-\frac{\tau+1}{2r}\right)\\ \textnormal{when}\:\beta_{\ell}\in i(0,\infty),\\
		\mbox{}\\
		(\tau+1)^{-\frac{1}{2}+i\q}\int_0^{\infty}(\tau+1)^{-ix}e^{-\eta x}{}_2\mathbf{F}_1\left(\frac{1}{2}+iq,\frac{1}{2}+iq,\frac{1}{2}+iq-ix;-\frac{\tau+1}{2r}\right)\,dx\\ \textnormal{when}\:\beta_{\ell}=0,
\end{cases}
\end{equation*}
with $N_{\beta_{\ell}}=\lceil (\re \beta_{\ell})^{-1}\rceil\in \N_0$, with ${}_2\mathbf{F}_1(a,b,c;\cdot )$ denoting regularized Gauss hypergeometric functions, see for example \cite{NIST:DLMF}[\S 15.2] for a definition, and $\zeta,\eta\in \C$ defined as follows:
\begin{align*}
	\zeta:=& -\frac{2^{\beta_{\ell}}\Gamma(-\beta_{\ell}+1)\Gamma(\frac{1}{2}+ \frac{\beta_{\ell}}{2}+i\q)}{\Gamma(\beta_{\ell}+1) \Gamma(\frac{1}{2}- \frac{\beta_{\ell}}{2}+i\q)}\frac{\alpha_-}{\alpha_+},\\
	\eta:=&\: i \sign(q)\left[\frac{\alpha_-}{\alpha_+}-\log 2+2\gamma_{\rm Euler}-\frac{\Gamma'(\frac{1}{2}+i\q)}{\Gamma(\frac{1}{2}+i\q)}\right],
	\end{align*}
where $\gamma_{\rm Euler}$ is the Euler--Mascheroni constant.

Note that $(\Psi^{\infty}_0)_{\ell m}(\cdot,\cdot;\q)Y_{\ell m}$ are solutions to \eqref{eq:CSFmink} with $A=\widehat{A}$ and $\h=0$, that are smooth with respect to $(\tau,\rho_{\infty},\theta,\varphi)$; see Appendix \ref{sec:tailfunctmink}.

In the $|Q|=M$ case, we also define the \emph{event horizon tail function}:
\begin{align*}
\Psi^{+}(\tau,r,\theta,\varphi;\q):=&\:\sum_{\substack{\ell\in \N_0\\\ell(\ell+1)<\q^2}}\sum_{|m|\leq \ell}\Psi^+_{\ell m}(\tau,r)Y_{\ell m}(\theta,\varphi),\\
\Psi_{\ell m}^+(\tau,r):=&\: \frac{\mathfrak{w}_{\ell}(1+(r-1)^{-1};-\q)}{\mathfrak{w}_{\ell}^0(1+(r-1)^{-1};-\q)}(\Psi^{\infty}_0)_{\ell m}(\tau,1+(r-1)^{-1};\q),\\
\mathfrak{H}_{\ell m}[\psi]\in &\: \C.
\end{align*}
In the case $|Q|<M$, we simply set $\mathfrak{h}_{\ell m}[\psi]=0$. We define
\begin{equation*}
\Psi:=\Psi^{\infty}+\Psi^+.
\end{equation*}
We also introduce:
\begin{equation*}
G^{(1)}:=\frac{1}{i\q}(K^{-1}-T^{-1})((g^{-1})^{\mu\nu}D_{\mu}D_{\nu}(r^{-1}\Psi)).
\end{equation*}
From the estimates in \S \ref{sec:boxPsismall}, it will follow that the above integral is in fact well-defined.

For later convenience, we introduce the following notation: 
\begin{align*}
	|\mathfrak{I}[\psi]|^2:=&\:\sum_{\substack{\ell\in \N_0\\\ell(\ell+1)<\q^2}}\sum_{\substack{m\in \Z\\ |m|\leq \ell}}|\mathfrak{I}_{\ell m}[\psi]|^2,\\
	|\mathfrak{H}[\psi]|^2:=&\:\sum_{\substack{\ell\in \N_0\\\ell(\ell+1)<\q^2}}\sum_{\substack{m\in \Z\\ |m|\leq \ell}}|\mathfrak{H}_{\ell m}[\psi]|^2.
\end{align*}

\subsection{Tail functions as approximate solutions}
\label{sec:boxPsismall}
Define
\begin{equation*}
	G[f]:=-(g^{-1}_{M,Q})^{\mu\nu}\:(^{\widehat{A}}D_{\mu})(^{\widehat{A}}D_{\nu})(r^{-1}f).
\end{equation*}

In this section, we show that $G[\Psi^{\infty}]$ and $G[\Psi^{+}]$ decay suitably fast in $r$ and $\tau$ and these decay rates are preserved under commutation with either $(\tau+1) T$ or $(\tau+1) K$ and $(r-r_+)X$. In this sense, the tail functions $\Psi^+$ and $\Psi^{\infty}$ may each be considered approximate solutions to \eqref{eq:maineqradfield}.
\begin{proposition}
\label{prop:estinhomG}
Fix $r_+=1$. Let $n_1,n_2\in \N_0$. Then there exists a constant $C=C(n_1,n_2,\widetilde{\h})>0$ such that:
\begin{equation}
\label{eq:erroresteqforPsiinf}
\left|r^3((r-1)X)^{n_1}((\tau+1)T)^{n_2}G[\Psi^{\infty}]\right|(\tau,r,\theta,\varphi)\leq 
C ((\tau+1)r^{-1}+2)^{-\frac{1}{2}-\frac{1}{2}\re \beta_{\ell}}(\tau+1)^{-\frac{3}{2}-\frac{1}{2}\re \beta_{\ell}}.
\end{equation}
and
\begin{equation}
\label{eq:erroresteqforPsihor}
\left|r^3((r-1)X)^{n_1}((\tau+1)K)^{n_2}G[\Psi^{+}]\right|(\tau,r,\theta,\varphi)\leq C ( (\tau+1)(1-r^{-1})+2)^{-\frac{1}{2}-\frac{1}{2}\re \beta_{\ell}}(\tau+1)^{-\frac{3}{2}-\frac{1}{2}\re \beta_{\ell}}.
\end{equation}
\end{proposition}
\begin{proof}
In view of \eqref{eq:maineqradfieldconfhor}, the estimates in \eqref{eq:erroresteqforPsihor} follow immediately from the estimates \eqref{eq:erroresteqforPsiinf}, after replacing $q$ with $-q$, $T$ with $K$ and $r=s_{\infty}+1$ with $s_{+}+1=\frac{1}{r-1}+1$.

We will first prove \eqref{eq:erroresteqforPsiinf} with $n_1=n_2=0$. 

Denote $f(\tau,r)=(\mathfrak{w}_{\ell}^0(r;\q))^{-1}(\Psi^{\infty}_0)_{\ell m}(\tau,r)$, $\mathfrak{w}(r)=\mathfrak{w}_{\ell}(r;\q)$ and  $\mathfrak{w}^0(r)=\mathfrak{w}_{\ell}^0(r;\q)$. Then we can write $\Psi^{\infty}_{\ell m}=\mathfrak{w}\cdot f$,so
\begin{multline}
\label{eq:maineqerror}
r^3(g^{-1})^{\mu \nu}D_{\mu}D_{\nu}(r^{-1}\mathfrak{w}fY_{\ell m})Y_{\ell m}^{-1}=2r^2 \frac{d(\Omega^2 \mathfrak{w})}{dr}X f-2r^2(1-\h)X(\mathfrak{w}Tf) \\
+\mathfrak{w}r^2 \left[\Omega^2X^2f-2i\q(1-\h)r^{-1}Xf-\h \widetilde{\h} T^2f+\left(\frac{d\h}{dr}-2i\q r^{-1}\h \widetilde{\h}\right)Tf\right]\\
=\Omega^2\frac{\mathfrak{w}}{\mathfrak{w}^0}\left[2r^2\frac{d\mathfrak{w}^0}{dr}Xf-2r^2X(\mathfrak{w}^0 Tf)+\mathfrak{w}^0r^2\left(X^2f-2i\q r^{-1}Xf\right)\right]\\
+2r^2\mathfrak{w}^0\frac{d(\Omega^2\mathfrak{w}(\mathfrak{w}^0)^{-1})}{dr}Xf-2r^2(1-\Omega^2)\mathfrak{w}XTf-2r^2\mathfrak{w}^0X(\mathfrak{w}(\mathfrak{w}^0)^{-1})Tf+2r^2\h X(\mathfrak{w}Tf)\\
+2iq\mathfrak{w}r(1-\Omega^2-\h )Xf-\mathfrak{w}r^2\h \widetilde{\h} T^2f+\mathfrak{w}r^2\left(\frac{d\h}{dr}-2i\q r^{-1}\h \widetilde{\h}\right)Tf\\
=2r^2\mathfrak{w}^0\frac{d(\Omega^2\mathfrak{w}(\mathfrak{w}^0)^{-1})}{dr}Xf-2r^2(1-\Omega^2)\mathfrak{w}XTf-2r^2\mathfrak{w}^0X(\mathfrak{w}(\mathfrak{w}^0)^{-1})Tf\\
+2r^2\h X(\mathfrak{w}Tf)-2i\q\mathfrak{w}r(1-\Omega^2-\h )Xf-\mathfrak{w}r^2\h \widetilde{\h} T^2f+\mathfrak{w}r^2\left(\frac{d\h}{dr}-2i\q r^{-1}\h \widetilde{\h}\right)Tf,
\end{multline}
where we arrived at the final equality by using that $\mathfrak{w}^0=(\Psi^{\infty}_0)_{\ell m}$ is a solution to \eqref{eq:maineqradfield} with $M=0$ by Propositions \ref{prop:hypgeomrealbeta}, \ref{prop:imaginarybetatail} and \ref{prop:hypgeombeta0}, so:
\begin{equation*}
	2r^2\frac{d\mathfrak{w}^0}{dr}Xf-2r^2X(\mathfrak{w}^0 Tf)+\mathfrak{w}^0r^2\left(X^2f-2i\q r^{-1}Xf\right)=0.
\end{equation*}

From Propositions \ref{prop:hypgeomrealbeta}, \ref{prop:imaginarybetatail} and \ref{prop:hypgeombeta0} it follows moreover that:
\begin{align*}
r^2|\h \widetilde{\h} \mathfrak{w}T^2f|\lesssim &\: r^{\frac{1}{2}+\frac{1}{2}\re \beta_{\ell}}(\tau+2r)^{-\frac{1}{2}-\frac{1}{2}\re \beta_{\ell}}(\tau+1)^{-\frac{5}{2}-\frac{1}{2}\re \beta_{\ell}},\\
r^2|\mathfrak{w}\left(\frac{d\h}{dr}+2iq r^{-1}\h \widetilde{\h}\right)Tf|\lesssim &\: r^{-\frac{1}{2}+\frac{1}{2}\re \beta_{\ell}}(\tau+2r)^{-\frac{1}{2}-\frac{1}{2}\re \beta_{\ell}}(\tau+1)^{-\frac{3}{2}-\frac{1}{2}\re \beta_{\ell}},\\
2r^2\left|\mathfrak{w}^0\frac{d(\Omega^2\mathfrak{w}(\mathfrak{w}^0)^{-1})}{dr}Xf\right|\lesssim &\:r^{\frac{1}{2}+\frac{1}{2}\re \beta_{\ell}}(\tau+2r)^{-\frac{3}{2}-\frac{1}{2}\re \beta_{\ell}}(1+\tau)^{-\frac{1}{2}-\re \frac{1}{2}\beta_{\ell}},\\
2r^2\left|(1-\Omega^2)\mathfrak{w}XTf\right|\lesssim &\:r^{\frac{3}{2}+\frac{1}{2}\re \beta_{\ell}}(\tau+2r)^{-\frac{3}{2}-\frac{1}{2}\re \beta_{\ell}}(1+\tau)^{-\frac{3}{2}-\re \frac{1}{2}\beta_{\ell}},\\
2r^2\left|\mathfrak{w}^0X(\mathfrak{w}(\mathfrak{w}^0)^{-1})Tf\right|\lesssim &\: r^{\frac{1}{2}+\frac{1}{2}\re \beta_{\ell}}(\tau+2r)^{-\frac{1}{2}-\frac{1}{2}\re \beta_{\ell}}(1+\tau)^{-\frac{3}{2}-\re \frac{1}{2}\beta_{\ell}},\\
2r^2|\h X(\mathfrak{w}Tf)|\lesssim &\:r^{-\frac{1}{2}+\frac{1}{2}\re \beta_{\ell}}(\tau+2r)^{-\frac{1}{2}-\frac{1}{2}\re \beta_{\ell}}(1+\tau)^{-\frac{3}{2}-\re \frac{1}{2}\beta_{\ell}}\\
+2\left|i\q \mathfrak{w}r(1-\Omega^2-\h )Xf\right|\lesssim &\: r^{\frac{1}{2}+\frac{1}{2}\re \beta_{\ell}}(\tau+2r)^{-\frac{1}{2}-\frac{1}{2}\re \beta_{\ell}}(1+\tau)^{-\frac{3}{2}- \frac{1}{2}\re\beta_{\ell}}.
\end{align*}
We therefore have that:
\begin{equation*}
|r^3(g^{-1})^{\mu \nu}D_{\mu}D_{\nu}(r^{-1}\mathfrak{w}f)|\lesssim  r^{\frac{1}{2}+\frac{1}{2}\re \beta_{\ell}}(\tau+2r)^{-\frac{1}{2}-\frac{1}{2}\re \beta_{\ell}}(1+\tau)^{-\frac{3}{2}-\re \frac{1}{2}\beta_{\ell}}.	
\end{equation*}
The case $n_1+n_2> 0$ follows from the fact that the above estimates are all preserved under commutation with $(\tau+1)T$ and $rX$.
\end{proof}

\begin{corollary}
	Let $r_+=1$ and $s_+=(r-1)^{-1}$. The following constants are well-defined:
\begin{align}
\label{eq:defBinf}
	B_{\infty,\ell m}:=&\int_1^{\infty}\mathfrak{w}_{\ell}(r;\q)e^{2iq\int_{r_{\sharp}}^rr'^{-1}\Omega^{-2}(r')(1-\h(r'))\,dr'}\Bigg[ rT^{-1}(G_{\ell m}[\Psi^{\infty}])(0,r)\\ \nonumber
+&\left(2(1-\h)X \Psi^{\infty}_{\ell m}+\h \widetilde{\h} T\Psi^{\infty}_{\ell m}-\left(\frac{d\h}{dr}-2i \q r^{-1} \h\widetilde{\h}\right)\Psi^{\infty}_{\ell m}\right)(0,r)\Bigg]\,dr,\\
\label{eq:defBhor}
B_{+,\ell m}:=&\int_0^{\infty}\mathfrak{w}_{\ell}(s_++1;-\q)e^{-2iq\int_{1}^{s_+}(s'+1)^{-1}\Omega^{-2}(s'+1)(1-\h^+(s'+1))\,ds'}\Bigg[ rK^{-1}(G_{\ell m}[\Psi^{+}])(0,s_+)\\ \nonumber
+&\left(2(1-\h^+)\partial_{s_+}\Psi^{+}_{\ell m}+\h^+ \widetilde{\h}^+ K\Psi^{+}_{\ell m}-\left(\frac{d\h^+}{ds_+}-2i \q (s_++1)^{-1} \h^+\widetilde{\h}^+\right)\Psi^{+}_{\ell m}\right)(0,s_+)\Bigg]\,ds_+.
\end{align}
\end{corollary}
\begin{proof}
	From the decay estimates in Proposition \ref{prop:estinhomG}, it follows that $rT^{-1}(G_{\ell m}[\Psi^{\infty}])(0,r)$ and $rK^{-1}(G_{\ell m}[\Psi^{+}])(0,s_+)$ and decay suitably fast in $r$ and $s_+$ for the integrals in \eqref{eq:defBinf} and \eqref{eq:defBhor} to be well-defined. Note moreover that $B_{+,\ell m}$ can be obtained from $B_{\infty,\ell m}$ after expressing $B_{\infty,\ell m}$ as an integral with respect to the variable $r-1$ and then mapping $T$ to $K$, $r-1$ to $s_+$, $q$ to $-q$, $\Psi^{\infty}_{\ell m}$ to $\Psi^{+}_{\ell m}$ and $rT^{-1}(G_{\ell m}[\Psi^{\infty}])$ to $ rK^{-1}(G_{\ell m}[\Psi^{+}])$.
\end{proof}

\begin{remark}
\label{rm:nonvanishingB}
	In Corollary \ref{cor:Bconstnonzero}, we will show that $B_{\infty,\ell m}\neq 0$. Furthermore, in the case $|Q|=M$, it then follows immediately that $B_{+,\ell m}\neq 0$ if we choose $\h: (0,\infty)_{s_+}\to [0,\infty)$ so that $\h^+: (0,\infty)_{s_{\infty}}\to [0,\infty)$ satisfies $\h(s_+)=\h^+(s_{\infty})$.
\end{remark}

\section{Initial data for time integrals}
\label{sec:idatatimeint}
We define the following difference function:
\begin{equation*}
\widehat{\psi}=\psi-\Psi.
\end{equation*}
Then $(g^{-1})^{\mu\nu}D_{\mu}D_{\nu}(r^{-1}\widehat{\psi})=TKG^{(1)}$.

In this section, we construct the time integrals $(TK)^{-1}(\widehat{\psi})$. This construction fixes $\mathfrak{I}_{\ell m}[\psi]$ and also $\mathfrak{H}_{\ell m}[\psi]$ in the $|Q|=M$ case.
\subsection{Fixed spherical harmonic modes}
\label{eq:timeintellipticboundl}
In this section, we consider the difference function $\psi-\Psi^{\infty}$ and we will determine the values of $\mathfrak{I}_{\ell m}[\psi]$ that are required to make sense of $T^{-1}(\psi-\Psi^{\infty})$.

We will first construct $T^{-1}(\psi-\Psi^{\infty})|_{\Sigma_0}$ via the renormalized quantities:
\begin{equation*}
	T^{-1}\widecheck{\psi}_{\ell}:=\mathfrak{w}_{\ell}^{-1}T^{-1}(\psi-\Psi^{\infty})|_{\Sigma_0}.
\end{equation*}
Indeed, the requirement that $T(T^{-1}\psi_{\ell})|_{\Sigma_0}=(\psi_{\ell})|_{\Sigma_0}$ and \eqref{eq:maineqradfield} holds for $T^{-1}\widehat{\psi}$ result in the equation:
\begin{multline}
\label{eq:maineqcheckpsi0}
 -rT^{-1}(G_{\ell m}[\Psi^{\infty}])(0,\cdot)=\mathfrak{w}_{\ell}X(\Omega^2 XT^{-1}\widecheck{\psi}_{\ell m})+2\Omega^2 \frac{d \mathfrak{w}_{\ell}}{dr}X T^{-1}\widecheck{\psi}_{\ell m}-2i \q\mathfrak{w}_{\ell}(1-\h)r^{-1}X T^{-1}\widecheck{\psi}_{\ell m}\\
-2(1-\h)X \widehat{\psi}_{\ell m}-\h \widetilde{\h} T\widehat{\psi}_{\ell m}+\left(\frac{d\h}{dr}-2i \q r^{-1} \h\widetilde{\h}\right)\widehat{\psi}_{\ell m}.
\end{multline}
Equivalently, we can express:
\begin{multline}
\label{eq:maineqcheckpsi}
-\mathfrak{w}_{\ell}e^{-2i\q\int_{r_{\sharp}}^rr'^{-1}\Omega^{-2}(r')(1-\h(r'))\,dr'}rT^{-1}(G_{\ell m}[\Psi^{\infty}])(0,\cdot)\\
=X(\Omega^2\mathfrak{w}_{\ell}^2 e^{-2i\q\int_{r_{\sharp}}^rr'^{-1}\Omega^{-2}(r')(1-\h(r'))\,dr'}XT^{-1}\widecheck{\psi}_{\ell m})\\
+\mathfrak{w}_{\ell}e^{-2i\q\int_{r_{\sharp}}^rr'^{-1}\Omega^{-2}(r')(1-\h(r'))\,dr'}\left[-2(1-\h)X \widehat{\psi}_{\ell m}-\h \widetilde{\h} T\widehat{\psi}_{\ell m}+\left(\frac{d\h}{dr}-2i \q r^{-1} \h\widetilde{\h}\right)\widehat{\psi}_{\ell m}\right].
\end{multline}

We define the following constant:

\begin{proposition}
\label{prop:timeinvffixedmodeconstr}
Fix $r_+=1$. Suppose that $B_{\infty,\ell m}\neq 0$ and fix the constants $\mathfrak{I}_{\ell m}[\psi]$ appearing in the definition of $\Psi^{\infty}$ as follows:
\begin{multline*}
	\mathfrak{I}_{\ell m}[\psi]:=B_{\infty,\ell m}^{-1}\int_{1}^{\infty}\mathfrak{w}_{\ell}(r;q)e^{-2i\q\int_{r_{\sharp}}^rr'^{-1}\Omega^{-2}(r')(1-\h(r'))\,dr'}\\
	\times \left[2(1-\h)X {\psi}_{\ell}+\h \widetilde{\h} T{\psi}_{\ell}-\left(\frac{d\h}{dr}-2i \q r^{-1} \h\widetilde{\h}\right){\psi}_{\ell}\right](0,r)\,dr.
\end{multline*}
Then there exists a unique solution $T^{-1}(\psi_{\ell m}-\Psi^{\infty}_{\ell m})|_{\Sigma_0}\in C^{\infty}([r_+,\infty))$ to \eqref{eq:maineqcheckpsi0} such that
\begin{align*}
\Omega^2XT^{-1}(\psi_{\ell}-\Psi^{\infty}_{\ell})|_{\Sigma_0}(r_+,\cdot)\equiv &\:0,\\
\lim_{r\to \infty}\mathfrak{w}_{\ell}^{-1}(r;\q)T^{-1}(\psi_{\ell}-\Psi^{\infty}_{\ell})|_{\Sigma_0}(r,\theta,\varphi)=&\: 0.
\end{align*}
We denote with
\begin{equation*}
	T^{-1}(\psi_{\ell}-\Psi^{\infty}_{\ell})
\end{equation*}
the solution to \eqref{eq:CSF} with $A=\widehat{A}$ and $G=G[\Psi^{\infty}]$, arising from the initial data:
\begin{equation*}
	\left(T^{-1}(\psi_{\ell}-\Psi^{\infty}_{\ell})|_{\Sigma_0}, (\psi_{\ell}-\Psi^{\infty}_{\ell})|_{\Sigma_0}\right).
\end{equation*}
	Then $T(T^{-1}(\psi_{\ell}-\Psi^{\infty}_{\ell}))=\psi_{\ell}-\Psi^{\infty}_{\ell}$.
\end{proposition}
\begin{proof}
For the sake of notational convenience, we will drop the $m$ in the subscripts $\psi_{\ell m}$. We can write \eqref{eq:maineqcheckpsi} as follows:
\begin{equation}
\label{eq:maineqcheckpsi1}
X(\Omega^2\mathfrak{w}_{\ell}^2 e^{-2i\q\int_{r_{\sharp}}^rr'^{-1}\Omega^{-2}(r')(1-\h(r'))\,dr'}XT^{-1}\widecheck{\psi}_{\ell})(0,r)=r\widetilde{G}_{\ell},
\end{equation}
with
\begin{multline*}
r\widetilde{G}_{\ell}:=-\mathfrak{w}_{\ell}e^{-2i\q\int_{r_{\sharp}}^rr'^{-1}\Omega^{-2}(r')(1-\h(r'))\,dr'}rT^{-1}(G_{\ell m}[\Psi^{\infty}])(0,\cdot)\\
+\mathfrak{w}_{\ell}e^{-2i\q\int_{r_{\sharp}}^rr'^{-1}\Omega^{-2}(r')(1-\h(r'))\,dr'}\left[2(1-\h)X \widehat{\psi}_{\ell}+\h \widetilde{\h} T\widehat{\psi}_{\ell}-\left(\frac{d\h}{dr}-2i \q r^{-1} \h\widetilde{\h}\right)\widehat{\psi}_{\ell}\right].
\end{multline*}
Since \eqref{eq:maineqcheckpsi1} is a second-order ODE, we need to impose two boundary conditions to obtain a unique solution. We take as our first boundary condition $(\Omega^2XT^{-1}\widecheck{\psi}_{\ell})(1)=0$, which together with the boundedness of $\mathfrak{w}_{\ell}$ at $r=1$ implies that:
\begin{equation*}
XT^{-1}\widecheck{\psi}_{\ell}(r)= e^{2i\q\int_{r_{\sharp}}^rr'^{-1}\Omega^{-2}(r')(1-\h(r'))\,dr'}\Omega^{-2}(r)\mathfrak{w}_{\ell}^{-2}(r)\int_1^rr \widetilde{G}_{\ell}(r')\,dr'.
\end{equation*}
Since $\mathfrak{w}_{\ell}$ is non-vanishing, the above expression is well-defined for all $r>1$.

Note that for $r\leq 2$ and $\delta>0$, we can estimate:
\begin{align*}
|XT^{-1}\widecheck{\psi}_{\ell}(r)|^2\lesssim &\: \Omega^{-3-\delta}(r-1) \int_{1}^r \Omega^{-1+\delta}|r\widetilde{G}_{\ell}(r')|^2\,dr',\\
|T^{-1}\widecheck{\psi}_{\ell}(r)|^2\lesssim &\: (r-1)^{-\delta} \int_{1}^r \Omega^{-1+\delta}|r\widetilde{G}_{\ell}(r')|^2\,dr'+|T^{-1}\widecheck{\psi}_{\ell}(2)|^2.
\end{align*}

The second boundary condition is $\lim_{r\to \infty}T^{-1}\widecheck{\psi}_{\ell}(r)=0$, which implies that
\begin{equation*}
T^{-1}\widecheck{\psi}_{\ell}(r)= -\int_{r}^{\infty}e^{2i\q\int_{r_0}^{r'}r''^{-1}\Omega^{-2}(r'')(1-\h(r''))\,dr''}\Omega^{-2}(r')\mathfrak{w}_{\ell}^{-2}(r')\int_1^{r'} r\widetilde{G}_{\ell}(r'')\,dr''dr'.
\end{equation*}

For $\re \beta_{\ell}< 1$, $r\widetilde{G}_{\ell}$ is integrable for suitably decaying ${\psi}_{\ell}$ so we can choose $\mathfrak{I}_{\ell m}[\psi]$ such that
\begin{equation}
\label{eq:vanishingtildeG}
\int_1^{\infty} r\widetilde{G}_{\ell}(r')\,dr'=0.
\end{equation}
We can guarantee \eqref{eq:vanishingtildeG} by taking:
\begin{multline*}
	\mathfrak{I}_{\ell m}[\psi]=B_{\infty,\ell m}^{-1}\int_{1}^{\infty}\mathfrak{w}_{\ell}e^{-2i\q\int_{r_{\sharp}}^rr'^{-1}\Omega^{-2}(r')(1-\h(r'))\,dr'}\\
	\times \left[2(1-\h)X {\psi}_{\ell}+\h \widetilde{\h} T{\psi}_{\ell}-\left(\frac{d\h}{dr}-2i \q r^{-1} \h\widetilde{\h}\right){\psi}_{\ell}\right](0,r)\,dr,
\end{multline*}
where we made use of the fact that $B_{\infty,\ell m}\neq 0$.

In that case, we obtain:
\begin{equation*}
T^{-1}\widecheck{\psi}_{\ell}(r)= \int_{r}^{\infty}e^{2i\q\int_{r_{\sharp}}^{r'}r''^{-1}\Omega^{-2}(r'')(1-\h(r''))\,dr''}\Omega^{-2}(r')\mathfrak{w}_{\ell}^{-2}(r')\int_{r'}^{\infty} r\widetilde{G}_{\ell}(r'')\,dr''dr'
\end{equation*}
Hence, by applying Cauchy--Schwarz and taking $0<\delta<1-\re \beta_{\ell}$, we obtain in the case $\beta_{\ell}\neq 0$:
\begin{multline*}
|T^{-1}\widecheck{\psi}_{\ell}(r)|\lesssim  \int_r^{\infty}r'^{-1-\re \beta_{\ell}}\cdot r'^{-\frac{1}{2}+\frac{\re \beta_{\ell}}{2}+\frac{\delta}{2}}\sqrt{\int_{r'}^{\infty}r''^{2-\re \beta_{\ell}-\delta}|r\widetilde{G}_{\ell}|^2(r'')\,dr''}\\
\lesssim r^{-\frac{1}{2}-\frac{\re \beta_{\ell}}{2}+\frac{\delta}{2}}\sqrt{\int_{r}^{\infty}r''^{2-\re \beta_{\ell}-\delta}|r\widetilde{G}_{\ell}|^2(r'')\,dr''}.
\end{multline*}
In the case $\beta_{\ell}=0$, we obtain instead:
\begin{multline*}
|T^{-1}\widecheck{\psi}_{\ell}(r)|\lesssim  \int_r^{\infty}r'^{-1}\cdot (\log (r'+1))^{-2}r'^{-\frac{1}{2}+\frac{\delta}{2}}\sqrt{\int_{r'}^{\infty}r''^{2-\delta}  |r\widetilde{G}_{\ell}|^2(r'')\,dr''}\\
\lesssim r^{-\frac{1}{2}+\frac{\delta}{2}}(\log (r+1))^{-1}\sqrt{\int_{r}^{\infty}r''^{2-\delta}|r\widetilde{G}_{\ell}|^2(r'')\,dr''}
\end{multline*}

Hence, 
\begin{align*}
|T^{-1}(\psi_{\ell}-\Psi^{\infty}_{\ell m})|_{\Sigma_0}|^2(r)\lesssim r^{\delta}\int_{r}^{\infty}r'^{2-\re \beta_{\ell}-\delta}|r\widetilde{G}_{\ell}|^2(r')\,dr',\\
|XT^{-1}(\psi_{\ell}-\Psi^{\infty}_{\ell m})|_{\Sigma_0}|^2(r)\lesssim  r^{-2+\delta}\int_{r}^{\infty}r'^{2-\re \beta_{\ell}-\delta}|r\widetilde{G}_{\ell}|^2(r')\,dr'.
\end{align*}

Now suppose that $\re \beta_{\ell}>1$. Then we estimate instead in the region $\{r\geq 2\}$:
\begin{multline*}
|T^{-1}\widecheck{\psi}_{\ell}(r)|\lesssim  \int_r^{\infty}r'^{-1-\re \beta_{\ell}}\int_{1}^{r'}  \left|r\tilde{G}_{\ell}(r'')\right|\,dr''\lesssim\int_r^{\infty}r'^{-1+\delta-\re \beta_{\ell}}\sqrt{\int_{1}^{r'} r''^{1-\delta} \left|r\tilde{G}_{\ell}(r'')\right|^2\,dr''} \,dr'\\
\lesssim r^{-\re \beta_{\ell}+\delta}\sqrt{\int_{1}^{\infty} r'^{1-\delta} \left|r\tilde{G}_{\ell}(r')\right|^2\,dr'}.
\end{multline*}
Hence, we can estimate for $r\geq 2$:
\begin{align*}
|T^{-1}(\psi_{\ell}-\Psi^{\infty}_{\ell m})|_{\Sigma_0}|^2(r)\lesssim  &\:  r^{1-\re\beta_{\ell}+\delta} \int_{1}^{\infty}r'^{1-\delta}|r'\widetilde{G}_{\ell}|^2(r')\,dr',\\
|XT^{-1}(\psi_{\ell}-\Psi^{\infty}_{\ell m})|_{\Sigma_0}|^2(r)\lesssim &\: r^{-1-\re\beta_{\ell}+\delta}\int_{1}^{\infty}r'^{1-\delta}|r'\widetilde{G}_{\ell}|^2(r')\,dr'. \qedhere
\end{align*}
\end{proof}

\begin{corollary}
\label{cor:timeinvffixedmodeconstrhor}
Let $|Q|=M$. Fix $r_+=1$. Suppose that $B_{+,\ell m}\neq 0$ and fix:
\begin{multline*}
	\mathfrak{H}_{\ell m}[\psi]:=B_{+,\ell m}^{-1}\int_{0}^{\infty}\mathfrak{w}_{\ell}(1+s_+;-\q)e^{2i\q\int_{r_{\sharp}+1}^{s_+}(s'+1)^{-1}\Omega^{-2}(s'+1)(1-\h^+(s'+1))\,ds'}\\
	\times \left[2(1-\h^+)\partial_{s_+} {\psi}_{\ell}+\h^+ \widetilde{\h}^+ K{\psi}_{\ell}-\left(\frac{d\h^+}{dr}+2i \q (s_++1)^{-1} \h^+\widetilde{\h}^+\right){\psi}_{\ell}\right](0,s_+)\,ds_+.
\end{multline*}
Then there exists a unique solution $K^{-1}(\psi_{\ell}-\Psi^{+}_{\ell})|_{\Sigma_0}\in C^{\infty}([0,r_+^{-1})_{\rho_{\infty}})$ to \eqref{eq:maineqcheckpsi0} such that
\begin{align*}
\lim_{r\to \infty}XK^{-1}(\psi_{\ell}-\Psi^{+}_{\ell})|_{\Sigma_0}(r,\theta,\varphi)\equiv &\:0,\\
\mathfrak{w}_{\ell}^{-1}\left(1+\frac{1}{r-1};-\q\right)K^{-1}(\psi_{\ell}-\Psi^{+}_{\ell})|_{\Sigma_0}(r_+,\cdot)\equiv &\: 0.
\end{align*}
We denote with
\begin{equation*}
	K^{-1}(\psi_{\ell}-\Psi^{+}_{\ell})
\end{equation*}
the solution to \eqref{eq:CSF} with $A=\widehat{A}$ and $G=G[\Psi^{+}]$, arising from the initial data:
\begin{equation*}
	\left(K^{-1}(\psi_{\ell}-\Psi^{+}_{\ell})|_{\Sigma_0}, (\psi_{\ell}-\Psi^{+}_{\ell})|_{\Sigma_0}-i\q K^{-1}(\psi_{\ell}-\Psi^{+}_{\ell})|_{\Sigma_0}\right).
\end{equation*}
	Then $K(K^{-1}(\psi_{\ell}-\Psi^{\infty}_{\ell}))=\psi_{\ell}-\Psi^{+}_{\ell}$.
\end{corollary}
\begin{proof}
We use the symmetry present in \eqref{eq:maineqradfieldhor} under the exchanges: $\rho_{\infty}\leftrightarrow \rho_+$, $s_+\leftrightarrow  s_{\infty}=r-1$,  $T\leftrightarrow K$, $\q\leftrightarrow -\q$, $\h\leftrightarrow \h^+$ and $\widetilde{\h}\leftrightarrow \widetilde{\h}^+$, so that we can simply repeat the proof of Proposition \ref{prop:timeinvffixedmodeconstr}.
\end{proof}

Write
\begin{equation*}
\widehat{\psi}^{(1)}_{\ell}:=-\frac{1}{i\q}\left( K^{-1}(\psi_{\ell}-\Psi^{+}_{\ell})-T^{-1}(\psi_{\ell}-\Psi^{\infty}_{\ell})\right).
\end{equation*}

Then $TK(\widehat{\psi}^{(1)}_{\ell})=\widehat{\psi}_{\ell}=\psi_{\ell}-\Psi$.

\subsection{Elliptic estimates}
\label{sec:ellipticesthighl}
In this section, we will show that we can sum the functions $\widehat{\psi}^{(1)}_{\ell}$ over $\ell$ to define:
\begin{equation*}
\widehat{\psi}^{(1)}(\tau,\cdot):=\sum_{\ell\in \N_0} \widehat{\psi}^{(1)}_{\ell}(\tau,\cdot)
\end{equation*}
as a function in an appropriate Sobolev space on $\Sigma_{\tau}$.

Then
\begin{equation*}
(\widehat{\psi}^{(1)})_{\ell m}(\tau,r)=\la \widehat{\psi}^{(1)}_{\ell}(\tau,r,\cdot), Y_{\ell m}\ra_{L^2(\s^2)}.
\end{equation*}

\subsubsection{Elliptic estimates for large angular frequencies}
The key estimates that will allow us to sum over $\ell$ are elliptic-type estimates for $\psi_{\geq L_0}$, with $L_0>0$ suitably large. Recall that $\widehat{\psi}_{\ell}=\psi_{\ell}$ for $\ell$ satisfying $\ell(\ell+1)\geq q^2$. We will take $r_+=1$ in this section.

\begin{proposition}
\label{prop:ellipticesthighfreq}
Let $L_0\in \N_0$, $\tau\geq 0$ and let $\phi$ be a solution to \eqref{eq:CSF} with $A=\widehat{A}$ and $G_A\equiv 0$. Let $1<p< 3$. Then for suitable large $L_0>0$, depending on $\q$ and $p$, there exists a constant $C=C(p,\q,L_0)>0$ such that for all $\tau\geq 0$:
\begin{align}
\label{eq:ellipticenergy1}
\int_{\Sigma_{\tau}} \mathcal{E}_{p-2}[\psi_{\geq L_0}]\,d\sigma dr\leq &\: C\int_{\Sigma_{\tau}} \mathcal{E}_{p}[T\psi_{\geq L_0}]\,d\sigma dr, \\
\label{eq:ellipticenergy2}
\int_{\Sigma_{\tau}} \mathcal{E}_{p-2}[\psi_{\geq L_0}]\,d\sigma dr\leq &\: C\int_{\Sigma_{\tau}} \mathcal{E}_{p}[K\psi_{\geq L_0}]\,d\sigma dr\quad \textnormal{if $|Q|=M$}.
\end{align}
\end{proposition}
\begin{proof}
Restrict to $\psi_{\ell}$, with $\ell\geq L_0$, where  $L_0>0$ will be chosen suitably large. Without loss of generality, we may assume that $\psi$ is smooth any compactly supported. The more general case follows from a standard density argument. For the sake of notational convenience, we drop the subscript $\geq L_0$ in $\psi_{\geq L_0}$ in the estimates below. Rearranging \eqref{eq:maineqradfield} then gives:
 \begin{multline}
 \label{eq:degelliptic}
X(\Omega^2X\psi)+r^{-2}\slashed{\Delta}_{\s^2}\psi-\frac{d\Omega^2}{dr} r^{-1}\psi+2i \q(1-\h)r^{-1}X\psi+\left[-iqr^{-1}\frac{d\h}{dr}+\q^2  \mathbbm{h}\widetilde{\mathbbm{h}}r^{-2}+i \q r^{-2}(1-\h)\right]\psi\\
=2(1-\h)XT\psi-\h \widetilde{\h} T^2\psi-\left(\frac{d\h}{dr}-2i \q r^{-1} \h\widetilde{\h}\right)T\psi=: r\tilde{G}[T\psi].
\end{multline}
We use that
\begin{equation*}
rX(\Omega^2X\psi)+r^{-1}\slashed{\Delta}_{\s^2}\psi-r\frac{d\Omega^2}{dr} r^{-1}\psi=X(r^2\Omega^2X\phi)+\slashed{\Delta}_{\s^2}\phi
\end{equation*}
and rearrange \eqref{eq:degelliptic} further:
\begin{equation}
 \label{eq:degelliptic2}
X(r^2\Omega^2X\phi)+\slashed{\Delta}_{\s^2}\phi=-2i \q(1-\h)rX\phi+\left[-i\q r\frac{d\h}{dr}-\q^2  \mathbbm{h}\widetilde{\mathbbm{h}}-3i \q (1-\h)\right]\phi+r^2\tilde{G}[T\psi].
\end{equation}

Multiplying both sides of \eqref{eq:degelliptic2} with $r^{-s}(r^{-1}\Omega)^{p'}X\phi$, taking the real parts of the products and applying the Leibniz rule, we obtain;
 \begin{multline*}
\frac{1}{2}r^{4-s}(r^{-1}\Omega)^{p'+2}X(|X\phi|^2)+2(r^{-1}\Omega)^{p'}r^{-s}\frac{d}{dr}(r^2\Omega^2)|X\phi|^2+\re(r^{-s}(r^{-1}\Omega)^{p'} \overline{X\phi}\cdot \slashed{\Delta}_{\s^2}\phi)\\
+r^{-s}(r^{-1}\Omega)^{p'}\re\left(\overline{X\phi}\cdot \left[i\q r\frac{d\h}{dr}+\q^2  \mathbbm{h}\widetilde{\mathbbm{h}}+3i \q (1-\h)\right]\phi\right)\\
=r^{-s}(r^{-1}\Omega)^{p'}\re \left( \overline{X\phi}\cdot r^2\tilde{G}[T\psi]\right).
\end{multline*}
We can further write:
\begin{multline*}
\frac{1}{2}r^{4-s}(r^{-1}\Omega)^{p'+2}X(|X\phi|^2)+2r^{-s}(r^{-1}\Omega)^{p'}\frac{d}{dr}(r^2\Omega^2)|X\phi|^2\\
=X\left(\frac{1}{2}r^{4-s}(r^{-1}\Omega)^{p'+2}|X\phi|^2\right)+\left[2r^{-s}\frac{d}{dr}(r^2\Omega^2)-\frac{1}{2}\frac{d}{dr}(r^{4-s}(r^{-1}\Omega)^{p'+2})\right]|X\phi|^2\\
=X\left(\frac{1}{2}r^4(r^{-1}\Omega)^{p'+2}|X\phi|^2\right)+\left[\left(3-\frac{p'}{2}\right)(r-M)+\left(p'+\frac{s}{2}\right)r\Omega^2\right]r^{-s}(r^{-1}\Omega)^{p'}|X\phi|^2.
\end{multline*}
Furthermore,
\begin{multline*}
\re(r^{-s}(r^{-1}\Omega)^{p'} \overline{X\phi}\cdot \slashed{\Delta}_{\s^2}\phi)=-\ell(\ell+1)\re(r^{-s}(r^{-1}\Omega)^{p'} \overline{X\phi}\cdot \phi)=-\frac{1}{2}\ell(\ell+1) X(r^{-s}(r^{-1}\Omega)^{p'}|\phi|^2)\\
+\frac{1}{4}\left[p'\frac{d}{dr}(\Omega^2r^{-2})-2sr^{-3}\Omega^2\right] \ell(\ell+1)r^{-s}(r^{-1}\Omega)^{p'-2}|\phi|^2\\
=-\frac{1}{2}\ell(\ell+1) X(r^{-s}(r^{-1}\Omega)^{p'}|\phi|^2)\\
+\frac{1}{4}\left[2p'(r-M)+2(- s-2p')(r-2M+r^{-1}e^2)\right] \ell(\ell+1)r^{-s-4}(r^{-1}\Omega)^{p'-2}|\phi|^2.
\end{multline*}
Combining the above estimates and integrating over $\s^2$, we therefore obtain:
 \begin{multline}
  \label{eq:degelliptic3}
\int_{\s^2}\left[\left(3-\frac{p'}{2}\right)(r-M)+\left(p'+\frac{s}{2}\right)r\Omega^2\right]r^{-s}(r^{-1}\Omega)^{p'}|X\phi|^2\\
+\frac{1}{2}\left[(-s-p')(r-M)+(2p'+s)(M-r^{-1}Q^2)\right]r^{-s-4}(r^{-1}\Omega)^{p'-2}|\snabla_{\s^2}\phi|^2\,d\sigma\\
= X\left[\int_{\s^2 }\frac{1}{2}r^{-s}(r^{-1}\Omega)^{p'}(\Omega^2r^2|X\phi|^2-|\snabla_{\s^2}\phi|^2)\,d\sigma \right]\\
-\int_{\s^2}r^{-s}(r^{-1}\Omega)^{p'}\re\left(\overline{X\phi}\cdot \left[-iqr \frac{d\h}{dr}+q^2  \mathbbm{h}\widetilde{\mathbbm{h}}+3i \q (1-\h)\right]\phi\right)\\
-r^{-s}(r^{-1}\Omega)^{p'}\re \left( \overline{X\phi}\cdot r^2\tilde{G}[T\psi]\right)\,d\sigma.
\end{multline}
The terms on the very LHS are non-negative definite if $p'>0$, $s\leq 0$ and $p'<|s|<2p'$ and $p'< 6$.
By Young's inequality, we can estimate:
\begin{multline*}
\left|r^{-s}(r^{-1}\Omega)^{p'}\re\left(\overline{X\phi}\cdot \left[-iqr\frac{d\h}{dr}+q^2  \mathbbm{h}\widetilde{\mathbbm{h}}+3i \q (1-\h)\right]\phi\right)\right|\leq \delta r^{-s}(r^{-1}\Omega)^{p'} (r-M)|X\phi|^2\\
+\frac{C \q^2}{4\delta} r^{-s}(r^{-1}\Omega)^{p'} (r-M)^{-1}|\phi|^2
\end{multline*}

Note that
\begin{multline*}
	\int_{\s^2}(r-M)r^{-s-4}(r^{-1}\Omega)^{p'-2}|\snabla_{\s^2}\phi|^2\,d\sigma\\
	\geq L_0(L_0+1)\int_{\s^2}\frac{(r-M)^2}{(r-M)^2-(M^2-Q^2)}(r-M)^{-1}r^{-s}(r^{-1}\Omega)^{p'}|\phi|^2\,d\sigma\\
	\geq L_0(L_0+1)\int_{\s^2}r^{-s}(r^{-1}\Omega)^{p'} (r-M)^{-1}|\phi|^2\,d\sigma.
\end{multline*}

We apply Young's inequality again to estimate:
\begin{equation*}
\left|r^{-s}(r^{-1}\Omega)^{p'}\re \left( \overline{X\phi}\cdot r^2\tilde{G}[T\psi]\right)\right|\lesssim  \delta r^{-s}(r^{-1}\Omega)^{p'} (r-M)|X\phi|^2+ \frac{ 1}{4\delta} r^{-s}(r^{-1}\Omega)^{p'}r^2 (r-M)^{-1}|r\tilde{G}[T\phi]|^2.
\end{equation*}
Finally, we integrate \eqref{eq:degelliptic3} in $r$ and observe that the boundary terms vanish by the boundary assumptions in the proposition and we take $L_0(L_0+1)\geq Cq^2 \delta^{-1}$ to obtain:
\begin{multline}
\label{eq:ellipticpart1}
\int_1^{\infty}\int_{\s^2}\left[\left(3-\frac{p'}{2}-2\delta\right)(r-M)+\left(p'+\frac{s}{2}\right)r\Omega^2\right]r^{-s}(r^{-1}\Omega)^{p'}|X\phi|^2\\
+\frac{1}{4}\left[(-s-p')(r-M)+(2p'+s)(M-r^{-1}Q^2)\right]r^{-s-4}(r^{-1}\Omega)^{p'-2}|\snabla_{\s^2}\phi|^2\,d\sigma dr\\
\leq C\int_1^{\infty}\int_{\s^2}r^{-s}(r^{-1}\Omega)^{p'}r^2 (r-M)^{-1}|r\tilde{G}[T\phi]|^2\,d\sigma dr\\
\leq C\int_1^{\infty}\int_{\s^2}r^{-s}(r^{-1}\Omega)^{p'-1} (|XT\psi|^2+r^{-4}|T^2\psi|^2+r^{-6}|T\psi|^2)\,d\sigma dr.
\end{multline}
Now \eqref{eq:ellipticenergy1} follows by applying \eqref{eq:ellipticpart1} with the choices $p'=3-p$ and $s= 1-p'-p=-2$, which implies that $1<p<3$.

In the case $|Q|=M$, \eqref{eq:ellipticenergy2} simply follows by repeating the above argument with $T$ replaced by $K$, $q$ replaced by $-q$ and $r$ replaced by $s_++1$.
\end{proof}

\begin{corollary}
\label{cor:elliptichighfreq}
Let $L_0,N\in \N_0$. Let $\tau\geq 0$ and let $\phi$ be a solution to \eqref{eq:CSF} with $G\equiv 0$.

Let $1<p<3$. Then, for suitably large $L_0>0$ depending on $q$ and $p$, there exists a constant $C=C(p,q,L_0,N)>0$ such that for all $\tau\geq 0$:
\begin{align}
\label{eq:ellipticenergy1ho}
\sum_{|\gamma|\leq N}\int_{\Sigma_{\tau}} \mathcal{E}_{p-2}[\mathbf{Z}^{\gamma}\psi_{\geq L_0}]\,d\sigma dr\leq &\: C\sum_{|\gamma|\leq N}\int_{\Sigma_{\tau}} \mathcal{E}_{p}[\mathbf{Z}^{\gamma}T\psi_{\geq L_0}]\,d\sigma dr, \\
\label{eq:ellipticenergy2ho}
\sum_{|\gamma|\leq N}\int_{\Sigma_{\tau}} \mathcal{E}_{p-2}[\mathbf{Z}^{\gamma}\psi_{\geq L_0}]\,d\sigma dr\leq &\: C\sum_{|\gamma|\leq N}\int_{\Sigma_{\tau}} \mathcal{E}_{p}[\mathbf{Z}^{\gamma}K\psi_{\geq L_0}]\,d\sigma dr\quad \textnormal{if $|Q|=M$}.
\end{align}
\end{corollary}
\begin{proof}
	We can immediately commute \eqref{eq:degelliptic} with angular momentum operators and $T$. To conclude \eqref{eq:ellipticenergy1ho} we therefore only need to commute \eqref{eq:degelliptic} with $(r-1)X$. It is straightforward to show that the highest order terms in \eqref{eq:degelliptic} have the same structure after commutation with $rX$ and non-negative constants can only get larger. See for example \cite{gaj22a}[Lemma 9.2, Corollary 9.4].
\end{proof}

\begin{corollary}
Let $r>1$, $\delta>0$ and $N\in \N_0$. Then
\begin{equation*}
\widehat{\psi}^{(1)}|_{\Sigma_0}(r,\cdot)=\sum_{\ell\in \N_0} \widehat{\psi}^{(1)}_{\ell}(r,\cdot)
\end{equation*}
is well-defined as an element of the Sobolev space $H^N(\s^2)$, provided
\begin{equation*}
\int_{\Sigma_{0}}\mathcal{E}_{2}[\psi]+ \mathcal{E}_{2}[\snabla_{\s^2}^{N}\psi]\,d\sigma dr<\infty.
\end{equation*}
\end{corollary}
\begin{proof}
We first recall that
\begin{equation*}
\widehat{\psi}^{(1)}_{\ell}=-\frac{1}{i\q}\left( K^{-1}\widehat{\psi}_{\ell}-T^{-1}\widehat{\psi}_{\ell}\right).
\end{equation*}
Hence, we can apply \eqref{eq:ellipticenergy1ho} and \eqref{eq:ellipticenergy2ho} with $\tau=0$, $\gamma_1=N$, $\gamma_2=\gamma_3=0$ and $p=2$ to conclude that for all $\ell\geq L$:
\begin{equation*}
\int_{\Sigma_{0}} \mathcal{E}_{0}[\snabla_{\s^2}^{k}\widehat{\psi}^{(1)}_{\ell}]\,d\sigma dr\leq 2 q^{-2}\int_{\Sigma_{0}} \mathcal{E}_{0}[\snabla_{\s^2}^{k}T^{-1}\widehat{\psi}_{\ell}]+\mathcal{E}_{0}[\snabla_{\s^2}^{k}K^{-1}\widehat{\psi}_{\ell}]\,d\sigma dr\\
\leq C \int_{\Sigma_{0}} \mathcal{E}_{2}[\snabla_{\s^2}^{k}\widehat{\psi}_{\ell}]\,d\sigma dr,
\end{equation*}
where $C>0$ is uniform in $\ell$. 

By the fundamental theorem of calculus in the $r$-direction and a Poincar\'e inequality on $\s^2$, it then follows that for all $r>1$:
\begin{equation*}
||\widehat{\psi}^{(1)}_{\ell}||^2_{H^N(\s^2)}\leq C  \int_{\Sigma_{0}} \mathcal{E}_{0}[\snabla_{\s^2}^{N}\widehat{\psi}_{\ell}]+ \mathcal{E}_{0}[\snabla_{\s^2}^{N}\widehat{\psi}_{\ell}]\,d\sigma dr.
\end{equation*}
Convergence of the sum over $\ell$ then follows by showing that the sums $\sum_{\ell=L}^{M} \widehat{\psi}^{(1)}_{\ell}(r,\cdot)$ are Cauchy sequences in $M$, which in turn follows from the above uniform estimate and the convergence of the sum $\sum_{\ell\in \N_0} \int_{\Sigma_{0}} \mathcal{E}_{0}[\snabla_{\s^2}^{N}\widehat{\psi}_{\ell}]+ \mathcal{E}_{0}[\snabla_{\s^2}^{N}\widehat{\psi}_{\ell}]\,d\sigma dr$, which we assumed.
\end{proof}

\subsubsection{Elliptic estimates for bounded angular frequencies}
We complement the elliptic estimates from Corollary \ref{cor:elliptichighfreq} with additional elliptic estimates for bounded angular frequencies, which follow from the construction in Proposition \ref{prop:timeinvffixedmodeconstr}.
\begin{proposition}
\label{prop:boundedangelliptic}
Let $N,L_0\in \N_0$. Let $\tau\geq 0$ and let $\phi$ be a solution to \eqref{eq:CSF}. Let $2-\re \sqrt{1-4\q^2}<p<3$ and $\epsilon>0$. Then there exists a constant $C=C(p,q,L_0,\epsilon)>0$ such that for all $\tau\geq 0$:
\begin{align}
\label{eq:ellipticenergy1boundl}
\sum_{|\gamma|\leq N}\int_{\Sigma_{\tau}} \mathcal{E}_{p-2-\epsilon}[\mathbf{Z}^{\gamma}\widehat{\psi}_{\leq L_0 -1}]\,d\sigma dr\leq &\: C\sum_{|\gamma|\leq N}\int_{\Sigma_{\tau}}\mathcal{E}_{p}[\mathbf{Z}^{\gamma}T\widehat{\psi}_{\leq L_0 -1}]+(r^{-1}\Omega)^{-p+2}|\mathbf{Z}^{\gamma}(G_{\widehat{A}})_{\leq L_0-1}|^2\,d\sigma dr, \\
\label{eq:ellipticenergy2boundl}
\sum_{|\gamma|\leq N}\int_{\Sigma_{\tau}} \mathcal{E}_{p-2-\epsilon}[\mathbf{Z}^{\gamma}\widehat{\psi}_{\leq L_0 -1}]\,d\sigma dr\leq &\: C\sum_{|\gamma|\leq N}\int_{\Sigma_{\tau}} \mathcal{E}_{p}[\mathbf{Z}^{\gamma}K\widehat{\psi}_{\leq L_0 -1}]+(r^{-1}\Omega)^{-p+2}|\mathbf{Z}^{\gamma}(G_{\widehat{A}})_{\leq L_0-1}|^2\,d\sigma dr,
\end{align}
where \eqref{eq:ellipticenergy2boundl} only holds for $|Q|=M$.
\end{proposition}
\begin{proof}
Suppose first that $N=0$. We apply the estimates in the proof of Proposition \ref{prop:timeinvffixedmodeconstr} with $T^{-1}(\psi_{\ell m}-\Psi^{\infty}_{\ell m})|_{\Sigma_0}$ replaced by $\psi|_{\Sigma_{\tau}}$ and with $\delta=3-p$ to obtain in the case $\beta_{\ell}\neq 0$:
\begin{align*}
|\widehat{\psi}_{\ell}|^2(\tau, r)\leq &\: (r^{-1}\Omega)^{p-3} \int_{\Sigma_{\tau}} \mathcal{E}_{p}[T\widehat{\psi}_{\leq L-1}]+(r'^{-1}\Omega )^{2-p+2}|(G_{\widehat{A}})_{\leq L-1}|^2\,d\sigma dr',\\
|X\widehat{\psi}_{\ell}|^2(\tau, r)\leq &\:\Omega^{p-6}(r-1)r^{-p} \int_{\Sigma_{\tau}} \mathcal{E}_{p}[T\widehat{\psi}_{\leq L-1}]+(r'^{-1}\Omega )^{2-p+2}|(G_{\widehat{A}})_{\leq L-1}|^2\,d\sigma dr'.
\end{align*}
The condition $0<\delta<1+\re \beta_{\ell}$ from the proof of Proposition \ref{prop:timeinvffixedmodeconstr} then implies that $2-\re \beta_{\ell}<p<3$.

In the $\beta_{\ell}=0$ case, we see an additional logarithmic term inside the integrand on the RHS, due to the estimate $|\mathfrak{w}_{\ell}(r)|^2\leq (\log(1+r))^2 r$ in this case.

The estimate \eqref{eq:ellipticenergy1boundl} then follows by integrating  
\begin{equation*}
(r^{-1}\Omega)^{-(p'-2)}r^{-2}|\widehat{\psi}_{\ell}|^2+(r^{-1}\Omega)^{-(p'-2)}\Omega^2|X\widehat{\psi}_{\ell}|^2
\end{equation*}
and taking $p'<p$ to ensure convergence of the integrals. The derivation of \eqref{eq:ellipticenergy2boundl} proceeds entirely analogously, with the roles of $r-1$ and $(r-1)^{-1}$ reversed, applying Corollary \ref{cor:timeinvffixedmodeconstrhor} instead of Proposition \ref{prop:timeinvffixedmodeconstr}.

The case $N= 1$ then follows by applying the elliptic estimates in Proposition \ref{prop:ellipticesthighfreq} and using \eqref{eq:ellipticenergy1boundl} and \eqref{eq:ellipticenergy1boundl} with $N=0$ to control the lower-order orders, instead of the assumption $\ell\geq L_0$. The general $N\geq 1$ case can similarly be estimated inductively.
\end{proof}

\section{Energy estimates for $\widehat{\psi}$}
\label{sec:edectimint}
In this section, we will derive energy estimates for $\widehat{\psi}$ by using that $\widehat{\psi}=TK\widehat{\psi}^{(1)}$ and applying Corollary \ref{prop:edecaygeneralp}. 

In order to bound the energy norms for $\widehat{\psi}^{(1)}$ on the RHS of  \eqref{eq:energydecayp} by initial energy norms of $\widehat{\psi}$, we would need to apply the elliptic estimates of Corollary \ref{cor:elliptichighfreq} and Proposition \ref{prop:boundedangelliptic} with $p>3$. Since these elliptic estimates are in general only valid for $2<p< 3$, we cannot apply them directly. 

Instead, will first multiply the initial data $(\widehat{\psi}^{(1)}|_{\Sigma_0},T\widehat{\psi}|_{\Sigma_0})$ with smooth cut-off functions, so that they they are compactly supported and supported away from $r=r_+$. We will then apply energy estimates for the cut off data before deriving energy estimates for the original data via an interpolation argument.

In the remainder of this section, we will set $r_+=1$. Let $\chi_{\rm initial}(r)=1$ for all $1+2(r_{\rm initial}-1)\leq r\leq \frac{1}{2(r_{\rm initial}-1)}$ and $\chi_{\rm initial}(r)=0$ for all $r<r_{\rm initial}$ and $r> \frac{1}{r_{\rm initial}-1}$. Let $\hpsi^{(1)}_{\Sc}$ be a solution to \eqref{eq:CSF} with inhomogeneity $G^{(1)}_{\Sc}$, defined as follows.
\begin{equation*}
G^{(1)}_{\Sc}=\xi_{\tau_*} (TK)^{-1}G_{\widehat{A}}[\Psi],
\end{equation*}
with $\xi_{\tau_*}: [0,\infty)_{\tau}\to [0,1]$ a smooth cut-off function such that $\xi_{\tau_*}(\tau)=1$ for $\tau\leq \tau_*$ and $\xi_{\tau_*}(\tau)=0$ for $\tau\geq \tau_*+1$.

We consider the following corresponding initial data for $\hpsi^{(1)}_{\Sc}$:
\begin{equation*}
\chi_{\rm initial}(\widehat{\psi}^{(1)}|_{\Sigma_0}, (T\widehat{\psi}^{(1)})|_{\Sigma_0}).
\end{equation*}
Then by finite speed of propagation, we have that
\begin{equation*}
\hpsi^{(1)}_{\Sc}\equiv \hpsi^{(1)}\quad \textnormal{in $D^+\left(\Sigma_{0}\cap\left\{1+2(r_{\rm initial}-1)\leq r\leq \frac{1}{2(r_{\rm initial}-1)}\right \}\right)\cap \{\tau\leq \tau_*\}$}
\end{equation*}
Given $\tau_*\geq 0$ and $\tilde{R}>1$, we can therefore find $r_{\rm initial}(\tau_*,\tilde{R})$ such that
\begin{equation*}
\hpsi^{(1)}_{\Sc}|_{\Sigma_{\tau}\cap\{1+\tilde{R}^{-1}\leq r\leq \tilde{R}\}}\equiv \hpsi^{(1)}|_{\Sigma_{\tau}\cap\{1+\tilde{R}^{-1}\leq r\leq \tilde{R}\}}
\end{equation*}
for all $\tau\leq \tau_*$.

Furthermore, it is straightforward to show that we can choose $r_{\rm initial}(\tau_*,\tilde{R})$, such that
\begin{equation*}
(r_{\rm initial}(\tau_*,\tilde{R})-1)^{-1} \lesssim \tilde{R}+ (1+\tau_*).
\end{equation*}

\begin{proposition}
\label{prop:energydecayptieminv}
Let $N_1\in \N_0$, $N_2\in \N_0$. Let $1< p<2$. For all $\nu>0$ and $\epsilon>0$ suitably small, there exists a constant $C=C(\nu,\epsilon, N_1,N_2,p)>0$ such that:
\begin{multline}
\label{eq:energydecayptieminv}
(1+\tau)^{3-p+2N_1-2\epsilon-\nu}\sum_{|\gamma|\leq N_2}\int_{\Sigma_{\tau}}\mathcal{E}_{p-3\epsilon}[\mathbf{Z}^{\gamma}(TK)^{N_1}\widehat{\psi}]\,d\sigma dr \\
\leq C\sum_{m\leq N_1+1}\sum_{|\gamma|+m\leq N_1+N_2+1}\sum_{j\leq 2(N_1+1)}\int_{\Sigma_{0}} \mathcal{E}_{3-\frac{\epsilon}{2}}[\mathbf{Z}^{\gamma}(TK)^mT^j\widehat{\psi}]\,d\sigma dr+C(|\mathfrak{I}[\psi]|^2+|\mathfrak{H}[\psi]|^2)\\
=: CE_{N_1,N_2,\epsilon}[\psi]
\end{multline}
and additionally,
\begin{equation}
\label{eq:energydecayptieminvmin1}
(1+\tau)^{-2\epsilon}\sum_{|\gamma|\leq N_2}\int_{\Sigma_{\tau}} \mathcal{E}_{1-5\epsilon}[\mathbf{Z}^{\gamma}T^{-1}\widehat{\psi}]\,d\sigma dr\leq CE_{-1,N_2,\epsilon}[\psi].
\end{equation}
Furthermore, for $N_1\in \N_0\cup\{-1\}$, $N_2\in \N_0$:
\begin{align}
\label{eq:energydecayptieminvfaster}
(1+\tau)^{3+\re\sqrt{1-4\q^2}+2N_1-8\epsilon-\nu}\int_{\Sigma_{\tau}}\sum_{|\gamma|\leq N_2}\mathcal{E}_{\epsilon-\re\sqrt{1-4\q^2}}[\mathbf{Z}^{\gamma}(TK)^{N_1}\widehat{\psi}]\,d\sigma dr\leq &\:C E_{N_1+1,N_2+1,\epsilon}[\psi],\\
\label{eq:energydecayptieminvfasterDT}
(1+\tau)^{3+\re\sqrt{1-4\q^2}+2N_1-8\epsilon-\nu}\int_{\Sigma_{\tau}}\sum_{|\gamma|\leq N_2}\mathcal{E}_{2+\epsilon-\re\sqrt{1-4\q^2}}[D_T\mathbf{Z}^{\gamma}(TK)^{N_1}\widehat{\psi}]\,d\sigma dr\leq &\: C E_{N_1+1,N_2+2,\epsilon}[\psi].
\end{align}
\end{proposition}
\begin{proof}
With the above choice of $r_{\rm initial}(\tau_*,\tilde{R})$, we obtain the following estimate:
\begin{multline*}
\sum_{m\leq N_1}\sum_{|\gamma|+m\leq N_1+N_2}\sum_{j\leq 2N_1}\int_{\Sigma_{0}} \mathcal{E}_{1+\epsilon}[\mathbf{Z}^{\gamma}(TK)^mT^j\widehat{\psi}^{(1)}_{\Sc}]\,d\sigma dr\\
\lesssim \left[ \tilde{R}^{2\epsilon}+ (1+\tau_*)^{2\epsilon}\right]\sum_{m\leq N_1}\sum_{|\gamma|+m\leq N_1+N_2}\sum_{j\leq 2N_1}\int_{\Sigma_{0}} \mathcal{E}_{1-\epsilon}[\mathbf{Z}^{\gamma}(TK)^mT^j\widehat{\psi}^{(1)}]\,d\sigma dr
\end{multline*}
We apply the elliptic estimates in Corollary \ref{cor:elliptichighfreq} and Proposition \ref{prop:boundedangelliptic} along $\Sigma_0$ to further estimate the RHS in terms of $\widehat{\psi}$ and obtain:
\begin{multline*}
\sum_{m\leq N_1}\sum_{|\gamma|+m\leq N_1+N_2}\sum_{j\leq 2N_1}\int_{\Sigma_{0}} \mathcal{E}_{1+\epsilon}[\mathbf{Z}^{\gamma}(TK)^mT^j\widehat{\psi}_{\Sc}^{(1)}]\,d\sigma dr\\
\lesssim \left[ \tilde{R}^{2\epsilon}+ (1+\tau_*)^{2\epsilon}\right]\sum_{m\leq N_1}\sum_{|\gamma|+m\leq N_1+N_2}\sum_{j\leq 2N_1}\int_{\Sigma_{0}} \mathcal{E}_{3-\frac{\epsilon}{2}}[\mathbf{Z}^{\gamma}(TK)^mT^j\widehat{\psi}]\,d\sigma dr\\
+\left[ \tilde{R}^{2\epsilon}+ (1+\tau_*)^{2\epsilon}\right]\sum_{m\leq N_1}\sum_{|\gamma|+m\leq N_1+N_2}\sum_{j\leq 2N_1}\int_{\Sigma_{0}} (\Omega r^{-1})^{-1+\frac{\epsilon}{2} }|r\mathbf{Z}^{\gamma}(TK)^mT^j(TK)^{-1}G_{\widehat{A}}[\psi]|^2\,d\sigma dr.
\end{multline*}
First, we we apply \eqref{eq:horpestrmin1pest} to obtain for $ p= 1+\epsilon$:
\begin{multline*}
\sum_{|\gamma|\leq N}\int_{\Sigma_{\tau_*}\cap\{1+\tilde{R}^{-1}\leq r\leq \tilde{R} \}} \mathcal{E}_{1-\epsilon}[\mathbf{Z}^{\gamma}T^{-1}\psi]\,d\sigma dr\leq \sum_{|\gamma|\leq N}\int_{\Sigma_{\tau_*}} \mathcal{E}_{1-\epsilon}[\mathbf{Z}^{\gamma}\psi^{(1)}_{\Sc}]\,d\sigma dr\\
\leq C\sum_{|\gamma|\leq N}\sum_{n\leq N}\int_{\Sigma_{\tau_1}} \mathcal{E}_{1+\epsilon}[\mathbf{Z}^{\gamma}\psi^{(1)}_{\Sc}]\,d\sigma dr\\
+C\int_{\tau_1}^{\tau_2}\int_{\Sigma_{\tau}} \max\{(r^{-1}\Omega)^{-1-\epsilon}\rho_+^{1-2\epsilon}r^{-1+2\epsilon},1\}r^{-2}|\mathbf{Z}^{\gamma}(r^3G^{(1)}_{\Sc})|^2+\Omega^2r^{-2}(1-\upzeta) |\mathbf{Z}^{\gamma}T( r^3G^{(1)}_{\Sc}) |^2\,d\sigma dr d\tau\\
\leq \left[ \tilde{R}^{2\epsilon}+ (1+\tau_*)^{2\epsilon}\right]\sum_{|\gamma|+m\leq N_1+N_2}\sum_{j\leq 2N_1}\int_{\Sigma_{0}} \mathcal{E}_{3-\frac{\epsilon}{2}}[\mathbf{Z}^{\gamma}(TK)^mT^j\widehat{\psi}]\,d\sigma dr\\
+C\sum_{m\leq N_1}\sum_{|\gamma|+m\leq N_1+N_2}\sum_{j\leq 2N_1}\int_{\Sigma_{0}} (\Omega r^{-1})^{-1+\frac{\epsilon}{2} }|r\mathbf{Z}^{\gamma}(TK)^mT^j(TK)^{-1}G_{\widehat{A}}[\psi]|^2\,d\sigma dr\\
+C\int_{\tau_1}^{\tau_2}\int_{\Sigma_{\tau}} \max\{(r^{-1}\Omega)^{-1-\epsilon}\rho_+^{1-2\epsilon}r^{-1+2\epsilon},1\}r^{-2}|\mathbf{Z}^{\gamma}(r^3G^{(1)}_{\Sc})|^2+\Omega^2r^{-2}(1-\upzeta) |\mathbf{Z}^{\gamma}T( r^3G^{(1)}_{\Sc}) |^2\,d\sigma dr d\tau.
\end{multline*}
Note first of all that all the terms on the RHS above involving $G^{(1)}_{\Sc}$ are finite, by the estimates in Proposition \ref{prop:estinhomG} and can be estimated by $|\mathfrak{I}[\psi]|^2+|\mathfrak{H}[\psi]|^2$, using the definition of $\Psi$ appearing in $\widehat{\psi}=\psi-\Psi$.  By including a factor $(r^{-1}\Omega)^{4\epsilon}$ into the left hand side, replacing $\tilde{R}$ with a dyadic sequence $\tilde{R}_j$ and summing over $j$, we obtain:
\begin{multline*}
\sum_{|\gamma|\leq N_2}\int_{\Sigma_{\tau_*}} \mathcal{E}_{1-5\epsilon}[\mathbf{Z}^{\gamma}T^{-1}\psi]\,d\sigma dr\\
\leq C(1+\tau_*)^{2\epsilon}\sum_{|\gamma|\leq N_2}\left[\int_{\Sigma_{0}} \mathcal{E}_{3-\frac{\epsilon}{2}}[\mathbf{Z}^{\gamma}(TK)^mT^j\widehat{\psi}]\,d\sigma dr+|\mathfrak{I}[\psi]|^2+|\mathfrak{H}[\psi]|^2\right].
\end{multline*}
Since $\tau_*$ can be taken arbitrarily large, we conclude \eqref{eq:energydecayptieminvmin1}.

Now suppose that $N_1\geq 1$. Now we apply Corollary \ref{prop:edecaygeneralp} with $\q_1>\frac{1}{2}$ to $\psi^{(1)}_{\Sc}$ to obtain for $1< p<2$ and $0<\epsilon<p-1$:
\begin{multline*}
(1+\tau)^{1-p+2N_1-\nu}\int_{\Sigma_{\tau}}\sum_{|\gamma|\leq N_2}\mathcal{E}_{p}[\mathbf{Z}^{\gamma}(TK)^{N_1}\psi^{(1)}_{\Sc}]\,d\sigma dr \\
\leq C\sum_{m\leq N_1}\sum_{|\gamma|+m\leq N_1+N_2}\sum_{j\leq 2N_1}\int_{\Sigma_{0}} \mathcal{E}_{1+\epsilon}[\mathbf{Z}^{\gamma}(TK)^mT^j\psi^{(1)}_{\Sc}]\,d\sigma dr\\
+\sup_{\tau\in [0,\infty)}\int_{\Sigma_{\tau}} (1+\tau)^{2-\epsilon+2\max\{m-1,0\}}(\Omega r^{-1})^{2-p}r^{-2}|r^3\mathbf{Z}^{\gamma}T^j(TK)^{\max\{m-1,0\}}G^{(1)}_{\Sc}|^2\,d\sigma dr d\tau\\
+\int_{0}^{\infty} \int_{\Sigma_{\tau}} (1+\tau)^{2m}\Big[\max\{(r^{-1}\Omega)^{-1-\epsilon}\rho_+^{1-2\epsilon}r^{-1+2\epsilon},1\}r^{-2}|r^3\mathbf{Z}^{\gamma}T^j(TK)^mG^{(1)}_{\Sc}|^2\\
+(1-\upzeta)\Omega^2|r\mathbf{Z}^{\gamma}T^{j+1}(TK)^mG^{(1)}_{\Sc}|^2\Big]\,d\sigma dr d\tau.
\end{multline*}
and hence:
\begin{multline*}
(1+\tau_*)^{1-p+2N_1-\nu}\left[ \tilde{R}^{2\epsilon}+ (1+\tau_*)^{2\epsilon}\right]^{-1}\sum_{|\gamma|\leq N_2}\int_{\Sigma_{\tau_*}\cap\{1+\tilde{R}^{-1}\leq r\leq \tilde{R} \}}\sum_{|\gamma|\leq N_2}\mathcal{E}_{p}[\mathbf{Z}^{\gamma}(TK)^{N_1-1}\widehat{\psi}]\,d\sigma dr \\
\lesssim \sum_{m\leq N_1}\sum_{|\gamma|+m\leq N_1+N_2}\sum_{j\leq 2N_1}\int_{\Sigma_{0}} \mathcal{E}_{3-\frac{\epsilon}{2}}[\mathbf{Z}^{\gamma}(TK)^mT^j\widehat{\psi}]\,d\sigma dr\\
+\left[ \tilde{R}^{2\epsilon}+ (1+\tau_*)^{2\epsilon}\right]^{-1}\Bigg[\sup_{\tau\in [0,\infty)}\int_{\Sigma_{\tau}} (1+\tau)^{2-\epsilon+2\max\{m-1,0\}}(\Omega r^{-1})^{2-p}r^{-2}|r^3\mathbf{Z}^{\gamma}T^j(TK)^{\max\{m-1,0\}}G^{(1)}_{\Sc}|^2\,d\sigma dr d\tau\\
+\int_{0}^{\infty} \int_{\Sigma_{\tau}} (1+\tau)^{2m}\Big[\max\{(r^{-1}\Omega)^{-1-\epsilon}\rho_+^{1-2\epsilon}r^{-1+2\epsilon},1\}r^{-2}|r^3\mathbf{Z}^{\gamma}T^j(TK)^mG^{(1)}_{\Sc}|^2\\
+(1-\upzeta)\Omega^2|r\mathbf{Z}^{\gamma}T^{j+1}(TK)^mG^{(1)}_{\Sc}|^2\Big]\,d\sigma dr d\tau\Bigg].
\end{multline*}
We can remove the factor $(\Omega^{-1} r)^{2\epsilon}$ on the LHS at the expanse of introducing a factor $(\Omega r^{-1})^{4\epsilon}$ inside the integral by applying the above energy estimate in dyadic $r$-intervals $[\tilde{R}_j,\tilde{R}_{j+1}]\cup [1+\tilde{R}_{j+1}^{-1},1+\tilde{R}_{j}^{-1}]$ and then summing over $j$. 

All the terms on the RHS above involving $G^{(1)}_{\Sc}$ are finite, by the estimates in Proposition \ref{prop:estinhomG} and can be estimated by $|\mathfrak{I}[\psi]|^2+|\mathfrak{H}[\psi]|^2$, using the definition of $\Psi$ appearing in $\widehat{\psi}=\psi-\Psi$. Furthermore, the constants in the estimates do not depend on $\tau_*$ so we can take $\tau_*$ arbitrarily large and we conclude \eqref{eq:energydecayptieminv}.

To derive \eqref{eq:energydecayptieminvfaster}, we apply the elliptic estimates in Corollary \ref{cor:elliptichighfreq} and Proposition \ref{prop:boundedangelliptic} along $\Sigma_{\tau}$ with $p=2-\re\sqrt{1-4\q^2}+\epsilon$ to obtain for $\tau\leq \tau_*$:
\begin{multline*}
\sum_{|\gamma|\leq N_2}\int_{\Sigma_{\tau}}\mathcal{E}_{\epsilon-\re\sqrt{1-4q^2}}[\mathbf{Z}^{\gamma}(TK)^{N_1}\widehat{\psi}]\,d\sigma dr\lesssim \sum_{|\gamma|\leq N_2}\int_{\Sigma_{\tau}}\mathcal{E}_{2-\re\sqrt{1-4q^2}+\epsilon}[\mathbf{Z}^{\gamma}(TK)^{N_1+1}\widehat{\psi}]\,d\sigma dr\\
+\int_{\Sigma_{\tau}}(r^{-1}\Omega)^{\re\sqrt{1-4q^2}-\epsilon}|\mathbf{Z}^{\gamma}(rG)|^2\,d\sigma dr.
\end{multline*} 
Let $N_1\geq 0$. We then split the integral over $\Sigma_{\tau}$ into a region where $\Omega r^{-1}\leq (1+\tau)^{-1}$ and $\Omega r^{-1}\geq (1+\tau)^{-1}$ to obtain:
\begin{multline*}
 \sum_{|\gamma|\leq N_2}\int_{\Sigma_{\tau}}\mathcal{E}_{2-\re\sqrt{1-4q^2}+\epsilon}[\mathbf{Z}^{\gamma}(TK)^{N_1+1}\widehat{\psi}]\,d\sigma dr\\
\lesssim  (1+\tau)^{-1-\re\sqrt{1-4\q^2}}\sum_{|\gamma|\leq N_2+1}\int_{\Sigma_{\tau}} \mathcal{E}_{1+\epsilon}[\mathbf{Z}^{\gamma}(TK)^{N_1}\widehat{\psi}]+\mathcal{E}_{1+\epsilon}[\mathbf{Z}^{\gamma}T(TK)^{N_1}\widehat{\psi}]\,d\sigma dr\\
 +(1+\tau)^{-1-\re\sqrt{1-4\q^2}}\sum_{|\gamma|\leq N_2+1}\int_{\Sigma_{\tau}} (\Omega r^{-1})^{1 -\epsilon}(|\mathbf{Z}^{\gamma}(TK)^{N_1}(rG)|^2+ |\mathbf{Z}^{\gamma}(TK)^{N_1}(rG)|^2)\,d\sigma dr\\
 + (1+\tau)^{1-\re \sqrt{1-4\q^2}} \sum_{|\gamma|\leq N_2}\int_{\Sigma_{\tau}}\mathcal{E}_{1-\epsilon}[\mathbf{Z}^{\gamma}(TK)^{N_1+1}\widehat{\psi}]\,d\sigma dr.
\end{multline*}
Then we apply \eqref{eq:energydecayptieminv} with $N_1$ replaced by $N_1+1$ and separately with $N_2$ replaced by $N_2+1$ to conclude \eqref{eq:energydecayptieminvfaster}. The $N_1=-1$ case proceeds analogously.

Finally, to obtain \eqref{eq:energydecayptieminvfasterDT}, we note that 
\begin{align*}
	|D_TX\psi|\lesssim&\: |KX\psi|+(r-1)|X\psi|\quad r\leq R,\\
	|D_TX\psi|\lesssim&\: |TX\psi|+r^{-1}|X\psi|\quad r> R.
\end{align*}
Hence, it follows from \eqref{eq:maineqradfieldconf} and \eqref{eq:maineqradfieldconfhor} that
\begin{equation*}
	\mathcal{E}_{p}[D_T\psi]\lesssim \sum_{|\gamma|\leq 1}\mathcal{E}_{p-2}[\mathbf{Z}^{\gamma}\psi].
\end{equation*}
We then apply \eqref{eq:energydecayptieminvfaster} with $N_2$ replaced by $N_2+1$ to obtain the desired result.
\end{proof}
\subsection{Proof of Theorem \ref{thm:mainthmenergy}}
\label{sec:pfthm2}
First, note that 
\begin{multline*}
	\int_{\Sigma_{\tau}}\mathcal{E}_{p}[\mathbf{Z}^{\gamma}((1
+\tau)TK)^{N_1}\Psi^{\infty}]\,d\sigma dr\\
\sim 	\sum_{\ell(\ell+1)<\q^2}\sum_{m\in \Z, |m|\leq \ell}|\mathfrak{I}_{\ell m}[\psi]|^2\int_{\Sigma_{\tau}}\mathcal{E}_{p}[\mathbf{Z}^{\gamma}((1
+\tau)TK)^{N_1}((\Psi_0^{\infty})_{\ell m}Y_{\ell m})]\,d\sigma dr.
\end{multline*}
We can therefore apply Proposition \ref{prop:Psi0inftygrowthdecay} to obtain for $0\leq p\leq 2$:
 	\begin{equation}
 	\label{eq:mainPsiinfendecay}
	\int_{\Sigma_{\tau}}\mathcal{E}_{p-3\epsilon}[\mathbf{Z}^{\gamma}((1
+\tau)TK)^{N_1}\Psi^{\infty}]\,d\sigma dr\leq  \begin{cases}C|\mathfrak{I}[\psi]|^2(1+\tau)^{p-2-2\re \beta_{\ell}}\quad (\beta_{\ell}\neq 0),\\
	C|\mathfrak{I}[\psi]|^2\frac{(1+\tau)^{p-2}}{\log^2(1+\tau)}\quad (\beta_{\ell}= 0)
  \end{cases}	
\end{equation}
When $|Q|=M$, we similarly obtain:
\begin{equation}
 	\label{eq:mainPsihorendecay}
	\int_{\Sigma_{\tau}}\mathcal{E}_{p-3\epsilon}[\mathbf{Z}^{\gamma}((1
+\tau)TK)^{N_1}\Psi^{+}]\,d\sigma dr\leq  \begin{cases}C|\mathfrak{H}[\psi]|^2(1+\tau)^{p-2-2\re \beta_{\ell}}\quad (\beta_{\ell}\neq 0),\\
	C|\mathfrak{H}[\psi]|^2\frac{(1+\tau)^{p-2}}{\log^2(1+\tau)}\quad (\beta_{\ell}= 0).
  \end{cases}	
\end{equation}

By Young's inequality, we have for any $\delta>0$ that:
	\begin{equation*}
		(1-(1+\tau)^{-\delta})\mathcal{E}_{p}[\Psi]-((1+\tau)^{\delta}-1)\mathcal{E}_{p}[\hpsi]\leq \mathcal{E}_{p-3\epsilon}[\psi]\leq (1+(1+\tau)^{-\delta})\mathcal{E}_p[\Psi]+(1+(1+\tau)^{\delta})\mathcal{E}_{p}[\hpsi].
	\end{equation*}
	Hence, we conclude the estimates for $1<p<2$ by combining the above with \eqref{eq:energydecayptieminv} and for $0\leq p<1-\epsilon$ by applying additionally \eqref{eq:energydecayptieminvfaster} and then performing a standard interpolation argument.
	
	To obtain the $p=2$ estimate, we apply \eqref{eq:maineqradfield} applied to $(TK)^{-1}\hpsi$, to estimate:
	\begin{equation*}
		\mathcal{E}_2[\mathbf{Z}^{\gamma}(TK)^{N_1}\hpsi]\lesssim \mathcal{E}_1[\mathbf{Z}^{\gamma}(TK)^{N_1}\hpsi]+\sum_{|\gamma|\leq 2}\mathcal{E}_0[\mathbf{Z}^{\gamma}(TK)^{N_1-1}\hpsi]+r^{-2}|\mathbf{Z}^{\gamma}(TK)^{N_1}(r^3G^{(1)})|^2.
	\end{equation*}
Hence, by interpolating between \eqref{eq:energydecayptieminvmin1} and \eqref{eq:energydecayptieminvfaster}, we obtain:
\begin{equation}
\label{eq:nindegendecaypsihat}
	\int_{\Sigma_{\tau}}\mathcal{E}_{2}[\mathbf{Z}^{\gamma}((1
+\tau)TK)^{N_1}\hpsi]\,d\sigma dr\leq CE_{N_1,N_2+2,\epsilon}[\psi](1+\tau)^{-1+8\epsilon+\nu}.
\end{equation}
We conclude the proof by combining the above with \eqref{eq:mainPsiinfendecay} and \eqref{eq:mainPsihorendecay}.

\section{Late-time tails}
\label{sec:tails}
In this section, we prove the main result of the present paper, which can be interpret as a determination of the precise late-time asymptotics of solutions $\phi$ to \eqref{eq:CSF} with $G_A=0$.
\begin{proposition}
\label{prop:tails}
Let $N_1\in \N_0\cup\{-1\}$, $N_2\in \N_0$. Let $l\in\{0,1\}$. For arbitrarily small $\nu>0$ and $\epsilon>0$, there exists a constant $C=C(\nu,\epsilon, N_1,N_2,p)>0$ such that:
\begin{multline}
\label{eq:mainpointwisedecay}
\sum_{|\gamma|\leq N_2}||D_T^l\mathbf{Z}^{\gamma}(TK)^{N_1}(\psi-\Psi)||_{L^{\infty}}\\
\leq C  \sum_{k\leq 2}\sqrt{E_{N_1+1,N_2+1+l,\epsilon}[\snabla_{\s^2}^k\psi]}(\tau+1)^{-1+\frac{1}{2}\nu+4\epsilon-N_1-\frac{l}{2}}(r^{-1}\Omega\tau+1)^{-\frac{1}{2}-\frac{1}{2}\re \beta_{0}}.
\end{multline}
If $|Q|<M$, we moreover have for all $l\in \N_0$:
\begin{multline}
\label{eq:mainpointwisedecaysubext}
\sum_{|\gamma|\leq N_2}||X^l\mathbf{Z}^{\gamma}(TK)^{N_1}(\psi-\Psi^{\infty})||_{L^{\infty}}\\
\leq C \kappa_+^{-\frac{1}{2}} \sum_{k\leq 2}\sqrt{E_{N_1+1,N_2+l+1,\epsilon}[\snabla_{\s^2}^k\psi]}(\tau+1)^{-1+\frac{1}{2}\nu+4\epsilon-N_1}r ^{\frac{1}{2}+\frac{1}{2}\re\beta_0}(r^{-1}\tau+1)^{-\frac{1}{2}-\frac{1}{2}\re \beta_{0}}\\
+Ce^{-c\kappa_+\tau}\sum_{|\gamma|\leq N_2}\sqrt{\int_{\Sigma_{0}\cap \{r\leq r_H\}}\mathcal{E}_2[\snabla_{\s^2}^kX^l\mathbf{Z}^{\gamma}(TK)^{N_1}\psi]\,d\sigma dr}.
\end{multline}
\end{proposition}
\begin{proof}
	We will derive $L^{\infty}$-estimates for $\widehat{\psi}$ by using the energy estimates of Proposition \ref{prop:energydecayptieminv}.
	
	Consider first the region $\{r\geq r_H\}$ with $r_H>r_+$. Then we can estimate via the fundamental theorem of calculus:
	\begin{equation*}
	\int_{\s^2}r^{-1}|\widehat{\psi}|^2(\tau,r,\cdot),d\sigma\lesssim \int_{\Sigma\cap\{r\geq r_H\}}\mathcal{E}_{0}[\widehat{\psi}]\,d\sigma dr.
	\end{equation*}
	Furthermore, for all $r\geq r_H$, we have:
\begin{equation*}
		\int_{\s^2} |\widehat{\psi}|^2(\tau,r,\cdot),d\sigma\lesssim \int_{\Sigma_{\tau}\cap\{r\geq r_H\}}\mathcal{E}_{1+\delta}[\widehat{\psi}]\,d\sigma dr+\int_{\s^2}|\widehat{\psi}|^2(\tau,r_H,\cdot),d\sigma.
		\end{equation*}
From \eqref{eq:energydecayptieminv}, it then follows that for $r\geq r_+$ and arbitrarily small $\nu$:
\begin{equation*}
	\int_{\s^2}|\widehat{\psi}|^2(\tau,r,\cdot)\lesssim C(1+\tau)^{-2+\delta+5\epsilon+\nu}E_{0,0,\epsilon}[\psi].
\end{equation*}
We moreover that by the fundamental theorem of calculus:
\begin{equation*}
		\int_{\s^2}(r^{-1}\Omega)^{\re \beta_0+1}|\widehat{\psi}|^2(\tau,r,\cdot),d\sigma\leq \int_{\Sigma_{\tau}}\mathcal{E}_{\epsilon-\re \beta_0}[\widehat{\psi}]\,d\sigma dr.
	\end{equation*}
We then apply \eqref{eq:energydecayptieminvfaster} to conclude that:
\begin{equation}
\label{eq:auxpointwest}
	\int_{\s^2}(r^{-1}\Omega)^{\re \beta_0+1}|\widehat{\psi}|^2(\tau,r,\cdot)\lesssim C(1+\tau)^{-3-\re \beta_0+8\epsilon+\nu}E_{1,1,\epsilon}[\psi].
\end{equation}
To obtain an estimate for $D_T\psi$, we instead apply \eqref{eq:energydecayptieminvfasterDT}.
We combine the above estimates and apply a standard Sobolev inequality on $\s^2$ to conclude \eqref{eq:mainpointwisedecay} with $N_1=N_2=0$. The general $N_1$ and $N_2$ cases follows entirely analogously.

In the $|Q|<M$ case, we can improve the above estimates by removing the factor $\Omega$ in \eqref{eq:auxpointwest} and by making use of the red-shift effect.

From a standard red-shift energy estimate in $\{r\leq r_H\}$ with $r_H-r_+>0$ suitably small, see for example \cite{lecturesMD}[\S 7] for a general analysis and \cite{gaj22a}[Proposition 6.7], it follows that there exists a numerical constant $c>0$, such that:
\begin{multline*}
	c \kappa_+ \int_{\tau}^{\infty}\int_{\Sigma_{\tau'}\cap \{r\leq r_H\}}\mathcal{E}_2[\hpsi]\,d\sigma drd\tau'-\int_{\Sigma_{\tau}\cap \{r\leq r_H\}}\mathcal{E}_2[\hpsi]\,d\sigma dr\\
	\leq C \int_{\tau}^{\infty}\int_{\Sigma_{\tau'}\cap\{r_H\leq r\leq r_H+\frac{1}{2}(r_H-r_+)\}}\mathcal{E}_2[\hpsi]\,d\sigma drd\tau'+C\int_{\tau}^{\infty}\int_{\Sigma_{\tau'}\cap \{r\leq r_H+\frac{1}{2}(r_H-r_+)\}}|G_{\widehat{A}}[\hpsi]|^2\,d\sigma drd\tau'.
\end{multline*}
Let $f(\tau):=\int_{\Sigma_{\tau}\cap \{r\leq r_H\}}\mathcal{E}_2[\hpsi]\,d\sigma dr$, then this implies that:
\begin{multline*}
	\frac{d}{d\tau}\left[e^{c\kappa_+\tau}f(\tau)\right]\leq C e^{c\kappa_+\tau}\int_{\Sigma_{\tau'}\cap\{r_H\leq r\leq r_H+\frac{1}{2}(r_H-r_+)\}}\mathcal{E}_2[\hpsi]\,d\sigma drd\tau'\\
+Ce^{c\kappa_+\tau}\int_{\Sigma_{\tau}\cap \{r\leq r_H+\frac{1}{2}(r_H-r_+)\}}|G_{\widehat{A}}[[\Psi]|^2\,d\sigma drd\tau'
\end{multline*}
Hence,
\begin{multline*}
	f(\tau)\leq Ce^{-c\kappa_+\tau}\int_{\Sigma_{0}\cap \{r\leq r_H\}}\mathcal{E}_2[\hpsi]\,d\sigma dr+C \kappa_+^{-1}\int_{\Sigma_{\tau'}\cap\{r_H\leq r\leq r_H+\frac{1}{2}(r_H-r_+)\}}\mathcal{E}_2[\hpsi]\,d\sigma drd\tau'\\
	+C\kappa_+^{-1}\int_{\Sigma_{\tau}\cap \{r\leq r_H+\frac{1}{2}(r_H-r_+)\}}|G_{\widehat{A}}[[\Psi]|^2\,d\sigma drd\tau'.
\end{multline*}
By \eqref{eq:energydecayptieminvfaster}, we therefore have that:
\begin{equation*}
	\int_{\Sigma_{\tau}\cap \{r\leq r_H\}}\mathcal{E}_2[\hpsi]\,d\sigma dr\leq C\kappa_+^{-1}(1+\tau)^{-3-\re \beta_0+\epsilon+\nu}E_{1,1,\epsilon}[\psi]+Ce^{-c\kappa_+\tau}\int_{\Sigma_{0}\cap \{r\leq r_H\}}\mathcal{E}_2[\hpsi]\,d\sigma dr.
\end{equation*}
Then we can apply the fundamental theorem of calculus in $\{r\leq r_H\}$ to conclude that:
\begin{equation*}
	\int_{\s^2}r^{-\re \beta_0-1+\epsilon}|\widehat{\psi}|^2(\tau,r,\cdot)\lesssim C\kappa_+^{-1}(1+\tau)^{-3-\re \beta_0+\epsilon+\nu}E_{1,1,\epsilon}[\psi]+Ce^{-c\kappa_+\tau}\int_{\Sigma_{0}\cap \{r\leq r_H\}}\mathcal{E}_2[\hpsi]\,d\sigma dr.
\end{equation*}
 The above improvement applies also after commuting the above estimates with $\mathbf{Z}^{\gamma}(TK)^{N_1}$ as well as $X^{l}$ and we conclude \eqref{eq:mainpointwisedecaysubext} after applying a standard Sobolev inequality on $\s^2$.
\end{proof}

\begin{corollary}
\label{cor:Bconstnonzero}
	$B_{\infty,\ell m}\neq 0$ and $B_{+,\ell m}\neq 0$ for all $\ell\in \N_0$ such that $\ell(\ell+1)<q^2$.
\end{corollary}
\begin{proof}
Note first of all that $K^{-1}\Psi^{\infty}_{\ell m}$ is well-defined. Suppose that $B_{\infty,\ell m}=0$. Then, in the case of zero initial data for $\psi_{\ell m}$, we have by Proposition \ref{prop:timeinvffixedmodeconstr} that $T^{-1}\Psi_{\infty}$ is also well-defined. Hence, $\widehat{\psi}_{\ell}^{(1)}=i\q\sum_{|m|\leq \ell}(T^{-1}\Psi^{\infty}_{\ell m}-K^{-1}\Psi^{\infty}_{\ell m})Y_{\ell m}$ is well-defined. Since $TK \widehat{\psi}_{\ell}^{(1)}=\sum_{|m|\leq \ell}\Psi^{\infty}_{\ell m}Y_{\ell m}$, the decay estimates in \eqref{eq:mainpointwisedecay} imply decay for $\Psi^{\infty}_{\ell m}$ and hence also $(\Psi^{\infty}_0)_{\ell m}$ that is faster than what follows from Proposition \ref{eq:asymptPsiinfty0}, which is a contradiction.  We must therefore have $B_{\infty,\ell m}\neq 0$.

In the $|Q|=M$ case, an analogous contradiction argument with $B_{+,\ell m}$ taking the role of $B_{\infty,\ell m}$ and $\Psi^{+}_{\ell m}$  replacing $\Psi^{\infty}_{\ell m}$ leads to the statement: $B_{+,\ell m}\neq 0$.
\end{proof}
\subsection{Proof of Theorem \ref{thm:mainthmpoint}}
\label{sec:proofthm1}
Theorem \ref{thm:mainthmpoint}(i) follows immediately from Proposition \ref{prop:tails} and Theorem \ref{thm:mainthmpoint}(ii) follows from Proposition \ref{eq:asymptPsiinfty0}.

\section{Instabilities and energy concentration}
\label{sec:instab}
In this section, we establish lower bound estimates for $r^2X\psi$ and the energies $\mathcal{E}_2[\psi]$, following from the leading-order behaviour of $r^2X\Psi$.
\subsection{Proof of Theorem \ref{thm:mainpointwinst}}
\label{sec:pfthm3}
	We first consider (i). We apply \eqref{eq:nindegendecaypsihat} in combination with Proposition \ref{prop:Psi0inftygrowthdecay}(ii) to conclude the desired statements in the region $r\geq r_I$. The region $r\leq r_H$, then follows immediately after transforming $r-1\mapsto (r-1)^{-1}$ and $\q\mapsto -\q$.
	
	Now consider (ii). Again, without loss of generality, we can consider $(\tau,\infty)$. Let $k=0$. Then we apply \eqref{eq:maineqradfieldconf} at $\rho_{\infty}=0$ with $\h\equiv 0$ when $r\geq r_I$ to obtain:
	\begin{equation*}
		-2(r^2X\hpsi)(\tau,\infty,\cdot)=\int_{0}^{\tau}r^3(G_{\widehat{A}})(\tau',\infty)+(\slashed{\Delta}_{\s^2}+i\q) \hpsi(\tau',\infty,\cdot)\,d\tau'
	\end{equation*}
	By \eqref{eq:mainpointwisedecay} and \eqref{eq:erroresteqforPsiinf}, we therefore have that:
	\begin{equation}
	\label{eq:decayr2Xhpsi}
		|(r^2X\hpsi)(\tau,\infty)|\leq \sqrt{E_{1,1,\epsilon}[\psi]}(\tau+1)^{\frac{1}{2}\nu+4\epsilon}.
	\end{equation}
	Note that 
	\begin{equation*}
	r^2X\psi_{\ell m}=r^2X\Psi_{\ell m}+r^2X\hpsi_{\ell m}=\mathfrak{I}_{\ell m}[\psi]r^2X(\Psi^{\infty}_0)_{\ell m}+\ldots,
	\end{equation*}
	where the terms in $\ldots$ can be bounded by $(\tau+1)^{\frac{1}{2}\nu+4\epsilon}$ by Proposition \ref{prop:tails}.
	
	Taking the $\limsup$ and applying Proposition \ref{prop:Psi0inftygrowthdecay}(i) then results in the $k=0$ estimates in (ii).
	
	To obtain the $k\geq 1$ case, we consider $\partial_{\rho_{\infty}}^kr^3(G_{\widehat{A}})$ and we establish estimates for $\partial_{\rho_{\infty}}^{k+1}\hpsi$ inducively. Then we apply Proposition \ref{prop:Psi0inftygrowthdecay}(i).

\subsection{Transient instabilities}
The estimates in the present paper are uniform in $\kappa_+$, which allows us to obtain lower bounds on the size of $X$-derivatives along the event horizon, which may be interpreted as \emph{transient instabilities}.
\begin{proposition}
\label{prop:transinstb}
	Let $\psi$ be a solution to \eqref{eq:CSF} with $G_A=0$ and $A=-\frac{Q}{r}d\tau$ and $\kappa_+>0$. Let $k\in \N_0$. Fix the initial data for $\psi$ independently of $\kappa_+$. Let $\psi_M$ be a solution arising from the same data, but with $\kappa_+=0$. Then, there exists  suitably small $0<\delta_1<\delta_2$, depending only on the initial data for $\psi$ and on $k$, such that for $r_+\kappa_+\leq \delta_1$ and $\ell(\ell+1)<\q^2$:
	\begin{align*}
		\sup_{\delta_1 \kappa_+^{-1}\leq \tau\leq \delta_2 \kappa_+^{-1}}(\tau+1)^{-\frac{1}{2}-k+\frac{1}{2}\re\beta_0}|X^{k+1}\psi_{\ell m}|(\tau,r_+)\geq&\:  \frac{|\mathfrak{c}_k|}{2}|\mathfrak{H}_{\ell m}[\psi_M]|\quad &(\beta_{0}\neq 0),\\
		\sup_{\delta_1 \kappa_+^{-1}\leq \tau\leq \delta_2 \kappa_+^{-1}}\log(1+\tau)\tau^{-\frac{1}{2}-k}|X^{k+1}\psi_{\ell m}|(\tau,r_+)\geq&\:  \frac{|\mathfrak{c}_k|}{2}|\mathfrak{H}_{\ell m}[\psi_M]|\quad &(\beta_{0}= 0).
	\end{align*}
\end{proposition}
\begin{proof}
Fix $Q$ and let $\psi^M$ denote a solution to \eqref{eq:CSF} with $G_A=0$ and $M=|Q|$, such that $(\psi^M|_{\Sigma_0},\partial_{\tau}\psi^M|_{\Sigma_0})=(\psi|_{\Sigma_0},\partial_{\tau}\psi|_{\Sigma_0})$. Then
\begin{multline*}
	(g^{-1}_{M,M})^{\mu\nu}(^{\widehat{A}}D)_{\mu}(^{\widehat{A}}D)_{\nu}(r^{-1}(\psi-\psi^M))=(r_--r_+)X((r-r_+)r^{-2}X\psi)+2(Q^2-M^2)r^{-4}\psi=:F_{M,Q}.
\end{multline*}
Then we can estimate: $|F_{M,Q}|^2\leq \kappa_+^2r^{-6}\sum_{|\gamma|=1}\mathcal{E}_2[\mathbf{Z}^{\gamma}\psi]$, so by standard local energy estimates, we obtain:
\begin{multline*}
	\sup_{\tau'\in [0,\tau]}\int_{\Sigma_{\tau'}}\mathcal{E}_2[\psi-\psi^M]\,d\sigma dr\leq C\kappa_+\sum_{|\gamma|=1}\int_0^{\tau}\int_{\Sigma_{\tau'}}\sqrt{\mathcal{E}_2[\psi-\psi^M]}\sqrt{\mathcal{E}_2[\mathbf{Z}	^{\gamma}\psi]}\,d\sigma drd\tau'\\
	\leq C\kappa_+ \sqrt{\int_0^{\tau}\int_{\Sigma_{\tau'}}\mathcal{E}_2[\psi-\psi^M]\,d\sigma drd\tau'}\sqrt{\sum_{|\gamma|=1}\int_0^{\tau}\int_{\Sigma_{\tau'}}\mathcal{E}_2[\mathbf{Z}	^{\gamma}\psi]\,d\sigma drd\tau'}\\
	\leq \frac{1}{2}\sup_{\tau'\in [0,\tau]}\int_{\Sigma_{\tau'}}\mathcal{E}_2[\psi-\psi^M]\,d\sigma dr+\frac{C^2}{2}(\kappa_+\tau)^2\sum_{|\gamma|=1}\sup_{\tau'\in [0,\tau]}\int_{\Sigma_{\tau'}}\mathcal{E}_2[\mathbf{Z}	^{\gamma}\psi]\,d\sigma dr.
\end{multline*}
By applying Theorem \ref{thm:mainthmenergy}, we therefore have that for $\epsilon>0$ arbitrarily small:
\begin{equation*}
	\int_{\Sigma_{\tau}}\mathcal{E}_2[\psi-\psi^M]\,d\sigma dr\leq C(\kappa_+\tau)^{2} E_{0,3,\epsilon}[\psi](1+\tau)^{-2\re \beta_0}.
\end{equation*}
 
 We now apply the fundamental theorem of calculus in the $r$-direction together with Theorem \ref{thm:mainthmenergy} to estimate:
 \begin{multline*}
 	\int_{\s^2}|\psi-\psi_M|^2(\tau,r_+,\cdot)\,d\sigma\leq C\sqrt{\int_{\Sigma_{\tau}\cap\{r\leq r_H\}}\mathcal{E}_2[\psi-\psi_M]\,d\sigma dr}\sqrt{\int_{\Sigma_{\tau}\cap\{r\leq r_I\}}\mathcal{E}_0[\psi-\psi_M]\,d\sigma dr}\\
 	\leq C(\kappa_+\tau)E_{1,3,\epsilon}[\psi](1+\tau)^{-1-\re\beta_0}.	
 \end{multline*}
In particular, we obtain:
 \begin{equation*}
 |(\psi_{\ell m}-(\psi_M)_{\ell m})(\tau,r_+)|\\
 	\leq C(\kappa_+\tau)^{\frac{1}{2}}\sqrt{\sum_{k=0}^2E_{1,3,\epsilon}[\snabla_{\s^2}^k\psi]}(1+\tau)^{-\frac{1}{2}(1+\re\beta_0)}.
 \end{equation*}
Now let $k=0$. From \eqref{eq:decayr2Xhpsi}, it follows in particular that:
	\begin{equation*}
		|(r^2X((\psi_M)_{\ell m}-\Psi^+_{\ell m}))(\tau,r_+)|\leq C\sqrt{E_{1,1,\epsilon}[\psi]}(\tau+1)^{\frac{1}{2}\nu+4\epsilon}.
	\end{equation*}
 
We apply \eqref{eq:maineqradfieldconfhor} with $\widetilde{\h}=0$ in $r\leq r_H$ to obtain:
\begin{equation*}
		2T(r^2X(\psi-\psi_M))(\tau,r_+,\cdot)=\kappa_+r^2X\psi(\tau,r_+)+(\slashed{\Delta}_{\s^2}-i\q) (\psi-\psi_M).
	\end{equation*}
We can rewrite the above as follows:
\begin{equation*}
		2T(e^{\kappa_+ \tau}r^2X(\psi-\psi_M))(\tau,r_+,\cdot)=\kappa_+e^{\kappa_+ \tau}r^2X\psi_M(\tau,r_+)+e^{\kappa_+ \tau}(\slashed{\Delta}_{\s^2}-i\q) (\psi-\psi_M).
	\end{equation*}
Integrating the above therefore gives for $\kappa_+\tau<1$:
\begin{equation*}
		2|r^2X(\psi-\psi_M)_{\ell m}|(\tau,r_+,\cdot)\leq C(\kappa_+ \tau)^{\frac{1}{2}}\sqrt{E_{1,3,\epsilon}[\psi]}(1+\tau)^{\frac{1}{2}(1-\re\beta_0)}
	\end{equation*}
and hence, we obtain:
\begin{equation*}
		2|r^2X(\psi-\Psi^+)_{\ell m}|(\tau,r_+,\cdot)\leq C(\kappa_+ \tau)^{\frac{1}{2}}\sqrt{E_{1,3,\epsilon}[\psi]}(1+\tau)^{\frac{1}{2}(1-\re\beta_0)}+C\sqrt{E_{1,1,\epsilon}[\psi]}(1+\tau)^{\frac{1}{2}\nu+4\epsilon}
	\end{equation*}
	
	Now, let $\delta_2>0$ be suitably small, so that for $r_+ \delta_2^{-1}\leq \tau\leq \delta_2 \kappa_+^{-1}$:
	\begin{equation*}
		(1+\tau)^{-\frac{1}{2}+\frac{\re \beta_0}{2}}|r^2X(\psi-\Psi^+)_{\ell m}|(\tau,r_+)\leq \frac{|\mathfrak{c}_0|}{4}|\mathfrak{H}_{\ell m}[\psi_M]|.
	\end{equation*}
	If $\beta_0\in [0,\infty)$, we have that
	\begin{equation*}
		(1+\tau)^{-\frac{1}{2}+\frac{\re \beta_0}{2}}(1+\delta_{\beta_0 0}\log (1+\tau))|r^2X\Psi^+_{\ell m}|(\tau,r_+)\geq |\mathfrak{c}_0||\mathfrak{H}_{\ell m}[\psi_M]|,
	\end{equation*}
	so the stated estimates follow immediately.
	
	Suppose that $\beta_0\in i(0,\infty)$. Note that for $\sigma=\im \beta_{\ell} \log(r_++\tau)$, we have by the mean value theorem that for $\delta_1<\delta_0$:
	\begin{multline*}
		\sup_{\delta_1 \kappa_+^{-1}\leq \tau \leq \delta_2 \kappa_+^{-1}}(r_++\tau)^{-1}|r^2X\Psi^+_{\ell m}|^2(\tau,r_+)\\
		\geq \frac{1}{\im \beta_{\ell}}\frac{1}{\log(\frac{r_++\delta_2 \kappa_+^{-1}}{r_++\delta_1 \kappa_+^{-1}})}\int_{\im \beta_{\ell}\log(r_++\delta_1 \kappa_+^{-1})}^{\im \beta_{\ell}\log(r_++\delta_2 \kappa_+^{-1})}(r_++\tau(\sigma))^{-1}|r^2X\Psi^+_{\ell m}|^2(\tau(\sigma),r_+)\,d\sigma\\
		\geq \mathfrak{c}_0^2 |\mathfrak{H}_{\ell m}[\psi_M]|^2-C|\mathfrak{H}_{\ell m}[\psi_M]|^2\frac{1}{\im \beta_{\ell}}\frac{1}{\log(1+\frac{\frac{\delta_2}{\delta_1}-1}{1 +\delta_1^{-1}\kappa_+ r_+})}\geq \frac{9}{16}\mathfrak{c}_0^2 |\mathfrak{H}_{\ell m}[\psi_M]|^2,
	\end{multline*}
	where we applied \eqref{eq:expvaluePsi0infest} to arrive at the second inequality above and we assumed that $\frac{\delta_2}{\delta_1}$ is suitably large and $\kappa_+\delta_1^{-1}\leq r_+^{-1}$ to arrive at the last inequality.
	
	We conclude that there exist $0<\delta_1<\delta_2$ suitably small with $\kappa_+\leq \delta_1r_+^{-1}$, such that:
	\begin{equation*}
		\sup_{\delta_1 \kappa_+^{-1}\leq \tau \leq \delta_2 \kappa_+^{-1}}(1+\tau)^{-\frac{1}{2}}|r^2X\psi_{\ell m}|(\tau,r_+)\geq \frac{|\mathfrak{c}_0|}{2}|\mathfrak{H}_{\ell m}[\psi_M]|.
	\end{equation*}
	To obtain the the desired estimates for $k'\geq 1$, we proceed by induction: we commute \eqref{eq:maineqradfieldconfhor} with $\partial_{\rho_+}^{k'}$ and then integrate in $\tau$, applying the estimates for $k\leq k'-1$.
\end{proof}

\appendix

\section{Global tail functions on Minkowski}
\label{sec:tailfunctmink}
In this section, we will construct solutions to \eqref{eq:CSFmink} with $A=\widehat{A}$ and $\h=0$, that will play an important role in the late-time asymptotics when $M\neq 0$.

It will be useful to work in $(\tau, \varrho, \theta,\varphi)$ coordinates, where
\begin{equation*}
	\varrho=\frac{\tau+1}{2r}.
\end{equation*}
Note that $\varrho$ is preserved by the conformal Killing vector field $S+T$, with $S=u\partial_u+r\partial_r$ and $T=\partial_u$ in $(u,r,\theta,\varphi)$ coordinates.

Denote $\mathfrak{X}:=\partial_{\varrho}$ and $\mathfrak{T}:=\partial_{\tau}$ with respect to $(\tau,\varrho, \theta,\varphi)$ coordinates. Then we can express:
\begin{align*}
	\mathfrak{X}:=&-\frac{\tau+1}{2\varrho^2}X,\\
	\mathfrak{T}:=&\:T+\frac{1}{2\varrho}X,
\end{align*}
or equivalently,
\begin{align*}
	X=&-\frac{2\varrho^2}{\tau+1}\mathfrak{X},\\
	T=&\: \mathfrak{T}+\frac{\varrho}{\tau+1}\mathfrak{X}.
\end{align*}

We will be considering solutions to \eqref{eq:CSFmink} with $A=\widehat{A}$ and $\h=0$ of the form:
\begin{equation*}
	(\tau+1)^{-s}F(\varrho)Y_{\ell m}(\theta,\varphi),
\end{equation*}
where $F$ is a function to be determined.

Note that \eqref{eq:CSFmink} commutes with the vector field $S+T$ and the solutions of the above form are eigenfunctions of the operator $S+T$:
\begin{equation*}
	(S+T)((\tau+1)^{-s}F(\varrho)Y_{\ell m}(\theta,\varphi))=-s \cdot (\tau+1)^{-s}F(\varrho)Y_{\ell m}(\theta,\varphi).
\end{equation*}

\begin{lemma}
\label{lm:solnminkhypg}
\begin{enumerate}[label=\emph{(\roman*)}]\mbox{}
	\item Let $s\in \C\setminus  \Z$. The function $(\tau+1)^{s}F(\varrho)Y_{\ell m}(\theta,\varphi)$ is a solution to \eqref{eq:CSFmink} with $A=\widehat{A}$ and $\h=0$ that is smooth in $\varrho$ at $\varrho=0$ if and only if
	\begin{equation*}
		F(\varrho)= B\cdot  {}_2\mathbf{F}_1(a,b,c;-\varrho),
	\end{equation*}
	with ${}_2\mathbf{F}_1(a,b,c;\cdot)$ the regularized Gauss hypergeometric function, $B\in \C$ and
	\begin{align*}
		a=&\:\frac{1}{2}\pm \frac{\beta_{\ell}}{2}+i\q,\\
		b=&\:\frac{1}{2}\mp \frac{\beta_{\ell}}{2}+i\q,\\
		c=&\: s+1.
	\end{align*}
	In particular, we obtain the following simplifications for particular choices of $s$: let $a=\frac{1}{2}+ \frac{\beta_{\ell}}{2}+i\q$, then
	\begin{align*}
		(\tau+1)^sF(\varrho)=&\:B(1+\tau)^{-\frac{1}{2}-\frac{\beta_{\ell}}{2}+i\q}(1+\varrho)^{-\frac{1}{2}-\frac{\beta_{\ell}}{2}-i\q}\quad (s=b-1),\\
				(\tau+1)^sF(\varrho)=&\:B(1+\tau)^{-\frac{1}{2}+\frac{\beta_{\ell}}{2}+i\q}(1+\varrho)^{-\frac{1}{2}+\frac{\beta_{\ell}}{2}-i\q}\quad (s=a-1).
	\end{align*}
	\item Let $\beta_{\ell}=0$ and define
	\begin{equation*}
		H(\tau,\varrho):=(\tau+1)^{-\frac{1}{2}+i\q}(1+\varrho)^{-\frac{1}{2}-i\q}(\log(\varrho+1)+\log (\tau+1)).
	\end{equation*}
	Then $H(\tau,\varrho)Y_{\ell m}(\theta,\varphi)$ is a solution to \eqref{eq:CSFmink} with $A=\widehat{A}$ and $\h=0$ that is smooth in $\varrho$ at $\varrho=0$.
	\end{enumerate}
\end{lemma}
\begin{proof}
Let $\psi$ be a solution to \eqref{eq:CSFmink} with $A=\widehat{A}$ and $\h=0$ and with $\slashed{\Delta}_{\s^2}\psi=-\ell(\ell+1)\psi$. Then:
\begin{multline*}
	0=X^2\psi-2(T+i\q r^{-1})X\psi+(i\q-\ell(\ell+))r^{-2}\psi\\
	=4\varrho^3(\varrho+1)(\tau+1)^{-2}\mathfrak{X}^2\psi+4\varrho^2(\tau+1)^{-1}\mathfrak{T}\mathfrak{X}\psi+4\varrho^2(\tau+1)^{-2}(2(1+i\q)\varrho+1)\mathfrak{X}\psi+(i\q-\ell(\ell+1))4\varrho^2(\tau+1)^{-2}\psi.
\end{multline*}	
Equivalently,
\begin{equation}
\label{eq:CSFminkrhotau}
	0= \varrho(\varrho+1)\mathfrak{X}^2\psi+(\tau+1)\mathfrak{T}\mathfrak{X}\psi+(2(1+i\q)\varrho+1)\mathfrak{X}\psi+(i\q-\ell(\ell+1))\psi.
\end{equation}
Now let $\psi=(\tau+1)^s F(\varrho)$. Then $F$ satisfies the following ODE:
\begin{equation*}
	\varrho(\varrho+1)\frac{d^2F}{d\varrho^2}+((2+2i\q)\varrho+1+s)\frac{dF}{d\varrho}+(i\q-\ell(\ell+1))F=0.
\end{equation*}
By taking $z=-\varrho$ we obtain a hypergeometric equation:
\begin{equation*}
	z(1-z)\frac{d^2F}{dz^2}+(-(2+2iq)z+1+s)\frac{dF}{dz}-(i\q-\ell(\ell+1))F=0.
\end{equation*}
The corresponding solutions that are smooth at $z=0$ for non-integer $c$ are multiples of regularized Gauss hypergeometric functions ${}_2\mathbf{F}_1(a,b,c;-\varrho)$ with
\begin{align*}
		a+b+1=&\:2+2i\q,\\
		ab=&\:i\q-\ell(\ell+1),\\
		c=&\: s+1,
	\end{align*}
see for example \cite{NIST:DLMF}[\S 15.10]. We can easily solve for $a$ and $b$ to obtain the expressions in the statement of the theorem.

The simplifications in the cases $s=a-1$ and $s=b-1$ follow from \cite{NIST:DLMF}[\S 15.4].

Finally, it can easily be verified that $H(\tau,\varrho)$ solves \eqref{eq:CSFminkrhotau}. Alternatively, we can take the $s$-derivative of $(\tau+1)^{s}{}_2\mathbf{F}_1(a,a,s+1;-\varrho)$ evaluate at $s=a-1$ and add a constant multiple of $(1+\tau)^{-\frac{1}{2}+i\q}(1+\varrho)^{-\frac{1}{2}-i\q}$ to obtain $H(\tau,\varrho)$. Indeed, $\frac{d}{ds}((\tau+1)^{s}{}_2\mathbf{F}_1(a,a,s+1;-\varrho))$ must be a solution since the equation \eqref{eq:CSFminkrhotau} does not depend on $s$, the corresponding difference quotient is a solution, and ${}_2\mathbf{F}_1(a,a,a;-\varrho)=\frac{1}{\Gamma(a)}(1+\varrho)^{-a}$.
\end{proof}

\begin{proposition}
\label{prop:hypgeomrealbeta}
	Let $\beta_{\ell}\in (0,1)$ and $\alpha_+,\alpha_-\in \C$. Let $\mathfrak{w}_{\ell}^0(r)=\alpha_+ r^{\frac{1}{2}+\frac{\beta_{\ell}}{2}+i\q}+\alpha_-r^{\frac{1}{2}-\frac{\beta_{\ell}}{2}+i\q}$ and let $r\geq R$, so that $\mathfrak{w}_0(r)\neq 0$. Let $N_{\beta_{\ell}}\in \N_0$, such that $N_{\beta_{\ell}}\beta_{\ell}>1$.

	Define
	\begin{equation}
	\label{eq:Psi0infimbeta}
		(\Psi^{\infty}_0)_{\ell m}(\tau,r(\tau,\varrho)):=\beta_{\ell}\sum_{n=0}^{N_{\beta_{\ell}}}\zeta^n(1+\tau)^{-a+2i\q-n\beta_{\ell}}{}_2\mathbf{F}_1\left(a,b,c_n;-\varrho\right),
	\end{equation}
	with
	\begin{align*}
		a=&\:\frac{1}{2}+ \frac{\beta_{\ell}}{2}+i\q,\\
		b=&\:\frac{1}{2}- \frac{\beta_{\ell}}{2}+i\q,\\
		c_n=&\: \frac{1}{2}+i\q-(n+\frac{1}{2})\beta_{\ell},\\
		\zeta=&\: -\frac{2^{\beta_{\ell}}\Gamma(-\beta_{\ell}+1)\Gamma(a)}{\Gamma(\beta_{\ell}+1) \Gamma(b)}\frac{\alpha_-}{\alpha_+}.
	\end{align*}
	Then $\Psi^{\infty}_0=(\Psi^{\infty}_0)_{\ell m}Y_{\ell m}$ is smooth with respect to the coordinate chart $(\tau,\varrho_{\infty},\theta,\varphi)$ (where $\varrho_{\infty}=\frac{1}{r}$) and it is a solution to \eqref{eq:CSFmink} with $A=\widehat{A}$ and $\h=0$.
	
	Furthermore, $f(\tau,r):=(\mathfrak{w}^0_{\ell})^{-1}(r)(\Psi^{\infty}_0)_{\ell m}(\tau,r)$ satisfies: for all $k,l\leq N$
	\begin{align}
	\label{eq:divisionwimproveddec0}
		|(rX)^k ((\tau+1)T)^l f|\leq &\: C_N (\tau+1)^{-\frac{1}{2}-\frac{1}{2}\beta_{\ell}}(\tau+1+2r)^{-\frac{1}{2}-\frac{1}{2}\beta_{\ell}},\\
	\label{eq:divisionwimproveddec1}
		|(rX)^k ((\tau+1)T)^l Tf|\leq &\: C_N (\tau+1)^{-\frac{3}{2}-\frac{1}{2}\beta_{\ell}}(\tau+1+2r)^{-\frac{1}{2}-\frac{1}{2}\beta_{\ell}},\\
			\label{eq:divisionwimproveddec2}
		|(rX)^k ((\tau+1)T)^l Xf|\leq &\: C_N (\tau+1)^{-\frac{3}{2}-\frac{1}{2}\beta_{\ell}}(\tau+1+2r)^{-\frac{1}{2}-\frac{1}{2}\beta_{\ell}}.
	\end{align}
\end{proposition}
\begin{proof}
Note first of all that by Lemma 	\ref{lm:solnminkhypg} and linearity of \eqref{eq:CSFmink}, $\Psi^{\infty}_0$ is a solution to \eqref{eq:CSFmink} with $A=\widehat{A}$ and $\h=0$ . Furthermore, $T^{m_1}(r^2X)^{m_2}\Psi^{\infty}_0$ is well-defined for all $m_1,m_2\in \N_0$.

Let $\varrho_0>0$ be arbitrarily large and suppose that $\varrho\leq \varrho_0$. Then \eqref{eq:divisionwimproveddec0} --\eqref{eq:divisionwimproveddec2} follow immediately, using that the decay rates of $(\Psi^{\infty}_0)_{\ell m}(\tau,r(\tau,\varrho))$ are preserved under commutation with $(\tau+1)\mathfrak{T}$ and $\varrho \partial_{\varrho}$.

Now suppose that $\varrho>\varrho_0$. The hypergeometric functions ${}_2\mathbf{F}_1(a,b,2-b-n\beta_{\ell};-\varrho)$ satisfy the following asymptotics when $\varrho\to \infty$ (combine example \cite{NIST:DLMF}[\S 15.8] with \cite{NIST:DLMF}[\S 15.2]): for $n\geq 1$:
\begin{multline}
\label{eq:hypgeomasymp}
	\frac{-\sin (\pi \beta_{\ell})}{\pi}{}_2\mathbf{F}_1(a,b,1-a+2iq-n\beta_{\ell};-\varrho)=\frac{\varrho^{-a}}{\Gamma(\beta_{\ell}+1)\Gamma(b)\Gamma(-(n+1)\beta_{\ell})}(1+O_{\infty}(\varrho^{-1}))\\
	-\frac{\varrho^{-b}}{\Gamma(-\beta_{\ell}+1)\Gamma(a)\Gamma(-n\beta_{\ell})}(1+O_{\infty}(\varrho^{-1}))
\end{multline}
On the other hand, we obtain for $n=0$:
\begin{multline}
\label{eq:hypgeomasymp2}
	\frac{-\sin (\pi \beta_{\ell})}{\pi}{}_2\mathbf{F}_1(a,b,1-a+2i\q;-\varrho)=	\frac{-\sin (\pi \beta_{\ell})}{\pi}{}_2\mathbf{F}_1(a,b,b;-\varrho)=\frac{-\sin (\pi \beta_{\ell})}{\pi \Gamma(b)}(1+\varrho)^{-a}\\
=\frac{(1+\varrho)^{-a}}{\Gamma(b)\Gamma(\beta_{\ell}+1)\Gamma(-\beta_{\ell})},
	\end{multline}
	where we applied Euler's reflection formula to arrive at the last equality.
	
From the above, it follows that \eqref{eq:divisionwimproveddec0} and \eqref{eq:divisionwimproveddec1} hold, since decay rates are preserved under commutation with $(\tau+1)T$ and $\varrho \partial_{\varrho}$ and $T\mathfrak{w}_0^{-1}=0$.

In order to conclude \eqref{eq:divisionwimproveddec2}, we will need to choose the constant $\zeta$ appropriately, since $X\mathfrak{w}_0^{-1}\neq 0$.

Consider \eqref{eq:hypgeomasymp}. Note that we can rewrite the terms with a factor $\varrho^{-b}$ as follows:
\begin{multline*}
	- \frac{\varrho^{-b}(1+\tau)^{-a+2i\q}}{\Gamma(-\beta_{\ell}+1)\Gamma(a)}\sum_{n=1}^{N_{\beta_{\ell}}}\zeta^n \frac{(1+\tau)^{-n\beta_{\ell}}}{\Gamma(-n\beta_{\ell})}\\
=- \zeta \frac{(2r)^{b}(1+\tau)^{-a+2i\q}(1+\tau)^{-b}}{\Gamma(-\beta_{\ell}+1)\Gamma(a)}\sum_{m=0}^{N_{\beta_{\ell}}+1}\zeta^m \frac{(1+\tau)^{-(m+1)\beta_{\ell}}}{\Gamma(-(m+1)\beta_{\ell})}\\
= -\zeta\frac{\Gamma(\beta_{\ell}+1) \Gamma(b)}{\Gamma(-\beta_{\ell}+1)\Gamma(a)}(2r)^b \left[(1+\tau)^{-a}\sum_{m=0}^{N_{\beta_{\ell}}+1}\zeta^m \frac{(1+\tau)^{-a+2i\q-m\beta_{\ell}}}{\Gamma(\beta_{\ell}+1) \Gamma(b)\Gamma(-(m+1)\beta_{\ell})}\right].
\end{multline*}
Hence, using that $N_{\beta_{\ell}} \beta_{\ell}>1$, we obtain:
\begin{multline*}
	2^{-a}\frac{-\sin (\pi \beta_{\ell})}{\pi}\sum_{n=0}^{N_{\beta_{\ell}}}\zeta^n(1+\tau)^{-a+2i\q-n\beta_{\ell}}{}_2\mathbf{F}_1(a,b,1-a+2i\q-n\beta_{\ell};-\varrho)\\
	=\left[r^a-\zeta\frac{\Gamma(\beta_{\ell}+1) \Gamma(b)}{2^{\beta}\Gamma(-\beta_{\ell}+1)\Gamma(a)}r^b\right]\cdot \left[\sum_{m=0}^{N_{\beta_{\ell}}}\zeta^m \frac{(1+\tau)^{-1-(m+1)\beta_{\ell}}}{\Gamma(\beta_{\ell}+1) \Gamma(b)\Gamma(-(m+1)\beta_{\ell})}\right]\\
	+(\tau+1)^{-\frac{1}{2}-\frac{\beta_{\ell}}{2}}O_{\infty}(\varrho^{-\frac{3}{2}-\frac{\beta_{\ell}}{2}}).
	\end{multline*}
Now choose $\zeta$ so that
\begin{equation*}
\zeta\frac{\Gamma(\beta_{\ell}+1) \Gamma(b)}{2^{\beta}\Gamma(-\beta_{\ell}+1)\Gamma(a)}=-\frac{\alpha_-}{\alpha+}.
\end{equation*}
Then
\begin{multline*}
	2^{-a}\alpha_+\frac{-\sin (\pi \beta_{\ell})}{\pi \beta_{\ell}}(\Psi^{\infty}_0)_{\ell m}(\tau,r)=\mathfrak{w}_0(r)\cdot \left[\sum_{m=0}^{N_{\beta_{\ell}}}\zeta^m \frac{(1+\tau)^{-1-(m+1)\beta_{\ell}}}{\Gamma(\beta_{\ell}+1) \Gamma(b)\Gamma(-(m+1)\beta_{\ell})}\right]\\
	+(\tau+1)^{-\frac{1}{2}-\frac{\beta_{\ell}}{2}}O_{\infty}\left(\left(\frac{\tau}{2r}\right)^{-\frac{3}{2}-\frac{\beta_{\ell}}{2}}\right).
	\end{multline*}
	Finally, we divide the above expression by $\mathfrak{w}_0$ and apply $X$ to conclude that \eqref{eq:divisionwimproveddec2} holds.
\end{proof}

When $\beta_{\ell}\in i(0,\infty)$, we instead consider an infinite sum over $n$. We first prove a lemma that characterizes the necessary properties of the relevant hypergeometric functions.
\begin{lemma}
\label{lm:unifomestnhypgeom}
	Let $\beta \in i[0,\infty), q\in \R$ and $x\in (0,\infty)$. Define $F_x(\varrho;\beta):={}_2\mathbf{F}_1(\frac{1}{2}+ \frac{\beta}{2}+i\q,\frac{1}{2}- \frac{\beta}{2}+i\q,\frac{1}{2}-\frac{\beta}{2}+i(\q+x);-\varrho)$. 
	\begin{enumerate}[label=\emph{(\roman*)}]
	\item Let $\varrho_0>0$. Then for all $N\in \N_0$, there exists a constant $C_{N,\varrho_0,\q}>0$, such that for all $0\leq m\leq N$ :
\begin{equation}
\label{eq:boundedrhoFnestho}
\left|((\varrho\partial_{\varrho})^m \left(\Gamma\left(\frac{1}{2}-\frac{\beta}{2}+iq-ix\right)F_x(\varrho;x)-(1+\varrho)^{-a}\right)\right|\leq C_{\q,\varrho_0, N} \varrho(1+\varrho)^{-\frac{3}{2}},
\end{equation}
with
\begin{equation}
\label{eq:boundedrhoFnesthoplus}
	\left|\frac{1}{e^{\frac{\pi}{2}x}|\Gamma(\frac{1}{2}-\frac{\beta}{2}+iq-ix)|}+1\right|\leq C_{\q,\beta,N} (x+1)^{-2}e^{-\frac{\pi}{2}x}.
\end{equation}
\item Let $\varrho\geq \varrho_0$ and let $\epsilon>0$ be arbitrarily small. Then for $\varrho_0$ suitably small and for all $N\in \N_0$, there exists a constant $C_{N,\q,\varrho_0,\epsilon}>0$, such that for all $0\leq m,l\leq N$ :
\begin{align}
\label{eq:largerhoFnestho1}
		\Bigg|(x\partial_{x})^l(\varrho\partial_{\varrho})^m&\left(\frac{\sin (\pi \im \beta)}{\pi}F(\varrho;x)-\frac{\varrho^{-a}}{\Gamma(b)\Gamma(c_x-a)}+\frac{\varrho^{-b}}{\Gamma(a)\Gamma(c_x-b)}\right)\Bigg|\\  \nonumber
		\leq &\: C_{N,\q, \varrho_0,\epsilon} \varrho^{-1} e^{(\frac{\pi}{2}+\epsilon) x}\quad(\beta\in i(0,\infty)),\\
		\label{eq:largerhoFnestho2}
		\Bigg|(x\partial_{x})^l(\varrho\partial_{\varrho})^m&\left(F(\varrho;x)-\frac{\varrho^{-a}}{\Gamma(a)\Gamma(c_x-a)}\left(\log \varrho+2\gamma_{\rm Euler}-\frac{\Gamma'(a)}{\Gamma(a)}-\frac{\Gamma'(c_x-a)}{\Gamma(c_x-a)}\right)\right)\Bigg|\\ \nonumber
		\leq &\: C_{N,\q, \varrho_0,\epsilon} \varrho^{-1}\log \varrho e^{(\frac{\pi}{2}+\epsilon) x}\quad  (\beta=0).
\end{align}

\end{enumerate}
\end{lemma}
\begin{proof}
Let $c_x:=	\frac{1}{2}-\frac{\beta}{2}+i\q-i x$, $a=\frac{1}{2}+ \frac{\beta}{2}+i\q$ and $b=\frac{1}{2}- \frac{\beta}{2}+i\q$. Note that $\re a=\re b=\re c_x=\frac{1}{2}$. Note that $c_x-b=-ix$.

We can express:
\begin{multline*}
	F(\varrho;x)=(1+\varrho)^{-a}{}_2\mathbf{F}_1\left(a,c_x-b,c_x;\frac{\varrho}{1+\varrho}\right)\\
	=\frac{(1+\varrho)^{-a}}{\Gamma(c_x)}\left(1+\frac{\Gamma(c_x)}{\Gamma(a)\Gamma(c_x-b)}\sum_{k=1}^{\infty}\frac{\Gamma(a+k)\Gamma(c_x+k-b)}{\Gamma(c_x+k)\Gamma(k+1)}\left(\frac{\varrho}{1+\varrho}\right)^k\right),
\end{multline*}
see for example \cite{NIST:DLMF}[\S 15.8.1].

By the Stirling approximation formula, we can alternatively estimate for $z\in \C$, with $z=re^{i\theta}$, with $r> 0$ and $\theta=\arg(z)\in (-\pi+\epsilon_0,\pi-\epsilon_0)$ for some $\epsilon_0>0$:
\begin{multline*}
	\log \Gamma(z)=z\log z-z+\frac{1}{2}\log(2\pi z^{-1})+O_{\infty}(z^{-1})\\
	=\left[(\re z) (\log r-1)-\theta \im z+\frac{1}{2}\log (2\pi r^{-1})\right] +i(\im z \log r+\theta \re z-\frac{\theta}{2})+O_{\infty}(z^{-1})
\end{multline*}
Hence, for $z=c_x+k-b$, $\re z= k$, $\im z=-x$, $r=(k^2+x^2)^{\frac{1}{2}}$ and $\theta=-\arctan (x k^{-1}):=\theta_k$, so
\begin{multline*}
	\re \log \Gamma(c_x+k-b)=\frac{k}{2} \left(\log (k^2+x^2)-2\right)+\theta_k x-\frac{1}{4}\log (4\pi^2 (k^2+x^2))+O_{\infty}((x^2+k^2)^{-\frac{1}{2}}).
\end{multline*}
On the other hand, for $z=c_k+k$, we have $\re z=k+\frac{1}{2}$, $\im z=-x+q-\frac{1}{2}\im \beta$, $r=((k+\frac{1}{2})^2+(x+\frac{1}{2}\im \beta-q)^2)^{\frac{1}{2}}$ and $\theta=-\arctan ((x+\frac{1}{2}\im \beta-q) (k+\frac{1}{2})^{-1}):=\tilde{\theta}_k$, so
\begin{multline*}
	\re \log \Gamma(c_x+k)=\frac{k+\frac{1}{2}}{2} \left(\log \left((k+\frac{1}{2})^2+(x+\frac{1}{2}\im \beta-q)^2\right)-2\right)+\tilde{\theta}_k (x+\frac{1}{2}\im \beta-q)\\
	-\frac{1}{4}\log (4\pi^2 ((k+\frac{1}{2})^2+(x+\frac{1}{2}\im \beta-q)^2)+O_{\infty}((x^2+k^2)^{-\frac{1}{2}}).
\end{multline*}
We therefore obtain:
\begin{multline*}
\re \log \Gamma(c_x+k-b)-\re \log \Gamma(c_x-b)-\re \log \Gamma(c_x+k)=\left(\theta_k-\tilde{\theta}_k+\frac{\pi}{2}\right) x+\frac{k}{2} \log (1+O((k^2+x^2)^{-\frac{1}{2}}))\\
	-\frac{1}{4}\log \left((x^{-2}k^2+1)(1+O((k^2+x^2)^{-\frac{1}{2}}))\right)+\frac{1}{4}\log \left(1+O((k^2+x^2)^{-\frac{1}{2}}\right)+O_{\infty}(1)\\
	=\left(\theta_k-\tilde{\theta}_k+\frac{\pi}{2}\right) x-\frac{1}{4}\log(x^{-2}k^2+1)+O_{\infty}(1)
\end{multline*}

Furthermore,
\begin{multline*}
	-\log \Gamma(c_x)=-\frac{1}{4} \left(\log \left(\frac{1}{4}+(x+\im \beta-q)^2\right)-2\right)+ (x+\im \beta-q)\arctan (4(x+\im \beta-q))\\
	+\frac{1}{4}\log \left(4\pi^2 \left(\frac{1}{4}+(x+\im \beta-q)^2\right)\right)-i\left(-(x+\im \beta)\log x-\frac{1}{2}\arctan (4(x+\im \beta-q))\right)\\
	+O_{\infty}(x^{-1}),
\end{multline*}
so
\begin{equation*}
	\frac{1}{\Gamma(c_x)}=-e^{\frac{1}{2}(\im \beta-q-1)}e^{\frac{\pi}{2}x(1+O_{\infty}(x^{-2}))}e^{+i(x+\im \beta)\log x}.
\end{equation*}

Finally, we can estimate as above the following $x$-independent factor:
\begin{equation*}
	\left|\frac{\Gamma(a+k)}{\Gamma(k+1)}\right|\lesssim e^{\frac{k}{2}\log \left(\frac{(k+1)^2+(\im a)^2}{(k+1)^2}\right)}(k+1)^{-\frac{1}{2}}\lesssim (k+1)^{-\frac{1}{2}}.
\end{equation*}

Putting the above together, we obtain the following estimate, which is uniform in $k,n$:
\begin{equation*}
\left|\frac{\Gamma(c_x)}{\Gamma(a)\Gamma(c_x-b)}\frac{\Gamma(a+k)\Gamma(c_x+k-b)}{\Gamma(c_x+k)\Gamma(k+1)}\right|\leq C_{\q} (k+1)^{-\frac{1}{2}}.
\end{equation*}
Hence, for $\varrho\leq \varrho_0$:
\begin{equation*}
	\left|\Gamma(c_x)F(\varrho;x)-(1+\varrho)^{-a}\right|\leq C_{\q,\varrho_0}\varrho(1+\varrho)^{-\frac{1}{2}-1}.
\end{equation*}
It is straightforward to commute the above estimates with $(\varrho\partial_{\varrho})^m$, which introduces harmless polynomial factors in $k$, in order to conclude \eqref{eq:boundedrhoFnestho}.

Now consider the region $\varrho\geq \varrho_0$. First, consider the case $\beta\in i(0,\infty)$.

Then we instead apply \cite{NIST:DLMF}[\S 15.8.2] to express:
\begin{equation*}
	-\frac{\sin (\pi \im \beta)}{\pi}F(\varrho;x)=\varrho^{-a}\frac{{}_2\mathbf{F}_1(a,a-c_x+1,a-b+1;-\varrho^{-1})}{\Gamma(b)\Gamma(c_x-a)}-\varrho^{-b}\frac{{}_2\mathbf{F}_1(b,b-c_x+1,b-a+1;-\varrho^{-1})}{\Gamma(a)\Gamma(c_x-b)}.
\end{equation*}
We can write:
\begin{multline*}
	\frac{{}_2\mathbf{F}_1(a,a-c_x+1,a-b+1;-\varrho^{-1})}{\Gamma(b)\Gamma(c_x-a)}=\frac{1}{\Gamma(b)\Gamma(c_x-a)\Gamma(a)\Gamma(a-c_x+1)}\sum_{k=0}^{\infty}\frac{\Gamma(a+k)\Gamma(a-c_x+1+k)}{\Gamma(a-b+1+k)\Gamma(k+1)}(-\varrho)^{-k}\\
	=\frac{-i \sinh(\pi (x+\im \beta_{\ell}))}{\pi}\frac{1}{\Gamma(b)\Gamma(a)}\sum_{k=0}^{\infty}\frac{\Gamma(a+k)\Gamma(a-c_x+1+k)}{\Gamma(a-b+1+k)\Gamma(k+1)}(-\varrho)^{-k},
\end{multline*}
where we applied Euler's reflection formula to arrive at the last equality.

We apply Stirling's formula, as above, to obtain:
\begin{multline*}
	\left|\frac{\Gamma(a-c_x+1+k)}{\Gamma(a-b+1+k)}\right|\lesssim e^{\frac{1}{2}(k+\frac{1}{2})\log [\frac{(x+\im\beta-\q)^2+(k+1)^2}{(\im\beta)^2+(k+1)^2}] }e^{-(x+\im\beta-\q)\arctan(\frac{x+\im\beta-\q}{k+1})}\\
	= e^{\frac{x}{2}\frac{k+\frac{1}{2}}{x}\log [\frac{(x+\im\beta-\q)^2+(k+1)^2}{(\im\beta)^2+(k+1)^2}] }e^{-(x+\im\beta-\q)\arctan(\frac{x+\im\beta-\q}{k+1})}.
\end{multline*}
Suppose that $k+1\leq \delta x$ for $\delta>0$. Then using that for $x\leq \delta$, $x\log (1+x^{-1})\leq \delta \log(1+\delta^{-1})$, we obtain for $\delta$ suitably small:
\begin{equation*}
	\sinh(\pi (x+\im \beta_{\ell}))\left|\frac{\Gamma(a-c_x+1+k)}{\Gamma(a-b+1+k)}\right|\lesssim e^{C \delta x}e^{\frac{\pi}{2} x}e^{-\frac{\pi}{2}(1-C\delta)x}\leq e^{C\delta x}.
\end{equation*}
Now suppose that $k+1> \delta x$. Then
\begin{equation*}
	\sinh(\pi (x+\im \beta_{\ell}))\left|\frac{\Gamma(a-c_x+1+k)}{\Gamma(a-b+1+k)}\right|\lesssim e^{\frac{\pi}{2}x}e^{C k}\leq e^{C\delta^{-1}k}.
	\end{equation*}
	By repeating the above estimates with the roles of $a$ and $b$ reversed and taking $\varrho_0$ suitably large, depending on $\delta$, we obtain:
		\begin{equation*}
		\Bigg|\frac{\sin (\pi \im \beta)}{\pi}F(\varrho;x)-\frac{\varrho^{-a}}{\Gamma(b)\Gamma(c_x-a)}+\frac{\varrho^{-b}}{\Gamma(a)\Gamma(c_x-b)}\Bigg|\leq C_{\epsilon,\q} \varrho^{-1} e^{(\frac{\pi}{2}+\epsilon) x}.
	\end{equation*}
	It is straightforward to apply $(\varrho\partial_{\varrho})^m$ and $(x\partial_{x})^l$ to the above estimates to conclude \eqref{eq:largerhoFnestho1}.

Now consider the case $\beta=0$. We apply \cite{NIST:DLMF}[\S 15.8.8] together with Euler's reflection formula to express:
\begin{multline*}
	F(\varrho;x)=\frac{1}{\Gamma(a)^2}\varrho^{-1}\sum_{k=0}^{\infty}\frac{\Gamma(a+k)}{\Gamma(k+1)^2\Gamma(c_x-a-k)}\varrho^{-k}\left(\log \varrho+2\frac{\Gamma'(k+1)}{\Gamma(k+1)}-\frac{\Gamma'(a+k)}{\Gamma(a+k)}-\frac{\Gamma'(c_x-a-k)}{\Gamma(c_x-a-k)}\right)\\
	=-i\frac{\pi\sinh(\pi x)}{\Gamma(a)^2}\varrho^{-1}\sum_{k=0}^{\infty}\frac{\Gamma(a+k)\Gamma(a-c_x+k+1)}{\Gamma(k+1)^2}(-\varrho)^{-k}\left(\log \varrho+2\frac{\Gamma'(k+1)}{\Gamma(k+1)}-\frac{\Gamma'(a+k)}{\Gamma(a+k)}-\frac{\Gamma'(c_x-a-k)}{\Gamma(c_x-a-k)}\right).
\end{multline*}
The terms in the $k$-sum can be estimated as in the $\im \beta>0$ case above to conclude:
\begin{equation*}
		\Bigg|F(\varrho;x)-\frac{\varrho^{-a}}{\Gamma(a)\Gamma(c_x-a)}\left(\log \varrho+2\gamma_{\rm Euler}-\frac{\Gamma'(a)}{\Gamma(a)}-\frac{\Gamma'(c_x-a)}{\Gamma(c_x-a)}\right)\Bigg|\leq C_{\epsilon,\q} \varrho^{-1} \log \varrho e^{(\frac{\pi}{2}+\epsilon) x}.
	\end{equation*}
	It is straightforward to apply $(\varrho\partial_{\varrho})^m$ and $(x\partial_{x})^l$ to the above estimates to conclude \eqref{eq:largerhoFnestho2}.
\end{proof}

\begin{proposition}
\label{prop:imaginarybetatail}
	Let $\beta_{\ell}\in i(0,\infty)$ and $\mathfrak{w}_{\ell}^0(r)=\alpha_+ (\beta_{\ell})r^{\frac{1}{2}+\frac{\beta_{\ell}}{2}+i\q}+\alpha_-(\beta_{\ell})r^{\frac{1}{2}-\frac{\beta_{\ell}}{2}+i\q}$, with $\alpha_{\pm}(is)\in \C$ for $s\in \R$ satisfying $\alpha_-(is)=\alpha_+(-is)$.
	
	Define
	\begin{equation*}
		\zeta(\beta_{\ell}):= -\frac{2^{\beta_{\ell}}\Gamma(-\beta_{\ell}+1)\Gamma(\frac{1}{2}+ \frac{\beta_{\ell}}{2}+i\q)}{\Gamma(\beta_{\ell}+1) \Gamma(\frac{1}{2}- \frac{\beta_{\ell}}{2}+i\q)}\frac{\alpha_-(\beta_{\ell})}{\alpha_+(\beta_{\ell})},
	\end{equation*}
	so that $\zeta(-\beta_{\ell})=\frac{1}{\zeta(\beta_{\ell})}$, and assume that there exists a $\eta_0=\eta(|q|)>0$, such that
	\begin{equation}
	\label{eq:condzeta}
		\sign(q) \log \left|\frac{\alpha_-(\beta_{\ell})}{\alpha_+(\beta_{\ell})}\right|\leq-\frac{\pi}{2}\im \beta_{\ell}-\eta_0 \im \beta_{\ell}+\frac{1}{2}\log \left(\left|\frac{\cosh(\pi(|q|-\frac{1}{2}\im \beta_{\ell})}{\cosh(\pi(|q|+\frac{1}{2}\im \beta_{\ell})}\right|^2\right).
	\end{equation}
Then the following sum is a well-defined smooth function with respect to $(\tau,\varrho_{\infty},\theta,\varphi)$:	
\begin{multline}
\label{eq:psiinftyimbetapos}
		(\Psi^{\infty}_0)_{\ell m}(\tau,r(\tau,\varrho)):=\sign(\q)\beta_{\ell}\sum_{n=0}^{\infty}(\zeta(\sign(\q)\beta_{\ell}))^n(1+\tau)^{-\frac{1}{2}-(n+\frac{1}{2})\sign(\q)\beta_{\ell}+i\q}\\
		\times {}_2\mathbf{F}_1\left(\frac{1}{2}+\sign(\q)\frac{\beta_{\ell}}{2}+iq,\frac{1}{2}-\sign(\q)\frac{\beta_{\ell}}{2}+iq,\frac{1}{2}+i\q-(n+\frac{1}{2})\sign(q)\beta_{\ell};-\varrho\right).
	\end{multline}
	The function $\Psi^{\infty}_0=(\Psi^{\infty}_0)_{\ell m}Y_{\ell m}$ is a solution to \eqref{eq:CSFmink} with $A=\widehat{A}$ and $\h=0$.
	
	Furthermore, $f(\tau,r):=(\mathfrak{w}^0_{\ell})^{-1}(r)(\Psi^{\infty}_0)_{\ell m}(\tau,r)$ satisfies: for all $k,l\leq N$
	\begin{align}
	\label{eq:divisionwimproveddecim0}
		|(rX)^k ((\tau+1)T)^l f|\leq &\: C_N (\tau+1)^{-\frac{1}{2}}(\tau+1+2r)^{-\frac{1}{2}},\\
	\label{eq:divisionwimproveddecim1}
		|(rX)^k ((\tau+1)T)^l Tf|\leq &\: C_N (\tau+1)^{-\frac{3}{2}}(\tau+1+2r)^{-\frac{1}{2}},\\
			\label{eq:divisionwimproveddecim2}
		|(rX)^k ((\tau+1)T)^l Xf|\leq &\: C_N (\tau+1)^{-\frac{3}{2}}(\tau+1+2r)^{-\frac{1}{2}}.
	\end{align}
\end{proposition}
\begin{proof}
	Without loss of generality, we may assume that $q>0$. If $q<0$, we simply repeat the argument with $\beta_{\ell}$ replaced by $-\beta_{\ell}$. We can express:
	\begin{equation*}
		(\Psi^{\infty}_0)_{\ell m}(\tau,r(\tau,\varrho))=\beta_{\ell}\sum_{n=0}^{\infty}\zeta^n(1+\tau)^{-a+2i\q-n\beta_{\ell}}F_{n \im \beta_{\ell}}(\varrho;\beta_{\ell}),
	\end{equation*}
	with $F_{x}(\varrho;\beta)$ defined in Lemma \ref{lm:unifomestnhypgeom}.
	
	By assumption on $\zeta$, \eqref{eq:boundedrhoFnestho} and \eqref{eq:boundedrhoFnesthoplus}, we can estimate:
	\begin{equation*}
		\log \zeta+\frac{1}{n}\log |F_{n \im \beta_{\ell}}(\varrho;\beta_{\ell})|\leq -\eta_0 \im \beta_{\ell}+\frac{\log C_{\varrho,\q}}{n}.
	\end{equation*}
	Hence, $|\zeta|F_{n \im \beta_{\ell}}(\varrho;\beta_{\ell})|^{\frac{1}{n}}|<e^{-\eta_0 \im \beta_{\ell}}<1$, so for $\varrho\leq \varrho_0$:
	\begin{equation*}
		|(\Psi^{\infty}_0)_{\ell m}(\tau,r(\tau,\varrho))|\leq  C_{\varrho,\q}(1+\tau)^{-\frac{1}{2}}|\beta_{\ell}|\sum_{n=0}^{\infty}e^{-\eta_0 \im \beta_{\ell}}\leq  C_{\varrho,\q}\eta_0.
	\end{equation*}
	This means that the series converges. Commutation with $\mathfrak{T}$ and $\varrho\mathfrak{X}$ is straightforward and we conclude that $(\Psi^{\infty}_0)_{\ell m}$ is smooth.
	
	Since $\Psi^{\infty}_{0}$ is a well-defined limit of a sum of solutions to \eqref{eq:maineqradfield} for $M=0$ by Lemma \ref{lm:solnminkhypg}, we conclude by linearity that $\Psi^{\infty}_{0}$ must itself by a solution.

The remaining estimates in the statement of the proposition follow by repeating the steps in the proof of Proposition \ref{prop:hypgeomrealbeta} with $N_{\beta_{\ell}}$ replaced by $\infty$, using \eqref{eq:largerhoFnestho1}.
\end{proof}

Finally, we consider the case $\beta_{\ell}=0$. We could interpret the limit of \eqref{eq:Psi0infimbeta} as $\beta_{\ell}\to 0$ as the limit of a Riemann sum, evaluated along $x=n\im \beta$, to obtain an integral over $(0,\infty)$. In the proposition below, we will instead directly define $(\Psi^{\infty}_{0})_{\ell m}$ as an integral.
\begin{proposition}
	\label{prop:hypgeombeta0}
	Let $\beta_{\ell}=0$ and $\alpha_+,\alpha_-\in \C$. Let $\mathfrak{w}_{\ell}^0(r)=r^{\frac{1}{2}+i\q}(\alpha_+ \log r+\alpha_-)$.
	
Define
	 \begin{equation}
	\label{eq:defPsiinftybeta0}
	 	(\Psi^{\infty}_{0})_{\ell m}(\tau,r(\tau,\varrho)):=(\tau+1)^{-\frac{1}{2}+i\q}\int_0^{\infty}(\tau+1)^{-ix}e^{-\eta x}{}_2\mathbf{F}_1\left(\frac{1}{2}+iq,\frac{1}{2}+iq,\frac{1}{2}+iq-ix;-\varrho\right)\,dx.
	 \end{equation}

	with
	\begin{equation*}
			\eta:=i \sign(q)\left[\frac{\alpha_-}{\alpha_+}-\log 2+2\gamma_{\rm Euler}-\frac{\Gamma'(\frac{1}{2}+i\q)}{\Gamma(\frac{1}{2}+i\q)}\right],
	\end{equation*}
	and with $\gamma_{\rm Euler}$ the Euler--Mascheroni constant. Assume that there exists a $\eta_0=\eta(|q|)>0$, such that
	\begin{equation}
	\label{eq:condeta}
		\sign(q) \im \left(\frac{\alpha_-}{\alpha_+}\right)\leq -\frac{\pi}{2}\left(1-\tanh(\pi |q|)\right)+\eta_0 .
	\end{equation}
		
	Then $\Psi^{\infty}_0=(\Psi^{\infty}_0)_{\ell m}Y_{\ell m}$ is well-defined and smooth with respect to the coordinate chart $(\tau,\varrho_{\infty},\theta,\varphi)$ and it is a solution to \eqref{eq:CSFmink} with $A=\widehat{A}$ and $\h=0$.
	
	Furthermore, $f(\tau,r):=(\mathfrak{w}^0_{\ell})^{-1}(r)(\Psi^{\infty}_0)_{\ell m}(\tau,r)$ satisfies: for all $k,l\leq N$
	\begin{align}
	\label{eq:divisionwimproveddec0beta0}
		|(rX)^k ((\tau+1)T)^l f|\leq &\: C_N (\log r)^{-1}(\tau+1)^{-\frac{1}{2}}(\tau+1+2r)^{-\frac{1}{2}},\\
	\label{eq:divisionwimproveddec1beta0}
		|(rX)^k ((\tau+1)T)^l Tf|\leq &\: C_N (\log r)^{-1}(\tau+1)^{-\frac{3}{2}}(\tau+1+2r)^{-\frac{1}{2}},\\
			\label{eq:divisionwimproveddec2beta0}
		|(rX)^k ((\tau+1)T)^l Xf|\leq &\: C_N (\log r)^{-1}(\tau+1)^{-\frac{3}{2}}(\tau+1+2r)^{-\frac{1}{2}}.
	\end{align}
\end{proposition}
\begin{proof}
We can express:
\begin{equation*}
	(\Psi^{\infty}_{0})_{\ell m}(\tau,\varrho)=(\tau+1)^{-\frac{1}{2}+i\q}\int_0^{\infty}(\tau+1)^{-ix}e^{-\sign(\q)\eta x}F_x(\varrho;0)\,dx,
\end{equation*}
with $F_{x}(\varrho;\beta)$ defined in Lemma \ref{lm:unifomestnhypgeom}.

Without loss of generality assume that $q>0$. In order to conclude that the above integral is well-defined, we will verify that the integrand decays exponentially in $x$. Note that
\begin{equation*}
	\re \eta=-\im \left(\frac{\alpha_-}{\alpha_+}\right)+\im \frac{\Gamma'(\frac{1}{2}+iq)}{\Gamma(\frac{1}{2}+iq)}.
\end{equation*}
Furthermore,
\begin{equation*}
	\im \frac{\Gamma'(\frac{1}{2}+iq)}{\Gamma(\frac{1}{2}+iq)}=\re \frac{d}{dq}\log \Gamma\left(\frac{1}{2}+iq\right)=\frac{1}{2}\frac{d}{dq}\log\left|\Gamma\left(\frac{1}{2}+iq\right)\right|^2=-\frac{1}{2\pi}\frac{d}{dq}\log(\cosh(\pi q))=-\frac{1}{2\pi}\tanh(\pi q).
\end{equation*}
Hence, by assumption on $\frac{\alpha_-}{\alpha_+}$
\begin{equation*}
	\re \eta\geq \frac{\pi}{2}+\eta_0
\end{equation*}
Applying \eqref{eq:boundedrhoFnestho} and \eqref{eq:boundedrhoFnesthoplus}, then implies that for $\varrho\leq \varrho_0$:
\begin{equation*}
	\left|e^{-\sign(\q)\eta x}F_x(\varrho;0)\right|\leq e^{-\eta_0 x}.
\end{equation*}
The absolute value of the integrand is therefore exponentially decaying, so the integral is well-defined. Furthermore, by Lemma \ref{lm:solnminkhypg}, the linearity of \eqref{eq:CSFminkrhotau} and the dominated convergence theorem, $\Psi^{\infty}_0$ must be a solution to \eqref{eq:CSFminkrhotau}. Since we can freely commute with $\varrho\mathfrak{X}$ and $\mathfrak{T}$, we moreover have that $\Psi^{\infty}_0$ is a smooth solution, so $(\Psi^{\infty}_0)_{\ell m}(\tau,r)Y_{\ell m}(\theta,\varphi)$ is a smooth solution to \eqref{eq:CSFmink}. 

It remains to consider the function $f=\mathfrak{w}_0^{-1}\Psi^{\infty}_0$ for $\varrho\geq \varrho_0$. Note first of all that by \eqref{eq:largerhoFnestho2}, we can write:
\begin{equation}
\label{eq:keysplitFrhoxbeta0}
	F(\varrho;x)=\frac{\varrho^{-a}}{\Gamma(\frac{1}{2}+iq)\Gamma(-ix)}\left(\log \varrho+2\gamma_{\rm Euler}-\frac{\Gamma'(\frac{1}{2}+iq)}{\Gamma(\frac{1}{2}+iq)}-\frac{\Gamma'(-ix)}{\Gamma(-ix)}\right)+\log \varrho e^{(\frac{\pi}{2}+\epsilon)x}R(\varrho,x),
\end{equation}
where $(x\partial_x)^{\ell} R(\varrho,x)=O_{\infty}(\varrho^{-1})$. We split
\begin{equation*}
	\log\varrho=\log(1+\tau)-\log r-\log 2
\end{equation*}

can integrate by parts in $x$, taking $\epsilon<\eta_0$ to estimate:
\begin{multline*}
	\int_0^{\infty}(1+\tau)^{-ix}e^{-(\eta-\frac{\pi}{2}-\epsilon)x}\log (1+\tau) R(\varrho,x)\,dx=-\frac{1}{i}\int_0^{\infty} \frac{d}{dx}(e^{-ix \log (1+\tau)})e^{-(\eta-\frac{\pi}{2}-\epsilon)x} R(\varrho,x)\,dx\\
	=-iR(\varrho,0)-i\int_0^{\infty}e^{-ix \log (1+\tau)}\frac{d}{dx}\left(e^{-(\eta-\frac{\pi}{2}-\epsilon)x}R(\varrho,x)\right)\,dx.
\end{multline*}
The right-hand side can be bounded by $\varrho^{-\frac{3}{2}}$ and after applying $X$, we conclude that:
\begin{equation*}
	\left|X\left((\mathfrak{w}_{\ell}^0)^{-1}(1+\tau)^{-\frac{1}{2}+iq}\int_0^{\infty}(1+\tau)^{-ix}e^{-(\eta-\frac{\pi}{2}-\epsilon)x}\log (1+\tau) R(\varrho,x)\,dx\right)\right|\leq C (1+\tau)^{-2}(\log r)^{-1},
\end{equation*}
as desired. It is straightforward to commute with $rX$ and $(\tau+1) T$.

We will now consider the remaining term in \eqref{eq:keysplitFrhoxbeta0} and show that we can factor out $\mathfrak{w}_{\ell}^0(r)$, by choosing $\eta$ appropriately. As above, we first split $\log\varrho=\log(1+\tau)-\log r-\log 2$ and integrate by parts the $\log(1+\tau)$ contribution. We obtain:
\begin{multline*}
	\frac{\varrho^{-\frac{1}{2}-iq}}{\Gamma(\frac{1}{2}+iq)}\int_0^{\infty}(1+\tau)^{-ix}e^{-\eta x}\log (1+\tau) \frac{1}{\Gamma(-ix)}\,dx=-\frac{\varrho^{-\frac{1}{2}-iq}}{i \Gamma(\frac{1}{2}+iq)}\int_0^{\infty} \frac{d}{dx}(e^{-ix \log (1+\tau)})e^{-\eta x} \frac{1}{\Gamma(-ix)}\,dx\\
	=\frac{\varrho^{-\frac{1}{2}-iq}}{i \Gamma(\frac{1}{2}+iq)}\int_0^{\infty} \left[-\eta+i\frac{\Gamma'(-ix)}{\Gamma(-ix)}\right]\frac{1}{\Gamma(-ix)}(1+\tau)^{-ix}e^{-\eta x}\,dx,
\end{multline*}
where we used that the boundary term at $x=0$ vanishes since $\frac{1}{\Gamma(-ix)}$ vanishes at $x=0$.

The full integral then becomes:
\begin{equation*}
	-\frac{\varrho^{-\frac{1}{2}-iq}(\tau+1)^{-\frac{1}{2}+i\q}}{\Gamma(\frac{1}{2}+iq)}\int_0^{\infty}(1+\tau)^{-ix} \frac{e^{-\eta x}}{\Gamma(-ix)}\left[\log r+\log 2 -2\gamma_{\rm Euler}+\frac{\Gamma'(\frac{1}{2}+iq)}{\Gamma(\frac{1}{2}+iq)}-i\eta\right]\,dx.
\end{equation*}
Now take
\begin{equation*}
	-i\eta=\frac{\alpha_-}{\alpha_+}-\log 2+2\gamma_{\rm Euler}-\frac{\Gamma'(\frac{1}{2}+iq)}{\Gamma(\frac{1}{2}+iq)}.
\end{equation*}
Then we can write the integral as follows.
\begin{equation*}
	-2^{\frac{1}{2}+iq}\frac{(\tau+1)^{-1}}{\Gamma(\frac{1}{2}+iq)}\mathfrak{w}_{\ell}^0(r)\int_0^{\infty}(1+\tau)^{-ix} \frac{e^{-\eta x}}{\Gamma(-ix)}\,dx.
\end{equation*}
After dividing by $\mathfrak{w}_{\ell}^0(r)$, the above expression is independent of $r$ and hence, the $X$-derivative vanishes. We conclude \eqref{eq:divisionwimproveddec0beta0}--\eqref{eq:divisionwimproveddec2beta0} for $\varrho\geq \varrho_0$.

We integrate by parts and use that ${}_2\mathbf{F}_1\left(\frac{1}{2}+iq,\frac{1}{2}+iq,\frac{1}{2}+iq;-\varrho\right)=\frac{(1+\varrho)^{-\frac{1}{2}-iq}}{\Gamma(\frac{1}{2}+iq)}$to rewrite:
\begin{multline*}
	(\Psi^{\infty}_{0})_{\ell m}(\tau,r(\tau,\varrho))=i((\log(1+\tau))^{-1}(\tau+1)^{-\frac{1}{2}+i\q}\\
\times \int_0^{\infty}\frac{d}{dx}((\tau+1)^{-ix})e^{-\eta x}{}_2\mathbf{F}_1\left(\frac{1}{2}+iq,\frac{1}{2}+iq,\frac{1}{2}+iq-ix;-\varrho\right)\,dx\\
=-i((\log(1+\tau))^{-1}(\tau+1)^{-\frac{1}{2}+i\q}\frac{(1+\varrho)^{-\frac{1}{2}-iq}}{\Gamma(\frac{1}{2}+iq)}\\
-i((\log(1+\tau))^{-1}(\tau+1)^{-\frac{1}{2}+i\q}\\
\times \int_0^{\infty}(\tau+1)^{-ix}\frac{d}{dx}\left(e^{-\eta x}{}_2\mathbf{F}_1\left(\frac{1}{2}+iq,\frac{1}{2}+iq,\frac{1}{2}+iq-ix;-\varrho\right)\right)\,dx
\end{multline*}
We therefore obtain:
\begin{equation*}
|(rX)^l((1+\tau)T)^m\Psi^{\infty}_0|\lesssim 	(\log(1+\tau))^{-1}((1+\tau)^{-\frac{1}{2}-\frac{1}{2}\re\beta_{\ell}}.
\end{equation*}
For $\varrho\leq \varrho_0$, $r^{-1}\lesssim (1+\tau)^{-1}$, so \eqref{eq:divisionwimproveddec0beta0}--\eqref{eq:divisionwimproveddec2beta0} then follow from the above estimates.
\end{proof}

\begin{proposition}
\label{eq:asymptPsiinfty0}
	The functions $(\Psi_0^{\infty})_{\ell m}(\tau,r)$, defined in Propositions \ref{prop:hypgeomrealbeta}, \ref{prop:imaginarybetatail} and \ref{prop:hypgeombeta0} satisfy the following late-time asymptotics:
	\begin{enumerate}[label=\emph{(\roman*)}]
		\item In the limit $r\to \infty$:
		\begin{align*}
			(\Psi_0^{\infty})_{\ell m}(\tau,\infty)=&\:\beta_{\ell}(1+\tau)^{-\frac{1}{2}-\frac{\beta_{\ell}}{2}+i\q}\sum_{n=0}^{N_{\beta_{\ell}}}\frac{\zeta^n(\beta_{\ell})(1+\tau)^{-n\beta_{\ell}}}{\Gamma\left(\frac{1}{2}+i\q-(n+\frac{1}{2})\beta_{\ell}\right)}\\
			=&\: \frac{\beta_{\ell}(1+\tau)^{-\frac{1}{2}-\frac{\beta_{\ell}}{2}+i\q}}{\Gamma\left(\frac{1}{2}+i\q-\frac{1}{2}\beta_{\ell}\right)}+O((1+\tau)^{-\frac{1}{2}-\frac{3\beta_{\ell}}{2}}),\quad (\beta_{\ell}\in (0,1)),\\
			(\Psi_0^{\infty})_{\ell m}(\tau,\infty)=&\:\sign(\q)\beta_{\ell}(1+\tau)^{-\frac{1}{2}-\sign(\q)\frac{\beta_{\ell}}{2}+i\q}\sum_{n=0}^{\infty}\frac{\zeta^n(\sign(\q)\beta_{\ell})(1+\tau)^{-\sign(\q)n\beta_{\ell}}}{\Gamma\left(\frac{1}{2}+i\q-(n+\frac{1}{2})\sign(\q)\beta_{\ell}\right)},\quad (\beta_{\ell}\in i(0,\infty)),\\
			(\Psi_0^{\infty})_{\ell m}(\tau,\infty)=&\:\frac{-i}{\Gamma\left(\frac{1}{2}+i\q\right)}(\tau+1)^{-\frac{1}{2}+i\q}(\log(1+\tau))^{-1}+(\log(1+\tau))^{-2}O_{\infty}((1+\tau)^{-\frac{1}{2}}),\quad (\beta_{\ell}=0).
		\end{align*}
		\item At $r=r_0>0$:
\begin{align*}
			(\Psi^{\infty}_0)_{\ell m}(\tau,r_0)=&-2^{\frac{1}{2}+\frac{1}{2}\beta_{\ell}}\frac{\pi \beta_{\ell}}{\alpha_+\sin (\pi \beta_{\ell})}(1+\tau)^{-1-\beta_{\ell}}\mathfrak{w}_{\ell}^0(r)\\
\times & \left[\sum_{n=0}^{N_{\beta_{\ell}}}\zeta^n(\beta_{\ell}) \frac{(1+\tau)^{-n\beta_{\ell}}}{\Gamma(\beta_{\ell}+1) \Gamma(\frac{1}{2}-\frac{1}{2}\beta_{\ell}+iq)\Gamma(-(n+1)\beta_{\ell})}\right]\\
	+&\:O_{\infty}\left((1+\tau)^{-2}\right)=-2^{\frac{1}{2}+\frac{1}{2}\beta_{\ell}}\frac{\pi \beta_{\ell}}{\alpha_+\sin (\pi \beta_{\ell}) \Gamma(\beta_{\ell}+1) \Gamma(\frac{1}{2}-\frac{1}{2}\beta_{\ell}+iq)\Gamma(-\beta_{\ell})}\mathfrak{w}_0(r)(1+\tau)^{-1-\beta_{\ell}}\\
	&+O_{\infty}\left((1+\tau)^{-1-2\beta_{\ell}}\right),\quad (\beta_{\ell}\in (0,1)),\\
			(\Psi_0^{\infty})_{\ell m}(\tau,\infty)=&\:-2^{\frac{1}{2}+\frac{1}{2}\sign(\q)\beta_{\ell}}\frac{\pi \beta_{\ell}}{\alpha_-\sin (\pi \beta_{\ell})}(1+\tau)^{-1}\mathfrak{w}_{\ell}^0(r)\\
			\times& \left[\sum_{n=0}^{\infty}\zeta^n(\sign(\q)\beta_{\ell}) \frac{(1+\tau)^{-\sign(\q)n\beta_{\ell}}}{\Gamma(\sign(\q)\beta_{\ell}+1) \Gamma(\frac{1}{2}-\frac{1}{2}\sign(\q)\beta_{\ell}+iq)\Gamma(-(n+1)\sign(\q)\beta_{\ell})}\right]\\
	+&\:O_{\infty}\left((1+\tau)^{-2}\right),\quad (\beta_{\ell}\in i(0,\infty)),\\
	(\Psi^{\infty}_0)_{\ell m}(\tau,r_0)=&-\frac{2^{\frac{1}{2}+iq}}{\Gamma(\frac{1}{2}+iq)}\mathfrak{w}_{\ell}^0(r)\left[\log(1+\tau))^{-2}(\tau+1)^{-1-\sign(\q)\beta_{\ell}}+(\log(1+\tau))^{-3}O_{\infty}((\tau+1)^{-1}\right]\\
	+&\:O_{\infty}((1+\tau)^{-2}),\quad (\beta_{\ell}=0).
		\end{align*}
	\end{enumerate}
\end{proposition}
\begin{proof}
	We first consider the case $\beta_{\ell}\in (0,1)$. Then
	\begin{equation*}
		(\Psi_0^{\infty})_{\ell m}(\tau,\infty)=\beta_{\ell}(1+\tau)^{-\frac{1}{2}-\frac{\beta_{\ell}}{2}+i\q}\sum_{n=0}^{N_{\beta_{\ell}}}\frac{\zeta^n(1+\tau)^{-n\beta_{\ell}}}{\Gamma\left(\frac{1}{2}+i\q-(n+\frac{1}{2})\beta_{\ell}\right)}=\frac{\beta_{\ell}(1+\tau)^{-\frac{1}{2}-\frac{\beta_{\ell}}{2}+i\q}}{\Gamma\left(\frac{1}{2}+i\q-\frac{1}{2}\beta_{\ell}\right)}+O((1+\tau)^{-\frac{1}{2}-\frac{3\beta_{\ell}}{2}}).
	\end{equation*}
	Let $r_0>0$ be arbitrary. Then:
	\begin{multline*}
	(\Psi^{\infty}_0)_{\ell m}(\tau,r_0)=-2^{\frac{1}{2}+\frac{1}{2}\beta_{\ell}}\frac{\pi \beta_{\ell}}{\alpha_+\sin (\pi \beta_{\ell})}(1+\tau)^{-1-\beta_{\ell}}\mathfrak{w}_0(r)\cdot \left[\sum_{n=0}^{N_{\beta_{\ell}}}\zeta^n \frac{(1+\tau)^{-n\beta_{\ell}}}{\Gamma(\beta_{\ell}+1) \Gamma(\frac{1}{2}-\frac{1}{2}\beta_{\ell}+iq)\Gamma(-(n+1)\beta_{\ell})}\right]\\
	+O_{\infty}\left((1+\tau)^{-2}\right)\\
	=-2^{\frac{1}{2}+\frac{1}{2}\beta_{\ell}}\frac{\pi \beta_{\ell}}{\alpha_+\sin (\pi \beta_{\ell}) \Gamma(\beta_{\ell}+1) \Gamma(\frac{1}{2}-\frac{1}{2}\beta_{\ell}+iq)\Gamma(-\beta_{\ell})}\mathfrak{w}_0(r)(1+\tau)^{-1-\beta_{\ell}}+O_{\infty}\left((1+\tau)^{-1-2\beta_{\ell}}\right).
	\end{multline*}
Now let $\beta_{\ell}\in i(0,\infty)$. Then the expressions from the $\beta_{\ell}\in (0,1)$ case remain valid, but with $\beta_{\ell}$ replaced by $\sign(\q)\beta_{\ell}$ and $N_{\beta_{\ell}}$ replaced with $\infty$.

Finally, we consider the case $\beta_{\ell}=0$. We can express:
\begin{equation*}
	(\Psi^{\infty}_0)_{\ell m}(\tau,\infty)=(\tau+1)^{-\frac{1}{2}+i\q}\int_0^{\infty}(\tau+1)^{-ix}e^{-\eta x}\frac{1}{\Gamma\left(\frac{1}{2}+iq-ix\right)}\,dx.
\end{equation*}
We write $(\tau+1)^{-ix}=i(\log(1+\tau))^{-1}\frac{d}{dx}\left(e^{-ix\log(1+\tau)}\right)$ and integrate by parts to obtain:
\begin{multline*}
	(\Psi^{\infty}_0)_{\ell m}(\tau,\infty)=-i(\tau+1)^{-\frac{1}{2}+i\q}(\log(1+\tau))^{-1}\frac{1}{\Gamma\left(\frac{1}{2}+iq\right)}\\
	-i(\tau+1)^{-\frac{1}{2}+i\q}(\log(1+\tau))^{-1}\int_0^{\infty}(\tau+1)^{-ix}\frac{d}{dx}\left(e^{-\eta x}\frac{1}{\Gamma\left(\frac{1}{2}+iq-ix\right)}\right)\,dx\\
	=-i(\tau+1)^{-\frac{1}{2}+i\q}(\log(1+\tau))^{-1}\frac{1}{\Gamma\left(\frac{1}{2}+iq\right)}+(\log(1+\tau))^{-2}O_{\infty}((1+\tau)^{-\frac{1}{2}}),
\end{multline*}
where we obtained the last equality by integrating by parts in $x$ again to pull out another factor $(\log(1+\tau))^{-1}$.

We also obtain along $r=r_0$:
\begin{equation*}
	(\Psi^{\infty}_0)_{\ell m}(\tau,r_0)=-2^{\frac{1}{2}+iq}\frac{(\tau+1)^{-1}}{\Gamma(\frac{1}{2}+iq)}\mathfrak{w}_{\ell}^0(r)\int_0^{\infty}(1+\tau)^{-ix} \frac{e^{-\eta x}}{\Gamma(-ix)}\,dx+O_{\infty}((1+\tau)^{-2})
\end{equation*}
We integrate by parts and use that $\frac{1}{\Gamma(-ix)}$ vanishes at $x=0$ to obtain:
\begin{multline*}
	\int_0^{\infty}(1+\tau)^{-ix} \frac{e^{-\eta x}}{\Gamma(-ix)}\,dx=-i\frac{1}{\log(1+\tau)}\int_0^{\infty}(1+\tau)^{-ix}\frac{d}{dx}\left(\frac{e^{-\eta x}}{\Gamma(-ix)}\right)\,dx\\
	=\frac{1}{(\log(1+\tau))^2}\int_0^{\infty}\frac{d}{dx}\left((1+\tau)^{-ix}\right)\frac{d}{dx}\left(\frac{e^{-\eta x}}{\Gamma(-ix)}\right)\,dx\\
	=-\frac{1}{(\log(1+\tau))^2}\frac{d}{dx}\Big|_{x=0}\left(\frac{1}{\Gamma(-ix)}\right)+O_{\infty}((\log(1+\tau))^{-3})\\
=\frac{1}{(\log(1+\tau))^2}+O_{\infty}((\log(1+\tau))^{-3}).
\end{multline*}
We therefore conclude that:
\begin{multline*}
	(\Psi^{\infty}_0)_{\ell m}(\tau,r_0)=-\frac{2^{\frac{1}{2}+iq}}{\Gamma(\frac{1}{2}+iq)}\mathfrak{w}_{\ell}^0(r)\left[\log(1+\tau))^{-2}(\tau+1)^{-1}+(\log(1+\tau))^{-3}O_{\infty}((\tau+1)^{-1}\right]\\
	+O_{\infty}((1+\tau)^{-2}). \qedhere
	\end{multline*}
\end{proof}

\begin{proposition}
\label{prop:Psi0inftygrowthdecay}
Define the following constants: let $k\in \N_0$, then
\begin{align}
\mathfrak{c}_k:=&\:\begin{cases}
(-1)^k2^{-(k+1)}\frac{\beta_{\ell}\Gamma(\frac{1}{2}+\frac{\beta_{\ell}}{2}+i\q+k+1)}{\Gamma(\frac{1}{2}+\frac{\beta_{\ell}}{2}+i\q)\Gamma(\frac{1}{2}-\frac{\beta_{\ell}}{2}+i\q)},\quad (\beta_{\ell}\in (0,1)),\\
2^{-(k+1)}|\beta_{\ell}|\\
\times \sqrt{\sum_{n=0}^{\infty}|\zeta(\beta_{\ell})|^{\sign(\q)2n}\left|\partial_z^{k+1}\: ({}_2\mathbf{F}_1)\right|^2\left(\frac{1+\sign(\q)\beta_{\ell}}{2}+i\q,\frac{1-\sign(\q)\beta_{\ell}}{2}+i\q, \frac{1+(2n-1)\sign(\q)\beta_{\ell}}{2}+i\q;0\right)},\\
\quad (\beta_{\ell}\in i(0,\infty)),\\
(-1)^k2^{-(k+1)}\frac{i\Gamma\left(\frac{1}{2}+iq+k+1\right)}{\Gamma(\frac{1}{2}+i\q)^2},\quad (\beta_{\ell}=0),
\end{cases}\\
	\mathfrak{c}_{\rm int}:=&\:\begin{cases}
\frac{\left(\q^2+\frac{1}{4}(1+\beta_{\ell})^2\right)|\beta_{\ell}|^2}{2(2+\beta_{\ell})|\Gamma\left(\frac{1}{2}+i\q-\frac{\beta_{\ell}}{2}\right)|^2},\quad (\beta_{\ell}\in (0,1)),\\
\frac{|\beta_{\ell}|^2}{2}\sum_{n=0}^{\infty}\int_0^{\infty} |\zeta|^{\sign(\q)2n}\left|\partial_z\: ({}_2\mathbf{F}_1)\right|^2\left(\frac{1+\sign(\q)\beta_{\ell}}{2}+i\q,\frac{1-\sign(\q)\beta_{\ell}}{2}+i\q, \frac{1+(2n-1)\sign(\q)\beta_{\ell}}{2}+i\q;-\varrho\right)\,d\varrho,\\
\quad (\beta_{\ell}\in i(0,\infty)),\\
\frac{\frac{1}{4}+\q^2}{8|\Gamma(\frac{1}{2}+i\q)|^2},\quad (\beta_{\ell}=0).
\end{cases}
\end{align}
	\begin{enumerate}[label=\emph{(\roman*)}]
		\item The following bounds hold for $(r^2X)^{k+1}(\Psi^{\infty}_0)_{\ell m}$ with $k\in \N_0$: there exist constants $C=C(\q,\beta_{\ell})>0$ and $\widetilde{C}=\widetilde{C}(\q)>0$, such that:
		\begin{align}
		\label{eq:lowboundrealbeta}
		\left|(r^2X)^{k+1}(\Psi^{\infty}_0)_{\ell m}(\tau,\infty)-\mathfrak{c}_k (1+\tau)^{\frac{k+1}{2}+i\q-\frac{\beta_{\ell}}{2}}\right|\leq &\: C (1+\tau)^{\frac{k+1}{2}-\frac{3\beta_{\ell}}{2}}\quad (\beta_{\ell}\in (0,1)),\\
		\label{eq:lowboundimbeta1}	
		\limsup_{\tau\to \infty}(1+\tau)^{-\frac{k+1}{2}}|(r^2X)^{k+1}(\Psi^{\infty}_0)_{\ell m}|(\tau,\infty)\geq &\:\mathfrak{c}_k  \quad (\beta_{\ell}\in i(0,\infty)),\\
			\label{eq:lowboundimbeta2}	
	|(r^2X)^{k+1}(\Psi^{\infty}_0)_{\ell m}|(\tau,\infty)\leq &\: C (1+\tau)^{\frac{k+1}{2}} \quad (\beta_{\ell}\in i(0,\infty)) ,\\
		\label{eq:lowboundzerobeta}	
		\left|(r^2X)^{k+1}(\Psi^{\infty}_0)_{\ell m}(\tau,\infty)-\mathfrak{c}_k \frac{(1+\tau)^{\frac{k+1}{2}+i\q}}{\log(1+\tau)}\right|\leq &\: \widetilde{C}(\log(1+\tau))^{-2}(1+\tau)^{\frac{k+1}{2}}\quad (\beta_{\ell}=0).
		\end{align}
\item Let $R>0$ be arbitrary. Then there exist constants $C=C(\q,\beta_{\ell},R)>0$ and $\widetilde{C}=\widetilde{C}(\q,R)>0$, such that:
\begin{align}
		\label{eq:enboundrealbeta}
		\left|\int_0^{\infty}r^2|X(\Psi^{\infty}_0)_{\ell m}|^2(\tau,r)\,dr-\mathfrak{c}_{\rm int} (1+\tau)^{-\beta_{\ell}}\right|\leq &\: C (1+\tau)^{-2\beta_{\ell}}\quad (\beta_{\ell}\in (0,1)),\\
		\label{eq:enboundimbeta1}	
		\limsup_{\tau\to \infty}\int_0^{\infty}r^2|X(\Psi^{\infty}_0)_{\ell m}|^2(\tau,r)\,dr\geq &\:\mathfrak{c}_{\rm int}  \quad (\beta_{\ell}\in i(0,\infty)),\\
			\label{eq:enboundimbeta2}	
	\sup_{\tau\in [0,\infty)}\int_0^{\infty}r^2|X(\Psi^{\infty}_0)_{\ell m}|^2(\tau,r)\,dr\leq &\: C  \quad (\beta_{\ell}\in i(0,\infty)) ,\\
		\label{eq:enboundzerobeta}	
		\left|\int_0^{\infty}r^2|X(\Psi^{\infty}_0)_{\ell m}|^2(\tau,r)\,dr-\frac{\mathfrak{c}_{\rm int}}{\log^2(1+\tau)}\right|\leq &\: \widetilde{C}(\log(1+\tau))^{-3}\quad (\beta_{\ell}=0).
		\end{align}
		\item Let $p\in (0,2]$ and consider the energy density $\mathcal{E}_p$ from \S \ref{sec:iedassm} and with $\Omega^2=1-2Mr^{-1}+Q^2$, $Q^2\leq M^2$. Let $N\in \N_0$. Then there exist constants $C_N=C_N(\q,\beta_{\ell},p,N)>0$ and $\widetilde{C}_N=\widetilde{C}_N(\q,R,p,N)>0$, such that for all $k_1+k_2+k_3\leq N$:
	\begin{equation}
	\int_{r_+}^{\infty}\int_{\s^2}\mathcal{E}_{p}[\snabla_{\s^2}^{k_1}(rX)^{k_2}((\tau+1)T)^{k_3}(\Psi^{\infty}_0)]\,d\sigma dr\leq \begin{cases}
 	C_{N}(\tau+1)^{p-2-\re \beta_{\ell}}\quad &(\beta_{\ell}\neq 0),\\
 	\widetilde{C}_N\frac{(\tau+1)^{p-2}}{\log^2(1+\tau)}\quad &(\beta_{\ell}= 0).
 \end{cases}
\end{equation}
	\end{enumerate}
\end{proposition}
\begin{proof}
We will consider the case $k=0$. We will omit the subscript $\ell m$ in the notation $(\Psi^{\infty}_0)_{\ell m}$. The $k\geq 1$ follows from a straightforward generalization. Note that for $\beta_{\ell}\in (0,\infty)$:
\begin{multline*}
(r^2X\Psi^{\infty}_0)(\tau,\varrho)=-\frac{\tau+1}{2}(\mathfrak{X}\Psi^{\infty}_0)(\tau,\infty)=-\frac{\beta_{\ell}}{2}\sum_{n=0}^{N_{\beta_{\ell}}}\zeta^n(1+\tau)^{\frac{1}{2}-\frac{\beta_{\ell}}{2}+i\q-n\beta_{\ell}}{}_2(\partial_{\varrho}{F}_{n \beta})(\varrho;\beta_{\ell})\\
=-\frac{\beta_{\ell}}{2}(1+\tau)^{\frac{1}{2}-\frac{\beta_{\ell}}{2}+i\q}(\partial_{\varrho}{F}_{0})(\varrho;\beta_{\ell})(1+O((1+\tau)^{-\beta_{\ell}})).
\end{multline*}
Note that
\begin{equation*}
	{F}_{0}(\varrho;\beta_{\ell})={}_2\mathbf{F}_1\left(\frac{1}{2}+i\q+\frac{\beta_{\ell}}{2},\frac{1}{2}+i\q-\frac{\beta_{\ell}}{2},\frac{1}{2}+i\q-\frac{\beta_{\ell}}{2};-\varrho\right)=\frac{(1+\varrho)^{-\frac{1}{2}-i\q-\frac{\beta_{\ell}}{2}}}{\Gamma\left(\frac{1}{2}+i\q-\frac{\beta_{\ell}}{2}\right)}.
\end{equation*}
Hence,
\begin{equation*}
	(\partial_{\varrho}F_{0})(\varrho;\beta_{\ell})=-\left(\frac{1}{2}+i\q+\frac{\beta_{\ell}}{2}\right)\frac{(1+\varrho)^{-\frac{3}{2}-i\q-\frac{\beta_{\ell}}{2}}}{\Gamma\left(\frac{1}{2}+i\q-\frac{\beta_{\ell}}{2}\right)}.
\end{equation*}
We conclude that
\begin{equation*}
	(r^2X\Psi^{\infty}_0)(\tau,\infty)=\frac{\beta_{\ell}\left(\frac{1}{2}+i\q+\frac{\beta_{\ell}}{2}\right)}{2\Gamma\left(\frac{1}{2}+i\q-\frac{\beta_{\ell}}{2}\right)}(1+\tau)^{\frac{1}{2}-\frac{\beta_{\ell}}{2}+i\q}(1+O((1+\tau)^{-\beta_{\ell}})).
\end{equation*}
Furthermore, using that $r^{-2}dr=-2(\tau+1)^{-1}d\varrho$ along constant $\tau$ level sets, we can express
\begin{multline*}
	\int_{R}^{\infty}r^2|X\Psi^{\infty}_0|^2(\tau,r)\,dr=\frac{(\tau+1)}{2}\int_{0}^{\frac{\tau+1}{2R}}|\mathfrak{X}\Psi^0_{\infty}|^2(\tau,\varrho)\,d\varrho\\
	=\frac{|\beta_{\ell}|^2(\tau+1)^{- \beta_{\ell}}}{2}\int_{0}^{\frac{\tau+1}{2R}}|(\partial_{\varrho}F_{0})(\varrho;\beta_{\ell})|^2(1+O((\tau+1)^{-\beta_{\ell}}))\,d\varrho\\
	=\frac{\left(\q^2+\frac{1}{4}(1+\beta_{\ell})^2\right)|\beta_{\ell}|^2(\tau+1)^{- \beta_{\ell}}}{2|\Gamma\left(\frac{1}{2}+i\q-\frac{\beta_{\ell}}{2}\right)|^2}\int_{0}^{\frac{\tau+1}{2R}}(1+\varrho)^{-3-\beta_{\ell}}(1+O((\tau+1)^{-\beta_{\ell}}))\,d\\
	=\frac{\left(\q^2+\frac{1}{4}(1+\beta_{\ell})^2\right)|\beta_{\ell}|^2(\tau+1)^{- \beta_{\ell}}}{2(2+\beta_{\ell})|\Gamma\left(\frac{1}{2}+i\q-\frac{\beta_{\ell}}{2}\right)|^2}(1+O((\tau+1)^{-\beta_{\ell}})).
\end{multline*}
By using that $r^{-1}= 2\varrho (\tau+1)^{-1}$ it is also straightforward to obtain for $p\in [0,2)$ the existence of a constant $C>0$, such that:
\begin{equation*}
	\int_{R}^{\infty}r^p|X(\Psi^{\infty}_0)_{\ell m}|^2+r^{-2}(|(\Psi^{\infty}_0)_{\ell m}|^2+|T(\Psi^{\infty}_0)_{\ell m}|^2)(\tau,r)\,dr\leq C(\tau+1)^{p-2-\beta_{\ell}}.
\end{equation*}
It is moreover straightforward to commute with $(\tau+1)T$ and $rX$ and obtain for $\ell(\ell+1)<q^2$ and $k_1+k_2+k_3=N\in \N_0$ the existence of a constant $C_N>0$, such that:
\begin{equation*}
	\int_{r_+}^{\infty}\int_{\s^2}\mathcal{E}_p[\snabla_{\s^2}^{k_1}(rX)^{k_2}((\tau+1)T)^{k_3}(\Psi^{\infty}_0)]\,d\sigma dr\leq C_N(\tau+1)^{p-2-\beta_{\ell}}.
\end{equation*}

Now suppose that $\beta_{\ell}\in i(0,\infty)$. Without loss of generality, we will assume that $q>0$. We use the general expressions above to obtain:
\begin{equation*}
	|r^2X\Psi^{\infty}_0|(\tau,\varrho)\leq \frac{|\beta_{\ell}|}{2}(1+\tau)^{\frac{1}{2}}\sum_{n=0}^{\infty}|\zeta|^{n}\left|\partial_{\varrho}F_{n\im \beta_{\ell}}(\varrho;\beta_{\ell})\right|,
\end{equation*}
where the right-hand side converges. This follows from the proof of Proposition \ref{prop:imaginarybetatail} where $\partial_{\varrho}F_{n\im \beta_{\ell}}$ is replaced by $F_{n\im \beta_{\ell}}$, but the additional $\varrho$-derivative only introduces additional polynomial growth in $n$, which does not affect the convergence of the relevant series.

Furthermore, by applying the Minkowski inequality, we can also estimate:
\begin{equation*}
	\frac{(\tau+1)}{2}\int_0^{\frac{\tau+1}{2R}}|\mathfrak{X}\Psi^0_{\infty}|^2(\tau,\varrho)\,d\varrho\leq  \frac{|\beta_{\ell}|^2}{2}\left(\sum_{n=0}|\zeta|^n ||\partial_{\varrho}F_{n\im \beta_{\ell}}||_{L^2_{\varrho}[0, \frac{\tau+1}{2R}]}\right)^2.
\end{equation*}
By Lemma \ref{lm:unifomestnhypgeom}, there exists a constant $C_{\q, \varrho_0}>0$, such that we can estimate:
\begin{equation*}
	||\partial_{\varrho}F_{n\im \beta_{\ell}}||_{L^2_{\varrho}[0, \varrho_0]}\leq C_{\q, \varrho_0}e^{\frac{\pi}{2}n \im \beta_{\ell}}.
\end{equation*}
Furthermore, we can apply Lemma \ref{lm:unifomestnhypgeom} to obtain moreover for $\varrho\geq \varrho_0$:
\begin{equation*}
	||\partial_{\varrho}F_{n\im \beta_{\ell}}||_{L^2_{\varrho}[\varrho,\infty)}\leq C_{\q, \varrho_0,\epsilon}\varrho^{-3}e^{(\frac{\pi}{2}+\epsilon)n \im \beta_{\ell}}.
\end{equation*}
 for a constant $C_{\q ,\varrho_0,\epsilon}>0$. By the properties assumed on $\zeta$ in Proposition \ref{prop:imaginarybetatail}, we therefore conclude that there exists a constant $C_{\q,\beta_{\ell}}$ such that:
 \begin{equation*}
	\frac{(\tau+1)}{2}\int_0^{\frac{\tau+1}{2R}}|\mathfrak{X}\Psi^0_{\infty}|^2(\tau,\varrho)\,d\varrho\leq  C_{\q,\beta_{\ell}}.
	\end{equation*}
 Furthermore, it is straightforward to extend the above estimate to obtain for $\ell(\ell+1)<q^2$ and $k_1+k_2+k_3=N\in \N_0$ the existence of a constant $C_N>0$, such that for $p\in [0,2]$:
\begin{equation*}
	\int_{r_+}^{\infty}\int_{\s^2}\mathcal{E}_p[\snabla_{\s^2}^{k_1}(rX)^{k_2}((\tau+1)T)^{k_3}(\Psi^{\infty}_0)]\,d\sigma dr\leq C_N(\tau+1)^{p-2}.
\end{equation*}
 
 In order to obtain an appropriate lower bound for $|r^2X\Psi^{\infty}_0|^2$, we instead write:
\begin{multline*}
	|r^2X\Psi^{\infty}_0|^2(\tau,\varrho)=\frac{|\beta_{\ell}|^2}{4}(1+\tau)\sum_{k\in \Z}\sum_{n,m\in \N_0}(\tau+1)^{-k \beta_{\ell}} \zeta^n \overline{\zeta}^m(\partial_{\varrho}{F}_{n \im \beta_{\ell}})(\varrho;\beta_{\ell})\cdot (\partial_{\varrho}\overline{{F}_{m \im \beta_{\ell}}})(\varrho;\beta_{\ell})\\
	=\frac{|\beta_{\ell}|^2}{4}(1+\tau)\sum_{n,m\in \N_0}e^{-i(n-m)\sigma} \zeta^n \overline{\zeta}^m(\partial_{\varrho}{F}_{n \im \beta_{\ell}})(\varrho;\beta_{\ell})\cdot (\partial_{\varrho}\overline{{F}_{m \im \beta_{\ell}}})(\varrho;\beta_{\ell}),
\end{multline*}
with $\sigma=\im \beta_{\ell}\log(1+\tau)$. Since the above series is uniformly absolutely-convergent, it follows that:
\begin{multline}
\label{eq:expvaluePsi0infest}
	\left|S^{-1}\int_0^S(1+\tau(\sigma))^{-1}|r^2X\Psi^0_{\infty}|^2(\tau(\sigma),\varrho))\,d\sigma-\frac{|\beta_{\ell}|^2}{4}\sum_{n=0}^{\infty}|\zeta|^{2n}|\partial_{\varrho}{F}_{n \im \beta}|^2(\varrho;\beta_{\ell})\right|\\
	\leq S^{-1}\frac{|\beta_{\ell}|^2}{4}\sum_{\substack{n,m\in \N_0\\n\neq m}}\frac{2}{|n-m|}|\zeta|^{n+m} |\partial_{\varrho}{F}_{n \im \beta_{\ell}}|(\varrho;\beta_{\ell})\cdot |\partial_{\varrho}{F}_{m \im \beta_{\ell}})|\varrho;\beta_{\ell})\leq C S^{-1}.
\end{multline}

Using that for any continuous $f:(0,\infty)\to \R$:
\begin{equation}
\label{eq:limsupid}
	\limsup_{x\to \infty}f(x)\geq \lim_{L\to \infty}\frac{1}{L}\int_0^Lf(x)\,dx,
\end{equation}
we therefore conclude that
\begin{multline*}
\limsup_{\tau\to \infty}(1+\tau)^{-1}|r^2X\Psi^0_{\infty}|^2(\tau,0)	\geq \frac{|\beta_{\ell}|^2}{4}\sum_{n=0}^{\infty}|\zeta|^{2n}|\partial_{\varrho}{F}_{n \im \beta}|^2(0;\beta_{\ell})\\
=\frac{|\beta_{\ell}|^2}{4}\sum_{n=0}^{\infty}|\zeta|^{2n} \left|\partial_z\: ({}_2\mathbf{F}_1)\right|^2\left(\frac{1+\beta_{\ell}}{2}+i\q,\frac{1-\beta_{\ell}}{2}+i\q, \frac{1+(2n-1)\beta_{\ell}}{2}+i\q;0\right).
\end{multline*}
Furthermore, we can express:
\begin{multline*}
	\int_{R}^{\infty}r^2|X\Psi^{\infty}_0|^2(\tau,r)\,dr=\frac{(\tau+1)}{2}\int_{0}^{\frac{\tau+1}{2R}}|\mathfrak{X}\Psi^0_{\infty}|^2(\tau,\varrho)\,d\varrho\\
	=\frac{|\beta_{\ell}|^2}{2}\sum_{k\in \Z}e^{-i k\sigma  } \int_0^{\frac{\tau+1}{2R}}\sum_{n,m\in \N_0,\: n-m=k} \zeta^n \overline{\zeta}^m(\partial_{\varrho}{F}_{n \im \beta})(\varrho;\beta_{\ell})\cdot (\partial_{\varrho}\overline{{F}_{m \im \beta}})(\varrho;\beta_{\ell})\,d\varrho.
\end{multline*}
As above, we have that:
\begin{equation*}
	\limsup_{\tau \to \infty}\int_{R}^{\infty}r^2|X\Psi^{\infty}_0|^2(\tau,r)\,dr\geq \frac{|\beta_{\ell}|^2}{2}\sum_{n=0}^{\infty}\int_0^{\frac{\tau+1}{2R}} |\zeta|^{2n}|\partial_{\varrho}\overline{{F}_{n \im \beta}}|^2(\varrho;\beta_{\ell})\,d\varrho.
\end{equation*}
By Lemma \ref{lm:unifomestnhypgeom}, it follows moreover that $\limsup_{\tau \to \infty} \int_{\frac{\tau+1}{2R}}^{\infty} |\zeta|^{2n}|\partial_{\varrho}\overline{{F}_{m \beta}}|^2(\varrho;\beta_{\ell})\,d\varrho=0$, so we conclude:
\begin{multline*}
	\limsup_{\tau \to \infty}\int_{R}^{\infty}r^2|X\Psi^{\infty}_0|^2(\tau,r)\,dr\\
	\geq \frac{|\beta_{\ell}|^2}{2}\sum_{n=0}^{\infty}\int_0^{\infty} |\zeta|^{2n}\left|\partial_z\: ({}_2\mathbf{F}_1)\right|^2\left(\frac{1+\beta_{\ell}}{2}+i\q,\frac{1-\beta_{\ell}}{2}+i\q, \frac{1+(2n-1)\beta_{\ell}}{2}+i\q;-\varrho\right)\,d\varrho.
\end{multline*}
The right-hand side is clearly strictly greater than 0.

Finally, we consider $\beta_{\ell}=0$. Then, after integrating by parts in $x$, we obtain:
\begin{multline*}
	(r^2X\Psi^{\infty}_0)(\tau,r(\tau,\varrho))=-\frac{\tau+1}{2}\mathfrak{X}\Psi^{\infty}_0(\tau,\varrho)=\frac{1}{2}(\tau+1)^{\frac{1}{2}+i\q}\int_0^{\infty}(\tau+1)^{-ix}e^{-\eta x}(\partial_{\varrho}F_x)(\varrho;0)\,dx\\
	=-\frac{i}{2} (\log(1+\tau))^{-1}(\tau+1)^{\frac{1}{2}+i\q}(\partial_{\varrho}\: F_0)(\varrho;0)(1+O((\log(1+\tau))^{-1})\\
	=\frac{i\left(\frac{1}{2}+iq\right)}{2\Gamma(\frac{1}{2}+i\q)}\frac{(\tau+1)^{\frac{1}{2}+i\q}}{\log(1+\tau)}(1+\varrho)^{-\frac{3}{2}-i\q}(1+O((\log(1+\tau))^{-1}).
\end{multline*}
In particular,
\begin{equation*}
		(r^2X\Psi^{\infty}_0)(\tau,\infty)=\frac{i\left(\frac{1}{2}+iq\right)}{2\Gamma(\frac{1}{2}+i\q)}\frac{(\tau+1)^{\frac{1}{2}+i\q}}{\log(1+\tau)}(1+O((\log(1+\tau))^{-1})
\end{equation*}
and
\begin{multline*}
	\int_R^{\infty}r^2|X\Psi^{\infty}_0|^2(\tau,r)\,dr=\frac{\frac{1}{4}+\q^2}{4|\Gamma(\frac{1}{2}+i\q)|^2}(\log(1+\tau))^{-2}\int_0^{\frac{\tau+1}{2R}}(1+\varrho)^{-3}(1+O((\log(1+\tau))^{-1})\,d\varrho\\
	=\frac{\frac{1}{4}+\q^2}{8|\Gamma(\frac{1}{2}+i\q)|^2}(\log(1+\tau))^{-2}(1+O((\log(1+\tau))^{-1}).
\end{multline*}
 Furthermore, after applying Lemma \ref{lm:unifomestnhypgeom}, it is straightforward to extend the above estimate to obtain for $\ell(\ell+1)<q^2$ and $k_1+k_2+k_3=N\in \N_0$ the existence of a constant $C_N>0$, such that for $p\in [0,2]$:
\begin{equation*}
	\int_{r_+}^{\infty}\int_{\s^2}\mathcal{E}_p[\snabla_{\s^2}^{k_1}(rX)^{k_2}((\tau+1)T)^{k_3}(\Psi^{\infty}_0)]\,d\sigma dr\leq C_N(\log(1+\tau))^{-2}(\tau+1)^{p-2}. \qedhere
\end{equation*}
\end{proof}

\begin{corollary}
Let $\alpha_+,\alpha_-\in \C$ be defined as in Proposition \ref{prop:statsol}.
Then \eqref{eq:condzeta} and \eqref{eq:condeta} are satisfied.
\end{corollary}
\begin{proof}
This follows immediately from Corollary \ref{cor:kappanotsmallstatsoln}.
\end{proof}

\bibliographystyle{alpha}


	\end{document}